\pdfoutput=1

\documentclass[11pt]{article}

\usepackage{titlesec}
\usepackage{geometry}
\usepackage{multirow}
\usepackage{fancyhdr}
\usepackage[numbers]{natbib}
\usepackage{graphicx}
\usepackage[english]{babel}
\usepackage{amsmath}
\usepackage{amssymb}
\usepackage{amsfonts}
\usepackage{amsthm}
\usepackage{thmtools} 
\usepackage{thm-restate}
\usepackage{accents}
\usepackage[noend]{algpseudocode}
\usepackage{caption}
\usepackage{subcaption}

\usepackage[dvipsnames]{xcolor}
\definecolor{cb-black}      {RGB}{  0,   0,   0}
\definecolor{cb-blue-green} {RGB}{  0,  073,  073}
\definecolor{cb-green-sea}  {RGB}{  0, 146, 146}
\definecolor{cb-rose}       {RGB}{255, 109, 182}
\definecolor{cb-salmon-pink}{RGB}{255, 182, 119}
\definecolor{cb-purple}     {RGB}{ 73,   0, 146}
\definecolor{cb-blue}       {RGB}{ 0, 109, 219}
\definecolor{cb-lilac}      {RGB}{182, 109, 255}
\definecolor{cb-blue-sky}   {RGB}{109, 182, 255}
\definecolor{cb-blue-light} {RGB}{182, 219, 255}
\definecolor{cb-burgundy}   {RGB}{146,   0,   0}
\definecolor{cb-brown}      {RGB}{146,  73,   0}
\definecolor{cb-clay}       {RGB}{219, 209,   0}
\definecolor{cb-green-lime} {RGB}{ 36, 255,  36}
\definecolor{cb-yellow}     {RGB}{255, 255, 109}

\usepackage{tikz}
\usepackage[compat=newest]{yquant}
\usepackage[colorlinks,linkcolor=cb-burgundy,urlcolor=cb-burgundy,citecolor=cb-burgundy,destlabel]{hyperref}
\usepackage[nameinlink,capitalise]{cleveref}
\usepackage{nicefrac}
\usepackage{stmaryrd}
\usepackage[normalem]{ulem}
\usepackage{tikz}
\usepackage{lipsum}
\usepackage{enumitem}
\usepackage{soul}
\usepackage[
raggedleft
]{sidecap}   
\sidecaptionvpos{figure}{t} 

\usetikzlibrary{external}
\usetikzlibrary{quotes}
\usetikzlibrary{calc}
\usetikzlibrary{fit}

\makeatletter
\let\OldStatex\Statex
\renewcommand{\Statex}[1][3]{%
  \setlength\@tempdima{\algorithmicindent}%
  \OldStatex\hskip\dimexpr#1\@tempdima\relax}
\makeatother

\algnewcommand{\Inp}{\textbf{Input:}\space}
\algnewcommand{\Out}{\textbf{Output:}\space}
\newcommand{\Input}{\Statex[-1] \Inp }
\newcommand{\Output}{\Statex[-1] \Out }
\newcommand{\Blank}{\Statex[-1]}

\makeatletter
\providecommand\theHALG@line{\thealgorithm.\arabic{ALG@line}}
\makeatother

\newcommand{\loadfig}[1]{\includegraphics{figures/tikz-#1}} 

\newtheorem{theorem}{Theorem}
\counterwithin{figure}{section}
\counterwithin{table}{section}
\counterwithin{theorem}{section}

\newtheorem{proposition}[theorem]{Proposition}
\newtheorem{lemma}[theorem]{Lemma}
\newtheorem{corollary}[theorem]{Corollary}
\newtheorem{definition}[theorem]{Definition}
\newtheorem{problem}[theorem]{Problem}

\theoremstyle{remark}

\declaretheoremstyle[
  notefont=\mdseries, notebraces={(}{)},
  bodyfont=\normalfont,
  postheadspace=0em,
  headpunct=
]{algostyle}
\theoremstyle{algostyle}
\newtheorem{procedure}[theorem]{Procedure}
\newtheorem{algorithm}[theorem]{Algorithm}

\newcommand{\ket}[1]{|#1\rangle}
\newcommand{\bra}[1]{\langle #1|}

\newcommand*{\f}{\mathbb{F}}

\newcommand*{\ip}[1]{ \langle #1 \rangle } 
\newcommand*{\Tr}{\mathrm{Tr}} 

\newcommand{\comm}[1]{ \llbracket #1 \rrbracket } 
\renewcommand*{\P}{\mathcal{P}} 

\newcommand*{\nmax}{n_{\scriptscriptstyle  \updownarrow}} 

\renewcommand*{\ni}{{n_{\inn}}} 
\newcommand*{\no}{{n_{\out}}} 

\newcommand*{\ki}{\ni} 
\newcommand*{\ko}{\no} 

\newcommand*{\nr}{n_r} 
\newcommand*{\nm}{n_m} 
\newcommand*{\nM}{n_M} 

\newcommand*{\Nu}{N_\mathrm{u}} 
\newcommand*{\Ncnd}{N_\mathrm{cnd}} 

\newcommand*{\Ci}{L} 
\newcommand*{\Co}{R} 

\newcommand*{\tki}{{\tilde n_{\inn}}} 
\newcommand*{\tko}{{\tilde n_{\out}}} 
\newcommand*{\tCi}{\tilde B} 
\newcommand*{\tCo}{\tilde D} 

\newcommand*{\bCi}{\bar B} 
\newcommand*{\bCo}{\bar D} 

\newcommand*{\Z}{\mathcal{Z}} 
\newcommand*{\X}{\mathcal{X}} 
\newcommand*{\Zs}{\mathcal{Z}^S} 
\newcommand*{\Xs}{\mathcal{X}^S} 
\newcommand*{\Zm}{\mathcal{Z}^M} 
\newcommand*{\Xm}{\mathcal{X}^M} 
\newcommand*{\Zcap}{\mathcal{Z}^\cap} 
\newcommand*{\Xcap}{\mathcal{X}^\cap} 
\newcommand*{\Zdelta}{\mathcal{Z}^\Delta} 
\newcommand*{\Xdelta}{\mathcal{X}^\Delta} 
\newcommand*{\Lin}{\mathcal{L}_{\comm{n,k_\inn,C^\ast_\inn}}}
\newcommand*{\Lout}{\mathcal{L}_{\comm{n,k_\out,C_\out}}}

\newcommand*{\inn}{\mathrm{in}}
\newcommand*{\out}{\mathrm{out}}
\newcommand*{\runtime}{run time}

\newcommand*{\runtimes}{run times}
\newcommand*{\bitstring}{bit string}
\newcommand*{\bitstrings}{bit strings}
\newcommand*{\bitruntime}{bit-string run time}

\newcommand*{\ca}{\color{cb-blue}}
\newcommand*{\cb}{\color{cb-brown}}
\newcommand*{\cc}{\color{cb-green-sea}}

\newcommand{\algemph}[1]{\colorbox{cb-blue-light!35!white}{#1}}
\newcommand{\emphspecific}[1]{\colorbox{cb-green-lime!10!white}{#1}}

\makeatletter
\def\smalloverbrace#1{\mathop{\vbox{\m@th\ialign{##\crcr\noalign{\kern3\p@}%
  \tiny\downbracefill\crcr\noalign{\kern3\p@\nointerlineskip}%
  $\hfil\displaystyle{#1}\hfil$\crcr}}}\limits}
\makeatother

\DeclareEmphSequence{\bfseries,\mdseries}

\title{
Stabilizer circuit verification
}
\author{Vadym Kliuchnikov, Michael Beverland, Adam Paetznick}

\begin{document}

\maketitle

\begin{abstract}
The ubiquity of stabilizer circuits in the design and operation of quantum computers makes techniques to verify their correctness essential.
The simulation of stabilizer circuits, which aims to replicate their behavior using a classical computer, is known to be efficient and provides a means of testing correctness.
However, simulation is limited in its ability to examine the exponentially large space of possible measurement outcomes.
We propose a comprehensive set of efficient classical algorithms to fully characterize and exhaustively verify stabilizer circuits with Pauli unitaries conditioned on parities of measurements.
We introduce, as a practical characterization, a general form for such circuits and provide an algorithm
to find a general form of any stabilizer circuit.
We then provide an algorithm for checking the equivalence of stabilizer circuits.
When circuits are not equivalent our algorithm suggests modifications for reconciliation.
Next, we provide an algorithm that characterizes the logical action of a (physical) stabilizer circuit on an encoded input.
All of our algorithms provide relations of measurement outcomes among corresponding circuit representations.
Finally, we provide an analytic description of the logical action induced by measuring a stabilizer group, 
with application in correctness proofs of code-deformation protocols including lattice surgery and code switching.
\end{abstract}

\newpage

\tableofcontents

\newpage
\section{Motivation and main results}
\label{sec:intro}

It is widely appreciated that future generations of large-scale quantum computers will require Quantum Error Correction (QEC) techniques to correct hardware faults. 
However it is less well appreciated that techniques to identify and correct \emph{software} faults will also be necessary. 
At scale, quantum computers will be extremely complex systems. 
They will require testing and debugging of high-level quantum algorithms, intermediate-level instruction sets,
and low-level physical circuits. 
Moreover, at each level, exploration of new approaches and optimization of existing approaches is ongoing.
Classically tractable techniques to test quantum circuits are needed now to enable this design and research process.
While several methods have been proposed to verify universal quantum circuits~\cite{Yamashita2010,burgholzer2020,Viamontes2007,niemann2014,wang2008,smith2019,amy2018,hong2022,Burgholzer2022,sundaram2022rich,Rand2018Qwire}, this problem is QMA complete~\cite{janzing2005} and so these approaches are expected to break down for large circuits.
We seek efficient verification methods for important but non-universal quantum circuits.


One extensively studied approach for testing and debugging quantum circuits involves quantum simulation, i.e., replicating the behavior of a quantum circuit by using a classical computer.
Simulation of general quantum circuits is widely believed to be intractable.
Fortunately, efficient simulation is possible for an important subclass known as stabilizer circuits~\cite{Gottesman1998, AaronsonGottesman2004}.
Testing via simulation is simple: simulate a stabilizer circuit and check if the observed outcomes match expectations.
This method of testing has a substantial drawback; measurement outcomes of a quantum circuit are non-deterministic.
Moreover, gates in a circuit may depend on those outcomes so that the gate sequence itself is non-deterministic, requiring an exponential number of simulations to prove that results hold for all pathways through a circuit.

For limited classes of stabilizer circuits, such as sequences of Clifford unitaries and deterministic measurements, random outcomes do not appear and a single simulation captures the full behavior.
Increasingly, however, stabilizer circuits proposed for large-scale quantum computers
rely on interpreting random outcomes.
Leading approaches for operations on encoded qubits prescribe an intermediate-level instruction set entirely based on measurements with random outcomes~\cite{Horsman2012, Beverland2022Disjoint, haner2022spacetime, fowler2019low, Litinski2019gameofsurfacecodes}.
New high-performing codes are defined by low-level sequences of measurements that also have random outcomes, even in the absence of errors~\cite{Hastings2021dynamically, Haah2022boundarieshoneycomb,davydova2022floquet,Aasen2022Adiabatic,Gidney2021faulttolerant,Gidney2022benchmarkingplanar, Paetznick2023Floquet}.
Some types of physical qubits such as those based on photons~\cite{bombin2021interleaving} or Majorana wires~\cite{Bonderson2008MeasurementOnly} rely on measurements rather than unitaries as entangling operations.

\begin{figure}[h]
\begin{minipage}[c][6cm][t]{0.36\textwidth}
     \includegraphics[scale=0.375]{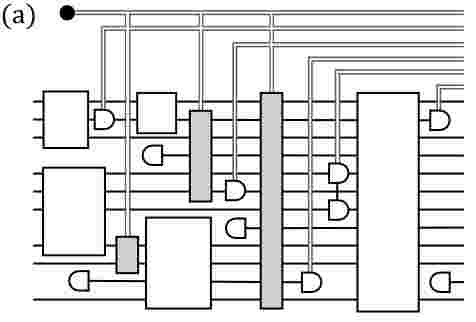}
     \includegraphics[scale=0.375]{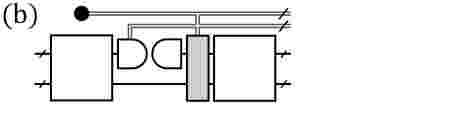}
\end{minipage}
\hfill 
\begin{minipage}[c][6cm][t]{0.59\textwidth}
 \caption[Stabilizer and general form circuits (high-level)]{
         \label{fig:stabilizer-circ-and-gen-form-circ}
         {(a)} Stabilizer circuits in this work are built from: allocations and deallocations of zero-state qubits (starting and ending wires next to \rotatebox[origin=c]{180}{D}- and D-like symbols), allocations of random bit outcomes (black circle), Clifford unitaries (white rectangles), non-destructive Pauli measurements (D-like symbol with continued wire), Pauli unitaries conditioned on parities of earlier measurement and random bit outcomes (gray rectangles).
         \\\hspace{\textwidth}
         {(b)} 
         General form circuit.
         For any stabilizer circuit, there is a general form circuit (as depicted) that matches the action on any input state for every possible path through the circuit.
         This forms the basis of our algorithms.
         }
\end{minipage}
\end{figure}

In this work we introduce classical tests of stabilizer circuits for all input states and outcome \bitstrings{}.
We choose the definition of stabilizer circuits in \cref{fig:stabilizer-circ-and-gen-form-circ}(a) for complexity-theoretic reasons as discussed later in this section.
In \cref{sec:stabilizer-circuit-general-form} we prove one of our central results: a \emph{general form circuit} illustrated in \cref{fig:stabilizer-circ-and-gen-form-circ}(b) captures the action of any stabilizer circuit. 
We leverage this to build the algorithms in \cref{fig:paper-plan} to efficiently solve the following problems.
\newpage
\begin{enumerate}[topsep=8pt,itemsep=4pt,partopsep=4pt, parsep=4pt]
    \item[\textbf{1.}] \textbf{What is the action of a stabilizer circuit?} 
\end{enumerate}
We would like to have a compact description for any stabilizer circuit that characterizes its action on the input qubits given any circuit outcome. 
This is directly achieved by our general form algorithm (\cref{fig:paper-plan}) which provides an equivalent general form circuit, which for any outcome of the original circuit has the same action on every input state for a corresponding outcome of the general form circuit.
We find the general form much easier interpret and compose than the typical characterization of quantum channels in terms of Choi states.

\begin{enumerate}[topsep=8pt,itemsep=4pt,partopsep=4pt, parsep=4pt]
    \item[\textbf{2.}] \textbf{Do two stabilizer circuits have the same action?} 
\end{enumerate}
To verify that a circuit achieves a desired goal requires more still - we need to compare against a reference circuit.
Given two stabilizer circuits, and a \textit{specific} outcome bitstring for each, 
verifying that circuits have the same action be accomplished using existing methods.
One option is to perform stabilizer simulation of Choi circuits and then compare the output stabilizer states~\cite{GMC2014}.
A second option is to write ZX diagrams for the Choi circuits and reduce them to a simplified pair of rGS–LC diagrams~\cite{Backens2014}.

However, the action of both the circuit and the reference will depend on their respective outcomes.
We would like to know if two circuits are have the same action for \emph{all} outcomes, and if they do, which outcomes of the two circuits correspond to the same action.
By comparing general forms, our verification algorithm (\cref{fig:paper-plan}) achieves this while existing methods do not.

\begin{enumerate}[topsep=8pt,itemsep=4pt,partopsep=4pt, parsep=4pt]
    \item[\textbf{3.}] \textbf{What is the action of a logical operation circuit on the encoded qubits?} 
\end{enumerate}
Often a circuit is constructed in order to implement some logical operation on encoded information which can be used to protect against the effects of noise.  
To characterize such a circuit we need to identify its \emph{logical action}.
We propose doing so with our logical action algorithm (\cref{fig:paper-plan}) to find a general form circuit equivalent to the action on the encoded information.

\begin{figure}[ht]
    \centering
    \includegraphics[width=0.9\linewidth]{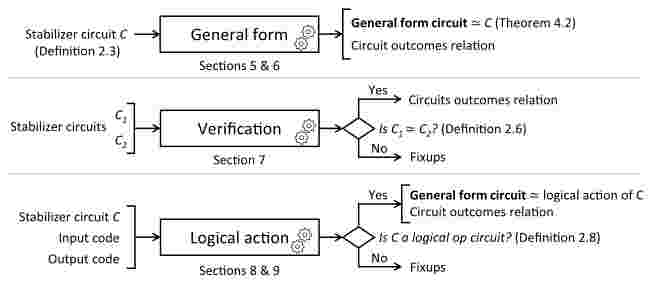}
    \caption[Stabilizer circuits and general form]{\label{fig:paper-plan}
    Algorithms to test stabilizer circuits for all input states and outcome \bitstrings{}.
    }
\end{figure}

Let us briefly sketch how the algorithms in \cref{fig:paper-plan}, which solve these three problems, work.
The general form algorithm (\cref{sec:general-form-algo}) combines two key ingredients: an enhanced stabilizer simulation of Choi states, and a generalization of a bipartite norm form for stabilizer states~\cite{Audenaert2005}.
General form circuits are then used as the basis of the verification and logical action algorithms.
Our verification algorithm (\cref{sec:equality-of-general-forms}) maps both input circuits to general forms, and checks if those general forms are equivalent.
The general form is not unique; the bulk of the algorithm involves checking the equivalence of the general form circuits and computing the relation between circuit outcomes.
The logical action algorithm (\cref{sec:logical-general-form}) first forms a composite by sandwiching the input circuit with encoding and unencoding circuits of the input and output codes, respectively.
It then obtains a general form for the composite circuit, which characterizes the logical action.
Validating that the circuit is indeed a logical operation uses the relation between outcomes of the composite circuit and its general form.

As noted above, a key ingredient of our work is an enhanced version of stabilizer simulation. The stabilizer simulation problem solved in the existing literature is, roughly, to find the probability of a specific subset of outcome bitstrings of a circuit and to provide an efficient description of the output state~\cite{AaronsonGottesman2004,Gidney2021,GraphStateSim, guan2019stabilizer,Beaudrap2022faststabiliser}.
We formalize this problem in \cref{sec:stab-simulation} as \emph{outcome-specific} stabilizer simulation.
By considering the Choi states of a stabilizer circuit, outcome-specific simulation can be used to efficiently characterize the circuit action over all possible input states.

Dealing with all possible measurement outcomes requires more care.
Pauli operations controlled on random bits along with the probabilistic outcomes of quantum measurements cause each circuit simulation to take a different pathway.
Accordingly, we introduce the \emph{outcome-complete} stabilizer simulation problem, and show that it can be efficiently reduced to the outcome-specific stabilizer simulation. 
In \cref{sec:stab-simulation-all}, we show that outcome-complete stabilizer simulation can be performed by using any outcome-specific stabilizer simulation as a subroutine which is called many times.
We also find a new algorithm to solve the outcome-complete case which is more efficient than using outome-specific subroutines.

Our outcome-complete simulation algorithm and the three algorithms in \cref{fig:paper-plan} are computationally efficient in that their \runtime{}s scale polynomially with the parameters of the input circuits.
Throughout this paper we present detailed fine-grained analyses of the \runtime{} of each algorithm. 
Here it is convenient to consider families of circuits with growing maximum width $n$ and depth $D$, and where all stabilizer operations appear with constant non-zero density. 
Then, the \runtime{} of the aforementioned algorithms for outcome-complete stabilizer simulation, general form, verification and logical action, are all $O(n^3 D^2)$.
For context, previously known outcome-specific simulations have complexity $O(n^3 D)$ and using them for outcome-complete simulation results in an $O(n^4 D^2)$ \runtime{}.

In some important special cases, the general form can be used to \emph{analytically} characterize the logical action. 
In \cref{sec:code-deformation-action} we consider the logical action of circuits which measure the generators of a stabilizer code when the initial state is encoded in a different stabilizer code~(code deformation). We find a general form and outcome relation for the logical action, expressed in terms of a common symplectic basis of the two stabilizer groups. 
This basis may be of independent interest.
We illustrate our analytic results with an example of logical $XX$ measurement on two copies of the repetition code via lattice surgery.

Lastly, our techniques require a restriction on the classical conditions of stabilizer circuits: we only allow Pauli unitaries conditioned on previous outcome parities (\cref{fig:stabilizer-circ-and-gen-form-circ}(a)).
This restriction is motivated by the computational complexity of simulation; efficient simulation is necessary for efficient characterization.
As a consequence of the Gottesman-Knill theorem~\cite{Gottesman1998},
a broader class of stabilizer circuits than we consider, that also includes Clifford gates conditioned on previous outcomes, can be sampled efficiently.
Characterization requires a notion of simulation stronger than sampling, one that can determine the probability distribution of circuit outcomes.
Somewhat surprisingly, this stronger simulation of stabilizer circuits is \#P-hard~\cite{ExtendedCliffordCircuits}.
Pauli gates conditioned on the AND of previous outcomes can encode arbitrary Boolean function evaluation.
Limiting to Pauli gates conditioned on parities admits a reformulated Gottesman-Knill theorem that can be leveraged into an efficient algorithm for characterization.

Despite the restriction to partity conditions, circuits from this set describe important components of high-level quantum algorithms, such as QROM lookups~\cite{Babbush2018Electronic}.
They also include intermediate-level implementations that, for example, transform algorithms according to architectural constraints~\cite{Beverland2022Disjoint}. 
Of course, such circuits also describe the vast majority of logical operations in existing QEC schemes.
We propose that future design of quantum circuits should adhere to this restriction, where possible, in order to permit characterization and verification.

\newpage
\section{Mathematical preliminaries}\label{sec:preliminaries}

Here we review some relevant mathematical concepts and specify notation.

\subsection{Pauli and Clifford unitaries, stabilizer codes and Bell states}
\label{sec:Paulis-Cliffords-Stabilizers-Bell}

The \emph{Pauli unitaries} on $n$ qubits are $i^k \{I,X,Y,Z\}^{\otimes n}$ for $k=0,1,2,3$, where $I,X,Y,Z$ are the single-qubit Pauli matrices.
We call any group of Pauli operators \emph{a Pauli group}, while the set of all Pauli unitaries on $n$ qubits forms a group known as \emph{the Pauli group} $\mathcal{P}_n$.
We also find it useful to label the subset of Hermitian Pauli unitaries $\pm \{I,X,Y,Z\}^{\otimes n}$, which we call the \emph{Pauli observables}.


A \emph{stabilizer group} is a group of commuting Pauli observables not containing $-I$.
We say a state is stabilized by a set of Pauli observables if it is a +1 eigenstate of all elements of the set.
A \emph{stabilizer code} is a sub-space of a Hilbert space consisting of all vectors which are stabilized by a stabilizer group $S$.
A \emph{stabilizer state} is a state on $n$ qubits which is stabilized by a set of $n$ independent Pauli observables.

For any two Pauli unitaries, we write that $\comm{P,Q} = 0 \in \f_2$ if $P,Q$ commute and $\comm{P,Q} = 1 \in \f_2$ otherwise:
$$
    PQ = (-1)^{\comm{P,Q}}QP = (-1)^{\comm{Q,P}}QP.
$$

We will find use for the following identity,
\begin{align}
\label{eq:exponentiated-pauli-image}
e^{i \frac{\pi}{4} P} Q e^{-i \frac{\pi}{4} P} = (iP)^{\comm{P,Q}}Q.
\end{align}

For any subgroup $G$ of the $n$-qubit Pauli group $\P_n$, we define $G^\perp$ as 
$$
 G^\perp = \{ P \in \P_n :  \comm{ P, g } = 0, \forall g \in G \}.
$$
Note that $G^\perp$ is a group because $\comm{PQ,R} = \comm{R,PQ} = \comm{P,R} + \comm{Q,R}$.

An $n$-qubit Clifford unitary is any $n$-qubit unitary $C$ such that $\forall P \in \P_n : CPC^\dagger \in \P_n$.  We refer to $CPC^\dagger$ as the \emph{image} of $P$ under $C$ and $C^\dagger P C$ as the \emph{preimage} of $P$ under $C$.

For computational basis states we define 
$$
 \left(|a_1\rangle \otimes \ldots \otimes |a_{j-1}\rangle  \otimes |a_{j+1}\rangle \otimes \ldots \otimes |a_{n}\rangle \right) \otimes_j |b\rangle = |a_1\rangle \otimes \ldots \otimes |a_{j-1}\rangle 
 \otimes |b\rangle \otimes |a_{j+1}\rangle \ldots \otimes |a_{n}\rangle.
$$
We then extend $\otimes_j$ to a linear map from $\mathbb{C}^{2^{n-1}} \times \mathbb{C}^2$
into $\mathbb{C}^{2^n}$ by linearity. Similarly we define $\otimes_j$ from 
$n-1$ and $1$ qubit linear operators into $n$ qubit linear operators as 
$$
 (A \otimes_j B) (|\phi\rangle \otimes_j |\psi\rangle) = (A|\phi\rangle) \otimes_j (B|\psi \rangle).
$$

We define the \emph{controlled-Pauli operator} for any commuting Pauli observables $P_1,P_2$ as: 
\begin{equation}
\label{eq:controlled-pauli}
\Lambda(P_1,P_2)
=
\frac{I+P_1}{2} +  \frac{I-P_1}{2} \cdot P_2.
\end{equation}
The requirement that $\comm{P,Q} = 0$ ensures that the matrix defined by the equation above is unitary and 
Hermitian.
We will make use of the following relations:
\begin{align}
\label{eq:controlled-pauli-hermitian}
\left(\Lambda(P_1,P_2) \right)^\dagger & = \Lambda(P_1,P_2) ,\\
\label{eq:controlled-pauli-order}
\Lambda(P_1,P_2) & = \Lambda(P_2,P_1),\\
\label{eq:controlled-pauli-stabilized}
\Lambda(P_1,P_2) \ket{\psi} & = \ket{\psi} ~~\text{if}~ P_1 \ket{\psi} = \ket{\psi} ~\text{or}~ P_2 \ket{\psi} = \ket{\psi},\\
\label{eq:controlled-pauli-image}
\Lambda(P_1,P_2)~Q~\Lambda(P_1,P_2) & = P_1^{\comm{P_2,Q}}~Q~P_2^{\comm{P_1,Q}},
\end{align}
where any operator $P$ raised to the power zero is interpreted as the identity.

We say that a stabilizer code on $n$ qubits and with $k$ logical qubits is specified by
an $n$ qubit Clifford unitary $C$,
if images $C Z_1 C^\dagger, \ldots, C Z_{n-k} C^\dagger$ are generators of the code's stabilizer group, 
$C Z_{n-k+1}C^\dagger, \ldots, C Z_{n}C^\dagger$  are representatives of logical $Z$ operators 
for logical qubits $1,\ldots,k$ and $C X_{n-k+1}C^\dagger, \ldots, C X_{n}C^\dagger$ are representatives of logical $X$ operators for logical qubits $1,\ldots,k$. 
We call such $C$ an \emph{encoding unitary} of the stabilizer code.
Note that $C$ is not unique; images $C X_1 C^\dagger, \ldots, C X_{n-k} C^\dagger$ are arbitrary.
We say that such code is \emph{$\comm{\mathbf{n,k,C}}$ code}.
Code $\comm{n,k,C}$ is a subspace of code $\comm{n,k',C'}$
if stabilizer group of $\comm{n,k',C'}$ is a subgroup of $\comm{n,k,C}$.
We say that a state is \emph{encoded in $\comm{n,k,C}$ with syndrome $s$} if the state is stabilized by 
$(-1)^{s_1} C Z_1 C^\dagger, \ldots, (-1)^{s_{n-k}} C Z_{n-k} C^\dagger$. Syndrome $s$ is given by $n-k$
dimensional binary vector.

A \emph{symplectic basis} of a Pauli group is a choice of generators that mimics the commutation relations of pairwise $X$ and $Z$ operators.  Such a basis is useful for selecting $X$ and $Z$ images of a Clifford unitary.

\begin{definition}[Symplectic basis]
\label{def:symplectic-basis}
A sequence of $2m$ Pauli operators form a symplectic basis if they have the same commutation relations as 
$X_1,Z_1,\ldots,X_m,Z_m,$.
\end{definition}

The $2n$-qubit \emph{Bell state} is defined as
\begin{equation}
\label{eq:bell-state-dfn}
 |\mathrm{Bell}_n\rangle = \frac{1}{\sqrt 2^n} \sum_{k \in \{0,1\}^{n}} |k\rangle \otimes |k\rangle.
\end{equation}

The two-qubit Bell state $|\mathrm{Bell}_1\rangle$ is stabilized by 
$X\otimes X, 
-Y\otimes Y, 
Z\otimes Z$
Using the notation $Y^* = -Y$ for element-wise complex conjugation of a matrix, 
the stabilizers of  $|\mathrm{Bell}_n\rangle$ are 
$P^*\otimes P
=
P\otimes P^*$
for all Pauli matrices $P \in \{I,X,Y,Z\}^{\otimes n}$.
This stabilizer group is generated by $Z_j \otimes Z_{n+j}$ and $X_j \otimes X_{n+j}$ for $j = 1, ..., n$.
Recall also that for any $n$-qubit unitary we have $(U\otimes I_n) |\mathrm{Bell}_n\rangle = (I_n \otimes U^T) |\mathrm{Bell}_n\rangle$, where $U^T$ is the transpose of $U$.

We use $n(U)$ to denote the number of qubits $n$ for which a unitary $U$ is defined, $\text{supp}(U)$ to denote the set of qubits on which $U$ acts non-trivially, and $|U|$ to denote the size of $\text{supp}(U)$. 

It can be useful to build a Clifford unitary from a Pauli measurement as follows.

\begin{proposition}[Measurement as unitary] 
\label{prop:measure-as-exp}
Let $|\psi\rangle$ be a state stabilized by a Pauli observable $(-1)^b Q$ for $b \in \{0,1\}$,
and let $P$ be a Pauli observable that anti-commutes with $Q$. 
The probability of outcome zero of measuring $P$ is $1/2$. 
For measurement outcome $r$, the resulting state is $Q^{r+b} e^{i\frac{\pi}{4}(iQP)}|\psi\rangle$.
\end{proposition}
\begin{proof}
The probability of outcome zero is equal to the probability of outcome one because: 
$$
\langle \psi | \frac{I+P}{2} | \psi \rangle
=
\langle \psi | (-1)^b Q\frac{I+P}{2} (-1)^b Q | \psi \rangle
=
\langle \psi | \frac{I-P}{2} | \psi \rangle.
$$
For outcome $r$, the state is equal to $\nicefrac{(I+(-1)^r P)}{\sqrt{2}} |\psi\rangle$ and proportional to
$$
Q^{r+b} e^{i\frac{\pi}{4}(iQP)}|\psi\rangle
=
\frac{Q^{r+b}}{\sqrt{2}}|\psi\rangle - \frac{Q^{r+b} Q P}{\sqrt{2}}|\psi\rangle
=
\frac{(-1)^{b(r+b)}}{\sqrt{2}}|\psi\rangle + \frac{ (-1)^{b+r + b(r+b+1)} P }{\sqrt{2}}|\psi\rangle,
$$
where we used that $Q^{r+b}\ket{\psi}=(-1)^{b(r+b)}\ket{\psi}$ and $Q^{r+b}P = (-1)^{r+b}P Q^{r+b}$.
\end{proof}

We include additional useful mathematical material in \cref{app:mathematical-material}.

\subsection{Stabilizer circuits, quantum instruments and Choi states}
\label{sec:circuits-and-choi-states}

Here we introduce several definitions that help to classify quantum circuits and identify sets of circuits which can be considered equivalent.
We are particularly interested in a particular subclass of circuits which we define as follows.

\begin{definition}[Stabilizer circuit]
\label{def:stabilizer-circuit}
A \emph{stabilizer circuit} is any sequence of the following elementary operations,
that we call \emph{stabilizer operations},
\begin{itemize}[noitemsep]
    \item allocations of qubits initialized to zero states,
    \item allocation of classical random bits distributed as fair coins, 
    \item Clifford and Pauli unitaries,
    \item non-destructive Pauli measurements,
    \item deallocation of qubits in zero states,
    \item Pauli unitaries conditioned on the parity of sets of measurement outcomes and classical random bits from earlier in the sequence.
\end{itemize}
Any stabilizer circuit starts with qubits in an arbitrary state that we call \emph{input qubits}. 
Qubits that remain after executing the circuit are \emph{output qubits}.
We call the sequence of all measurement outcomes and classical random bits allocated by the circuit the \emph{circuit outcome vector}.
\end{definition}

Note that destructive Pauli measurements, while not explicitly in this list, can be formed by following a single-qubit non-destructive Pauli $Z$ measurement with a conditional Pauli unitary and a de-allocation of that qubit.

It is classically efficient to calculate the state output by a stabilizer circuit for a particular input stabilizer state and a given circuit outcome vector~\cite{Gottesman1998, AaronsonGottesman2004}.
Note that due to the inclusion of destructive qubit measurements and qubit allocations,
a stabilizer circuit can have a different number of input and output qubits.
We find it convenient to explicitly include the allocation of classical bits in our definition, although the same effect could be created by allocating a qubit to the zero state and then immediately measuring it in the X basis.

We classify circuit outcomes into different types.
We call an outcome \emph{input dependent}, if the probability it is zero (conditioned on any previous outcomes) depends on the circuit input. 
We distinguish two classes of outcomes which are not input dependent: \emph{random} and \emph{redundant}. 
We call an outcome redundant if the probability that it is zero is either zero or one (conditioned on any sequence of previous outcomes), and we call the outcome random otherwise.

The \emph{action} of any stabilizer circuit (and indeed a much broader class of quantum channels) can be faithfully captured by a \emph{quantum instrument}~\cite{RudingerRibeill2022,Davies1970}, defined as follows.

\begin{restatable}[Quantum instrument\label{def:instrument}]{definition}{instrument}
A quantum instrument consists of a set $R$ and an associated set of completely positive non-zero maps 
$\{ \mathcal{Q}_r : r \in R \}$ where $\sum_{r \in R} \mathcal{Q}_r$ is a trace preserving map. 
The quantum instrument maps a density matrix $\rho$ to a joint quantum-classical state $\{ r, \rho_r \}$,
where $p_r = \Tr(\mathcal{Q}_r[\rho])$ is the probability of observing the outcome $r$ 
and $\rho_r = \mathcal{Q}_r[\rho] / p_r$ is the output state conditioned on observing that outcome.
\end{restatable}

The action of a stabilizer circuit is formally identified by assigning a label $r$ to each of the possible sets of outcomes (measurements and random classical bits)
that can be observed when running the circuit, 
and specifying the map $\mathcal{Q}_r$ that is enacted on the input qubits when $r$ is observed.
In our context, the instrument is defined by a stabilizer circuit in which all outcomes are retained. For each $r$, the map therefore has the form $\mathcal{Q}_r(\rho) = Q_r \rho Q_r^\dagger$ for linear operator $Q_r$.  We refer to a quantum instrument defined by a stabilizer circuit as a \emph{stabilizer instrument}. We also refer to $Q_r$ as the linear map enacted by the instrument for outcome $r$.

We can classify sets of stabilizer circuits by identifying whether or not they have equal action.
In certain cases it can be useful not just to identify if the actions of two stabilizer circuits are equal, but to consider a more general notion of equivalence.
Informally, we consider two quantum instruments to be equivalent if their actions are equal up to a regrouping and relabeling of their outcomes.
We make this more formal with the following two definitions.

\begin{definition}[Outcome compression map]
\label{def:outcome-compression}
Consider an instrument consisting of a set of outcomes $R$ and completely positive maps $\{ \mathcal{Q}_r : r \in R \}$.
Define a new quantum instrument with a set of outcomes $S$ and completely positive maps $\{ \mathcal{Q}'_s : s \in S \}$,
where each element of $s$ consists of a disjoint subset of elements of $R$, such that $S$ forms a partition of $R$ and $\mathcal{Q}_r = \alpha_{r,r'} \mathcal{Q}_{r'}$ if and only if $r,r'$ in $S$,
and $\mathcal{Q}'_s$ is then defined as $\mathcal{Q}'_s = \sum_{r \in s} \mathcal{Q}_r$. 
The map from $R$ to $S$, that maps element $r$ to $s$ such that $r$ is in $s$ is called the \emph{outcome compression map}.
The map that takes instrument $\mathcal{Q}$ and outputs instrument $\mathcal{Q}'$ is called the \emph{compression map}.
\end{definition}

\begin{definition}[Quantum instrument equivalence]
\label{def:instrument-equality}
Consider two quantum instruments, which after the application of the compression map have sets of outcomes $R$ and $R'$ and completely positive maps $\{ \mathcal{Q}_r : r \in R \}$, and $\{ \mathcal{Q}'_r : r \in R' \}$.
We say that the instruments are \emph{equivalent} if there is a bijection $f:R \rightarrow R'$ such that for all $r \in R$, $\mathcal{Q}_r = \mathcal{Q}'_{f(r)}$.
\end{definition}

Note that if for two quantum instruments  $\{ \mathcal{Q}_r : r \in R \}$, and $\{ \mathcal{Q}'_r : r \in R' \}$ there is bijection $f:R \rightarrow R'$ such that for all $r \in R$, $\mathcal{Q}_r = \mathcal{Q}'_{f(r)}$, then theses quantum instruments are equivalent. 

We also find it useful to define the Choi states.

\begin{definition}[Choi states]
\label{def:instrument-choi}
Consider a quantum instrument with the set $R$ and completely positive maps $\{  \mathcal{Q}_r : r \in R \}$,
where $\mathcal{Q}_r(\rho) = Q_r \rho Q_r^\dagger$ for some linear operator $Q_r$ 
from $\ki$ qubits into $\ko$ qubits.
Then, the Choi state given $r$ is the $(\ko+\ki)$-qubit state
$$
 (Q_r \otimes I) |\mathrm{Bell}_{\ki}\rangle = \frac{1}{\sqrt{2}^\ki} \sum_{k \in \{0,1\}^{\ki}} \underbrace{ Q_r |k\rangle}_{\ko} \otimes \underbrace{|k\rangle}_{\ki}.
$$
\end{definition}

As noted above, the condition $\mathcal{Q}_r(\rho) = Q_r \rho Q_r^\dagger$ is satisfied for any stabilizer instrument.
Furthermore, each Choi state of a stabilizer instrument is a stabilizer state, allowing it to be efficiently represented (by providing a set of generators of it's stabilizer group).

Given a circuit $\mathcal{C}$ for a quantum instrument, a Choi state can be prepared by initializing an appropriate number of Bell pairs and then applying $\mathcal{C}$ to the first half of each pair, as illustrated in~\cref{fig:choi-state-circuit}.  We call this the \emph{Choi circuit} $\text{Choi}(\mathcal{C})$ of $\mathcal{C}$.  
While the circuit $\mathcal{C}$ has $\ki$ input qubits and $\ko$ output qubits, the corresponding Choi circuit has zero input qubits and $(\ko + \ki)$ output qubits.

\begin{figure}[h]
\centering
  \loadfig{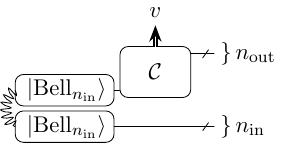}
  \caption[Choi state of a circuit]{
  \label{fig:choi-state-circuit}
The Choi circuit $\text{Choi}(\mathcal{C})$ of a circuit $\mathcal{C}$ with $\ni$ input qubits, $\no$ output qubits and outcome vector $v$.
We use the labels ``$\}\,\ni$'' and ``$\}\,\no$'' to indicate multi-qubit registers of $\ni$, $\no$ qubits, and a double wire for a multi-bit classical register.
The first operation in the circuit initializes qubits into a $2\ki$-qubit Bell state~(defined in \cref{eq:bell-state-dfn}).}
\end{figure}

While the Choi state fully specifies the action of the circuit for the fixed outcome vector, it does not typically provide a very intuitive operational interpretation of the action. 
In \cref{sec:stabilizer-circuit-general-form} we describe a circuit-based representation of the action which provides a more natural interpretation.
It will be fruitful to analyze Choi states as bipartite states with respect to bipartition of $(\ko + \ki)$ qubits into first $\ko$ and last $\ki$ qubits.
We discuss properties of stabilizer states related to such a bipartition in the next subsection.

\subsection{Encoding circuits, unencoding circuits and logical action}
\label{sec:encoding-circuits}

We make frequent use of small explicit circuits that take an input state and encode into, or unencode out of a quantum error correcting code (see \cref{fig:encoding-and-unencoding}).
For each, we provide an explicit encoding Clifford unitary $C$, which along with an integer $k$ defines a stabilizer code that we call $\comm{n,k,C}$ as described in \cref{sec:Paulis-Cliffords-Stabilizers-Bell}.
Specifically, the stabilizer generators are $C Z_i C^\dagger$ for $i \in [n-k]$, and $C Z_{n - k +j}C^\dagger, C X_{n - k +j}C^\dagger$ are the $j$th logical qubits' $Z$ and $X$ operators for $j \in [k]$.

When we write $\mathcal{E}(k,C)$ or ``encode $\comm{n,k,C}$'', we mean the circuit in \cref{fig:encoding-and-unencoding}(a) and when we write $\mathcal{E}^\dagger(k,C)$ or ``unencode $\comm{n,k,C}$'', we mean the circuit in \cref{fig:encoding-and-unencoding}(b).
Note that the circuit $\mathcal{E}(k,C)$ encodes with respect to a bitstring $m$, which is typically generated as random bits.
The circuit encodes in the stabilizer code with generators $(-1)^{m_i}C Z_i C^\dagger$ for $i \in [n-k]$.
Similarly, the circuit $\mathcal{E}^\dagger(k,C)$ unencodes any stabilizer code with generators given by $(-1)^{m_i}C Z_i C^\dagger$ for $i \in [n-k]$ for some $m_i =0,1$, and outputs the syndrome as the circuit's outcome vector $m$.

\begin{figure}[h]
\centering
 \begin{subfigure}[c]{0.48\textwidth}
    \centering
    \loadfig{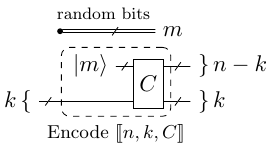}
    \caption{\label{fig:encoding-circuit} The encoding circuit $\mathcal{E}(k,C)$ for $\comm{n,k,C}$ }
\end{subfigure}
\hfill
\begin{subfigure}[c]{0.48\textwidth}
    \centering
    \loadfig{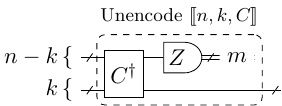}
    \caption{\label{fig:unencoding-circuit} The unencoding circuit $\mathcal{E}^{\dagger}(k,C)$ for $\comm{n,k,C}$}
\end{subfigure}
\caption[Encoding and unencoding circuits]{
A unitary $C$ and integers $n,k$ define encoding and unencoding circuits for $\comm{n,k,C}$.
In this and other circuit diagrams, measuring $Z$ on an $(n-k)$-qubit register implies destructively measuring each qubit in the $Z$ basis and recording the results as an $(n-k)$-bit vector $m$.
}
\label{fig:encoding-and-unencoding}
\end{figure}

We can use encoding and unencoding circuits to formally define the concept of a logical operation circuit and its logical action as follows.

\begin{definition}[Logical operation circuit]
\label{def:logical-operation-circuit}
A stabilizer circuit $\mathcal{C}$
is a \emph{logical operation circuit} with \emph{input code} $\comm{\ni,k_\inn,C_\inn}$\footnote{When considering a code $\comm{n,k,C}$ we assume that we are given integers $n$, $k$ and \emph{fixed} Clifford unitary $C$}
and \emph{output code} $\comm{\ko,k_\out,C_\out}$ if applying circuit $\mathcal{C}$ to any input state encoded in $\comm{\ni,k_\inn,C_\inn}$
results in an output state encoded in $\comm{\no,k_\out,C_\out}$.
\end{definition}

\begin{definition}[Logical action]
\label{def:logical-action}
Given a logical operation circuit~(\cref{def:logical-operation-circuit}) $\mathcal{C}$ with input code $\comm{n_\inn,k_\inn,C_\inn}$
and output code $\comm{n_2,k_\out,C_\out}$ its \emph{logical action} is the action of the circuit given in~\cref{fig:logical-instrument-conjugation}
when input code syndrome $s_\inn$ is zero.
\end{definition}

\begin{figure}[ht]
    \centering
    \loadfig{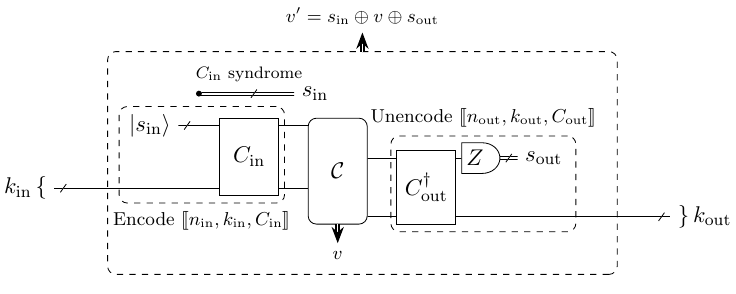}
    \caption[Circuit defining logical action]{\label{fig:logical-instrument-conjugation}
    The circuit $\mathcal{E}(k_\inn,C_\inn) \circ \mathcal{C} \circ \mathcal{E}^\dagger(k_\out,C_\out)$ with outcome vector $v'$ used for defining the logical action of a logical operation circuit $\mathcal{C}$ 
    with outcome vector $v$, input code $\comm{\ni, k_\inn, C_\inn}$ and output code $\comm{\no, k_\out, C_\out}$.
    }
\end{figure}

Note that according to our definition of a logical operation circuit and its logical action, the outcome $v' = 0^{\ni-k_\inn} \oplus v \oplus 0^{\no -k_\out}$
when $s_\inn = 0^{\ni-k_\inn}$ and so the outcome vector of the circuit in \cref{fig:logical-instrument-conjugation}
and the outcome vector $v$ of its sub-circuit $\mathcal{C}$ are directly related.

It can be useful to explicitly define a map from logical operator representatives to elements of the logical Pauli group of a stabilizer code.
Recall, that given a code $\comm{n,k,C}$ with stabilizer group $S$,  $n$-qubit Pauli unitary $P$ is in $S^\perp$
if and only if preimage $C^\dagger P C = Z^a \otimes Q$ for some Pauli unitary $Q$ from $\mathcal{P}_k$ and bit string $a$ from $\f^{n-k}_2$.
Using this fact, we define 
$
\mathcal{L}_{(k,C)} : S^{\perp} \rightarrow \mathcal{P}_k 
$
and 
$
\mathcal{G}_{(k,C)} : S^{\perp} \rightarrow \f^{n-k}_2
$
via the following equation:
\begin{equation}
\label{eq:logical-operator-map}
 C^\dagger P C = Z^{\mathcal{G}_{(k,C)}} \otimes (\mathcal{L}_{(k,C)}(P)).
\end{equation}
It is also useful to note that upon syndrome $s \in \f_2^{n-k}$ the
unencoding circuit $\mathcal{E}^\dagger(k,C)$ implements a linear map:
\begin{equation}
\label{eq:encoding-conjugation-action}
\mathcal{E}^\dagger(k,C):~ P ~ \xrightarrow[]{~~\text{outcome}~s~~}~ (-1)^{\ip{s,\mathcal{G}_{(k,C)}}} \mathcal{L}_{
(k,C)}(P).
\end{equation}

\subsection{Bipartite normal form of a stabilizer state}
\label{sec:bipartite}

Bipartite entanglement of stabilizer states is useful for reconstructing the action of a Clifford channel given the Choi state of the channel.
This connection is first clarified in \cref{sec:stabilizer-circuit-general-form}
and used for algorithmic purposes in \cref{sec:general-form-algo}.
We find it useful to restate a result from \cite{Haah2017} (which is a special case of a more general result in \cite{Audenaert2005}\footnote{Theorem~1 in \cite{Audenaert2005} applies not only to pure stabilizer states, but to mixed stabilizer states proportional to projectors on stabilizer codes, but we find the result in \cite{Haah2017} more convenient here.}):
\begin{theorem}[Theorem~II.24 from~\cite{Haah2017}]
\label{thm:bipartition}
Given any bipartition $M \sqcup M^c$ of $n$ qubits, and a stabilizer state $\ket{\psi}$ on $n$ qubits, there exists a tensor product of two Clifford operators (supported on $M$ and $M^c$ correspondingly) that transforms $\ket{\psi}$
into Bell pairs and zero states.
\end{theorem}

\begin{figure}[h]
\centering
  \loadfig{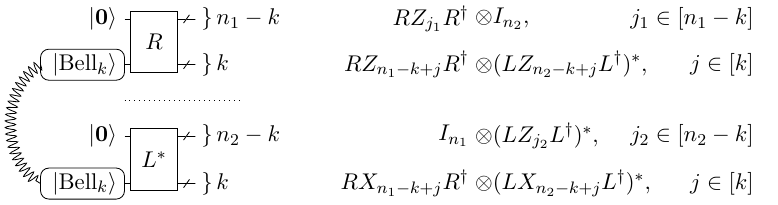}
  \caption[Stabilizer state bipartition]{
  \label{fig:bipartition}
Bipartite normal form of a stabilizer state $\ket{\psi}$ on $n_1 + n_2$ qubits; a corollary of \cref{thm:bipartition}.
On the left, the circuit diagram shows creation of a stabilizer state $\ket{\psi}$ starting from $k$ Bell pairs and using Clifford unitaries $\Co$ and $\Ci^\ast$ acting on, respectively, the first $n_1$ and last $n_2$
qubits. 
On the right, generators of the stabilizer group of $\ket{\psi}$ that follow from the circuit diagram.
}
\end{figure}

An immediate corollary of the above theorem is that any stabilizer state on $n_1 + n_2$ qubits can be prepared by a circuit as shown in~\cref{fig:bipartition}.
The relevant proofs in~\cite{Audenaert2005} and~\cite{Haah2017} are constructive and include polynomial time algorithms for constructing the circuit 
by solving the following:

\begin{problem}[Bipartite normal form of a stabilizer state]
\label{prob:partition-state}
Let $\ket{\psi}$ be the stabilizer state stabilized by a set of $n_1 + n_2$ commuting Pauli operators acting on $n_1 +n_2$ qubits.
Find an integer $k$ and Clifford unitaries $\Ci$ and $\Co$ such that the circuit in~\cref{fig:bipartition} produces the state $\ket{\psi}$.
\end{problem}

Moreover, the stabilizer group 
of any stabilizer states can be written as a product of three stabilizer groups 
\begin{equation} \label{eq:three-stab-groups}
\left( S_\Co \otimes I_{n_2} \right)
\cdot
\left\{ \Co (I_{n_1 - k } \otimes P) \Co^\dagger \otimes  (\Ci (I_{n_2 - k } \otimes P) \Ci^\dagger)^\ast : P \in \{I,X,Y,Z\}^{\otimes k} \right\}
\cdot
\left( I_{n_1} \otimes  S_\Ci^\ast \right),
\end{equation}
where groups $S_\Ci$ and $S_\Co$ are defined as: 
$$
S_\Co = \Co  \ip{Z_1,\ldots,Z_{n_1 - k}} \Co^\dagger, ~~~ S_\Ci = \Ci \ip{Z_1,\ldots,Z_{n_2 - k}}  \Ci^\dagger.
$$
Groups $S_\Co$ and $S_\Ci$ do not depend on the particular choice of Clifford unitaries $\Co$ and $\Ci$.
This is because $S_\Co \otimes I_{n_2}$ and $I_{n_1} \otimes S_\Ci^\ast$ are exactly the maximal subgroups of the stabilizer group of $\ket{\psi}$ supported on the first $n_1$ and the last $n_2$ qubits.

Another corollary of the constructive proofs in~\cite{Audenaert2005} and \cite{Haah2017} is an 
efficient algorithm for computing a Clifford unitary, given its Choi state. 
The algorithm is to first compute $\Ci$ and $\Co$ in~\cref{fig:bipartition} for the Choi state of a Clifford unitary $C$ considered as a bipartite state between the first and the last $n$-qubits.
Then return $C = \Co \Ci^\dagger$. 
Correctness of the algorithm follows from the following two observations.
First, for an $n$-qubit Clifford unitary its Choi state $(C \otimes I_n) \ket{\mathrm{Bell}_n}$ is locally equivalent to $n$ Bell pairs, and therefore $k=n$ in~\cref{fig:bipartition}.
Second, equality $(I \otimes \Ci^\ast) \ket{\mathrm{Bell}_n} = (\Ci^\dagger \otimes I) \ket{\mathrm{Bell}_n}$ implies that the Choi states of $C$ and $\Co \Ci^\dagger$ are equal and therefore 
$C = \Co \Ci^\dagger$.

\subsection{Linear algebra notation}
\label{sec:linear-algebra}

Here we provide some common linear algebra notation we use. 
All vectors we consider are column vectors such that an $n$-dimensional vector can also be thought as an $n\times 1$ matrix.
For matrices $A,B$ with the same number of rows, we use $(A|B)$ to denote a matrix with columns consisting of the columns of $A$ followed by the columns of $B$.
Similarly, for matrices  $A,B$ with the same number of columns, $(\frac{A}{B})$ is the matrix with rows consisting of the rows of $A$ followed by the rows of $B$.
In particular, for $n_1$- and $n_2$-dimensional vectors $v_1$ and $v_2$, the vector $(\frac{v_1}{v_2})$ is an $(n_1 + n_2)$-dimensional vector.
We use $A_n$ to refer to the $n$-th row of $A$, $A_{[n]}$ to refer to the matrix consisting of the first $n$ rows of $A$, and 
$A_{[n,\cdot]}$ to refer to the matrix consisting of all rows of $A$ starting from row $n$.
We use $[n]$ to denote a sequence of integers $1,\ldots,n$ and $[m,n]$ to denote the sequence $m,m+1,\ldots,n$.
More generally, for any sequences of integers $J$ and $I$, $A_{J,I}$ is the $|J|\times |I|$  matrix with entries $(A_{J,I})_{j,i} = A_{J_j,I_i}$ for $j \in [J], i \in [I]$.
Similarly, we use $A_{*,n}$ to refer to the $n$-th column of $A$, $A_{\ast,[n]}$ to refer to the matrix consisting of the first $n$ columns of $A$ and
$A_{\ast,[n,\cdot]}$ to refer to the matrix consisting of all columns of $A$ starting from column $n$.

The inner product $\ip{v_1,v_2} = v_1^T v_2$ acts on column vectors. 
The direct sum of two vectors $v_1, v_2$ is $v_1 \oplus v_2 = (\frac{v_1}{v_2})$.
When considering matrices with entries in $\f_2 = \{0,1\}$, the matrix $I_n$ is the $n\times n$ identity matrix and
$\mathbf{0}_{n_1\times n_2}$ is the $n_1\times n_2$ matrix with all zero entries.
For $n_1\times n_2 $ and $m_1 \times m_2$ matrices $A$ and $B$, the direct sum is defined as
$$
 (A \oplus B)(v_1 \oplus v_2) = (A v_1) \oplus (B v_2),~~
 A\oplus B = \left( \begin{array}{c|c}
    A & \mathbf{0}_{n_1 \times m_2}  \\ \hline
    \mathbf{0}_{m_1\times n_2} & B
\end{array} \right).
$$
When considering matrices with entries in $\mathbb{C}$, $I_n$ is the $2^n\times 2^n$ identity matrix.

We use $A^{-1}$ for the inverse of a square matrix.
We use $B^{(-1)}$ for the inverse of a rectangular matrix with full row or column rank (this is a right inverse when row rank is full, and a left inverse when column rank is full).
The \emph{row rank profile} of a rank-$r$ matrix $A$ is the lexicographically smallest sequence of $r$ row indices
$i_1 < i_2 < \ldots < i_{r-1}$ such that the corresponding rows of $A$ are linearly independent. 
When $A^T$ is in reduced row echelon form, the row rank profile is the sequence of columns of $A^T$ with leading ones.

\newpage
\section{Algorithmic preliminaries}

Here we cover some basic algorithmic concepts and definitions that underpin our main results.

\subsection{Stabilizer circuit parameters}
\label{sec:circuit-parameters}

Throughout this work, unless otherwise specified, we consider the \runtime{} of algorithms given input stabilizer circuits as defined in \cref{def:stabilizer-circuit} parameterized in terms of the following quantities:
\begin{itemize}[noitemsep]
    \item $\nmax$ -- the maximum number of qubits used throughout the circuit,
    \item $\nM$ -- the number of circuit outcomes,
    \item $\Nu$ -- the number of unitary operations ($\exp(P)$, $\Lambda(P,P')$),
    \item $\Ncnd$ -- the number of conditional Pauli unitary operators,
    \item $\nr$ -- the number of random outcomes,
    \item $\nm$ -- the number of input-dependent outcomes.
\end{itemize}

We further assume that conditional unitary and measurement operations are parameterized by Pauli operators with weight bounded by some fixed constant and that the weight of any bitstring indicating outcomes upon which Pauli operators are conditioned on is at most $\nmax$.

In \cref{sec:intro} we considered families of circuits with growing maximum width $n = \nmax$ and depth $D$, and where all stabilizer operations appear with constant non-zero density.
More specifically this assumes that $\nM$, $\Nu$, $\Ncnd$, $\nr$ are all non-zero and proportional to $n D$, and $\nm$ is zero for stabilizer simulation (which has no input qubits), but non-zero and proportional to $n$ for the other cases analyzed in \cref{sec:intro}.

\subsection{Two algorithmic cost models}

We consider two cost models when analyzing the \runtime{} complexity of algorithms in this paper.
The first is the standard sequential model in which which each bit-level operation has cost O(1).
In the second model, \bitstring{} level operations have a cost O(1). For example, the bit-wise XOR of two length-$n$ \bitstrings{} has cost O($n$) in the sequential model and cost O(1) in the second model.
More precisely, we consider bit-wise NOT, AND, OR, XOR, Hamming weight calculation and initialising the all-zero \bitstring{} $0^n$ to be $O(1)$ in the second model.
To distinguish the costs of these two models, we simply write \emph{`\runtime'} when referring to the \runtime{} in the standard sequential cost model, and \emph{`\bitruntime{}'} in the second model.
The second model is motivated by the fact that modern computing hardware can quickly perform
bit-wise operations on \bitstrings{} of length $32$, $64$, $128$, $256$ and even $512$.

When considering \bitstring{} \runtime{} it is important to keep in mind
how a given data-structure is organized into \bitstrings. 
For example, consider an $n\times n$ matrix with $\{0,1\}$ entries 
that is stored as an array of \bitstrings. 
We say that we store a matrix in a row-major order if each row corresponds to a bitstring. 
In this case, computing a sum of two rows of the matrix has \bitstring{} \runtime{} $O(1)$.
However, computing a sum of two columns has \bitstring{} \runtime{} $O(n)$, 
because it requires modifying individual bits of $n$ \bitstrings.

\subsection{Data structures for Pauli and Clifford unitaries}
\label{sec:data-structures}

Here we define a set of data structures to represent Pauli and Clifford unitaries, which can be used to build algorithms for stabilizer circuits.
We briefly mention alternative data structures that can be used and give some justification for the choices we have made.
However, our main algorithms do not require use of these specific data structures, provided the alternative data structures offer an equivalent set of procedures with the given time complexities (as described in the next subsection).

\begin{figure}[h]
    \centering
    \includegraphics[width=\textwidth]{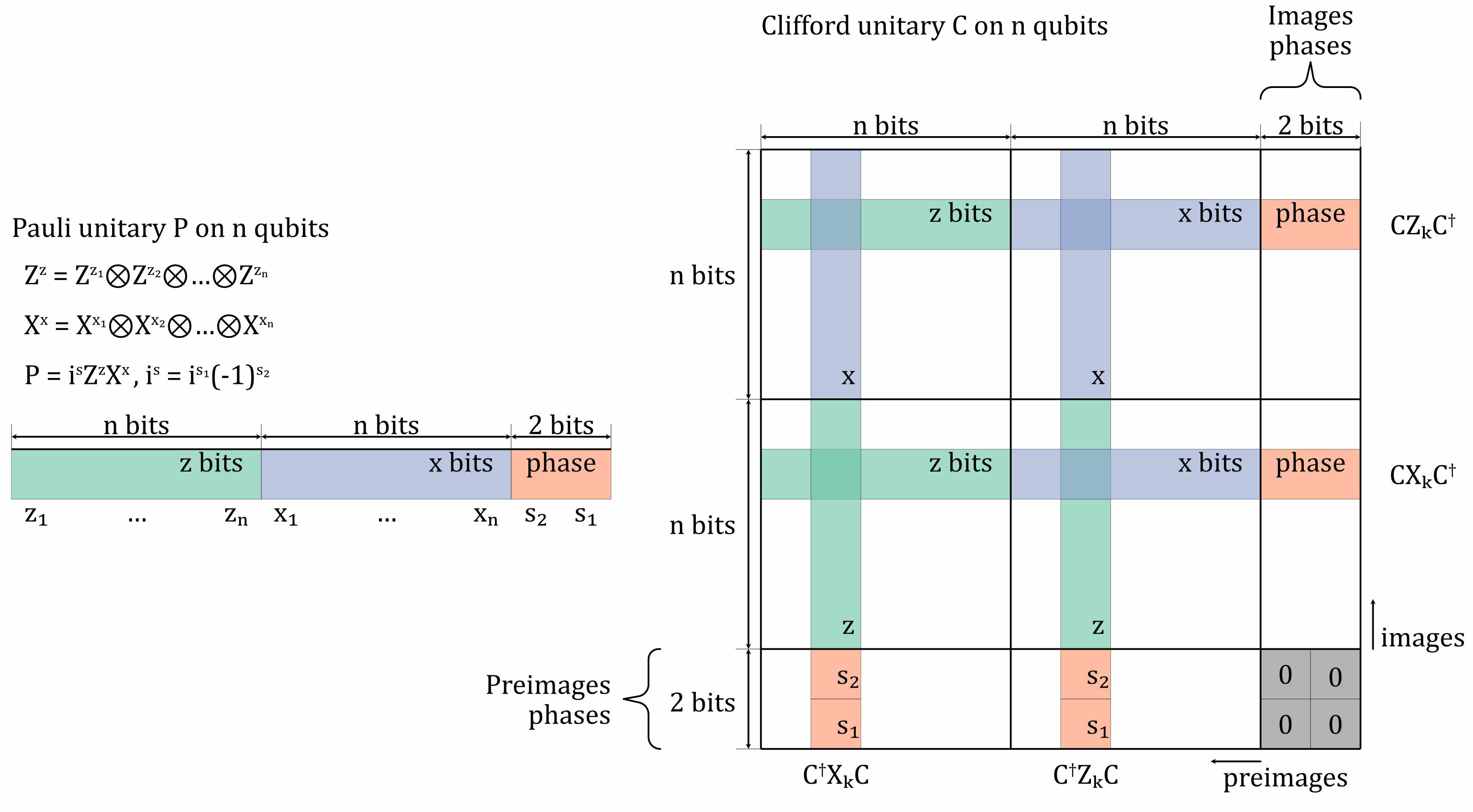}
    \caption[Description of Pauli and Clifford unitaries]{Description of Pauli unitary in terms of bit vectors as in \cref{def:desc-pauli} 
    and of a Clifford unitary using a binary matrix as in \cref{def:desc-clifford}.
    }
    \label{fig:desc-pauli-and-clifford}
\end{figure}

There are two common ways of encoding a one-qubit Pauli unitary using 
two bits \cite{AaronsonGottesman2004, DehaeneMoor2003}. 
The first encoding maps bits $(z,x)$ to powers $Z^{z} X^{x}$, and the second encoding maps all two-bit combinations to Pauli observables (Hermitian Pauli unitaries): 
$(0,0) \rightarrow I$, 
$(0,1) \rightarrow X$, 
$(1,0) \rightarrow Z$,
$(1,1) \rightarrow Y$.
We call the first encoding the \emph{power encoding} and the second encoding the \emph{Hermitian encoding}.

These encodings can be readily extended to form $n$-qubit Pauli unitary encodings using $2n$ bits, and there are two standard approaches to encode those bits.
The first, and most commonly used approach~\cite{AaronsonGottesman2004, DehaeneMoor2003}, is to encode an $n$-qubit Pauli 
operator into a \bitstring{} of length $2n$~(or two \bitstrings{} of length $n$) with a batch of all $z$ bits followed by a batch of all $x$ bits $(z_1,\ldots,z_n,x_1,\ldots,x_n)$.
We call this approach the \emph{batch encoding}.
The second approach is to encode a Pauli unitary on each qubit serially $(z_1,x_1,\ldots,z_n,x_n)$; we call this the \emph{serial encoding}.

To fully specify an $n$-qubit Pauli unitary, we must include a phase, which requires two bits to encode the power of $i$.
To encode an $n$-qubit Pauli observable, this cost can be reduced when a Hermitian encoding is used, because the phase can only be $\pm 1$.
On the other hand, two bits are still needed to encode the phase for the power encoding of an $n$-qubit Pauli observable. 
Ref.~\cite{AaronsonGottesman2004} uses Hermitian batch encoding of $n$-qubit Pauli observables, while Ref.~\cite{DehaeneMoor2003} uses power batch encoding. 
In what follows we use the batch power encoding of Pauli unitaries, see \cref{def:desc-pauli}.

\begin{definition}[Pauli unitary data structure]
\label{def:desc-pauli}
Two length-$n$ binary vectors $x,z$ and one length-2 binary vector $s=(s_1,s_2)$  describe a
Pauli unitary~(see~\cref{fig:desc-pauli-and-clifford}),
$$
 P =  i^{ s_1 + 2 s_2}\cdot Z^z X^x.
$$
The vectors $x(P)$, $z(P)$ and $s(P)$ are, respectively, the $\boldsymbol{x}$ \emph{bits}, the $\boldsymbol{z}$ \emph{bits}, and the \emph{phase} of $P$.
It is convenient to also interpret $s(P)$ as integer $s_1 + 2 s_2$ and write $s(P)/2$ for $s_1$.
\end{definition} 

Given this notation, we 
define $\text{supp}(z(P))$ and $\text{supp}(x(P))$ as the indices where $z(P)$ and $x(P)$, respectively, are non-zero.

A primary advantage of using the batch power encodings over the alternatives for our Pauli unitary data structure is that it is known to admit composition of two Clifford unitaries by way of matrix-matrix and matrix-vector multiplications, thus allowing fast matrix multiplication algorithms~\cite{Kerzner2020} to be used.
There are however some things which can be less efficient with this encoding.
For example, checking if a Pauli unitary is Hermitian requires computing $\ip{z(P),x(P)}$ with \runtime{} $O(n)$ and \bitstring{} \runtime{} $O(1)$.
This is because inner product $\ip{\cdot,\cdot}$ can be calculated by bit-wise AND, followed by the Hamming weight calculation.
By contrast, with any Hermitian encoding this check can be done in \runtime{} $O(1)$ (by checking the sign).
However we find that in our algorithms, checking if a Pauli unitary is Hermitian is infrequent, and only used as a test.
Another standard operation is to compute the commutator of two $n$-qubit Pauli unitaries $\comm{P,Q} = \ip{x(P),z(Q)} +  \ip{z(P),x(Q)}$, which with the batch power encoding data structure has
\runtime{} $O(n)$ and \bitruntime{} $O(1)$.

We will occasionally use another data-structure for the sparse encoding of Pauli unitaries.
This data structure consists of a pair, with the first part a length-$n$ ordered list\footnote{We require the list to be ordered, so that the support of the product of two sparse Pauli observables can be computed in linear time} of integers and the second part an $n$-qubit Pauli observable data structure. 
The first part identifies the qubits on which the Pauli observable acts non-trivially.

There are two common approaches to represent $n$-qubit Clifford unitaries compactly. 
The first approach is to store the images $C X_k C^\dagger$, $C Z_k C^\dagger$ of the generators of the Pauli group $X_k$, $Z_k$ for $k$ in $[n]$ \cite{AaronsonGottesman2004,DehaeneMoor2003}.
The second approach is to store preimages $C^\dagger X_k C$, $C^\dagger Z_k C$~\cite{Gidney2021}, which is advantageous when the Clifford unitary data structure is used for stabilizer simulation~\cite{Gidney2021},
as we discuss in \cref{sec:stab-simulation}.
In this paper we propse a third approach, in which both images and preimages are stored simultaneously, see \cref{def:desc-clifford}.

\begin{definition}[Clifford unitary data structure]
\label{def:desc-clifford}
A $(2n+2)\times(2n+2)$ matrix $M(C)$ over $\f_2$ describes an $n$-qubit Clifford
unitary $C$~(see~\cref{fig:desc-pauli-and-clifford}), where:
\begin{enumerate}
    \item for rows $k \in [n]$ of $M(C)$, $M(C)^T$, $Z$ images $P_k = C Z_k C^\dagger$ and $X$ preimages $Q_k = C^\dagger X_k C$,
    \begin{align}
        M(C)_k &= z(P_k)\oplus x(P_k) \oplus s(P_k)\label{eq:clifford-row} \\
        M(C)^T_k &= x(Q_k)\oplus z(Q_k) \oplus s(Q_k)\label{eq:clifford-column},
    \end{align}
    \item and rows $k \in [n+1,2n]$ of $M(C)$, $M(C)^T$ are defined analogously for $X$ images $P_k = C X_k C^\dagger$ and $Z$ preimages $Q_k = C^\dagger Z_k C$.
\end{enumerate}
\end{definition}


We can see that this is a consistent definition as follows.
The top-left $2n\times2n$ sub-matrix of $M(C)$ is a binary symplectic matrix.
Using the expression for the inverse of a binary-symplectic matrix 
\begin{equation}
\label{eq:symplectic-matrix-inverse}
\left(
\begin{array}{c|c}
      A_{z,x} & A_{x,x} \\ \hline
      A_{z,z} & A_{x,z}
\end{array}
\right)^{-1}
=
\left(
\begin{array}{c|c}
      A_{x,z}^T & A_{x,x}^T \\ \hline
      A_{z,z}^T & A_{z,x}^T
\end{array}\right)
, \text{ where } A_{z,x},A_{x,x},A_{z,z},A_{x,z} \text{ are } n \times n \text{ matrices over } \f_2 
\end{equation}
one can see that 
column $k$ of $M(C)$ encodes the preimage $C^\dagger X_k C$ with first $n$ bits 
corresponding to the $x$ bits, the next $n$ bits corresponding to the $z$ bits and the
last two bits corresponding to the phase~(see~\cref{fig:desc-pauli-and-clifford}).
Similarly, column $n+k$ encodes preimage
$C^\dagger Z_k C$ with the first $n$ bits corresponding to the $x$ bits, the next $n$ bits 
corresponding to the $z$ bits and the last two bits corresponding to the phase.
Therefore,~\cref{eq:clifford-row,eq:clifford-column} are mutually consistent.

Given an $n$-qubit Pauli observable supported on $k$ qubits, \cref{def:desc-clifford} admits calculation of the image $C P C^\dagger$ and preimage $C^\dagger P C$ with \runtime{} $O(kn)$.  
If $C$ is stored in column-major (row-major) order, then calculating the preimage (image) has \bitruntime{} $O(k)$.
 
In the algorithmic context, we say that the Pauli observable $P$ or Clifford unitary $C$ are given 
when descriptions satisfying \cref{def:desc-pauli} or \cref{def:desc-clifford}, respectively, are provided.
Similarly when we say that algorithm finds Pauli observable $P$ or Clifford unitary $C$, 
it means that algorithm computes descriptions that satisfy \cref{def:desc-pauli} or \cref{def:desc-clifford}.

\subsection{Procedures for Pauli and Clifford unitary manipulation}
\label{sec:pauli-and-clifford-procedures}

Here, in \cref{tab:data-structure-requirements}, we list the mathematical action and \runtime{} complexities for a set of key procedures.
For the purposes of \bitruntime{} analysis, we assume that each column of the Clifford unitary data structure \cref{def:desc-clifford} is stored as three \bitstrings{} corresponding to $z$, $x$ and sign bits.
Similarly, Pauli unitaries are stored as three \bitstrings{}.
We use an asterisk to mark \emph{basic procedures}, namely those which make explicit use of our data structures for Pauli and Clifford unitaries that can be achieved using data structures \cref{def:desc-pauli} and \cref{def:desc-clifford}.
All other procedures and algorithms in this paper are built using these basic procedures in \cref{tab:data-structure-requirements} without explicit reference to the underlying data structures.
Therefore, if one prefers alternative data structures to those we propose here, and provides alternative basic procedures using their preferred data structures which have the same action and \runtime{} as those in \cref{tab:data-structure-requirements}, then they can use our proposed algorithms with their preferred data structures and will achieve the same \runtime{} guarantees by calling their versions of the procedures in \cref{tab:data-structure-requirements}.

\begin{table}[p]
\centering
    \begin{tabular}{|c|c|c|c|}
    \hline 
    \multirow{2}{*}{Procedure}  & \multirow{2}{*}{Mathematical action} & \multirow{2}{*}{Run time} & Bit-string \\
    & & & \runtime{} \\
    \hline 
    \hline 
    \multirow{2}{*}{\texttt{init\_pauli$^*$(}$a,b,c$\texttt{)}} &  $P$, where $x(P)=a$, & \multirow{2}{*}{$O(n)$} & \multirow{2}{*}{$O(1)$} \\
     & $z(P)=b$ and $s(P)=c$ && \\
    \hline
    \texttt{x\_bits$^*$(}$P$\texttt{)} & $x(P)$ & $O(1)$ & $O(1)$  \\
    \hline 
    \texttt{z\_bits$^*$(}$P$\texttt{)} & $z(P)$ & $O(1)$ & $O(1)$  \\
    \hline 
    \texttt{xz\_phase$^*$(}$P$\texttt{)} & $s(P)$ & $O(1)$  & $O(1)$ \\
    \hline 
    \texttt{set\_x\_bits$^*$(}$P,a$\texttt{)} & $x(P) \leftarrow a$  & $O(n)$ & $O(1)$  \\
    \hline 
    \texttt{set\_z\_bits$^*$(}$P,a$\texttt{)} & $z(P) \leftarrow a$  & $O(n)$ & $O(1)$  \\
    \hline 
    \texttt{is\_hermitian$^*$(}$P$\texttt{)} & Is $P = P^\dagger$  & $O(n)$ & $O(1)$  \\
    \hline
    \texttt{commutator$^*$(}$P,Q$\texttt{)}  & $\comm{P,Q}$  & $O(n)$ & $O(1)$  \\
    \hline
    \texttt{prod\_pauli$^*$(}$P,Q$\texttt{)} & $PQ$  & $O(n)$ & $O(1)$  \\
    \hline
    \texttt{mult\_phase$^*$(}$P,j$\texttt{)} & $P \leftarrow i^j P$  & $O(1)$ & $O(1)$  \\
    \hline
    \hline 
    \multirow{2}{*}{\texttt{init\_cliff$^*$(}} &  $C$, where $\forall j \in [n]$:  & &
    \\
    \multirow{2}{*}{$P_1,Q_1,\dots P_n,Q_n$\texttt{)}} & $C X_j C^\dagger =P_j$,  & $O(n^\omega)$ & $O(n^2)$ \\
     & and $C Z_j C^\dagger =Q_j$  &  & \\
    \hline 
     & $U_A$, where $\forall j \in [n]$:  & &
    \\
    \texttt{init\_cliff\_css$^*$(}$A,B$\texttt{)} & $U_A X_j U_A^\dagger =X^{A_j}$, & $O(n^2)$ & $O(n^2)$   \\
    where $B^T = A^{-1}$ &  and $U_A Z_j U_A^\dagger =Z^{B_j}$  &  & 
    \\
    \hline 
    \texttt{copy\_cliff$^*$(}$C$\texttt{)} & $C,C$ & $O(n^2)$ & $O(n)$ \\
    \hline 
    \texttt{num\_qubits$^*$(}$C$\texttt{)} & $n$ & $O(1)$ & $O(1)$ \\
    \hline 
    \texttt{image$^*$(}$C,P$\texttt{)} &  $CPC^\dagger$ & $O(n \cdot |P|))$ & $O(n \cdot |P|)$ 
    \\
    \hline 
    \texttt{preimage$^*$(}$C,P$\texttt{)}  & $C^\dagger P C$ & $O(n \cdot |P|)$   & $O(|P|)$  
    \\
    \hline 
    \texttt{inverse\_cliff$^*$(}$C$\texttt{)} & $C^\dagger$ & $O(n^2)$ & $O(n^2)$ \\
    \hline 

    \texttt{prod\_cliff$^*$(}$C,C'$\texttt{)} & $C C'$ & $O(n^\omega)$ & $O(n^2)$ \\
    \hline 
    \texttt{left\_mult\_exp$^*$(}$C,P$\texttt{)} & $C \leftarrow \exp{\!(P)}~C$ & $O(n \cdot |P|)$ & $O(|P|)$ \\

    \hline 
    \texttt{left\_mult\_swap(}$C,i,j$\texttt{)} & $C \leftarrow \text{SWAP}_{i,j} ~ C$ & $O(n)$ & $O(1)$ \\
    \hline 
    \texttt{right\_mult\_swap$^*$(}$C,i,j$\texttt{)} & $C \leftarrow  C ~ \text{SWAP}_{i,j}$ & $O(n)$ & $O(n)$ \\
    \hline 
    \texttt{tensor\_prod$^*$(}$C,C''$\texttt{)} & $C \otimes C''$ & $O((n+l)^2)$ & $O((n+l)^2)$ \\
    \hline
    \texttt{add\_qubits$^*$(}$C,m$\texttt{)} & $C \leftarrow C \otimes I_m$ &  $O((n+m)^2)$  &  $O((n+m)^2)$
    \\
    \hline 
    \texttt{remove\_qubits$^*$(}$C,m$\texttt{)} & $C \otimes I_m \leftarrow C$ & $O(n^2)$  & $O(n^2)$   
    \\
    \hline 
    \texttt{left\_mult\_pauli(}$C,P$\texttt{)} & $C \leftarrow PC$ & $O(n \cdot |P|)$ & $O(|P|)$ \\
    \hline 
    \texttt{left\_mult\_ctrl\_pauli(} & \multirow{2}{*}{$C \leftarrow \Lambda(P,Q)C$} & \multirow{2}{*}{$O(n \cdot(|P|+|Q|))$} & \multirow{2}{*}{$O(|P|+|Q|)$} \\
    $C,P,Q$\texttt{)}   & &  &  \\
    \hline
    {\texttt{disentangle(}$C,j$\texttt{)}} & $C \leftarrow C' \otimes_j I_1$, where:  & \multirow{2}{*}{$O(n^2)$} & \multirow{2}{*}{$O(n)$} \\
    {where $Z_j C |0^n\rangle = C |0^n\rangle$} &  $(C'\otimes_j I_1)|0^n\rangle = C |0^n\rangle$ &&  \\
    \hline 
    \end{tabular}
    \caption[Data structure requirements]{
    Key procedures for Pauli and Clifford unitary manipulation.
    Asterisks mark basic procedures which can be implemented using the data structures in \cref{def:desc-pauli} and \cref{def:desc-clifford}.
    All procedures and algorithms with no asterisk (here and throughout this work) are expressed only in terms of asterisked procedures, independently of whichever underlying data structures are used.
    The procedure inputs are as follows:
    $P,Q, P_j$ and $Q_j$ are $n$-qubit Pauli observables;
    $a$ and $b$ are length-$n$ \bitstrings, while $c$ is a length-2 \bitstring;
    $A$ and $B$ are $n \times n$ binary matrices (with $A_j$ being the $j$th row of $A$ etc);
    $C$ and $C'$ are $n$-qubit Clifford unitaries, while
    $C''$ is an $l$-qubit Clifford unitary.
    By $\alpha \leftarrow \beta$, we mean that object $\alpha$ is replaced by object $\beta$.
    We define $I_m$ to be the $m$-qubit identity operator, $\Lambda(P,Q)$ as the generalized controlled Pauli from \cref{eq:controlled-pauli}, and $\text{SWAP}_{i,j}$ to be the operator that swaps the locations of qubits $i$ and $j$; $\omega$ is the matrix-multiplication exponent~\cite{ABH}.
    The procedures' details and \runtime{} are discussed in \cref{app:pauli-and-clifford-manipulation}. }
    \label{tab:data-structure-requirements}
\end{table}

In the remainder of this subsection we provide a high-level algorithmic description and an explanation of the time complexity of a few illustrative procedures.
For readability here and elsewhere in this paper, we describe algorithms at a high level in terms of the mathematical action of procedures. 
We include explicit checks in some of our algorithms which are always `true' in correct implementations, but may be useful for testing and debugging. 

First we describe the preimage procedure, which makes explicit reference to the data structures in \cref{def:desc-pauli} and \cref{def:desc-clifford}.

\begin{procedure}[\texttt{preimage$^*$}] \label{proc:image} 
\begin{algorithmic}[1]
\Blank
\Input Pauli observable $P$, Clifford unitary $C$. 
\Output $C^\dagger P C$.
\State Initialize Pauli $Q = \prod_{ k\,:\, z(P)_k \ne 0 } C^\dagger Z_k C \prod_{ k \,:\, x(P)_k \ne 0 } C^\dagger X_k C \cdot i^{s(P)}$.
\State Check $Q$ is Hermitian.
\State \Return $Q$.
\end{algorithmic}
\end{procedure}

In the \texttt{preimage$^*$} procedure, we assume that the Pauli and Clifford unitaries $P$ and $C$ are specified using the data structures in \cref{def:desc-pauli} and \cref{def:desc-clifford} respectively.
Let $n$ be the number of qubits that $P$ and $C$ act on.
Note that $x(P)$ and $z(P)$ both have at most $|P|$ non-zero entries.
The procedure first initializes a trivial $n$-qubit Pauli observable $Q$.
Next, the procedure walks through $j \in \text{supp}(z(P))$, and iteratively sets $Q \leftarrow Q~(C^\dagger Z_j C)$, where the Pauli preimage $C^\dagger Z_j C$ is read out from $j$th column of $M(C)$, that is $C$ stored using the data structure \cref{def:desc-clifford}.
Then this process is continued for $X$, by walking through $j \in \text{supp}(x(P))$, and iteratively setting $Q \leftarrow Q~(C^\dagger X_j C)$, where the Pauli preimage $C^\dagger X_j C$ is read out from $j+n$th column of $M(C)$.
The resulting Pauli $Q = C^\dagger P C$ is returned as the output of the algorithm.
For both $X$ and $Z$, each of the $O(|P|)$ Pauli multiplications takes $O(n)$ standard \runtime{} and $O(1)$ bit-string \runtime, such that the overall algorithm \runtime{} is $O(n|P|)$ or $O(|P|)$ in the standard or bit-string cost models respectively. 

Next we consider the \texttt{disentangle} procedure, which is expressed only in terms of asterisked procedures (which are referenced here by their mathematical action) in \cref{tab:data-structure-requirements}, with no explicit dependence on the underlying data structures.

\begin{procedure}[\texttt{disentangle}]\label{alg:disentanlgle}
\begin{algorithmic}[1]
\Blank
\Require $n$-qubit Clifford unitary $C$, integer $j \in [n]$ such that $Z_j C |0^n\rangle = C |0^n\rangle$.
\Ensure Replace $C$ by $C'\otimes_j I$ such that $(C'\otimes_j I)|0^n\rangle = C |0^n\rangle$.
\Statex[-1] $\triangleright$ First modify Clifford $C$ such that $C|0^n\rangle$ is unchanged, and $C Z_j C^\dagger = Z_j$.
\State Let $j'$ be the index of the first non-zero entry of $z(C^\dagger Z_j C)$. \label{line:ds:non-zero-entry}
\State $C \leftarrow \Lambda(CX_{j'}C^\dagger,  C Z_{j'} C^\dagger ~ Z_j) ~ C ~ \mathrm{SWAP}_{j,j'}$ .  \label{line:ds:z-image-fix}
\Statex[-1] $\triangleright$ Second modify Clifford $C$ such that $C|0^n\rangle$ is unchanged, and $C X_j C^\dagger = X_j$.
\If{bit $j$ of $z(C^\dagger X_j C)$ is zero} \label{line:ds:condition}
\State $C \leftarrow \Lambda(C Z_j C^\dagger, X_j ~ C X_j C^\dagger) ~ C$. \label{line:ds:x-image-fix-1}
\Else
\State $C \leftarrow \Lambda(C Z_j C^\dagger, X_j ~ C ~iZ_j X_j ~ C^\dagger)~ e^{i \frac{\pi}{4} C Z_j C^\dagger} ~ C$. \label{line:ds:x-image-fix-2}
\EndIf
\end{algorithmic}
\end{procedure}

We provide the correctness proof of \cref{alg:disentanlgle} in the appendix in~\cref{prop:disentangle}.
In what follows, we argue that its \runtime{} is $O(n^2)$ and its bit-string \runtime{} is $O(n)$ (as in \cref{tab:data-structure-requirements}) and also highlight the correspondence between the mathematical action and names of procedures.
In \cref{line:ds:non-zero-entry} the \runtime{} of \texttt{preimage} is $O(n)$ and finding a non-zero entry is $O(n)$.
In \cref{line:ds:z-image-fix} the \runtime{} of \texttt{right\_mul\_swap} is $O(n)$, two uses of \texttt{image} has \runtime{} $O(n)$,
\texttt{prod\_pauli} has \runtime{} $O(n)$ and \texttt{left\_mul\_ctrl\_pauli} has \runtime{} $O(n^2)$.
The \runtime{} of the first part of the algorithm is dominated by \texttt{left\_mul\_ctrl\_pauli} and is $O(n^2)$.
Similarly, the \runtime{} of the second part of the algorithm is dominated by \texttt{left\_mul\_ctrl\_pauli}, \texttt{left\_mul\_ctrl\_exp} and is $O(n^2)$.
A similar analysis shows that the bit-string \runtime{} is $O(n)$.

Our choice of the Clifford unitary data structure has two important features. 
The first feature is the preimage $C^\dagger P C$ calculation of a Pauli operator $P$ with bit-string \runtime{} $O(|P|)$. 
This is ensured by storing the preimage signs and storing the matrix $M(C)$ in column-major order.
This choice is motivated by the fact that \texttt{preimage} is a very common procedure in the algorithms considered in this paper.
The second feature is the fast calculation of the inverse of a Clifford unitary $C$.
This is ensured by storing both images and preimages. 
If only images or preimages are stored,
the inverse would require $O(n^\omega)$ (and would require $O(n^3)$ if we had chosen a Hermitian rather than power encoding).
Keeping standard and \bitstring{} \runtime{}s low requires a careful design of procedures that update the data structure for a Clifford unitary $C$.
This is illustrated by the pseudo-code of \texttt{left\_mul\_exp} and \texttt{left\_mul\_ctrl\_pauli} procedures 
discussed in \cref{app:pauli-and-clifford-manipulation}.

\subsection{Outcome-specific stabilizer simulation}
\label{sec:stab-simulation}

It is well known that stabilizer circuits can be efficiently simulated, including the original Gottesman-Knill simulation algorithm~\cite{Gottesman1998}, the improved version of this algorithm Gottesman-Aaronson~\cite{AaronsonGottesman2004}, and further practical improvements in~\cite{Gidney2021}. 
A precise statement of the problem that these existing stabilizer circuit algorithms solve is stated in \cref{def:specific-outcome-stab-sim}.

\begin{problem}[Outcome-specific stabilizer circuit simulation]
\label{def:specific-outcome-stab-sim}
Consider any stabilizer circuit with no input qubits and a length-$n_M$ outcome vector, along with any bit-string $\tilde v \in \{ 0,1\}^{n_M}$.
Find the vector of non-zero conditional probabilities $\overrightarrow{p} \in \{1,1/2\}^{n_M}$,
a Clifford unitary $\Co$ and
the unique outcome vector $v\in \{ 0,1\}^{n_M}$
which satisfy the following properties:
\begin{itemize}[noitemsep]
    \item for each $\overrightarrow{p_l} = 1/2$, the corresponding element $v_l$ is equal to $\tilde v_l$,
    \item $\overrightarrow{p_l}$ is the probability of obtaining outcome $v_l$ given 
previous outcomes $v_1,\ldots,v_{l-1}$,
    \item $\Co|0^{n(\Co)}\rangle$ is the output state of the circuit given the outcome $v$.
\end{itemize} 
\end{problem}

The fact that the conditional probabilities are always either 1/2 or 1 is a consequence of the fact that at each point in time the state is a stabilizer state, and measuring any Pauli operator either has a definite or a uniformly random outcome.
The fact that only those bits of the input \bitstring{} $\tilde v$ which correspond to conditional probabilities of 1/2 correspond to observed measurement outcomes in the \bitstring{} $v$ is then required to ensure that $v$ occurs with non-zero probability.
Note that any algorithm that solves this outcome-specific stabilizer circuit simulation problem can be used to sample faithfully from the quantum circuit, by choosing the input \bitstring{} $\tilde v$ uniformly  at random.

\begin{figure*}[p]
\begin{algorithm}[\texttt{Outcome-specific stabilizer circuit simulation}]
\label{alg:outcome-specific-stab-sim} 
\begin{algorithmic}[1]
\Blank
\Input 
\begin{itemize}[noitemsep,topsep=0pt]
\item a stabilizer circuit $\mathcal{C}$ with no input qubits, $n$ output qubits, and a length-$n_M$ outcome vector,
\item a length-$n_M$ vector $\tilde v$. 
\end{itemize}
\Output 
\begin{itemize}[noitemsep,topsep=0pt]
    \item a vector $\overrightarrow{p} \in \{1,1/2\}^{n_M}$ of conditional probabilities,
    \item a Clifford unitary $\Co$,
    \item and length-$n_M$ outcome vector $v$,
\end{itemize}
that satisfy the conditions of~\cref{def:specific-outcome-stab-sim}.

\State Initialize an empty vector $\overrightarrow{p}$, a zero-qubit Clifford unitary $\Co$, and a \emphspecific{vector $v = 0^{n_M}$}.

\For{$g$ in $\mathcal{C}$}

\hrulefill\Comment{allocation}
\If{$g$ allocates qubit $j$,\label{line:outcome-specific-allocate}}
  \State replace $\Co \leftarrow \Co \otimes_j I_2$.
\ElsIf{$g$ deallocates qubit $j$, where \emphspecific{$Z_j \Co |0\rangle = \Co |0\rangle$} ,\label{line:outcome-specific-deallocate}}
  \State \texttt{disentangle}(\Co,j),
  \State remove qubit $j$ from $\Co$.
\ElsIf{$g$ allocates a random bit\label{line:outcome-specific-random}}
  \State append $1/2$ to $\overrightarrow{p}$,
  \State\emphspecific{replace $v_{l} \leftarrow \tilde{v}_{l}$, where $l = |\overrightarrow{p}|$}.

\hrulefill\Comment{unitaries}
\ElsIf{$g$ is a unitary $U$\label{line:outcome-specific-unitary}} 
  \State replace $\Co \leftarrow U \Co$.
\ElsIf{$g$ applies a Pauli unitary $P$ if $\langle c\rangle = c_0$, \\$\quad\quad\quad\quad\quad$ where $\langle c\rangle$ is the parity of outcomes indicated by $c\in \mathbb{F}_2^{n_M}$,\label{line:outcome-specific-conditional}}
  \State replace $\Co \leftarrow P^{c_0 + 1} \Co$ \emphspecific{ followed by $\Co \leftarrow P^{\ip{c}} \Co$}.

\hrulefill\Comment{measurements}
\ElsIf{\label{line:outcome-specific-fast-measure}$g$ measures Pauli $P$ given a hint Pauli $P'$ such that $\comm{P,P'}=1$, \\$\quad\quad\quad\quad\quad$ and  $P'\Co|0^{n(\Co)}\rangle = \pm \Co|0^{n(\Co)}\rangle$}
  \State find $b'$ and $\alpha$ such that preimage $ \Co^\dagger P'\Co = \alpha Z^{b'}$,
  \State replace $\Co\leftarrow e^{(\alpha\pi /4) PP'} \Co$,
  \State allocate a random bit according to \cref{line:outcome-specific-random},
  \State \emphspecific{apply $P'$ conditioned on the value of the random bit}.
\ElsIf{$g$ measures Pauli $P$,}
\State find the preimage $Q = \Co^\dagger P \Co$. 
\If{$x(Q)=0$ (deterministic measurement)\label{line:outcome-specific-deterministic-measure}} 
  \State append $1$ to $\overrightarrow{p}$,
  \State\emphspecific{$v_l \leftarrow s(Q)/2$, where $l = |\overrightarrow{p}|$}.
\Else{ (uniformly random measurement)\label{line:outcome-specific-random-measure}}
  \State let $j$ be the position of first non-zero bit of $x(Q)$, 
  \State find image $P' = \Co Z_j \Co^\dagger$ (which anticommutes with $P$), \label{line:outcome-specific-image}
  \State measure $P$ with assertion that $P'$ is a hint Pauli according to \cref{line:outcome-specific-fast-measure}.
\EndIf
\EndIf
\EndFor
\State \Return vector $\overrightarrow{p}$, Clifford unitary $\Co$, vector $v_0$
\end{algorithmic}
\end{algorithm}
\end{figure*}

In what follows, we provide \cref{alg:outcome-specific-stab-sim} that solves the outcome-specific stabilizer circuit simulation problem.
Our algorithm uses some key procedures listed in \cref{tab:data-structure-requirements}.
It is also desirable to solve a stronger stabilizer circuit simulation problem which accounts for all possible measurement outcomes.
At first, such a problem may seem like it could not have an efficient solution; a naive specification of the output of the simulation would involve enumeration of all possible sequences of outcomes, which can grow exponentially with the circuit size.
Nonetheless, in \cref{sec:stab-simulation-all} we leverage a general form (introduced in \cref{sec:stabilizer-circuit-general-form}) to define an efficiently specified outcome-complete stabilizer circuit simulation problem, and provide an algorithm to efficiently solve that stronger simulation problem.

Our presentation~(\cref{alg:outcome-specific-stab-sim}) of the well-known outcome-specific stabilizer algorithm serves two goals.
First it clarifies how to use the Clifford unitary data structure as a ``black box'', including how the simulation complexity follows from the complexity of the procedures for Clifford and Pauli unitaries manipulation in \cref{tab:data-structure-requirements}.
Second, using this presentation, it is easy to see the differences and similarities between this outcome-specific and 
the outcome-complete stabilizer simulation algorithm we provide later. 
The differences between outcome-specific and 
the outcome-complete algorithms are indicated by highlights.

We conclude with a few remarks about correctness and \runtime{} of our presentation of the outcome-specific stabilizer simulation \cref{alg:outcome-specific-stab-sim}.
First note that the preimage calculation $\Co^\dagger P \Co$ is commonly used throughout the algorithm. 
The only place where the image calculation $\Co P \Co^\dagger$ is explicitly used is in~\cref{line:outcome-specific-image}.
However, \cref{line:outcome-specific-fast-measure} where the image is used requires knowing the image only up to a sign.
For this reason, in practice, one can use Clifford unitary data structure without image signs as in~\cite{Gidney2021}.
Another place where the image calculation is used is within procedure \hyperref[alg:disentanlgle]{$\texttt{disentangle}$}, but the number of
image calculations associated with it can also be reduced as we discuss next.

The \runtime{} of this simulation algorithm (and also the simulation algorithm introduced in \cref{sec:stab-simulation-all}) can be significantly reduced by reducing the number of qubit allocations and deallocations.
First it is beneficial to calculate the maximum number of qubits needed and allocate all of them in the beginning to reduce allocation costs.
Second we can keep track of which qubits were deallocated, but keep themaround and use them next time a qubit allocation is requested.
It is then sufficient to apply any deallocations at the very end of the simulation. 
We also discuss a more efficient way of deallocating sets of qubits together in~\cref{app:bulk-deallocation-of-qubits}.

In our presentation of \cref{alg:outcome-specific-stab-sim} we see that a non-deterministic measurement during stabilizer simulation reduces to the application of a Pauli exponential $e^{i \pi Q/4}$, which follows from \cref{prop:measure-as-exp}.
This ensures that we can simulate these measurements treating the Clifford unitary which tracks the state as a ``black box'' by referring only to the basic procedures in \cref{tab:data-structure-requirements}, and crucially without making any reference to the Clifford unitary's underlying data structure.
Note that having a `hint' Pauli $P'$, which anticommutes with the Pauli $P$ being measured and which has the system's state as an eigenstate (which triggers \cref{line:outcome-specific-fast-measure}), can reduce the \runtime{} to simulate a measurement. 
This improvement comes into play because the exponential $e^{(\alpha\pi /4) PP'}$ which must be applied to simulate the measurement can be applied more efficiently if $PP'$ is sparse, and when no hint Pauli $P'$ is provided, that which is found by the algorithm is typically dense. 
For example, if both $P$ and $P'$ have constant weight, then such a measurement can be simulated in $O(n)$ with a hint Pauli, compared with $O(n^2)$ without a hint.
For the simulation of deterministic measurements, we achieve the same $O(n)$ \runtime{} as is achieved in~\cite{Gidney2021}.

The \runtime{} of the outcome-specific simulation algorithm per Clifford operation is stated and compared to 
the \runtime{} of the outcome-complete simulation algorithm in~\cref{tab:outcome-complete-simulation}. 
It is informative to state the simulation \runtime{} per operation because the simulation proceeds through the circuit 
operation by operation. This way the reader can easily derive the \runtime{} of the simulation for specific applications,
taking into account the number of occurrences of each type of operation in the circuit.
\newpage
\section{General form for stabilizer circuits}
\label{sec:stabilizer-circuit-general-form}

The goal of this section is 
to demonstrate that the action of any stabilizer circuit is equivalent to that of a circuit of the form shown in \cref{fig:general-form}, with equivalence as in \cref{def:instrument-equality}.
Having a simple general form circuit that can capture the action of any stabilizer circuit forms a useful tool that will later allow us to build efficient algorithms to check circuit equivalence, analyze the logical action of circuits on information stored in a stabilizer code and more.
Note that general form stabilizer circuits are not unique, since distinct general form stabilizer circuits can have equivalent action. 
However, in \cref{sec:equality-of-general-forms} we present an efficient algorithm for checking the equivalence of two general forms.

\begin{figure}[h]
\centering
\loadfig{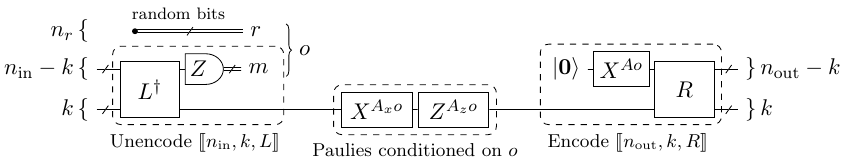}
\caption[General form]{
\label{fig:general-form}
A general form stabilizer circuit with $\ki$ input qubits and $\ko$ 
output qubits and $n_{r}$ classical random bits.
$\Ci$ and $\Co$ are the \emph{left} and \emph{right Clifford unitaries}, $A$, $A_x$, $A_z$ are the \emph{condition matrices} over $\f_2$, 
$r$ is an $n_{r}$-dimensional \emph{random bits vector} over $\f_2$, $k$ is the number of \emph{inner qubits}.
$A_x$ and $A_z$ are $k \times n_O$ matrices, while $A$ is an $(\ko - k) \times n_O$ matrix, where $n_O= (\ki-k + n_{r})$.
We call the length-$(\ki-k)$ vector $m$ the \emph{measurement outcome vector},
and the length-$n_O$ vector $o = r \oplus m$
the circuit's \emph{outcome vector}.
Applying $X^{v_x},Z^{v_z}$ with $k$-bit vectors $v_x = A_x o, v_z = A_z o$ to a $k$-qubit register corresponds to applying $X_j$ conditioned on $(v_x)_j = 1$ and $Z_j$ conditioned on $(v_z)_j = 1$ for $j \in [k]$.
}
\end{figure}

To describe the relation between a stabilizer circuit's outcome vector and an equivalent general from circuit's outcome vector we will need the following definition:
\begin{definition}[Split (reduced) echelon form]
\label{def:split-echelon-form}
A matrix $M$ of size $m\times n$ is in \emph{$(n_r,n_m)$-split echelon form} if $n_r + n_m = n$, $M$ has full column rank and 
the following conditions hold~(\cref{fig:outcome-mapping-structure}): 
\begin{itemize}[noitemsep]
    \item the \emph{left part} $M^T_{[n_r]}$\footnote{$M^T_{[n_r]}$ refers to the first $n_r$ rows of $M^T$, as opposed to the transposed of the first $n_r$ rows of $M$.} of $M$  is in reduced row echelon form,
    \item the \emph{right part} $M^{T}_{[n_r+1,n_r + n_m]}$ of $M$ is in row echelon form,
    \item for any index $j$ of a leading column of $M^T_{[n_r]}$ row $M_{j,[n_r+1,n_r + n_m]}$ is equal to $0^{n_m}$.
\end{itemize}
Given a matrix $M$ which is in $(n_r,n_m)$-split echelon form, if the right part $M^{T}_{[n_r+1,n_r + n_m]}$ is in reduced row echelon form, we further say that $M$ is in \emph{$(n_r,n_m)$-split reduced echelon form}.
\end{definition}

\begin{theorem}[General form stabilizer circuit]
\label{thm:general-form}
Given a stabilizer circuit $\mathcal{C}$(\cref{def:stabilizer-circuit}),
an equivalent
stabilizer circuit $\mathcal{C}_\mathrm{gen}$ can be written in the general form in \cref{fig:general-form}.
The outcome vector $v$ of the circuit $\mathcal{C}$ is related to the outcome vector $o$ of the general form circuit $\mathcal{C}_\mathrm{gen}$ by an invertible affine function. 
That is, 
there there exist a vector $v_0$ and a matrix $M$ such that when the relation $v = v_0 + Mo$ is satisfied, the circuits have the same action on any input state.
Additionally, $M$ is in $(n_r,n_m)$-split echelon form, where $n_r$ is the number of random bits of $\mathcal{C}_\mathrm{gen}$
and $n_m$ is the size of the measurement outcome vector of $\mathcal{C}_\mathrm{gen}$,
The entries of $v_0$ corresponding to the row rank profile (see \cref{sec:linear-algebra}) of the left part of $M$ are zero.
\end{theorem}
\begin{proof}
The result follows by induction on the elementary operations contained in 
the stabilizer circuit. The base case is the empty stabilizer circuit, which is trivially in the general form. 
The induction step follows from~\cref{lem:general-form-induction} below.
\end{proof}

\begin{figure}
    \centering
    \includegraphics[scale=0.5]{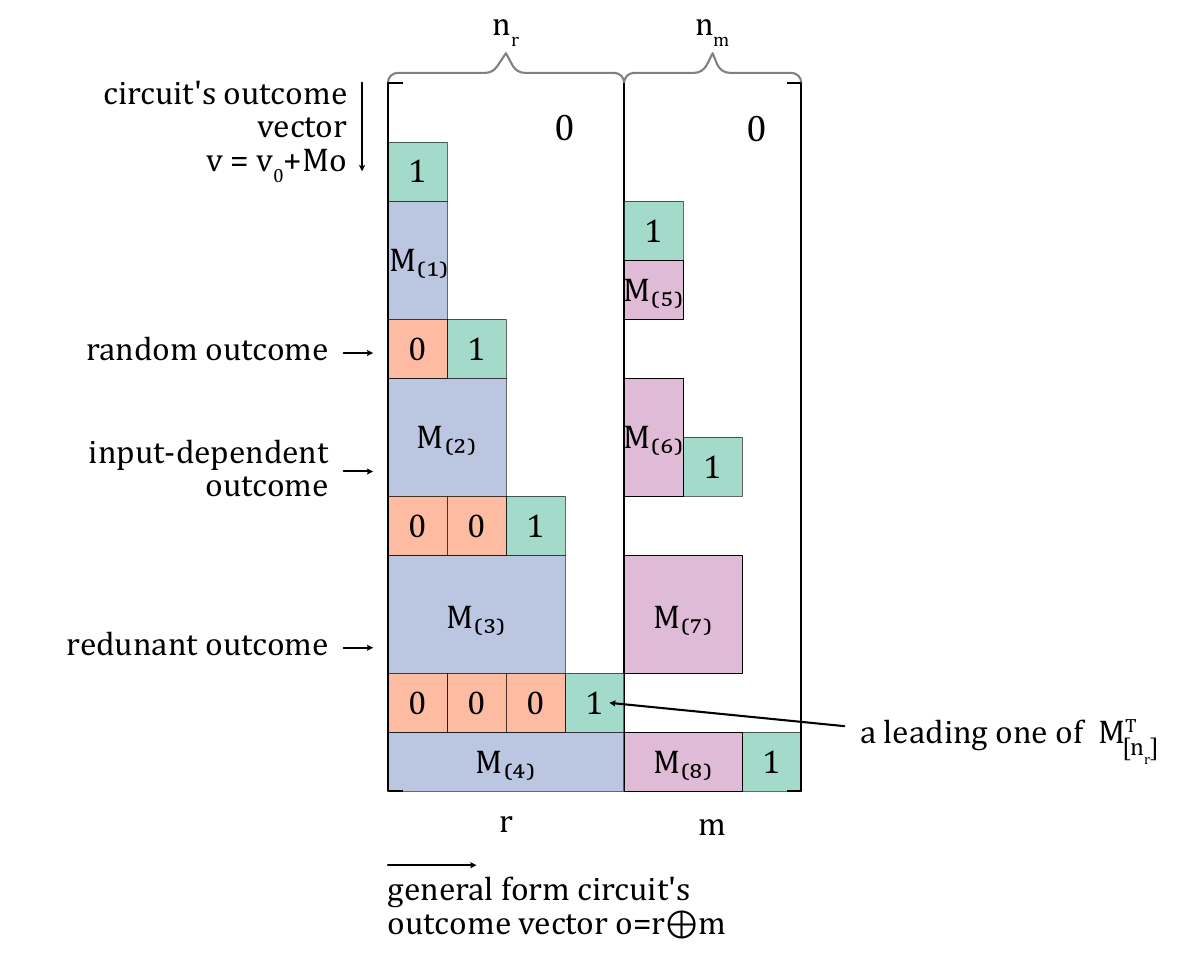}
    \caption[Split echelon form]{Matrix $M$ in $(n_r,n_m)$-split echelon form~(\cref{def:split-echelon-form}), which relates a stabilizer circuit $\mathcal{C}$ outcome vector $v$ to
    the general form circuit $\mathcal{C}_\mathrm{gen}$ outcome vector $o$ as described in \cref{thm:general-form}.
    We label rows corresponding to different kinds of stabilizer circuit outcomes.
    The indices of rows in the left (right) block containing a leading one, marked in green, form the row rank profile of the left (right) part of $M$ and correspond to random (input-dependent) outcomes.
    The remaining rows correspond to redundant outcomes.
    }
    \label{fig:outcome-mapping-structure}
\end{figure}

The structure of the matrix $M$ which appears in the outcome relation between a circuit and its equivalent general form circuit in \cref{fig:outcome-mapping-structure} and \cref{thm:general-form} reflect
the three kinds of stabilizer circuit outcomes which we defined in \cref{sec:circuits-and-choi-states}. 
(1) If the probability of an outcome being zero (conditioned on the previous circuit outcomes) is $\nicefrac{1}{2}$ for all input states, then the outcome is a \emph{random outcome}.
(2) If the probability of an outcome being zero (conditioned on the previous circuit outcomes) depends only on the input state, then it is an \emph{input-dependent outcome}.
(3) If the probability of an outcome being zero (conditioned on the previous circuit outcomes) is either zero or one for all input states, then it is a \emph{redundant outcome}.
Any row of $M$ that contains a leading one of $M^T_{[n_r]}$, i.e., which is in the row rank profile (see \cref{sec:linear-algebra}) of the left part of $M$, corresponds to a random outcome.
Any row of $M$ that contains a leading one of $M^{T}_{[n_r+1,n_r + n_m],J}$, i.e., which is in the row rank profile of the right part of $M$, corresponds to an input-dependent outcome.
The remaining rows correspond to redundant outcomes.

\begin{lemma}[Induction step for \cref{thm:general-form}]
\label{lem:general-form-induction}
Let $\mathcal{C}$ be a stabilizer circuit~(\cref{def:stabilizer-circuit}) which is equivalent to a general form circuit.
More specifically, suppose that $\mathcal{C}$ with outcome vector $v$ has the same action on any input state as a general form circuit $\mathcal{C}_\text{gen}$ with outcome vector $o$, where $v = v_0 + M o$ for some vector $v_0$ and matrix $M$,
and where the sets of outcomes of $\mathcal{C}_\text{gen}$ and $\mathcal{C}$ are $\f_2^{n_O}$ and $v_0 + M \f_2^{n_O}$ respectively. 

Then the circuit $\mathcal{C}' = \mathcal{C} \circ g$, consisting of $\mathcal{C}$ followed 
by a stabilizer operation $g$~(\cref{def:stabilizer-circuit}), is also equivalent to a general form circuit.
In particular, $\mathcal{C}'$ with outcome vector $v'$ has the same action on any input state as a general form circuit $\mathcal{C}'_\text{gen}$ with outcome vector $o'$, where $v' = v'_0 + M' o'$ for some vector $v_0'$ and matrix $M'$,
and the sets of outcomes of $\mathcal{C}'_\text{gen},\mathcal{C}'$ are $\f_2^{n'_O}$, $v'_0 + M' \f_2^{n'_O}$.
Moreover the stabilizer group of $\comm{\ni,k,\Ci}$ is a subgroup of the stabilizer group of $\comm{\ni,k',\Ci'}$.
Additionally, the split echelon form property of $M$ is inherited by $M'$.
Specifically, $M'$ is in $(n'_r,n'_m)$-split echelon form, where $n'_r$ and $n'_m$ are the number of random bits of $\mathcal{C}'_\mathrm{gen}$
and the size of the measurement outcome vector of $\mathcal{C}_\mathrm{gen}$ 
when 
$M$ is in $(n_r,n_m)$-split echelon form, where $n_r$ and $n_m$ are the number of random bits of $\mathcal{C}_\mathrm{gen}$
and the size of the measurement outcome vector of $\mathcal{C}_\mathrm{gen}$.
The entries of $v'_0$ corresponding to the row rank profile of the left part of $M'$ are zero 
when entries of $v_0$ corresponding to the row rank profile of the left part of $M$ are zero.
\end{lemma}

\begin{proof}
Recall that according to \cref{def:stabilizer-circuit}, operation $g$ can be one of the following:
\begin{itemize}[noitemsep]
\item[\textbf{(A)}] Clifford or Pauli unitary,
\item[\textbf{(B)}] Allocation of a qubit initialized to the zero state,
\item[\textbf{(C)}] Deallocation of a qubit in the zero state,
\item[\textbf{(D)}] Pauli unitary conditioned on the parity of a set of measurement outcomes and classical random bits,
\item[\textbf{(E)}] Allocation of a classical random bit distributed as a fair coin,
\item[\textbf{(F)}] Pauli measurement.
\end{itemize}
We will prove that the lemma holds in each of these cases.
We have reordered the cases from \cref{def:stabilizer-circuit} for convenience.
Our proof strategy is to use the fact that the action of the circuit $\mathcal{C}' = \mathcal{C} \circ g$ is equivalent to the circuit $\mathcal{C}_\text{gen} \circ g$ (where $\mathcal{C}_\text{gen}$ is a general form circuit equivalent to $\mathcal{C}$)
and to transform $\mathcal{C}_\text{gen} \circ g$ into a general form circuit via a series of simple action-preserving circuit transformations. 
In cases \textbf{(A)} - \textbf{(D)}, appending $g$ to the circuit does not modify the outcome vector and nor do the transformations that take $\mathcal{C}_\text{gen} \circ g$ to the general form circuit $\mathcal{C}'_\text{gen}$, implying that $M' = M$ and $m = m'$.
In cases \textbf{(A)} - \textbf{(E)}, we will see that, given a general form circuit $\mathcal{C}_\text{gen}$ for $\mathcal{C}$ as described in the lemma, we can construct a general form circuit $\mathcal{C}'_\text{gen}$ for $\mathcal{C}'$ such that the left Cliffords of the two general forms are equal, i.e., $\Ci = \Ci'$, and therefore so too are the stabilizer groups of $\comm{\ni,k,\Ci},\comm{\ni,k',\Ci'}$.
In what follows, we attach the prime symbol ``$\,'\,$'' to all of the parameters of the general form circuit $\mathcal{C}'_\text{gen}$, and when constructing it starting from $\mathcal{C}_\text{gen}$, we explicitly describe how primed objects are obtained, omitting those which are the same for both $\mathcal{C}_\text{gen}$ an $\mathcal{C}'_\text{gen}$.

\textbf{(A)} Operation $g$ is a Clifford or Pauli unitary $U$. 
In this case, the transformation consists of removing $\Co$ and $U$ at the end of $\mathcal{C}_\text{gen} \circ g$ and adding the unitary $\Co' = U \Co$ to form $\mathcal{C}'_\text{gen}$. 

\textbf{(B)} Operation $g$ is the allocation of a qubit initialized to the zero state.
Consider the case when we allocate a qubit which is placed before the first qubit, and initialize it to the zero state. 
The general case reduces to this case by applying SWAP gates, which is a Clifford unitary considered in \textbf{(A)}.
In this case, the transformation simply involves removing the Clifford $\Co$ at the end of $\mathcal{C}_\text{gen} \circ g$ and adding the unitary $\Co' = I \otimes \Co$, and defining $A'$ by appending a row of zeros at the beginning of the rows of $A$ to form a general form circuit with $\ko' = \ko + 1$.

\textbf{(C)} Operation $g$ is the deallocation of a qubit in the zero state.
Consider the case when we deallocate the first qubit, which is assumed to be in the zero state prior to the application of $g$. 
The general case again reduces to this case by applying SWAP gates.
Our strategy is to show that we can modify $\Co$ and $A$ in $\mathcal{C}_\text{gen}$ without changing the circuit's action such that $A$ has a zero first row and $\Co$ maps states $\ket{0}\otimes \ket{\phi}$ to states $\ket{0}\otimes \ket{\psi}$. 
When this is the case, we can form $\mathcal{C}'_\text{gen}$ by setting $A'$ to be $A$ without the first row, 
and obtaining $\Co'$ from $\Co$ defined by $(\Co' \ket{\psi}) \otimes \ket{0} = \Co (\ket{\psi}\otimes \ket{0})$.
The unitary $\Co'$ defined in this way has a Choi state that is a stabilizer state and is therefore a Clifford unitary according to
\href{https://arxiv.org/pdf/1904.01124.pdf\#lem.A.13}{Lemma~A.13} in \cite{Beverland2020}. 
Moreover, the unitary $\Co'$ can be efficiently computed given its Choi state as discussed in \cref{sec:bipartite}.

First we specify the modifications of $A$ and $\Co$.
The fact that the qubit is guaranteed to be in the zero state prior to the application of $g$ implies that $Z_1$ stabilizes the output state of $\mathcal{C}$ (and therefore also the output state of $\mathcal{C}_\text{gen}$).
Given that $Z_1$ stabilizes the output of $\mathcal{C}_\text{gen}$,
it must be that $Z_1$ is in the stabilizer group of $\comm{\no,k,\Co}$ and as such there must be some length-$(\ko - k)$ bit string $b$ such that $Z_1 = \Co (Z^b \otimes I_k) \Co^\dagger$. 
Let $E$ be an invertible $(\ko - k)$ by $(\ko - k)$ matrix that maps the first basis vector $e_1$ to $b$.
$\Co$ is modified to $\Co (U_{E^{-T}} \otimes I_k)$ 
(where $E^{-T}$ is the inverse transpose of $E$ and $U_{E^{-T}}$ is the unitary defined by $U_{E^{-T}}\ket{a} = \ket{E^{-T}a}$ for every length-$(\ko - k)$ bit string $a$) and $A$ is modified to $E^T A$.

Next we confirm that the action of $\mathcal{C}_\text{gen}$ on the input state is preserved by this modification of $A$ and $\Co$.
Indeed, insert $U^\dagger_{E^{-T}} U_{E^{-T}}$ in the original general form circuit following $X^{Ao}$,
and insert $U_{E^{-T}}$ before $X^{Ao}$. 
Applying $U_{E^{-T}}$ to the all zero state leaves it unchanged, 
$U^\dagger_{E^{-T}} U_{E^{-T}} = I$, and 
$U^\dagger_{E^{-T}} X^{Ar} U_{E^{-T}} = X^{E^T A r}$ (using \cref{prop:linear-reversible-clifford}), 
so the circuits are equivalent.

Lastly, we confirm that the modified $A$ and $\Co$ have the claimed form.
Indeed, $Z_1$ conjugated by $U_{E^{-T}}$ is equal to $Z^{E e_1} = Z^b$ (using \cref{prop:linear-reversible-clifford}),
and so $\Co U_{E^{-T}}$ maps $Z_1$ to $Z_1$, which implies that it maps $\ket{0}\otimes \ket{\phi}$ to states $\ket{0}\otimes \ket{\psi}$ as required.
The first row of $E^T A$ must be zero since otherwise there would be some outcome vector $o$ for which the the first qubit is not output in state $\ket{0}$. 
This completes the deallocation case.

\textbf{(D)} Operation $g$ is a Pauli unitary conditioned on the parity of a subset of outcomes of the outcome vector $v$ of circuit $\mathcal{C}$.
Specifically, $g$ is specified by a Pauli $P$, a bit string $c$ with the same length as $v$ and a bit $c_0$, such that $P$ is applied when $\ip{v,c} = c_0$. 
Expressed in terms of the outcome vector $o$ of the general form circuit $\mathcal{C}_\text{gen}$, the parity $\ip{v,c}$ is $\ip{Mo + v_0, c} = \ip{v_0,c} + \ip{ o, M^T c }$. 
The action of $\mathcal{C} \circ g$ is equivalent to applying $\Co^\dagger P \Co$ before $\Co$
in $\mathcal{C}_\text{gen}$, 
conditioned on the parity $\ip{o, M^T c}$ being equal to $\ip{v_0,c} + c_0$. 
We take the case of $\ip{v_0,c} + c_0 = 1$ (the case of $\ip{v_0,c}  + c_0 = 0$ is reduced to this by replacing $\Co$ with $P \Co$ -- crucially this replacement has no dependence on the outcome vector).
The conditional Pauli can be written as $(\Co^\dagger P \Co)^{\ip{o,M^T c}}$, which we will now show can be represented by modifying the matrices $A,A_x,A_z$. 
This follows from the observation that $u \ip {v,w}= (u v^T) w$, where $u$, $v$ and $w$ are column vectors (with $v$ and $w$ having the same dimension), and where $u v^T$ is a matrix obtained by the outer product of $u$ and $v$.
Indeed, if we write $(\Co^\dagger P \Co)$ as $X^{a_x \oplus b_x} Z^{a_z \oplus b_z}$ for $a_x,a_z$ being $\ko - k$ dimensional vectors 
over $\f_2$ and $b_x, b_z$ being $k$ dimensional vectors over $\f_2$, then we have
$$
(\Co^\dagger P \Co)^{\ip{o,M^T c}} =
(X^{(a_x c^T M) o} \otimes X^{(b_x c^T M)o})
(Z^{(a_z c^T M) o} \otimes Z^{(b_z c^T M)o}).
$$
Therefore, $A' = A + a_x c^T M$, $A'_z = A_z + b_z c^T M$, 
and $A'_x = A_x + b_x c^T M$. 
Note that $Z^{a_z}$ can be commuted past $X^{Ax}$ in the general form and so is applied to the all zero state and can be removed.

\textbf{(E)} Operation $g$ is the allocation of a classical random bit distributed as a fair coin.
Starting with $\mathcal{C}_\text{gen}$, we set $n'_r=n_r+1$ , and add the additional random bit to the end of $r$.
The matrices $A',A'_x,A'_z$ can be obtained from matrices $A,A_x,A_z$ by inserting a zero column after column $n_r$. 
The matrix $M'$ is obtained from $M$
by adding a zero column after column $n_r$ and appending row $e_{n'_r}$, where $e_{n'_r}$ is $n'_r$-th standard basis vector.
The vector $v'_0$ is obtained from $v_0$ by appending a $0$.
Note that $M'$ and $,v'_0$ have all the required properties when $M$ and $v_0$ do.

\textbf{(F)} Operation $g$ is a measurement of a Pauli operator $P$, which we separate into three cases:
\begin{itemize}[noitemsep]
    \item[\textbf{(F1)}] $\pm P$ belongs to the stabilizer group of $\comm{\no,k,\Co}$ (redundant outcome).
    \item[\textbf{(F2)}] $\pm P$ does not belong to the stabilizer group of $\comm{\no,k,\Co}$, but commutes with all of its elements (input-dependent outcome).
    \item[\textbf{(F3)}] $\pm P$ anti-commutes with at least one element of the stabilizer group of $\comm{\no,k,\Co}$ (random outcome).
\end{itemize}
Let the preimage be $Q = \Co^\dagger P \Co$.

In case \textbf{(F1)}, $Q$ must be equal to $(-1)^{b_0} Z^b \otimes I_k$ for 
a non-zero vector $b$ over $\f_2$ of dimension $\ko - k$, where $b_0 \in \{0,1\}$. 
In this case, the general form circuits $\mathcal{C}_\text{gen}$ and $\mathcal{C}'_\text{gen}$ are the same, 
however the parameters $v_0',M'$ of the outcome maps differ from $v_0,M$. 
The vector $v_0'$ equals $v_0$ with $b_0$ appended to the end, 
and the matrix $M'$ is equal to $M$ with a row $A^T b$ appended to the end,
because the result of measuring $P$ is equal to $b_0 + \ip{b,Ao} = b_0 + \ip{A^T b, o}$. 
Note that $M'$ has all the required properties when $M$ does.

In case \textbf{(F2)}, the operator $Q$ must be equal to $Z^b \otimes Q'$,
where $Q'$ is a non-trivial $k$-qubit Pauli operator and 
$b$ is a vector over $\mathbb{F}_2$ of dimension $\ko - k$. 
The action of $\mathcal{C}'$ is equivalent to measuring $Q'$ after applying $X^{A_x o}$ in $\mathcal{C}_\text{gen}$, with outcome equal to the outcome of measuring $P$ plus $\ip{Ao,b}$.
Let $C'$ be any $k$-qubit Clifford unitary such that $Q' = C' Z_1 (C')^\dagger$. 
This means that measuring $Q'$ is equivalent to applying $(C')^\dagger$, measuring $Z_1$, and then applying $C'$.
This equivalence implies that the general form circuit $\mathcal{C}'_\text{gen}$ has $\Co' = \Co C'$,
$\Ci' = \Ci C'$, and $k' = k-1$.

To derive $A'_x,A'_z,M',A'$ for the general form circuit $\mathcal{C}'_\text{gen}$, we follow several steps.
First, we note that $X^{A_x o} Z^{A_z o}$ followed by the measurement of $Q'$
is equivalent to applying $C'^\dagger$, followed by $X^{\tilde A_x o} Z^{\tilde A_z o}$, followed by
measuring $Z_1$, and finally applying $C'$, for some $\tilde A_x, \tilde A_z$ calculated using the identity:
$$
C' X^{ A_x o} Z^{ A_z o} C'^\dagger \simeq X^{\tilde A_x o} Z^{\tilde A_z o}.
$$
If we move the measurement of $Z_1$ before $X^{\tilde A_x o}$, it will add
$\langle \tilde A_x o, e_1\rangle$ to the measurement outcome.
Therefore, we can conclude that $X^{A_x r} Z^{A_z r}$ followed by the measurement of $Q'$ is equivalent to the following sequence of operations:
\begin{enumerate}[noitemsep]
\item $C'^\dagger$,
\item destructive measurement of $Z_1$ with outcome $o_{0}$,
\item applying $X^{A'_x o} Z^{A'_z o}$ on the last $k-1$ qubits, where $A'_x$, $A'_z$ are $\tilde A_x, \tilde A_z$ with the first row removed,
\item allocating a new first qubit in the zero state,
\item applying $X^{o_0 + \langle\tilde A_x o,e_1\rangle}$ to the allocated qubit,
\item applying $C'$.
\end{enumerate}
The outcome of measuring $Q'$ is $o_0 + \langle\tilde A_x o,e_1\rangle$, and outcome of measuring $Q$ is  $o_0 + \langle\tilde A_x o,e_1\rangle + \langle A o,b\rangle$.
The new outcome vector $o'$ is obtained by concatenating $o$ with $o_0$,
$o' = o \oplus o_0$, $M'$ is obtained by adding a zero column to $M$ and then adding row $(\tilde A_x^T e_1 + A^T b) \oplus 1$, and $v_0' = v \oplus 0$.
Note that $M'$ has all the required properties when $M$ does.
Finally, $A'$ is obtained by adding a zero column to $A$ and then adding row $(\tilde A_x^T e_1) \oplus 1$.
In this case, the stabilizer group of $\comm{\ni,k,\Ci}$ is a proper subgroup of the stabilizer group of $\comm{\ni,k',\Ci'}$.

In case \textbf{(F3)}, let $P'$ be a Pauli operator from the stabilizer group of $\comm{\no,k,\Co}$ that anti-commutes with $P$. 
In this case, the measurement outcome $r_0$ is a uniformly distributed random bit. 
Let $b'_0,b'$ be defined using the preimage of $P'$ as $(-1)^{b'_0} Z^{b'} \otimes I_k = \Co^\dagger P' \Co$.
The output state of $\mathcal{C}$ (which is the same as the output state of $\mathcal{C}_\text{gen}$) is stabilized by $(-1)^{\ip{Ao,b'}} P'$.
Using \cref{prop:measure-as-exp}, we see that measuring $P$ is equivalent to applying $(P')^{r_0 + \ip{A o,b'}} e^{i\frac{\pi}{4}(iP'P)}$.
For this reason, the case \textbf{(F3)} reduces to 
applying unitary $e^{i\frac{\pi}{4}(iP'P)}$ as in case \textbf{(A)}, 
allocating a classical random bit $r_0$ as in case \textbf{(E)}, 
and applying Pauli unitary $P'$ conditioned on the parity of the set of measurement outcomes indicated by $A^T b'$ and the random bit $r_0$ as in case \textbf{(D)}. 
As such the form of $\mathcal{C}'_\text{gen}$ is constructed from $\mathcal{C}_\text{gen}$ using those cases.
In this case, $\comm{\ni,k,\Ci}$ and $\comm{\ni,k',\Ci'}$ are the same.
\end{proof}

We say that a general form circuit $\mathcal{C}_\text{gen}$ with equivalent action to a stabilizer circuit $\mathcal{C}$ is a \emph{general form of} the circuit $\mathcal{C}_\text{gen}$.
Similarly, we say that a general form circuit $\mathcal{C}_\text{gen}$ with equivalent action to a given quantum instrument is a \emph{general form of} that quantum instrument.

The proof of \cref{lem:general-form-induction} can be reformulated into a polynomial time algorithm for computing the general form of a stabilizer circuit, however the step \textbf{(F2)} requires a somewhat involved computation that scales as $O(n^\omega)$ where $n$ is the number of qubits and $\omega = \log_2 7$ in practice~\cite{Albrecht2011}. 
Instead, in \cref{sec:general-form-algo} we show an algorithm that requires at most $O(n^2)$ 
runtime per gate.
Our algorithm takes advantage of the relation between a general form stabilizer circuit and a Choi circuit as shown in \cref{fig:general-form-choi}.

\begin{figure}[h]
\centering
  \loadfig{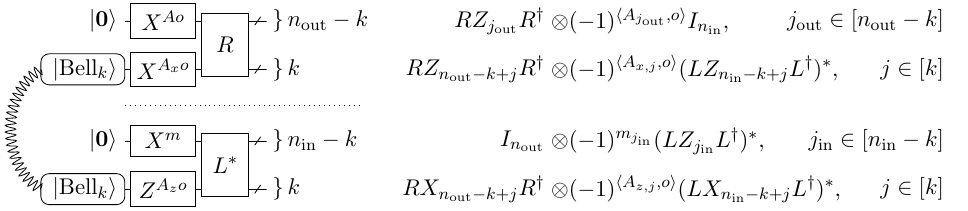}
\caption[General form's Choi state]{
\label{fig:general-form-choi}
The circuit and the stabilizer group generators of the Choi state of the general form of a stabilizer instrument with $\ki$ input qubits, $\ko$ 
output qubits and $n_r$ classical random bits.
Clifford unitaries $\Ci$ and $\Co$ act on $\ki$ and $\ko$ qubits, 
$A$, $A_x$, $A_z$ are matrices over $\f_2$, $r$ is an $n_r$ dimensional random vector over $\f_2$.
The general form outcome vector $o = r \oplus m$ has dimension $n_O= (\ki-k + n_r)$, matrix $A$ is $(\ko - k) \times n_O$ matrix, 
$A_x, A_z$ are $k \times n_O$ matrices.
}
\end{figure}

The circuit in \cref{fig:general-form-choi}, which produces the Choi state of a general form circuit $\mathcal{C}_\text{gen}$ for each outcome vector $o$, is derived in a few steps. 
First, recall that Choi states can be produced using the Choi circuit in \cref{fig:choi-state-circuit}, which involves first initializing $2\ki$ qubits in  $\ket{\mathrm{Bell}_{\ki}}$ 
followed by applying the general form circuit in \cref{fig:general-form} to the top $\ki$ qubits.
Second, we use the identity $(\Ci^\dagger \otimes I_{\ki})\ket{\mathrm{Bell}_{\ki}} = (I_{\ki} \otimes \Ci^\ast)\ket{\mathrm{Bell}_{\ki}}$ 
to remove $\Ci^\dagger$ acting on the first $\ki$ qubits and add $\Ci^\ast$ acting on the last $\ki$ qubits.
Third we observe that destructively measuring qubits $\ki-k$ through $\ki$ with outcome $m$ is the same as 
initialising qubits $\ki + 1,\ldots, 2\ki - k$ in the zero state before applying $X^m$ to them.
Finally, we remove $Z^{Ao}$ acting on qubits $\ki - k$ through $\ki$ and add $Z^{Ao}$ acting on qubits $2\ki - k$ through $2\ki$
right before $\Ci^\ast$. 
This is similar to removing $\Ci^\dagger$ form the first $\ki$ qubits and adding $\Ci^\ast$ to last $\ki$.

We conclude with a couple observations that may aid the conceptual understanding of the general form and the next sections of the paper.
The first observation concerns learning some parameters of a general form of a given circuit $\mathcal{C}$.
Note that for all measurement outcomes, the Choi states in \cref{fig:general-form-choi} are the same up to Pauli corrections, 
or, in other words, up to the signs of the Choi states' stabilizer group generators. 
This leads to a simple way of finding $k$ and also learning $\Ci$ and $\Co$ up to Pauli corrections as follows.
First, compute the circuit's Choi state $\ket{\Psi(o)}$ using \cref{alg:outcome-specific-stab-sim} for any outcome vector $o$. 
Second, apply \cref{thm:bipartition} and \cref{fig:bipartition} to identify the $k$, $\Ci(o)$ and $\Co(o)$ of that Choi state $\ket{\Psi(o)}$. 
One can see that when $o$ is the zero vector, the variables $\Ci(o)$ and $\Co(o)$ coincide with $\Ci$ and $\Co$ in \cref{fig:general-form-choi}, which implies that $\Ci(o)$ and $\Co(o)$ coincide with $\Ci$ and $\Co$ up to Pauli corrections for any $o$.
In \cref{sec:general-form-algo} we extend this approach to form an algorithm to efficiently fully compute a general form of a stabilizer circuit.

The second observation is about interpreting the general form circuit from the perspective of quantum error correction.
As already highlighted in \cref{fig:general-form} the general form circuit consists of three steps.
The first step measures the generators of a stabilizer code $\comm{\ni,k,\Ci}$, and then unencodes the code that was projected onto into the bottom $k$ qubits.
The second step applies Pauli operators to the bottom $k$ qubits conditioned on both the syndrome of $\comm{\ni,k,\Ci}$ measured in the first step and random bits $r$. 
The third step encodes the state of the $k$ bottom qubits into a new stabilize code $\comm{\no,k,\Co}$, with syndrome conditioned on both the syndrome of $\comm{\ni,k,\Ci}$ and the random bits $r$.

In summary, the general form provides a convenient description of the action of a circuit: the input observables that it measures, the output observables that it stabilizes, and a mapping of the remaining observables.
The input observables that the circuit measures are captured by the stabilizer group of $\comm{\ni,k,\Ci}$.
The output observables that the circuit stabilizes are captured by the stabilizer group of $\comm{\no,k,\Co}$ up to a sign.
The mapping of the remaining observable is captured by the left and right Clifford unitaries $\Ci$ and $\Co$ up to signs as:
$$
 \Ci Z_{\ni - k + j} \Ci^\dagger \rightarrow \pm \Co Z_{\no - k + j} \Co^\dagger,~~\Ci X_{\ni - k + j} \Ci^\dagger \rightarrow \pm \Co X_{\no - k + j} \Co^\dagger,~~j\in[k].
$$
The condition matrices $A,A_x,A_z$ and $M$ capture the dependence of the signs of the observables of the output state 
on the circuit outcomes.

\newpage
\section{Outcome-complete stabilizer circuit simulation}
\label{sec:stab-simulation-all}

In \cref{sec:stab-simulation} we provided an algorithm, similar to those already known in the literature~\cite{AaronsonGottesman2004}, 
to solve the stabilizer circuit simulation problem given a specific outcome vector.
Here we consider a stronger stabilizer circuit simulation problem which accounts for all possible outcome vectors.
To efficiently encode the output of this outcome-complete stabilizer simulation problem \cref{def:outcome-complete-stab-sim} (which naively would involve enumeration of an exponential number of possible outcome vectors) we leverage the general form in \cref{sec:stabilizer-circuit-general-form}.

\begin{problem}[Outcome-complete stabilizer circuit simulation]
\label{def:outcome-complete-stab-sim}
Consider any stabilizer circuit with no input qubits, $n$ output qubits, and a length-$n_M$ outcome vector. 
Find the vector of non-zero conditional probabilities $\overrightarrow{p} \in \{1,1/2\}^{n_M}$,
a Clifford unitary $\Co$,
matrices $A$ and $M$ and a vector $v_0$ with entries in $\f_2$
that satisfy the following properties:
\begin{itemize}[noitemsep]
    \item each possible outcome vector $v$ is an element of the set $\{M r : r \in \{0,1\}^{n_r} \}$, where $n_r =|\{ \overrightarrow{p}_k =1/2 : k \in [n_M] \}|$, 
    \item for any outcome vector, $\overrightarrow{p}_j$ is the probability of obtaining outcome $v_j$ given previous outcomes $v_1,\ldots,v_{j-1}$,
    \item $\Co | A r \rangle $ is the output state of the circuit given the outcome vector $v = v_0 + M r$.
\end{itemize} 
\end{problem}

This outcome-complete \cref{def:outcome-complete-stab-sim} differs from the outcome-specific \cref{def:specific-outcome-stab-sim} in that the random outcomes of the circuit are no longer selected ahead of time.  
Instead, those outcomes are represented by a bit-vector $r$ of arbitrary value.
Note that when a circuit has no input qubits ($\ki = 0$), its general form (\cref{fig:general-form}) simplifies to include only random bits $r$, state preparation based on those bits $X^{Ar}|0^n\rangle = |Ar\rangle$, and a subsequent Clifford unitary $\Co$, which sets the form of the output of \cref{def:outcome-complete-stab-sim}.
Before we consider algorithms to solve \cref{def:outcome-complete-stab-sim}, we discuss various aspects of the problem structure and explain how a solution to \cref{def:outcome-complete-stab-sim} for a given circuit can be used to solve any problem instance of \cref{def:specific-outcome-stab-sim} for the same circuit.

A number of degrees of freedom exist for solutions to \cref{def:outcome-complete-stab-sim} arising from the fact that the general form itself is not unique for a given circuit.
First we describe a degree of freedom is for the matrix $M$.
Note that given an invertible $n_r \times n_r$ matrix $B$, we can write set $\{ r : r \in \{0,1\}^{n_r} \}$ as $\{ Br' : r' \in \{0,1\}^{n_r} \}$.
Therefore for any such $B$, matrices $M' = MB$, $A' = AB$ and $\Co' = \Co$ satisfy the conditions in \cref{def:outcome-complete-stab-sim} when matrices $M$, $A$ and $\Co$ do.
Moreover, for any matrix $M^T$ there exists an invertible matrix $B$ such that $(M')^T = B^T M^T$ is in reduced row echelon form.
This allows us, without loss of generality, to focus on solutions to \cref{def:specific-outcome-stab-sim} for which the transpose of the matrix $M$ is in reduced row echelon form, which simplifies a number of linear-algebraic computations involving $M$.
Another way to see that it is sufficient to consider $M^T$ in reduced row echelon form, is to recall that \cref{thm:general-form} implies that $M$ is in $(n_r,0)$-split echelon form.

It is informative to consider the structure of the matrix $M$ in the solution to~\cref{def:outcome-complete-stab-sim}, where $M^T$ is in row-echelon form, in more detail (see \cref{fig:reduce-row-echelon-form}).
The number of rows and columns of $M$ is $n_M$ and $n_r$ respectively.
Since $M$ has full column rank (as shown in \cref{lem:general-form-induction}), for each $j \in [n_r]$ there is a row in $M$ with the form $0^{j-1}10^{n_r-j}$.
These leading-zero rows appear in order of increasing $j$, and the sequence $(l_1,\ldots,l_{n_r} )$ of row indices is the row rank profile of $M$.

\begin{figure}[ht]
    \centering
    \includegraphics[scale=0.5]{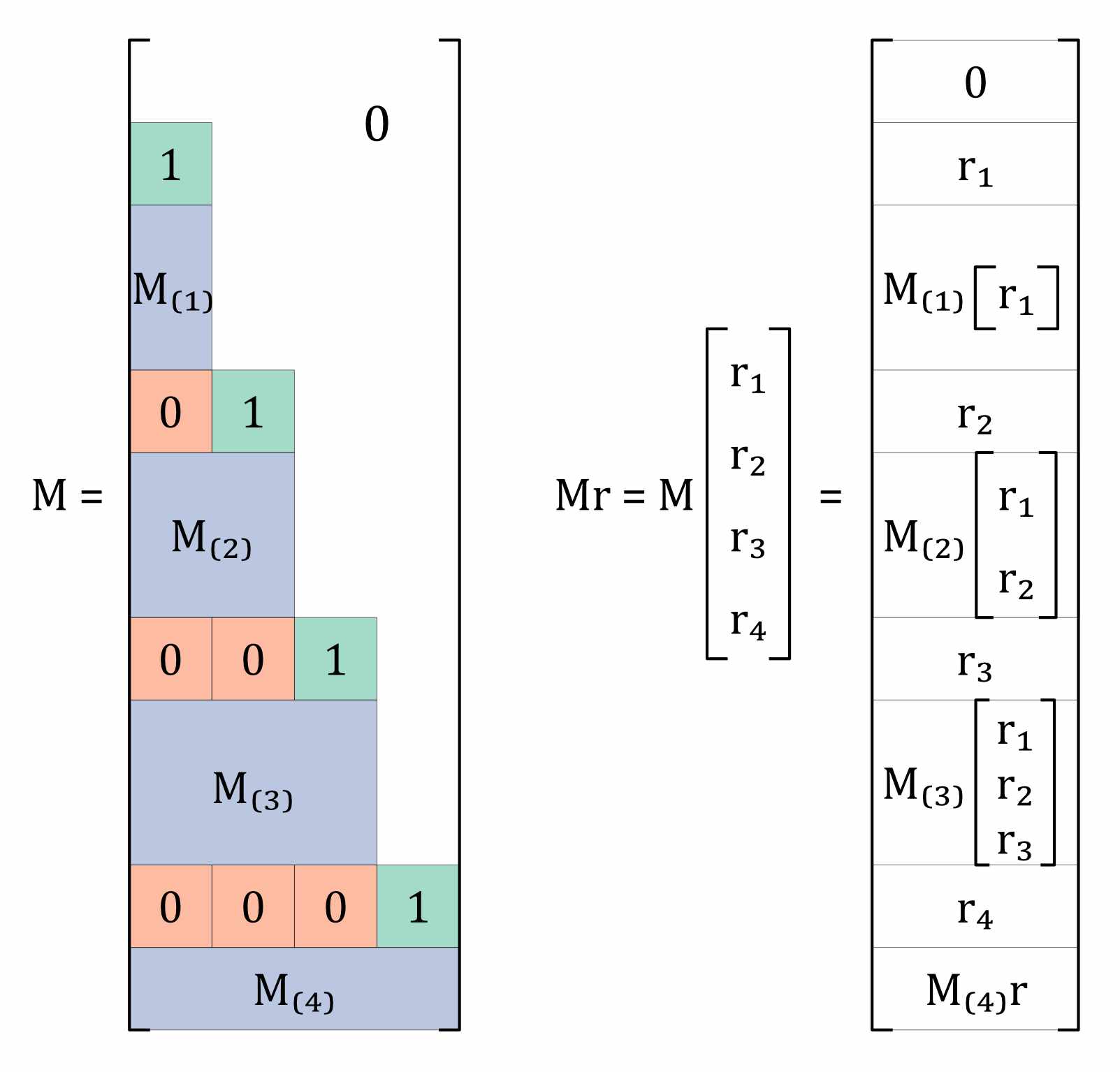}
    \caption[Outcome mapping example for outcome-complete simulation]{
    An illustration of an $n_M \times n_r$ matrix $M$ that specifies the output of an instance of the outcome-complete \cref{def:outcome-complete-stab-sim}, where $M^T$ is in reduced row-echelon form.
    For any such matrix $M$, for each $j \in [n_r]$ the $l_j$th row of $M$ is $0^{j-1}10^{n_r-j}$, where $l_j$ is the $j$th entry in the row rank profile of $M$.
    This structure ensures that when $M$ acts upon any length-$n_M$ \bitstring{} $r$, the $j$th entry of $r$ appears as the $l_j$th entry of $(Mr)$.
    }
    \label{fig:reduce-row-echelon-form}
\end{figure}

Next we see that the vector $\overrightarrow{p}$ in the solution to~\cref{def:outcome-complete-stab-sim} is redundant and can be constructed using the matrix $M$.
Specifically the entry $\overrightarrow{p}_l$ is equal to 1/2 for each $l$ in the row rank profile of $M$, and is equal to 1 for each $l$ not in the row rank profile of $M$. 
To see this, note that $r \in \{0,1\}^{n_r}$ is a uniformly random vector and 
consider the probability $\overrightarrow{p}_l$ of obtaining the outcome $v_l$ given previous outcomes $v_1,\ldots,v_{l-1}$, given the equation $v = v_0 + Mr$.
We see that $\overrightarrow{p}_l = 1/2$ when $l$ in the row rank profile of $M$ since the $l$th row of $M$ then has the form $0^{j-1}10^{n_r-j}$, and therefore the entry $v_l$ is set by $r_j$ (which is uniformly random), since $v_l = (v_0)_l+r_j$~(see~\cref{fig:reduce-row-echelon-form})
and $r_j$ did not appear in $v_1,\ldots,v_{l-1}$.
When $l$ not in the row rank profile of $M$ on the other hand, the $l$th row of $M$ is not equal to $0^{j-1}10^{n_r-j}$, and the entry $v_l$ equals $(v_0)_l$ plus a linear function of entries of $r$ that have already appeared in $v_1,\ldots,v_{l-1}$~(see~\cref{fig:reduce-row-echelon-form}), and therefore since $\overrightarrow{p}_l$ is the \textit{conditional} probability of $v_l$ given $v_1,\ldots,v_{l-1}$, we see that it must be $\overrightarrow{p}_l = 1$.

Another freedom in the solution to~\cref{def:outcome-complete-stab-sim} allows us to assume without loss of generality that the vector $(v_0)_l=0$ for $l$ in the row rank profile of $M$.
We can make this assumption because one can always replace $r$ by $r'=r+r_\text{shift}$ and $v_0$ by $v_0'=v_0+Mr_\text{shift}$ in $v=Mr +v_0$, giving $v = Mr'+v_0'$ for any $r_\text{shift}$.
Choosing $(r_\text{shift})_j= (v_0)_{l_j}$, where $l_j$ is the $j$th entry in row rank profile of $M$, we see that $(v_0')_l =0$ for all $l$ in the row rank profile of $M$.

We conclude with an explanation of how to solve the outcome-specific simulation~\cref{def:specific-outcome-stab-sim} using 
a solution to outcome-complete simulation \cref{def:outcome-complete-stab-sim}. 
As before, we assume that $M^T$ is in reduced row echelon form.
We first select the $r$ that corresponds to the specific outcome for \cref{def:specific-outcome-stab-sim} as follows.
For each $l$ such that the $l$th row of $M$ is $0^{j-1}10^{n_r-j}$, set $r_j = \tilde v_l + (v_0)_l$.
The solution to outcome specific simulation problem then has conditional probabilities $\overrightarrow{p}$, Clifford unitary $\Co X^{Ar}$ and outcome vector $v = Mr + v_0$.
This reduction from the outcome-complete simulation \cref{def:outcome-complete-stab-sim} to the outcome-specific simulation~\cref{def:specific-outcome-stab-sim} is 
instructive for understanding the reduction in the opposite direction discussed in \cref{sec:outcome-compelete-using-outcome-specific}.

\subsection{Solution with outcome-specific algorithm as a subroutine}
\label{sec:outcome-compelete-using-outcome-specific}

One approach to solve the outcome-complete \cref{def:outcome-complete-stab-sim} is~\cref{alg:outcome-complete-stab-sim-via-os} which calls the outcome-specific \cref{alg:outcome-specific-stab-sim} $n_r + 1$ times. 

\begin{figure*}[h]
\begin{algorithm}[\texttt{Outcome-complete via outcome-specific stabilizer simulation}] \label{alg:outcome-complete-stab-sim-via-os}
\begin{algorithmic}[1]
\Statex[-1] \Input A stabilizer circuit $\mathcal{C}$ with no input qubits, $n$ output qubits, and $\nM$ outcomes. 
\Output 
\begin{itemize}[noitemsep,topsep=0pt]
    \item a vector $\overrightarrow{p} \in \{1,1/2\}^{n_M}$ of conditional probabilities,
    \item a Clifford unitary $\Co$,
    \item matrices $A, M$ and vector $v_0$ with entries in $\f_2$,
\end{itemize}
that satisfy conditions of~\cref{def:outcome-complete-stab-sim}, additionally $M^T$ is in reduced row echelon form
and entries of $v_0$ corresponding to the row rank profile of $M$ are zero.
\State Let $\overrightarrow{p}^{(0)}$, $\Co^{(0)}$, $v^{(0)}$ be results of \hyperref[def:specific-outcome-stab-sim]{outcome-specific stabilizer simulation} of $\mathcal{C}$ with $\tilde v = 0^{n_M}$
\State Set $\Co \leftarrow \Co^{(0)}$, $v_0 \leftarrow v^{(0)}$, $ \overrightarrow{p} \leftarrow \overrightarrow{p}^{(0)}$, $j \leftarrow 1$
\For{$l \in [n_M]$ such that $\overrightarrow{p}_l=1/2$} \Comment $n_r$ loop iterations
    \State Let $\overrightarrow{p}^{(j)}$, $\Co^{(j)}$, $v^{(j)}$ be results of \hyperref[def:specific-outcome-stab-sim]{outcome-specific stabilizer simulation} of $\mathcal{C}$ with $\tilde v = e_l$
    \State Set column $j$ of $M$ to $v^{(j)} + v_0$
    \State Find \bitstring{} $a$ from equality $\ket{a} = \Co^\dagger \Co^{(j)}\ket{0^n}$ and set column $j$ of $A$ to $a$
    \State $j \leftarrow j + 1$ \Comment Ensures $M_l = 0^{j-1} 1 0^{n_r - j}$
\EndFor
\State \Return vector $\overrightarrow{p}$, Clifford unitary $\Co$,  matrices $A, M$ and vector $v_0$
\end{algorithmic}
\end{algorithm}
\end{figure*}

\cref{alg:outcome-complete-stab-sim-via-os} constructs a solution to \cref{def:outcome-complete-stab-sim} defined by $M$, $v_0$, $\Co$ and $A$, where $M^T$ is in reduced row echelon form, and $(v_0)_l=0$ for all $l$ in the row rank profile of $M$.
Such a solution always exists, as discussed above.
The algorithm begins by calling the outcome-specific \cref{alg:outcome-specific-stab-sim} with $\tilde v = 0^{n_M}$, obtaining a Clifford $\Co^{(0)}$, an outcome vector $v^{(0)}$
and a conditional probability vector $\overrightarrow{p}^{(0)}$.
From this, one can obtain some information about the outcome-complete solution, for example, we can straightforwardly identify $\overrightarrow{p} = \overrightarrow{p}^{(0)}$ since the conditional probability vector is independent of $\tilde v$.
The vector $\overrightarrow{p}$ allows us to identify the row rank profile of $M$, formed by sorting the set $\{l~|~\overrightarrow{p}_l=1/2,l \in [n_M] \}$ into ascending order, and $n_r$ is the size of the row rank profile of $M$.

Next, note that when we execute outcome-specific simulation algorithm for circuit $\mathcal{C}$ with input vector $\tilde v$, the outcome vector $v$ satisfies equality $v_{l_j} = (\tilde v)_{l_j}$ for $j \in [n_r]$. 
This is because $p_l = 1/2$ if and only if $l = l_j$ for some $j \in [n_r]$.
For input vector $\tilde v$, the outcome-specific simulation algorithm returns unitary $\Co^{(\tilde v)}$ and vector $v^{(\tilde v)}$ such that the output state upon outcome $v^{(\tilde v)}$ is $\Co^{(\tilde v)}|0^n\rangle = \Co|Ar^{(\tilde v)}\rangle$, $r^{(\tilde v)}$ must satisfy equation $v^{(\tilde v)} = Mr^{(\tilde v)} + v_0$ and $v^{(\tilde v)}_{l_j} = (\tilde v)_{l_j}$.
We have $\tilde v_{l_j} = v^{(\tilde v)}_{l_j} = (v_0)_{l_j} + M r^{(\tilde v)}_{l_j} = r^{(\tilde v)}_j$ by definition of $l_j$, because $M^T$ is in reduced row-echelon form and $(v_0)_{l_j}=0$.

Setting $\tilde v = 0$ implies that we compute action corresponding to $r=0$, $\Co^{(0)}|0^n\rangle = \Co|0^n\rangle$.
Setting $\tilde v = e_{l_j}$ gives us $\Co^{(e_{l_j})}|0^n\rangle = \Co|A e_j\rangle$, $v^{(e_{l_j})} = M e_j + v_0$ which allows us to compute columns $j$ of $A$ and $M$.
In total then we have used $n_r$ calls of the outcome-specific algorithm.

\subsection{Solution with explicit outcome-complete algorithm}

\begin{figure*}[p]
\begin{algorithm}[\texttt{Outcome-complete stabilizer circuit simulation}] \label{alg:outcome-complete-stab-sim}
\begin{algorithmic}[1]
\Blank
\Input A stabilizer circuit $\mathcal{C}$ with no input qubits, $n$ output qubits, and $\nM$ outcomes. 
\Output 
\begin{itemize}[noitemsep,topsep=0pt]
    \item a vector $\overrightarrow{p} \in \{1,1/2\}^{n_M}$ of conditional probabilities,
    \item a Clifford unitary $\Co$
    \item matrices $A, M$ and vector $v_0$ with entries in $\f_2$
\end{itemize}
that satisfy conditions of~\cref{def:outcome-complete-stab-sim}, additionally $M^T$ is in reduced row echelon form
and the entries of $v_0$ corresponding to the row rank profile of $M$ are zero.

\State Initialize an empty vector $\overrightarrow{p}$, zero-qubit Clifford unitary $\Co$, dimension-zero matrices $A, M$, vector $v_0$.

\For{$g$ in $\mathcal{C}$}

\hrulefill\Comment{allocation}
\If{$g$ allocates qubit $j$,\label{line:outcome-complete-allocate}}
  \State replace $\Co \leftarrow \Co \otimes_j I_2$,
  \State \algemph{insert a zero row into $A$ after row $j-1$.}
\ElsIf{$g$ deallocates qubit $j$, where \algemph{$Z_j \Co |Ar\rangle = \Co |Ar\rangle$ for all $r$},\label{line:outcome-complete-deallocate}}
  \State \algemph{let $j'$ be the position of the first non-zero bit of $z(\Co^\dagger Z_j \Co)$}
  \State \algemph{swap rows $j,j'$ of $A$ and remove row $j$ from $A$}
  \State \texttt{disentangle}($\Co,j$)
  \State remove qubit $j$ from $\Co$, 
\ElsIf{$g$ allocates a random bit\label{line:outcome-complete-random}}
  \State append $1/2$ to $\overrightarrow{p}$,
  \State \algemph{add a zero column to $A$,} 
  \State \algemph{append zero to the end of $v_0$},
  \State \algemph{add a zero column and row to $M$ and set the last bit in the row to $1$.}

\hrulefill\Comment{unitaries}
\ElsIf{$g$ is a unitary $U$\label{line:outcome-complete-unitary}} 
  \State replace $\Co \leftarrow U \Co$.
\ElsIf{$g$ applies a Pauli unitary $P$ if $\langle c\rangle = c_0$, \\$\quad\quad\quad\quad\quad$ where $\langle c\rangle$ is the parity of outcomes indicated by $c\in \mathbb{F}_2^{n_M}$,\label{line:outcome-complete-conditional}}
  \State replace $\Co \leftarrow P^{c_0 + 1} \Co$ \algemph{followed by $\Co \leftarrow P^{\ip{v_0,c}} \Co$ followed by $a \leftarrow M^T c$}
  \State \algemph{find preimage $\tilde{P}=\Co^\dagger P\Co$ } \label{line:outcome-complete-conditional-r} \Comment{apply $P$ conditioned on bits of $r$ indicated by $a$}
  \State \algemph{for all $j$ such that $x(\tilde{P})_j=1$, replace row $A_{j} \leftarrow A_{j} + a$.}

\hrulefill\Comment{measurements}
\ElsIf{$g$ measures Pauli $P$ given a hint Pauli $P'$ such that $\comm{P,P'}=1$ \\$\quad\quad\quad\quad\quad$ and  $P'\Co|0^{n(\Co)}\rangle = \pm \Co|0^{n(\Co)}\rangle$\label{line:outcome-complete-fast-measure}}
  \State find $b'$ and $\alpha$ such that preimage $ \Co^\dagger P'\Co = \alpha Z^{b'}$,
  \State replace $\Co\leftarrow e^{(\alpha^\ast\pi /4) PP'} \Co$ \algemph{followed by $a \leftarrow A^T b' \oplus 1$ }
  \State allocate a random bit according to \cref{line:outcome-complete-random},
  \State \algemph{apply $P'$ conditioned on bits of $r$ indicated by $a$, according to \cref{line:outcome-complete-conditional-r}.}\label{line:outcome-complete-conditional-in-m}
\ElsIf{$g$ measures Pauli $P$,}
\State find preimage $Q = \Co^\dagger P \Co$. 
\If{$x(Q)=0$ (deterministic measurement)\label{line:outcome-complete-deterministic-measure}} 
  \State append $1$ to $\overrightarrow{p}$,
  \State \algemph{add row $A^T z(Q)$ to $M$ and append $s(Q)/2$ to the end of $v_0$. }
\Else{ (uniformly random measurement)\label{line:outcome-complete-random-measure}}
  \State let $j$ be the position of first non-zero bit of $x(Q)$, 
  \State find image $P' = \Co Z_j \Co^\dagger$ (which anticommutes with $P$).\label{line::outcome-complete-image}
  \State measure $P$ with assertion $P'$ according to \cref{line:outcome-complete-fast-measure}.
\EndIf
\EndIf
\EndFor
\State \Return vector $\overrightarrow{p}$, Clifford unitary $\Co$,  matrices $A, M$ and vector $v_0$
\end{algorithmic}
\end{algorithm}
\end{figure*}

Here we offer an alternative approach to solve the outcome-complete \cref{def:outcome-complete-stab-sim} in the form of \cref{alg:outcome-complete-stab-sim}, which is more efficient than the iterative approach.

The outcome-complete \cref{alg:outcome-complete-stab-sim} is structured similarly to the outcome-specific \cref{alg:outcome-specific-stab-sim}.
Each operation in the circuit is handled sequentially and the quantum state of the system is tracked as the algorithm proceeds.
Most information about the state is encoded using a Clifford $\Co$, which is updated to reflect the corresponding action.
In addition, the state depends on the random outcomes in the circuit -- the effect of all possible outcomes are encoded by the the matrices $A$ and $M$, which are updated incrementally as the algorithm proceeds. 
For qubit and bit allocation, the additional steps serve to properly adjust the dimensions of $A$ and $M$. 
For random bit allocation in the outcome-specific simulation retrieves the value from $\tilde v$, 
while outcome-complete simulation adjusts $M$ so that last entry of $Mr$ is equal to the last entry of $r$ for all $r$.

For qubit deallocation, additional steps make sure that matrix $A$ is consistent with the modification of $\Co$ by the \texttt{disentangle} procedure.
In addition, in outcome complete-simulation it is required that the deallocated qubit is in the zero state prior to deallocation for all possible outcomes.
Unconditional unitary operations are handled identically in both simulation algorithms.
The most significant difference between the outcome-specific \cref{alg:outcome-specific-stab-sim} and the outcome-complete \cref{alg:outcome-complete-stab-sim} is in the handling of a conditional Pauli unitary, \cref{line:outcome-specific-conditional,line:outcome-complete-conditional}, respectively. Note that conditional Paulis are also used to handle measurements with random outcomes.

In the outcome-specific case, values of outcomes on which the Pauli $P$ is conditioned are known and phases of the Clifford $\Co$ can be adjusted by direct application of $P$.  In the outcome-complete case, $P$ is first commuted through $\Co$ by calculating the preimage $\Co^\dagger P \Co$. The $z$-part of this preimage commutes with input $|0^n\rangle$.  The $x$-part of this preimage identifies inputs that can be flipped in order to induce the correct phases on the output.

Both algorithms differentiate between random and deterministic measurements in the same way.
For deterministic measurement, the outcome-complete algorithm determines its dependence on the values of random bits,
while outcome-specific algorithm simply calculates the deterministic outcome value.
The random measurement is reduced to the measurement with a hint in both algorithms.
In outcome-specific case the conditional Pauli operator only depends on the random bit corresponding to the random outcome,
while in the outcome-complete case the conditional Pauli can depend on multiple random bits.
This is because the phase of the hint Pauli can depend on multiple random bits. 
See \cref{prop:measure-as-exp} for more details on conditional Pauli use in the random measurement.
The full set of differences between
\cref{alg:outcome-specific-stab-sim,alg:outcome-complete-stab-sim} is indicated by highlights.

We omit a detailed discussion of the correctness of \cref{alg:outcome-complete-stab-sim} because the algorithm mirrors the proof of 
\cref{lem:general-form-induction} for the special case where the number of input qubits is zero. 
There are two subtle details in the algorithm design that deserve a highlight.
First, in \cref{line:outcome-complete-conditional-in-m} we apply Pauli $P'$ conditioned on bits $r$
indicated by $a$.
This is different from applying $P$ conditioned on the outcome vector bits indicated by $c$ in \cref{line:outcome-complete-conditional}, 
which is reduced to applying $P$ conditioned on bits $r$ indicated by $a$ in \cref{line:outcome-complete-conditional-r}.
Second, we perform swap rows of $A$ during deallocation in \cref{line:outcome-complete-deallocate}. 
The goal here is to make modification of $A$ consistent with the modification of $\Co$ by the procedure \texttt{disentangle}.
Correctness of this step relies on the implementation details of the procedure \texttt{disentangle}
and the step might need to be modified if a different implementation of procedure \texttt{disentangle} is used.
The correctness of this step follows from propagating $X^{Ar}$ (representing current simulation state $\Co X^{Ar}\ket{0^n}$)
through Clifford unitaries applied within procedure \texttt{disentangle} using \cref{prop:pauli-power-conjuagtion-by-ctrl-pauli}.
Finally we note that $M^T$ constructed by \cref{alg:outcome-complete-stab-sim} is in reduced row-echelon normal form, 
which ensures that it has full column rank. This simplifies further linear-algebraic operations with $M$, such as computing left inverse of $M$.

\subsection{Comparison of solutions}

In \cref{tab:outcome-complete-simulation} we summarize the \runtime{} complexities of the outcome-specific \cref{alg:outcome-specific-stab-sim} and the outcome-complete \cref{alg:outcome-complete-stab-sim}.

\begin{table}[htp]
\small{
\centering
    \begingroup
    \setlength{\tabcolsep}{4.5pt}
    \begin{tabular}{|l|r|c|c|r|c|c|}
    \hline 
    \multirow{3}{*}{Stabilizer operation}  & \multicolumn{3}{c|}{ Outcome-complete~(\cref{alg:outcome-complete-stab-sim})} & \multicolumn{3}{c|}{ Outcome-specific~(\cref{alg:outcome-specific-stab-sim})} \\
    \cline{2-7}
     &\parbox[t]{2mm}{\multirow{2}{*}{\rotatebox[origin=c]{-90}{Line}}} & \multirow{2}{*}{Run time} & \bitstring{} & \parbox[t]{2mm}{\multirow{2}{*}{\rotatebox[origin=c]{-90}{Line}}} & \multirow{2}{*}{Run time} & \bitstring{} \\
     & &                           & \runtime{}  &  &                           & \runtime{}\\
    \hline
    \hline 
    Qubit allocation   & \ref{line:outcome-complete-allocate}&  \algemph{$O(\nmax ^2 + \nr \nmax )$} & $O(\nmax )$ 
                       & \ref{line:outcome-specific-allocate}&  \emphspecific{$O(\nmax ^2) $}    & $O(\nmax )$ \\
    \hline 
    Qubit deallocation & \ref{line:outcome-complete-deallocate}& \algemph{$O(\nmax ^2 + \nr \nmax )$} & $O(\nmax )$ 
                       & \ref{line:outcome-specific-deallocate}& \emphspecific{$O(\nmax ^2)$}         & $O(\nmax )$ \\
    \hline
    Random bit & \multirow{2}{*}{\ref{line:outcome-complete-random}} &  \multirow{2}{*}{\algemph{$O(n_M \nmax  + \nr \nmax )$}} & \multirow{2}{*}{\algemph{$O(\nmax )$}}
               & \multirow{2}{*}{\ref{line:outcome-specific-random}} &  \multirow{2}{*}{\emphspecific{$O(1)$}}             & \multirow{2}{*}{\emphspecific{$O(1)$}} \\
    allocation & & & & & & \\
    \hline 
    \hline 
    Unitary $P$, $\exp(P)$ & \multirow{2}{*}{\ref{line:outcome-complete-unitary}} & $O(\nmax |P|)$ & $O(|P|)$
                             & \multirow{2}{*}{\ref{line:outcome-specific-unitary}} & $O(\nmax |P|)$ & $O(|P|)$ \\ 
    \cline{1-1}
    \cline{3-4}
    \cline{6-7}
    Unitary $\Lambda(P,Q)$ & & $O(\nmax (|P|+|Q|))$ & $O(|P|+|Q|)$ & & $O(\nmax(|P|+|Q|))$ & $O(|P|+|Q|)$ \\
    \hline
    Pauli $P$ conditional& \multirow{3}{*}{\ref{line:outcome-complete-conditional}}& \algemph{$O(\nmax |P| + \nr|c|+$} & \algemph{{$O(|c|+$}} 
                             & \multirow{3}{*}{\ref{line:outcome-specific-conditional}}& \multirow{3}{*}{\emphspecific{$O(\nmax|P| + |c|)$}} & \multirow{3}{*}{\emphspecific{$O(|P|+|c|)$}}\\
    on outcomes  & & \algemph{$\nr \nmax )$} & \algemph{$\nmax )$} & & & \\
    indicated by $c$ & & & & & & \\
    \hline
    \hline 
    Measure Pauli $P$ & \multirow{2}{*}{\ref{line:outcome-complete-fast-measure}} & \algemph{$O(\nmax (|P|+|P'|)+$} & \multirow{2}{*}{\algemph{$O(\nmax )$}}
                      & \multirow{2}{*}{\ref{line:outcome-specific-fast-measure}} & \multirow{2}{*}{\emphspecific{$O(\nmax(|P|+|P'|)) $}} & \multirow{2}{*}{\emphspecific{$O(|P|+|P'|)$}}\\
    with a hint Pauli $P'$ & & \algemph{$\nr \nmax )$} & & & & \\
    \hline
    Measure Pauli $P$ & \multirow{2}{*}{\ref{line:outcome-complete-deterministic-measure}}&  \multirow{2}{*}{\algemph{$O(\nmax |P|+\nr \nmax )$}} &  \multirow{2}{*}{\algemph{$O(\nmax )$}}
                      & \multirow{2}{*}{\ref{line:outcome-specific-deterministic-measure}}&  \multirow{2}{*}{\emphspecific{$O(\nmax |P|)$}} &  \multirow{2}{*}{\emphspecific{$O(|P|)$}} \\
    (deterministic)   & & & & & & \\
    \hline
    Measure Pauli $P$ & \multirow{2}{*}{\ref{line:outcome-complete-random-measure}} & \multirow{2}{*}{\algemph{$O(\nmax ^2 + \nr \nmax )$}} & \multirow{2}{*}{$O(\nmax )$} 
                      & \multirow{2}{*}{\ref{line:outcome-specific-random-measure}} & \multirow{2}{*}{\emphspecific{$O(\nmax ^2)$}} & \multirow{2}{*}{$O(\nmax )$} \\
    (random)          & & & & & & \\
    \hline
    \end{tabular}
    \endgroup
}
    \caption[Complexity of outcome-complete stabilizer simulation]{
    Complexity of outcome-complete~(\cref{alg:outcome-complete-stab-sim}) and outcome-specific~(\cref{alg:outcome-specific-stab-sim}) stabilizer simulation for each type of stabilizer operation~(\cref{def:stabilizer-circuit}). 
    The integers $\nmax$, $n_r$ and $n_M$ are the maximum number of qubits, the number of random outcomes, and the total number of outcomes in the circuit respectively.
    We neglect the \runtime{} required to increase the sizes of matrices $A$ and $M$, or of vector $v_0$ in 
    the \runtime{} of measurement operations.  
    }
    \label{tab:outcome-complete-simulation}
\end{table}

In order to take advantage of constant-time bit string operations, we assume that each row of matrices $A$ and $M$ is represented by a bit string.  
In other words, matrices $A$ and $M$ are contiguous row-major arrays.
To compute matrix-vector products $M^T a$ and $A^T b$ it is sufficient to compute the sum of rows of $M$ and $A$ indicated by $a$ and $b$. 
This choice justifies the \bitruntime{} of conditional Paulis and measurements.

This choice of representation has tradeoffs. 
In particular, allocation and deallocation operations are limited by re-writing the corresponding arrays with new dimensions.
A linked-list implementation would yield lower \runtime{} complexities of those operations, but increase the \bitruntime{} complexities of measurement operations.

Workarounds for allocation and deallocation costs are available.
For example, maximum array sizes can be determined by pre-processing the circuit and re-using qubits when possible.
As a practical matter, this can reduce cost by allocating memory as a single step up-front.  
It also reduces the \runtime{} complexity by eliminating the need to re-dimension matrices.
If pre-processing is undesirable, one can use a qubit and random bit allocation strategy similar to those used for contiguous dynamically sized arrays, such as \texttt{std::vector} in C++.
Consequently, we exclude random bit allocation from measurement operation \runtime{}s in \cref{tab:outcome-complete-simulation}.

The \runtime{} per unitary operation is faster than the \runtime{} per measurement in \cref{alg:outcome-complete-stab-sim}.
One potential way to make these \runtimes{} more balanced is to represent the state during outcome complete simulation as 
$X^{A_x r} Z^{A_z r} \Co\ket{0^n}$ as opposed to  $\Co X^{A r} \ket{0^n} = \Co \ket{Ar}$.
We leave the investigation of this option for future work.

For our \runtime{} analysis we rely on both the image $\Co P \Co^\dagger$ and preimage $\Co^\dagger P \Co$ calculation being possible in $O(|P|\nmax)$. 
The only line where we calculate the image is~\cref{line::outcome-complete-image}, where it is sufficient to calculate image up-to a phase because it is used to derive a hint Pauli $P'$.
Omitting the phase calculation can further speed-up the implementation of outcome-complete simulation by a constant factor.

Lastly, we compare the \runtime{} (in terms of the circuit parameters defined in \cref{sec:circuit-parameters}) of the outcome-complete \cref{alg:outcome-complete-stab-sim} with that of \cref{alg:outcome-complete-stab-sim-via-os}, which is built from $n_r+1$ calls of the outcome-specific \cref{alg:outcome-specific-stab-sim}.
Comparing \cref{eq:outcome-complete-runtime,eq:outcome-complete-runtime-reduction} below, 
we find that outcome-complete simulation via~\cref{alg:outcome-complete-stab-sim} saves a factor of $\nr + 1$ over \cref{alg:outcome-complete-stab-sim-via-os} for contributions from unitaries, measurements with random outcomes and qubit deallocations.

The \runtime{} of the outcome-complete simulation \cref{alg:outcome-complete-stab-sim} is 
\begin{equation}\label{eq:outcome-complete-runtime}
O\left(\nmax\cdot\left(\Nu+(\Ncnd+\nM)(\nr+1) + \nmax(\nmax + \nr) \right)\right).
\end{equation}
This formula is a sum of contributions from \cref{tab:outcome-complete-simulation} including:
$O(\nmax \cdot \Nu)$ for unitaries,
$O(\nmax \cdot \Ncnd \cdot (\nr + 1))$ for conditional Paulis, and  $O(\nM(\nmax\nr+\nmax) + \nr\cdot\nmax^2)$ for measurements.
For this analysis we assume that matrices $M$, $A$ and $R$, are pre-allocated with \runtime{} costs $O(\nM \cdot \nr)$, $O(\nmax \cdot \nr)$ and $O(\nM^2)$, which are less than the other stabilizer operation costs.
Additionally we use a qubit deallocation strategy described in~\cref{app:bulk-deallocation-of-qubits} that contributes \runtime{} $O(\nmax^2 \cdot (\nmax + \nr))$.

An analogous analysis of the outcome-specific simulation \cref{alg:outcome-specific-stab-sim} yields a \runtime{} of
\begin{equation}\label{eq:outcome-specific-runtime}
O\left(\nmax\cdot\left(\Nu+\Ncnd+\nM + \nmax(\nmax + \nr) \right)\right).
\end{equation}
Using $n_r+1$ calls of outcome-specific \cref{alg:outcome-specific-stab-sim}, the \runtime{} of \cref{alg:outcome-complete-stab-sim-via-os} is then
\begin{equation}\label{eq:outcome-complete-runtime-reduction}
O\left(\nmax\cdot\left(\Nu(\nr+1)+(\Ncnd+\nM)(\nr+1) + \nmax(\nmax + \nr)(\nr+1) \right)\right).
\end{equation}
\cref{alg:outcome-complete-stab-sim-via-os} includes $n_r$ Clifford unitary multiplications, each with \runtime{} $O(\nmax^3)$. But that
contribution is less than the cost of $n_r+1$ outcome-specific simulations.


\newpage
\section{Computing a general form of a stabilizer circuit}
\label{sec:general-form-algo}

In \cref{thm:general-form} in \cref{sec:stabilizer-circuit-general-form} we showed that every stabilizer circuit, no matter how long and complex, has equivalent action to some circuit of the simple class of \emph{general form circuits} shown in~\cref{fig:general-form}.
In this section we present \cref{alg:general-from-choi}, which efficiently finds an explicit general form circuit with equivalent action for any input stabilizer circuit.
This algorithm is built upon the outcome-complete stabilizer simulation \cref{alg:outcome-complete-stab-sim} presented in the previous section.

Our general form \cref{alg:general-from-choi} produces a general form circuit $\mathcal{C}_\text{gen}$ with equivalent action to the input stabilizer $\mathcal{C}$ in two stages.
First, the algorithm computes a general form circuit $\mathcal{C}'_\text{gen}$ for the Choi circuit $\mathcal{C}'$ of $\mathcal{C}$~(\cref{fig:choi-state-circuit}).
Second, the algorithm computes the general form circuit $\mathcal{C}_\text{gen}$ from the Choi circuit's general form $\mathcal{C}'_\text{gen}$.
The first stage performs an outcome-complete simulation of the Choi circuit $\mathcal{C}'$, which can be done since the Choi circuit has no input qubits.
The description of the circuit can be made more compact by using an outcome compression procedure, which also ensures that the \runtime{} of the second stage depends only on the number of input and output qubits in $\mathcal{C}$, and not on its length.
The second stage relies on a subroutine to compute the bipartite normal form of the family of stabilizer states corresponding to the Choi circuit general form.

In the remainder of this section, we first specify the bipartite normal form \cref{prob:partition} for a family of stabilizer states.
This is a generalization of the bipartite normal form \cref{prob:partition-state} for a single stabilizer state, and its structure is crucial to understand our general form \cref{alg:general-from-choi}.
We then present \cref{alg:general-from-choi}, discuss its correctness and analyze its \runtime{}.
Lastly, in \cref{sec:partition} we provide the procedure which solves \cref{prob:partition}, and in \cref{sec:outcome-compression-map} we provide the compression map subroutine which can be optionally used following \cref{alg:general-from-choi} to make the output more compact.



\begin{problem}[Bipartite normal form of a family of stabilizer states]
\label{prob:partition}
Consider an $n$-qubit Clifford $C$ that defines a family of stabilizer states $\{ C|a\rangle \}_{a \in \{0,1\}^n }$ and an integer $n_1 < n$ that defines a bipartition into the first $n_1$ qubits and the remaining $n_2 = n-n_1$ qubits. 
Find an $n\times n$ matrix $F$ over $\f_2$, an integer $k$ and Clifford unitaries $\Ci,\Co$ such that
the circuit identity in~\cref{fig:bipartition-family-circuit} holds for all $a \in \{0,1\}^n$.
\end{problem}



\begin{figure}[h]
\centering
\loadfig{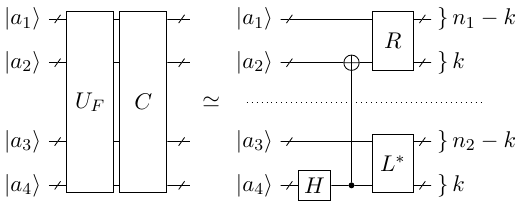}
\caption[Bipartite normal form for a family of stabilizer states]{
For any Clifford unitary $C$ there exist an $n\times n$ invertible matrix $F$ such that
the states produced by applying $C$ to computational basis states $\ket{Fa}$ can instead be produced (up to a global phase) by applying the sequence of smaller unitaries to $\ket{a}$ as shown.
Note that the Hadamards followed by CNOTs produce $k$ Bell states (up to Paulis) across the bipartition $[n_1],[n_1+1,n_1+n_2]$ of the $n_1+n_2$ qubits, and the Clifford unitaries $\Co$, $\Ci^\ast$ which follow act on qubits $[n_1],[n_1+1,n_1+n_2]$.
Also note that this figure coincides with \cref{fig:bipartition} when $a = 0^n$.
}
\label{fig:bipartition-family-circuit}
\end{figure}

Given a procedure that solves \cref{prob:partition}, the following \cref{alg:general-from-choi} finds the general form of a stabilizer circuit $\mathcal{C}$.

\begin{algorithm}[\texttt{Stabilizer circuit general form}]
\label{alg:general-from-choi}
\begin{algorithmic}[1]
\Blank
\Input
A Clifford circuit $\mathcal{C}$ with $\ki$ input qubits, $\ko$ output qubits, $n_M$ outcomes.
\Output
\begin{itemize}
\item A general form $\mathcal{C}_\text{gen}$ of $\mathcal{C}$ (i.e. of the form shown in \cref{fig:general-form}) with left Clifford unitary $\Ci$, right Clifford unitary $\Co$, and condition matrices $A,A_x,A_z$, $n_r$ random bits, $n_m$ measurement outcomes,
\item A matrix $M$ and vector $v_0$ over $\f_2$ defining the outcome relation, that is action of $\mathcal{C}_\text{gen}$ upon outcome $o$ is equivalent to the action of $\mathcal{C}$ upon outcome $v = v_0 + Mo$.
Additionally $M$ is in $(n_r,n_m)$-split reduced echelon form, 
and entries of $v_0$ corresponding to the row rank profiles of the left and right parts of $M$ are zero.
\end{itemize}

\Blank\label{line:choi-circuit-general-form} \hrulefill \Comment Find general form $\mathcal{C}'_\text{gen}$ of the Choi circuit of $\mathcal{C}$
\State\label{line:general-from-choi-circuit}Let $\mathcal{C}'$ be the Choi circuit of $\mathcal{C}$ (see~\cref{sec:circuits-and-choi-states}).
\State\label{line:general-form-simulation}Find the Clifford unitary $C$, matrices $A', M'$ and the vector $v_0$ over $\f_2$ corresponding to outcome-complete simulation of $\mathcal{C}'$ (\cref{alg:outcome-complete-stab-sim}), defining its general form $\mathcal{C}'_{\mathrm{gen}}$.
\Blank\label{line:choi-circuit-general-form-to-general-from-circuit}  \hrulefill \Comment Find a general form $\mathcal{C}_\text{gen}$ of $\mathcal{C}$ from $\mathcal{C}_\text{gen}'$
\State\label{line:general-form-partition}Find the integer $k$, the basis-change matrix $F$ that brings $C\ket{A'r'}$ (for all $r'$) into bipartite normal form with unitaries $\Ci,\Co$ over the qubit bipartition $[\ko],[\ko+1,\ko+\ki]$ (\cref{fig:bipartition-family-circuit}, \cref{prob:partition}, \cref{proc:partition}).
\State\label{line:general-form-basis-change}Let $\tilde C = CU_F$, set $k_1 = \ko - k$, $k_2 = \ki - k$, $n'_r$ to the number of columns of $A'$
\State\label{line:general-form-outcome-relation}Find an $n'_r\times n'_r$ invertible matrix $B$ and a $k_2 \times k_2$ matrix $A_\mu$ such that \Comment \cref{proc:block-reshape}
$$(F^{-1}A')_{[\ko + 1, \ko + k_2]} B = \left( \mathbf{0}_{k_2 \times (n'_r - k_2)} | A_\mu \right),$$
and such that $B$ is in $(n'_r - k_2,k_2)$-split reduced echelon form.
\State\label{line:general-form-left-clifford-update}$\Ci \leftarrow \Ci(U_{A_\mu} \otimes I_k)$
\State\label{line:general-form-matrix-assignment}Assign matrices $A, A_x, A_z, M$ as follows:
\begin{itemize}[noitemsep]
    \item $A = (F^{-1}A' B)_{[k_1]}$. 
    \item $A_x = (F^{-1}A' B)_{[k_1 + 1, \ko]}$.
    \item $A_z = (F^{-1}A' B)_{[\ko + k_2 + 1, \ko + \ki]}$.
    \item $M = M' B$.
\end{itemize}
\State \Return $\Ci,\Co,A,A_x,A_z$ defining $\mathcal{C}_\textrm{gen}$, and $M,v_0$ defining outcome map
\end{algorithmic}
\end{algorithm}

To verify the correctness of \cref{alg:general-from-choi}, it can be interpreted as a sequence of transformations.
In the first stage of the algorithm (consisting of \cref{line:general-from-choi-circuit,line:general-form-simulation}) we construct the Choi circuit $\mathcal{C}'$ of the input circuit and then find a general form $\mathcal{C}'_\text{gen}$ of $\mathcal{C}'$. 
In more detail, \cref{line:general-form-simulation} replaces the Choi circuit $\mathcal{C}'$ of $\mathcal{C}$ with an equivalent general form circuit $\mathcal{C}'_\text{gen}$ that 
allocates qubits $\ket{0^{\ko+\ki}}$, 
allocates random bits $r$,
applies Pauli unitaries $X^{A'r}$ conditioned on $r$, and then applies a Clifford unitary $C$,
by applying the outcome-complete stabilizer simulation~\cref{alg:outcome-complete-stab-sim}.

Now we enter the second stage of the algorithm (consisting of \cref{line:general-form-partition} to \cref{line:general-form-matrix-assignment}), in which the circuit $\mathcal{C}'_\text{gen}$ is further transformed into the Choi circuit $(\mathcal{C}_\text{gen})'$ of the general form (\cref{fig:general-form-choi}), from which the parameters of the desired general form $\mathcal{C}_\text{gen}$ of $\mathcal{C}$ can be read off.
To do so we begin by inserting $I = U_F U_{F^{-1}}$ between $X^{A'r}$ and $C$, after which $U_F$ is absorbed into $C$ (\cref{line:general-form-basis-change}) and $F^{-1}$ is absorbed into $A'$ via  $U_{F^{-1}} X^{A'r}\ket{0^{\ko+\ki}} = |F^{-1}A'r\rangle$ (\cref{line:general-form-outcome-relation}).
The basis change $F$ lets us replace the Clifford unitary $C$ with a circuit that creates $k$ Bell pairs from $\ket{0^{2k}}$ followed by $D\otimes B^\ast$, as in~\cref{fig:bipartition-family-circuit}. 
By the end of \cref{line:general-form-basis-change}, the transformation to $(\mathcal{C}_\text{gen})'$ is nearly complete, except that unitary $X^{F^{-1}A'r}$ is parameterized by $r$, the random outcomes of $\mathcal{C}'$, and not by $o$, the outcomes of the general form $\mathcal{C}_\text{gen}$.
~\cref{line:general-form-outcome-relation} finds a relation between the two, $r = B o$, which is then used to make the necessary adjustments on~\cref{line:general-form-matrix-assignment}. 
The relation matrix $B$ is obtained by using~\cref{proc:block-reshape} to ensure that $M = M'B$ is in $(n'_r-k_2,k_2)$-split reduced echelon form.
It follows from \cref{prop:matrix-structure-propagation} that $M$ is in $(n'_r-k_2,k_2)$-split reduced echelon form since $(M')^T$ is in $(n'_r,0)$-split echelon form and $B$ is in $(n'_r-k_2,k_2)$-split reduced echelon form.
We can successfully find $B$ in \cref{line:general-form-outcome-relation} according to~\cref{prop:relabelling}.
The proposition applies because according to \cref{lem:choi-state-criteria} the set of stabilizer states $\{ \tilde C|F^{-1} A'r'\rangle \}_{r \in \f_2^{n'_r}}$ 
has a certain property, called phase-completeness, defined in \cref{app:mathematical-material}~(\cref{def:phase-complete}).
In \cref{line:general-form-left-clifford-update} we use the identity $U_{A_\mu} X^m U^\dagger_{A_\mu} = X^{A_\mu m}$~(\cref{prop:css-cliiford-action}), note that $U^\dagger_{A_\mu}\ket{0^{k_2}} = \ket{0^{k_2}}$,
and absorb $U_{A_\mu}$ into $\Ci$.
The last circuit transformation step to get to $(\mathcal{C}_\text{gen})'$ is to propagate $X^{A_x o}, X^{A_z o}$ through the unitary that creates $k$ Bell pairs
in~\cref{fig:bipartition-family-circuit}. Then we have $a_1 = Ao$, $a_2 = A_x o$,  $a_3 = {( 0 | I_{\ki - k}) o}$,  $a_4 = A_z o$.
The entries of $v_0$ corresponding to the row rank profile of $M'$ are exactly the entries of $v_0$ corresponding to the row rank profiles of the left and the right parts of $M'$,
and are therefore zero.

We conclude with a \runtime{} analysis. 
In terms of the circuit parameters defined in \cref{sec:circuit-parameters}, the \runtime{} of the general form \cref{alg:general-from-choi} is 
\begin{equation}\label{eq:general-form-runtime}
O\left(\nmax\cdot\left(\Nu+(\Ncnd+\nM)(\nr+\nm + 1) + \nmax(\nmax + \nr)\right)\right).
\end{equation}
This expression is equal to \cref{eq:outcome-complete-runtime} (the \runtime{} of the outcome-complete simulation \cref{alg:outcome-complete-stab-sim}), but with $\nr$ replaced by $\nm + \nr$. 
This is because the \runtime{} of the general form \cref{alg:general-from-choi} is dominated by the outcome-complete simulation of the Choi circuit $\mathcal{C}'$ of the input circuit in \cref{line:general-form-simulation}.
The Choi circuit $\mathcal{C}'$ has most of the same parameters as the input circuit $\mathcal{C}$ except the number of random outcomes $\nr' = \nr + \nm$,
, the number of unitary operations $\Nu' \le \Nu + 2 \nmax$ and the maximum number of qubits used throughout the circuit $\nmax' \le 2\nmax$.

Let us now discuss the \runtime{} for the rest of the steps of \cref{alg:general-from-choi} and establish that they are all less than the \runtime{} of the outcome-complete simulation \cref{line:general-form-simulation} (and can therefore be neglected).
The bipartite normal form computation in \cref{line:general-form-partition} has \runtime{} $O((\ni+\no)^3)$ which is $O(\nmax^3)$~(\cref{eq:partition-runtime}).
The \runtime{} of finding the Clifford $\tilde C$ in \cref{line:general-form-basis-change} and of updating the Clifford $\Ci$ in \cref{line:general-form-left-clifford-update} is $O(\nmax^3)$.
The \runtime{} of computing the outcome relabelling matrix $B$ in \cref{line:general-form-outcome-relation} is $O(\nm(\nr+\nm))$~(\cref{eq:block-reshape-runtime}).
The \runtime{} of updating matrices $A$,$A_x$ and $A_z$ in \cref{line:general-form-matrix-assignment} is $O(\nmax \nm (\nr + 1))$~(\cref{eq:split-reduced-echelon-mult-complexity}),
and the \runtime{} of updating the matrix $M$ is $O(\nM\nm(\nr+1)))$~(\cref{eq:reduced-split-reduced-echelon-mult-complexity}).
The \runtime{} of optionally finding a more compact representation of the general form circuit using \cref{alg:compress-randomness} is $O((\nr+\nmax)\nmax^2)$~(\cref{eq:compression-map-runtime}).
The bound $\nm \le \nmax$ then implies that \runtime{} of the general form algorithm is dominated by the outcome-complete simulation of the Choi circuit $\mathcal{C}'$.

\subsection{Bipartite normal form of a family of stabilizer states}
\label{sec:partition}

Here we introduce the bipartite normal form \cref{proc:partition}, which forms a key subroutine of our stabilizer circuit general form \cref{alg:general-from-choi}.
This procedure solves the bipartite normal form \cref{prob:partition} for a family of stabilizer states, which was stated at the beginning of \cref{sec:general-form-algo} and which generalizes the previously studied bipartite normal form \cref{prob:partition-state} for a single stabilizer state stated in \cref{sec:bipartite}. 

To solve \cref{prob:partition}, the bipartite normal form
\cref{proc:partition} must find a basis-change matrix $F$ that ensures the following properties of the matrix $\tilde C = C U_F$:
\begin{itemize}[noitemsep]
    \item[(A)] images $\tilde C Z_1 \tilde C^\dagger,~\ldots,~\tilde C Z_{n_1-k} \tilde C^\dagger$ are supported on $[n_1]$,
    \item[(B)] images $\tilde C Z_{n_1+1} \tilde C^\dagger,~\ldots,~\tilde C Z_{n-k} \tilde C^\dagger$ are supported on the complement $[n_1 + 1,n]$,
    \item[(C)] images $\tilde C Z_{n_1-k+j} \tilde C^\dagger,\tilde C Z_{n-k+j} \tilde C^\dagger$ for $j \in [k]$ form a symplectic basis when restricted to $[n_1]$, and also form a symplectic basis when restricted to $[n_1 + 1,n]$.
\end{itemize}
To construct such a basis-change matrix $F$, \cref{proc:partition} uses two subroutines.
The first subroutine, \cref{proc:support-subgroup}, helps us find blocks of $F$ that ensure properties (A) and (B) hold, while the second subroutine, \cref{proc:symplectic-basis-for-group}, helps fill in the remaining parts of $F$ to ensure property (C) holds.
As an abstract algorithm,
\cref{proc:support-subgroup} computes the subgroup of a stabilizer group $G = \langle P_1,\ldots,P_m\rangle$ for which the support of every element is contained in some subset of qubits $K$.
We will see later how this helps to ensure (A) and (B) hold within \cref{proc:partition}.

\begin{procedure}[\texttt{support-restricted subgroup}] 
\label{proc:support-subgroup}
\begin{algorithmic}[1]
\Blank
\Input
\begin{itemize}
    \item A sequence of $n$-qubit commuting Pauli unitaries $P_1,\ldots,P_m$,
    \item a set $K \subset [n]$.
\end{itemize}
\Output
An $m'\times m$ full-row-rank matrix $F$ over $\f_2$ such that 
$$
    \langle P_1^{F_{j,1}}\ldots P_m^{F_{j,m}} : j\in [m'] \rangle = \{ P : P \in \langle P_1,\ldots, P_m \rangle,~\mathrm{supp}(P) \subseteq K\}. 
$$
\State Define $K^c = \{j_1, \ldots, j_{|K^c|}\}$, the complement of $K$.
\State\label{line:complement-support}Initialize an $m \times 2|K^c|$ matrix $A$ over $\f_2$ with rows
$$
    A_i = (z(P_i)_{j_1},\ldots,z(P_i)_{j_{|K^c|}},x(P_i)_{j_1},\ldots,x(P_i)_{j_{|K^c|}}), \text{ for } i \in [m].
$$
\State Find $F$ such that rows of $F$ form a basis of kernel of $A^T$ and \Return $F$.
\end{algorithmic}
\end{procedure}

\cref{proc:support-subgroup} works by noting that the desired products of $P_1,\ldots,P_m$ are those that have trivial support on the \emph{complement} of $K$. 
The matrix $A$, constructed on~\cref{line:complement-support}, identifies the support of $G$ on the complement.  
Then, the kernel of $A^T$ corresponds to indicator vectors $v$ for which the product $A^T v = 0$ has trivial support on $K^c$.
That is, the support of $P_1^{v_1}\ldots P_m^{v_m}$ is a subset of $K$.
Using a row reduced echelon form to compute the kernel basis~(\cref{eq:kernel-basis-complexity}), the run-time complexity of~\cref{proc:support-subgroup} is $O(n m^{\omega-1})$.
We call $F$ an \emph{indicator matrix}, and rows of $F$ \emph{indicator vectors}, because they indicate how to construct new Pauli unitaries 
$P_1^{F_{j,1}}\ldots P_m^{F_{j,m}}$
out of Pauli unitaries $P_1,\ldots,P_m$.

Next we provide \cref{proc:symplectic-basis-for-group}, which finds a symplectic basis of a Pauli group. 
We will see later how this helps to ensure property (C) holds within \cref{proc:partition}.

\begin{procedure}[\texttt{symplectic basis}] 
\label{proc:symplectic-basis-for-group}
\begin{algorithmic}[1]
\Blank
\Input Pauli unitaries $P_1,\ldots,P_{2m}$ such that $B_{i,j} = \comm{P_i,P_j}$ 
is an invertible matrix.
\Output
A $2m \times 2m$ matrix $F$ over $\f_2$ such that generators $\{P_1^{F_{j,1}} \ldots P_{2m}^{F_{j,2m}} : j\in [2m]\}$ form a symplectic basis~(\cref{def:symplectic-basis}).
\State Initialize $F = I_{2m}$, and Pauli unitaries $Q_j = P_j$ for $j\in[2m]$.
\For{$k \in [m]$}
  \State\label{line:anticommuting-partner}Find the minimum $j > 2k-1$ such that $\comm{Q_{2k-1}, Q_j} = 1$.
  \State Swap rows $F_{2k}$ and $F_j$, $Q_{2k}$ and $Q_j$
  \For {$i \in [2k+1,2m]$}
    \State\label{line:anticommutation-bits}let $a = \comm{Q_{2k-1}, Q_{i}}$ and $b = \comm{Q_{2k}, Q_{i}}$, 
    \State\label{line:row-adjustment}replace row $F_{i} \leftarrow F_{i} + a \cdot F_{2k-1} + b \cdot F_{2k}$,
    \State\label{line:commutation-adjustment}replace $Q_{i} \leftarrow Q_{i} Q_{2k-1}^a Q_{2k}^b$.
  \EndFor
\EndFor
\State \Return $2m \times 2m$ matrix $F$
\end{algorithmic}
\end{procedure}

\cref{proc:symplectic-basis-for-group} works by sequentially selecting each anticommuting pair (\cref{line:anticommuting-partner}), which always succeeds because we require that matrix $B$ is invertible.
Once a pair is selected, the remaining generators are adjusted in order to commute with both elements of the pair (\cref{line:commutation-adjustment}).
The complexity of \cref{line:commutation-adjustment} (which is $O(n)$ for \runtime{} and $O(1)$ for bit-string \runtime{}) limits the overall complexity, which is $O(nm^2)$ for \runtime{} and $O(m^2)$ for bit-string \runtime{}.  
If the $2m\times 2m$ matrix $B_{i,j} = \comm{P_i,P_j}$ is known ahead of time then \cref{line:anticommuting-partner,line:anticommutation-bits,line:commutation-adjustment} can be replaced with cheaper operations on $B$, reducing the \runtime{} to $O(m^3)$.
Finally, note that it is easy to modify the algorithm to simultaneously compute $F$ and $F^{-1}$.
This is because adding row $2k-1$ to row $i$ of $F$ in \cref{line:row-adjustment} is the same as $F \leftarrow (I + e_i e_j^T)F$,
and so $F^{-1}$ can be updated as $F^{-1} \leftarrow F^{-1}(I + e_i e_j^T)^{-1}$.

With~\cref{proc:support-subgroup} and \cref{proc:symplectic-basis-for-group} in hand we are in position to solve~\cref{prob:partition}.
\cref{proc:partition} outputs a basis-change matrix $F$ for a Clifford $C$ such that the $Z$ images of the Clifford $C U_F$ are partitioned into three subsets with specific properties~(see~\cref{fig:bipartition-family-blocks}).

\begin{procedure}[\texttt{Bipartite normal form of a family of stabilizer states}] 
\label{proc:partition}
\begin{algorithmic}[1]
\Blank
\Input An $n$-qubit Clifford $C$, integer $n_1$.
\Output An $n\times n$ matrix $F$ over $\f_2$, integer $k$ and Clifford unitaries $\Ci$, $\Co$ that satisfy conditions of \cref{prob:partition}.
\State\label{line:partition-supp1}Find a $k_1\times n$ indicator matrix $F_1$ for
$\langle C Z_1 C^\dagger,~\ldots,~C Z_n C^\dagger\rangle$ 
restricted to $[n_1]$ (\cref{proc:support-subgroup}).
\State\label{line:partition-supp2}Find a $k_2 \times n$ indicator matrix $F_2$ for
$\langle C Z_1 C^\dagger,~\ldots,~C Z_n C^\dagger\rangle$ 
restricted to $[n_1+1,n]$ (\cref{proc:support-subgroup}).
\State Set $k = n_1 - k_1$, set $n_2 = n - n_1$.
\State\label{line:partition-full-rank-completion}Let 
$F_{12} = \left(\begin{array}{c} F_1 \\ \hline F_2\end{array}\right)$,
the concatenation of rows from $F_1$ and $F_2$.
\State Complete $F_{12}$ into a full-rank $n\times n$ invertible matrix $\hat F$.
\State\label{line:partition-symplectic-basis}Let $\hat C = C U_{\hat F^{-1}}$, find a matrix $\tilde F$ for transforming
$\hat C Z_{k_1 + k_2 + 1} \hat C^\dagger,~\ldots,~\hat C Z_{n} \hat C^\dagger$ restricted to first $n_1$ qubits
into a symplectic basis (\cref{proc:symplectic-basis-for-group}).
\State\label{line:paritioned-basis} Let $A = ( I_{k_1 + k_2} \oplus \tilde F) \hat F$, set\footnote{We use notation $A_{[k:2:l]}$ for rows $k,k+2,\ldots$ of matrix $A$ } 
$$F = \left(\begin{array}{c} A_{[k_1]} \\ \hline  A_{[n-2k+1:2:n]} \\ \hline  A_{[k_1 +1, k_1 + k_2]} \\ \hline  A_{[n-k+2:2:n]} \end{array}\right)^{-1}.$$
\State\label{line:paritioned-matrix}Set $\tilde C = C U_F$.
\State\label{line:partition-right-clifford}Construct $\Co$ acting on $n_1$ qubits as follows~(see~\cref{fig:bipartition-family-blocks}):
    \begin{enumerate}[noitemsep]
    \item Set images $\Co Z_j \Co^\dagger$ so that $\Co Z_j \Co^\dagger \otimes I_{n_2} = \tilde C Z_j \tilde C^\dagger$, for $j \in [k_1]$.
    \item Set images $\Co Z_{k_1 + j} \Co^\dagger$ to $\tilde C Z_{k_1 + j} \tilde C^\dagger$ restricted to $K$, for $j \in [k]$.
    \item Set images $\Co X_{k_1 + j} \Co^\dagger$ to $\tilde C Z_{n_2 + k_1 + j} \tilde C^\dagger$ restricted to $K$, for $j \in [k]$.
    \item Complete $\Co$ to a full Clifford by setting (arbitrarily) images $\Co X_j \Co^\dagger$ for $j\in[k_1]$.
    \end{enumerate}
\State\label{line:partition-left-clifford}Construct $\Ci$ acting on $n_2$ as follows~(see~\cref{fig:bipartition-family-blocks}):
    \begin{enumerate}[noitemsep]
    \item Set images $\Ci^\ast Z_j (\Ci^\ast)^\dagger$ so that $I_{n_1} \otimes \Ci^\ast Z_{j} (\Ci^\ast)^\dagger  = \tilde C Z_{n_1 + j} \tilde C^\dagger$, for $j \in [k_2]$.
    \item Set images $\Ci^\ast Z_{k_2 + j} (\Ci^\ast)^\dagger$ so that 
    $$\Co Z_{k_1 + j} \Co^\dagger ~\otimes~ \Ci^\ast Z_{k_2 + j} (\Ci^\ast)^\dagger = \tilde C Z_{k_1 + j} \tilde C^\dagger, j \in [k].$$
    \item Set images $\Ci^\ast X_{k_2 + j} (\Ci^\ast)^\dagger$ so that 
    $$\Co X_{k_1 + j} \Co^\dagger ~\otimes~ \Ci^\ast X_{k_2 + j} (\Ci^\ast)^\dagger = \tilde C Z_{n_1 + k_2 + j} \tilde C^\dagger, j \in [k].$$
    \item Complete $\Ci$ to a full Clifford by setting (arbitrarily) images $\Ci^\ast X_j (\Ci^\ast)^\dagger$ for $j\in[k_2]$.
    \end{enumerate}
\State \Return matrix $F$, integer $k$, Clifford unitaries $\Ci,\Co$
\end{algorithmic}
\end{procedure}

\begin{figure}[h]
\centering
 \begin{subfigure}[c]{0.39\textwidth}
    \centering
    \loadfig{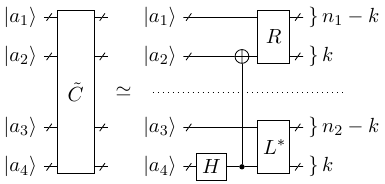}
    \caption{\label{fig:bipartition-family-equality} A Clifford unitary $\tilde C$ such that states $\tilde C \ket{a_1} \otimes \ldots \otimes \ket{a_4}$ can be produced by a simpler circuit on the right (up to a global phase) for some $k,\Ci,\Co$. }
\end{subfigure}
\hfill
\begin{subfigure}[c]{0.57\textwidth}
    \centering
    \includegraphics[width=\textwidth]{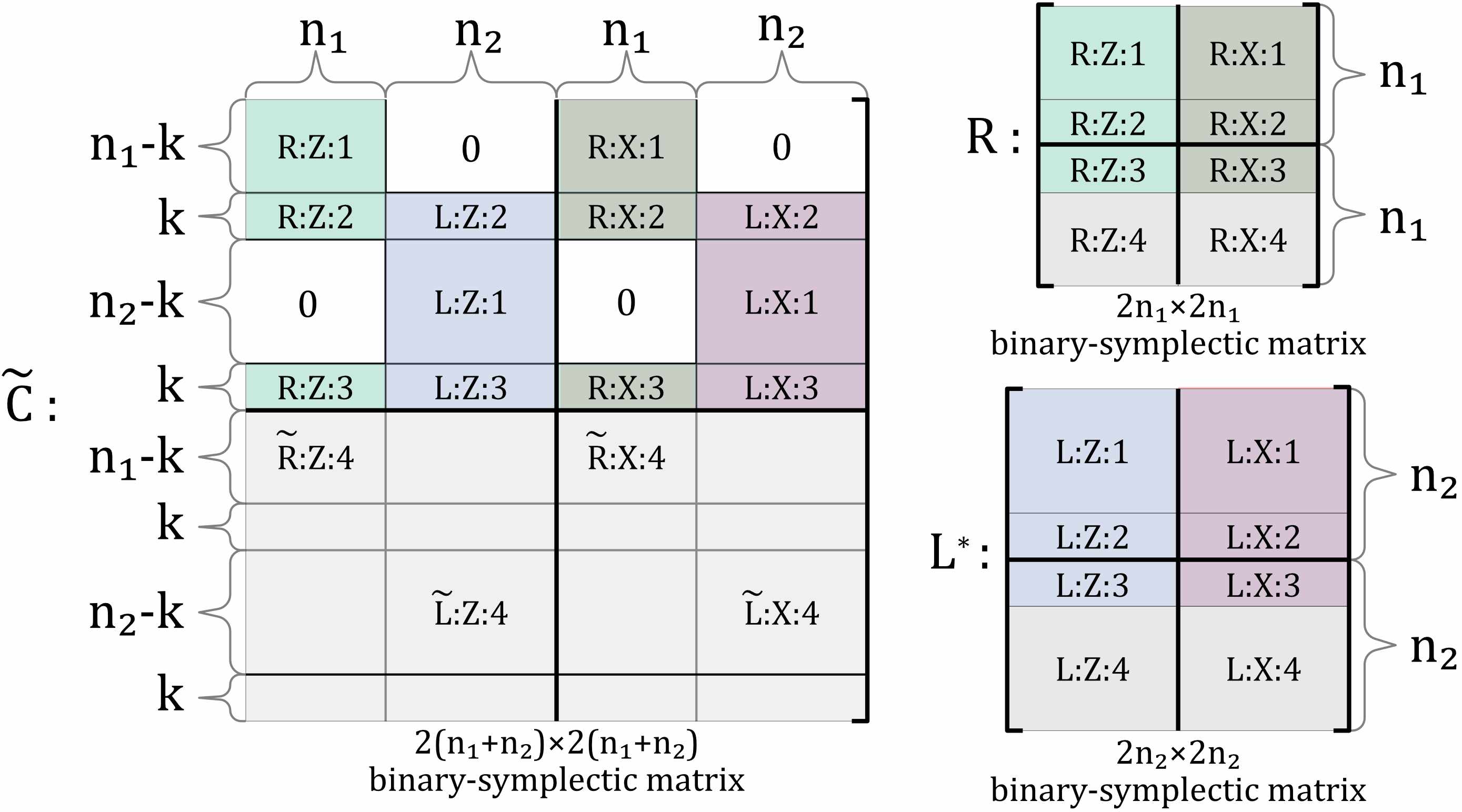}
    \caption{\label{fig:bipartition-family-blocks} Blocks of binary-symplectic representations of $\tilde C$, $\Ci$, $\Co$}
\end{subfigure}
\caption[Block structure of a Clifford unitary related to bipartite normal form]{
Given any Clifford unitary $\tilde C$, such that the circuit identity (a) holds for 
all computational basis states $\ket{a_1},\ldots,\ket{a_4}$,
Clifford unitaries $\Ci^\ast$, $\Co$ can be constructed using blocks of $\tilde C$ as shown in (b).
Blocks $\tilde{\text{\Co}}$:Z:4, $\tilde{\text{\Co}}$:X:4, $\tilde{\text{\Ci}}$:X:4, $\tilde{\text{\Ci}}$Z:4 do not appear in blocks of $D,B^\ast$,
however can be used to compute blocks ${\text{\Co}}$:Z:4, ${\text{\Co}}$:X:4, ${\text{\Ci}}$:X:4, ${\text{\Ci}}$:Z:4 that appear in $\Co,\Ci^\ast$ more efficiently~(see discussion around~\cref{eq:tilde-block-commutation}).
}
\end{figure}

Next we briefly discuss the correctness of \cref{proc:partition}.
The matrix $\tilde C$ in~\cref{line:paritioned-matrix} satisfies properties in \cref{prob:partition} since, by~\cref{prop:linear-reversible-clifford},
$$(C U_F) Z_j ( U^\dagger_F C^\dagger) = C Z^{(F^{-1})^T e_j} C^\dagger = C Z^{F^{-1}_j} C^\dagger = \prod_{i=1}^n (C Z_i C^\dagger)^{(F^{-1})_{j,i}}.$$
It is then sufficient to check that rows of $F^{-1}$ assigned in \cref{line:paritioned-basis} are ordered as required by \cref{prob:partition}
and also satisfy the required properties.
Note that $F^{-1}$ is matrix $A$ with permuted rows.
We constructed matrix $A$ such that first $k_1$ rows indicate products of $C Z_1 C^\dagger, \ldots,  C Z_n C^\dagger$
supported on first $n_1$ qubits~(\cref{line:partition-supp1}),
next $k_2$ rows indicate the products supported on the last $n_2$ qubits~(\cref{line:partition-supp2}), 
and the rest of rows indicate the products that form a symplectic basis when restricted to first $n_1$ qubits~(\cref{line:partition-symplectic-basis}).
\cref{fig:bipartition-family-blocks} shows how binary-symplectic parts of $\Ci,\Co$ can be constructed from $\tilde C$. 
We need to make sure that image and preimage phases of $\Ci,\Co$ are assigned so the that circuit in~\cref{fig:bipartition-family-circuit} is correct.
One way to ensure this is shown in steps 2 and 3 within~\cref{line:partition-left-clifford}.

For step 4 in \cref{line:partition-right-clifford,line:partition-left-clifford} we note that 
images  $\Co X_j \Co^\dagger$ for $j\in[k_1]$, $\Ci^\ast X_j (\Ci^\ast)^\dagger$ for $j\in[k_2]$ 
can be set arbitrarily. 
We can use blocks of binary-symplectic matrix of $\tilde C$ to make these steps more efficient. 
In \cref{fig:bipartition-family-blocks} we can set blocks $\Co$:$Z$:$4$, $\Co$:$X$:$4$ of $\Co$ to 
blocks $\tilde \Co$:$Z$:$4$, $\tilde \Co$:$X$:$4$ of $\tilde C$. 
This ensures that images $\Co X_j \Co^\dagger$ for $j\in[k_1]$ have correct commutation relations with
images of $\Co Z_j \Co^\dagger$ for $j\in[k_1]$:
\begin{equation}
\label{eq:tilde-block-commutation}
    \comm{ \Co X_j \Co^\dagger, \Co Z_{j'} \Co^\dagger } = \delta_{j,j'},~j,j'\in[k_1].
\end{equation}
It remains then to adjust images 
$\Co X_j \Co^\dagger$ for $j\in[k_1]$ so that they have correct commutation relations with 
$\Co X_j \Co^\dagger$, $\Co Z_j \Co^\dagger$ for $j\in[n_1-k+1,n_1]$
which has \runtime{} $O(kn^2)$.
Similarly we can set $\Ci^\ast X_j (\Ci^\ast)^\dagger$ for $j\in[k_2]$ more efficiently.

We discuss \runtime{} of \cref{proc:partition} in terms of parameters $n$ and $k$. 
The second parameter $k$ is computed by the procedure, however it is frequently known ahead of time in practice. 
The matrix multiplications, matrix inversions, full-rank completions and kernel basis calculations are all $O(n^\omega)$~(\cref{eq:kernel-basis-complexity},\cref{eq:full-rank-completion-complexity}).
Finding a symplectic basis in \cref{line:partition-symplectic-basis} is $O(nk^2)$ and constructing unitaries $\Ci$, $\Co$ is $O(k n^2)$.
We conclude that \runtime{} of \cref{proc:partition} is
\begin{equation} \label{eq:partition-runtime}
    O(n^\omega + kn^2).
\end{equation}
\subsection{Outcome compression}
\label{sec:outcome-compression-map}

The general form circuit in \cref{fig:general-form-choi} is not unique, in part, because there are degrees of freedom in the use of the random vector $r$. 
For example, two general form circuits may be equal up to remapping of random vectors.
More precisely, given any invertible square matrix $F$, we can produce a new general form circuit equivalent to a given general form circuit by substituting $r = Fr'$ and changing matrices $A, A_x, A_z$
to $A(F \oplus I_{\ki-k})$, $A_x(F \oplus I_{\ki-k})$, $A_z(F \oplus I_{\ki-k})$.
Additionally, a general form circuit may contain redundancy in $r$.
That is, for two different values $r_1$, $r_2$ of $r$ the actions of the two corresponding circuits may be identical.
This happens when $(r_1-r_2) \oplus 0^{\ki - k}$ belongs to the kernel of each of $A, A_x, A_z$.
To summarize these degrees of freedom, one can replace the first $n_r$ columns of  $A, A_x, A_z$
with another $n'_r$ columns, as long as the linear span is the same.
We can eliminate these degrees of freedom by ``compressing'' $r$ with the following algorithm.
This algorithm corresponds to constructing an outcome compression map for the quantum instrument corresponding to the general form circuit in~\cref{def:outcome-compression}.
This is also useful for computing an outcome compression map for the action of arbitrary stabilizer circuits.


\begin{procedure}[\texttt{compression map}] 
\label{alg:compress-randomness}
\begin{algorithmic}[1]
\Blank
\Input A general form circuit $\mathcal{C}_{\mathrm{gen}}$ in \cref{fig:general-form-choi}, with $n_r$ random bits, $\ki$ input qubits, $\ko$ output qubits,
$k$ inner qubits, condition matrices $A,A_x,A_z$ 
\Output
\begin{itemize}
    \item Integer $n'_r$, condition matrices $A',A'_x,A'_z$ that define a general form circuit $\mathcal{C}'_{\mathrm{gen}}$, equivalent to  $\mathcal{C}_{\mathrm{gen}}$, with $n'_r \le n_r$ random bits, same left Clifford unitary, same right Clifford unitary.
    \item Compression matrix $F_{o\rightarrow o'}$ over $\f_2$, embedding matrix $F_{o'\rightarrow o}$ over $\f_2$.
\end{itemize}
such that general form circuits $\mathcal{C}_{\mathrm{gen}}$ and $\mathcal{C}'_{\mathrm{gen}}$ have the same action for outcome vectors $o$ and $o' = F_{o\rightarrow o'} o$.
Furthermore, first $n'_r$ rows of $((A')^T | (A'_x)^T | (A'_z)^T)$ are in reduced row echelon form and have full rank, $F_{o\rightarrow o'} F_{o'\rightarrow o} = I$.
\State\label{line:compression-kernels-basis} Let $\tilde A$ be the first $n_r$ columns of
$$
\left(\begin{array}{c} A \\ \hline  A_x \\ \hline  A_z \end{array}\right).
$$
\State\label{line:compression-map}Find matrices $F,F^{-1}$ such that $\hat A^T = F^T \tilde A^T$ is in reduced row echelon form.
\State\label{line:compression-num-bits}Set $n'_r$ to be the number of non-zero columns of $\hat A$.
\State Set matrices
\begin{align*}
F_{o\rightarrow o'} \leftarrow & \left( \left(I_{n'_r} | 0_{ n'_r \times (n_r - n'_r)} \right) F^{-1} \right) \oplus I_{\ki - k}, \\
F_{o'\rightarrow o} \leftarrow &  \left( F \left(\begin{array}{c} I_{n'_r}, \\ \hline  0_{ (n_r - n'_r) \times n'_r} \end{array}\right) \right) \oplus I_{\ki - k}, \\
A'  \leftarrow & A F_{o'\rightarrow o},~~A'_x  \leftarrow A_x F_{o'\rightarrow o},~~A'_z  \leftarrow A_z F_{o'\rightarrow o}.
\end{align*}
\State \Return $n'_r$, condition matrices $A',A'_x,A'_z$, compression and embedding matrices $F_{o\rightarrow o'}$,$F_{o'\rightarrow o}$
\end{algorithmic}
\end{procedure}

The goal of \cref{alg:compress-randomness} is to first identify all random bit vectors 
that result in the same linear map being applied by the general form circuit~(\cref{line:compression-map}).
All vectors from the kernel of $\tilde A = \hat A F^{-1}$ lead to the same map being applied.
It is easy to identify the kernel of $\hat A$, those are exactly its zero columns,
that is columns from $n'_r + 1$ to $n_r$~(\cref{line:compression-num-bits}). 
The basis of kernel of $\tilde A$ is then $F e_{n'_r + 1}, \ldots, F e_{n_r}$.

Vectors from within each coset in $\f_2^{n_r} / \mathrm{ker}(A)$ also lead to the same map being applied.
Vectors from the distinct cosets result in the distinct maps applied by the general form circuit.
The basis $b_1 = F e_1,\ldots,b_{n'_r} = F e_{n'_r}$ corresponds 
to a basis $b_1 + \mathrm{ker}(\tilde A), \ldots, b_{n'_r} + \mathrm{ker}(\tilde A)$ of the coset vector space $\f_2^{n_r} / \mathrm{ker}(\tilde A)$.
Matrix $F$ defined in \cref{line:compression-map} maps standard basis vectors $e_1,\ldots,e_{n'_r}$ 
to the coset-basis vectors $b_1,\ldots,b_{n'_r}$, and the rest $e_{n'_r + 1},\ldots,e_{n_r}$ to the kernel basis. 
For this reason, columns $n'_{r}+1, \ldots, n_r$ of the products $A (F \oplus I_{\ki-k})$, $A_x (F \oplus I_{\ki-k})$,$A_z (F \oplus I_{\ki-k})$ are zero and so we can remove 
those columns and shrink the random bit vector of the general form to $n'_r$.
This is exactly the effect of matrix $F_{o\rightarrow o'}$.
Distinct values of $r'$ correspond to distinct cosets of $r$ and so the outcome space of 
an the quantum instrument corresponding to general from circuit is compressed as in~\cref{def:outcome-compression}.
Finally we notice that the first $n'_r$ rows of $((A')^T | (A'_x)^T | (A'_z)^T)$ are equal to first $n'_r$ rows of $\hat A$,
which ensures required properties of the rows.

An important edge case of the compression map~\cref{alg:compress-randomness} is when $n'_r = 0$ and $\ki = k$, that is, when the general form circuit's 
action is independent on the random bits and no observables of the input are measured.
In this case matrix $F_{o\rightarrow o'}$ is of size $n_r \times 0$ and matrices $A',A'_x,A'_z$ are of size $(\ko-k)\times 0$,$k\times 0$, $k\times 0$.
We use the convention that an $n_r \times 0$ matrix takes $n_r$-dimensional vectors as input and returns nothing as the output,
similarly $k\times 0$ takes nothing as an input and returns a zero vector of dimension $k$.
When zero-dimensional matrices are used in direct sum operation they pad the other matrix with rows or columns full of zeros. 
For example, if $F$ is $n_r \times 0$ matrix, then $F \oplus I_k = (0_{k\times n_r} | I_k)$. 


Additionally, a circuit's general form \cref{alg:general-from-choi} together with randomness compression~\cref{alg:compress-randomness}
can be used to find an outcome compression map for the quantum instrument corresponding to any stabilizer circuit $\mathcal{C}$. 
This is achieved in three steps.
First find the stabilizer circuit's general form circuit $\mathcal{C}_{\textrm{gen}}$ with $o \mapsto M o + v_0 $.
Second apply the compression map algorithm and find compression matrix $F_{o \rightarrow o'}$.
Third, find $M^{(-1)}$, the left inverse of $M$, that is a rectangular matrix such that $M^{(-1)} M = I$. 
This is always possible because $M$ has full column rank.
An outcome compression map for stabilizer circuit $\mathcal{C}$ with outcome vector $v$ is then $v \mapsto F_{o \rightarrow o'} M^{(-1)}(v-v_0)$.
The outcome compression map computes outcome $o'$ of the compressed general form corresponding to the circuit outcome $v$.
Computing the compression map is useful when checking equivalence of the action of two stabilizer circuits discussed in the next section.

The \runtime{} of \cref{alg:compress-randomness} is
\begin{equation} \label{eq:compression-map-runtime}
    O(\nr \no \min(\nr,\no)^{\omega-2}).
\end{equation}
This is because the \runtime{} of the row reduced echelon form calculation is exactly \cref{eq:compression-map-runtime}, 
as discussed in \cref{app:procedures},~\cref{eq:rref-complexity}.
The matrix $F^T$ is sparse with a dense sub-matrix of size $\nr \times O(\no)$, so there is no $O(\nr^2)$ contribution to the \runtime{} when $\nr > O(\no)$.

\newpage
\section{Stabilizer circuit verification}
\label{sec:equality-of-general-forms}

Given two stabilizer circuits, it is useful to know if they have equivalent action or not.
For example, we might want to verify that a circuit that meets certain hardware-related connectivity requirements performs the same action as a (simpler) reference circuit.
This motivates the following problem.
\begin{problem}[Stabilizer circuit comparison]
\label{prob:stabilizer-circuit-comparison}
Given two stabilizer circuits $\mathcal{C}_1$ and $\mathcal{C}_2$ with outcome vectors $v_1$ and $v_2$, find if they have equivalent action, i.e, if their corresponding quantum instruments are equivalent according to \cref{def:instrument-equality}.
If their actions are equivalent, find matrices $M_1$ and $M_2$ and vectors $u_1$ and $u_2$ such that:
(i) the action of $\mathcal{C}_1$ given outcome $v_1$ is the same as the action of $\mathcal{C}_2$ given outcome $v_2$ if and only if $M_1(v_1 + u_1) = M_2(v_2 + u_2)$, and
(ii) the action of $\mathcal{C}_j$ given outcome $v_j$ is the same as the action of $\mathcal{C}_j$ given outcome $v'_j$ if and only if $M_j(v_j - v'_j) = 0$.
\end{problem}

First let us understand this problem statement in the context of the quantum instrument equivalence \cref{def:instrument-equality}.
For two quantum instruments to be equivalent, there must be a bijection $f$ between outcomes of the outcome-compressed versions of the two instruments such that their action on any quantum state is the same for outcomes related by the bijection.
In \cref{prob:stabilizer-circuit-comparison}, we are tasked with finding the maps $v_1 \mapsto M_1(v_1 + u_1)$ and $v_2 \mapsto M_2(v_2 + u_2)$ such that when the output of these maps match, the circuits $\mathcal{C}_1$ and $\mathcal{C}_2$ have identical action on any input. 
These maps are in effect implementing the outcome compression for both circuits, but also an outcome relabeling that implements the bijection $f$.

Next, it can useful to consider a special set of scenarios in which not only are two input circuits equivalent, but where we can find an explicit map from the outcomes of one circuit to the other.
This is the case for example if the reference circuit $\mathcal{C}_2$ has no redundancy in its measurement outcomes,
then $\mathrm{ker}(M_2)$ is trivial making it possible to construct a map from the outcomes of $\mathcal{C}_1$ to the outcomes of $\mathcal{C}_2$. 
In this case there exists a left inverse $M^{(-1)}_2$ of $M_2$ and we can write $v_2 = v_0 + M v_1$
for $v_0 = u_2 + M_2^{(-1)}u_1$ and $M = M_2^{(-1)} M_1$.

A naive strategy aiming to solve \cref{prob:stabilizer-circuit-comparison} would be to construct general form circuits for $\mathcal{C}_1$ and $\mathcal{C}_2$ and to check equality of the objects comprising the general forms, namely their left and the right Clifford unitaries and the condition matrices.
However, this strategy would fail because general form circuits are not unique and therefore even if two circuits have identical action, the general form circuits we find for each may be different.
For this reason we develop an algorithm in \cref{sec:comparing-gen-form} to compare two general form circuits, thereby solving a special case of~\cref{prob:stabilizer-circuit-comparison} when $\mathcal{C}_1$ and $\mathcal{C}_2$ are general form circuits.
The solution to the general case of \cref{prob:stabilizer-circuit-comparison} reduces to this special case as shown below. 

\begin{algorithm}[\texttt{Stabilizer circuit comparison}] 
\label{alg:stabilizer-circuit-comparison}
\begin{algorithmic}[1]
\Blank
\Input 
\begin{itemize}
    \item Stabilizer circuits $\mathcal{C}_1$, $\mathcal{C}_2$ with the same number of input and output qubits.
\end{itemize}
\Output One of the following:
\begin{itemize}
    \item \texttt{False} (the stabilizer circuits are not equivalent).
    \item \texttt{True} (the stabilizer circuits are equivalent), matrices $M_1,M_2$ and vectors $u_1,u_2$ such that conditions in \cref{prob:stabilizer-circuit-comparison} are satisfied.
\end{itemize}
\State Find a general form circuit $\mathcal{C}_{\mathrm{gen},j}$ for $\mathcal{C}_j$ and outcome map $o_j \mapsto \tilde u_j + \tilde M_j o_j$ 
from the outcomes of $\mathcal{C}_{\mathrm{gen},j}$ to the outcomes of $\mathcal{C}_j$, for $j \in [2]$. \Comment \cref{alg:general-from-choi}
\State Solve the stabilizer circuit comparison~\cref{prob:stabilizer-circuit-comparison} for $\mathcal{C}_{\mathrm{gen},1}$ and $\mathcal{C}_{\mathrm{gen},2}$. \label{line:general-form-comparison}
\Comment{\cref{alg:general-form-circuit-comparison}}
\If{ $\mathcal{C}_{\mathrm{gen},1}$ and $\mathcal{C}_{\mathrm{gen},2}$ are equivalent, with matrices $\hat M_1,\hat M_2$ and vectors $\hat u_1, \hat u_2$ \\ ~~~~~~~~ defining the equivalent outcomes,}
\State\label{line:generic-comparison-matrices} $M_j \leftarrow \hat M_j \tilde M_j^{(-1)},~~u_j \leftarrow \hat u_j + \tilde M_j^{(-1)} \tilde u_j$ for $j \in [2]$ \Comment $\tilde M_j^{(-1)}$ is a left-inverse of $\tilde M_j$
\State \Return \texttt{True}, $M_1,M_2,u_1,u_2$
\Else 
\State \Return \texttt{False}
\EndIf
\end{algorithmic}
\end{algorithm}

We briefly discuss the correctness of \cref{alg:stabilizer-circuit-comparison}.
By definition, each input circuit is equivalent to its general form circuit, and therefore if one observes that the two general form circuits are inequivalent, this implies the two input circuits are also inequivalent.
When the input circuits are equivalent, we need to establish that the matrices $M_1,M_2$ and vectors $u_1,u_2$ 
have the required properties. 
This immediately follows from the fact that there is an affine one-to-one correspondence 
between a circuit's outcomes and its general form outcomes. 

To analyze the \runtime{} of \cref{alg:stabilizer-circuit-comparison} we assume that the stabilizer circuits $\mathcal{C}_1$ and $\mathcal{C}_2$
use bounded-weight Pauli operators and have the same properties as in our analysis of the general form \cref{alg:general-from-choi} \runtime{}
in~\cref{eq:general-form-runtime}. 
Let the parameters of the input circuit $\mathcal{C}_j$ be as in \cref{sec:circuit-parameters} but with superscript $(j)$.
The \runtime{} of \cref{alg:stabilizer-circuit-comparison} is 
\begin{equation}\label{eq:circuit-comparison-runtime}
    O\left(\sum_{j \in \{1,2\}}  \nmax^{(j)}\cdot\left(\Nu^{(j)}+(\Ncnd^{(j)}+\nM^{(j)})(\nr^{(j)}+\nm^{(j)} + 1) + \nmax^{(j)}(\nmax^{(j)} + \nr^{(j)})\right)\right).
\end{equation}
The \runtime{} of \cref{alg:stabilizer-circuit-comparison} is dominated by the cost of finding general forms for the input circuits rather than the cost of comparing those general forms.
The \runtime{} of \cref{line:generic-comparison-matrices} is also negligible because the left inverse of matrices $\tilde M_j$ is one-sparse.

When we find that two circuits are inequivalent, sometimes it is possible to find a ``correction'' unitary that can be appended to one of the circuits to achieve equivalence.
We have designed \cref{alg:stabilizer-circuit-comparison} in such a way that the correction unitary can be easily found when it exists~(\cref{tab:equality-correction-unitaries}).
We either find a Clifford correction unitary or a Pauli correction unitary conditioned on the circuit's general form outcomes.
The map $o_1= \tilde M_1^{(-1)}(v_1 + \tilde u_1)$ from the circuit's outcomes to its general form outcomes allows us to easily convert Pauli correction unitaries $X^{\ip{c_x,o_1}} Z^{\ip{c_z o_1}}$ 
conditioned on the general form's outcomes $o_1$ to the Pauli correction unitaries conditioned on the circuit's outcomes $v_1$:
\begin{equation}
X_j^{\ip{c_x,o_1}} Z_j^{\ip{c_z,o_1}} = X_j^{\ip{(\tilde M_1^{(-1)})^T c_x, v_1 + \tilde u_1}} Z_j^{\ip{(\tilde M_1^{(-1)})^T c_z, v_1 + \tilde u_1}}.
\end{equation}

The rest of the section is dedicated to solving~\cref{prob:stabilizer-circuit-comparison} for two general form circuits~(\cref{line:general-form-comparison}, \cref{{alg:stabilizer-circuit-comparison}}). 
This special case is also relevant for verifying the logical action of logical operation circuits, discussed in \cref{sec:logical-general-form}.

\subsection{Relating (un)encoding circuits}
\label{sec:compare-encoding-unencoding}

As discussed in~\cref{sec:stabilizer-circuit-general-form} and highlighted in~\cref{fig:general-form}, a general form circuit consists of three steps: 
(1) an unencoding circuit, (2) Pauli unitaries conditioned on both random bits and the syndrome measured by the unencoding circuit and (3) an encoding circuit. 
To compare two general form circuits, we first consider how the encoding and unencoding circuits of one general form can be related to those of the other general form.

\begin{figure}[htp]
    \centering
     \begin{subfigure}[b]{\textwidth}
        \centering
        \loadfig{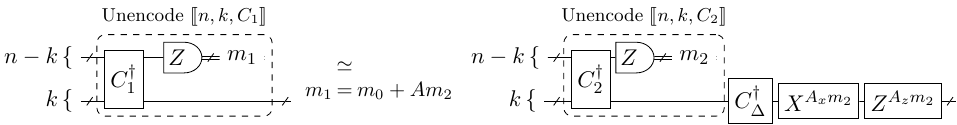}
        \caption{The relation between unencoding circuits for equivalent codes.}
        \label{fig:unencoding-circuits-compare}
    \end{subfigure}
    \begin{subfigure}[b]{\textwidth}
        \centering
        \loadfig{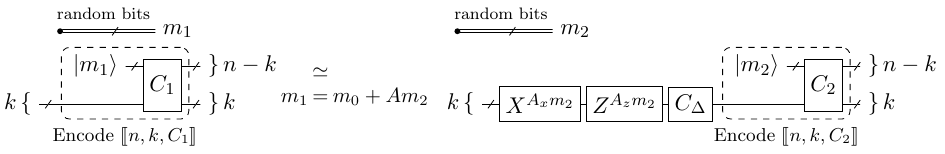}
        \caption{The relation between encoding circuits for equivalent codes.}
        \label{fig:encoding-circuits-compare}
    \end{subfigure}
    \begin{subfigure}[b]{\textwidth}
        \centering
        \includegraphics[scale=0.3]{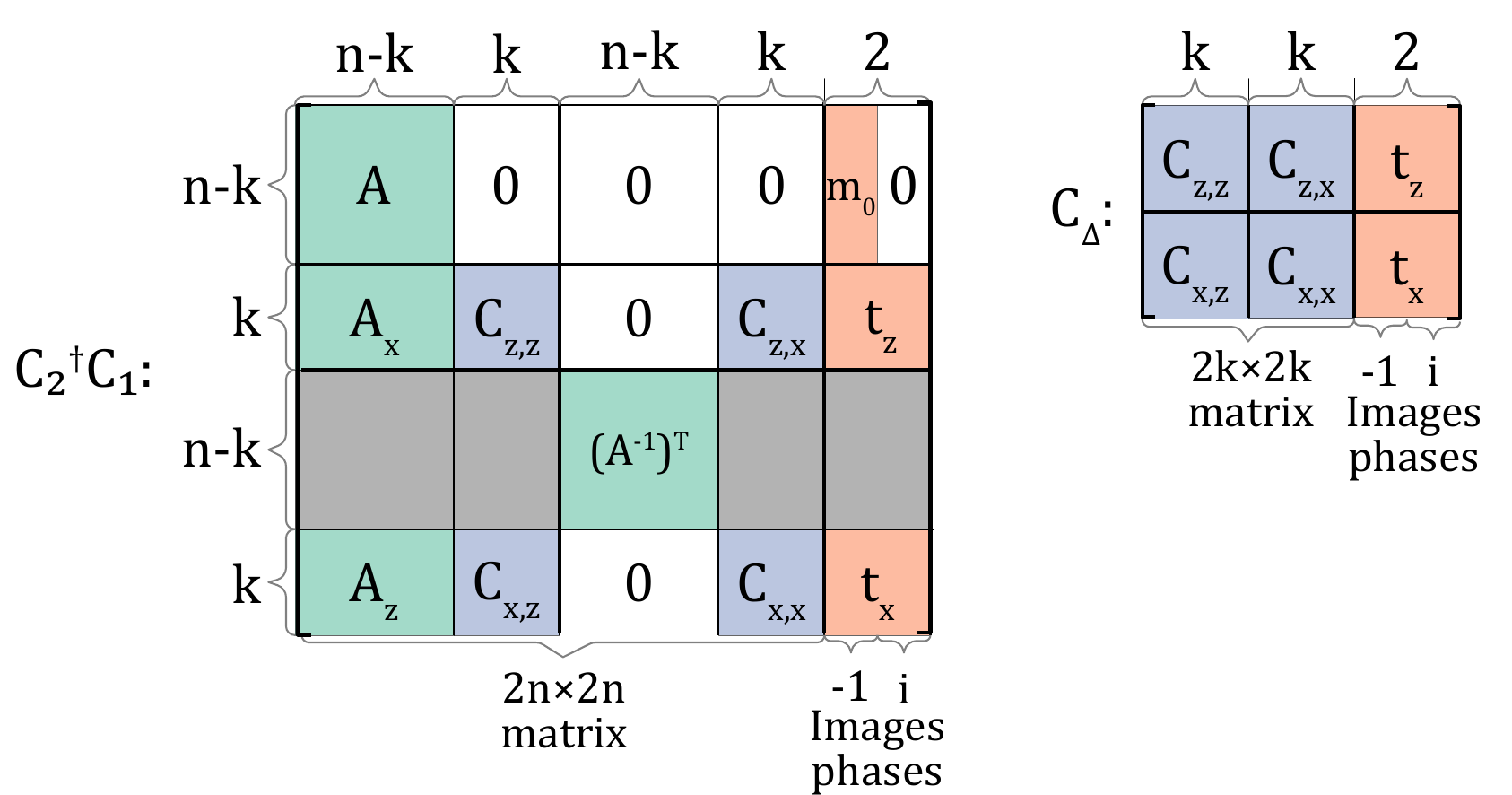}
        \caption{Deriving Clifford unitary $C_\Delta$, condition matrices $A,A_x,A_z$ and vector $m_0$ from $C_2^\dagger C_1$. See~\cref{prop:encoding-comparison-solution}
        for the correctness proof.} 
        \label{fig:encoding-comparison-solution}
    \end{subfigure}
    \caption[Comparing encoding and unencoding circuits]{Relating (un)encoding circuits. 
    When codes $\comm{n,k,C_1}$, $\comm{n,k,C_2}$ are equivalent,
        there exist Clifford unitary $C_\Delta$, matrices $A,A_x,A_z$ and vector $m_0$ such that the circuits in subfigures (a),(b) have 
        equivalent action. The map between outcomes is $m_1 = m_0 + A m_2$.}
    \label{fig:endcoding-unencoding-circuits}
\end{figure}

Recall that encoding and unencoding circuits for the code $\comm{n,k,C}$ are fully specified by the Clifford unitary $C$ and integers $n,k$.
We define the encoding circuit for $\comm{n_1,k_1,C_1}$ to be the circuit in the left in \cref{fig:encoding-circuits-compare}, similarly we define the unencoding circuit for $\comm{n_1,k_1,C_1}$ to be the circuit in the left in \cref{fig:encoding-circuits-compare}. 
We say that a pair of codes are equivalent\footnote{This definition of code equivalence is motivated by stabilizer instrument equivalence since the action of any circuit that measures any complete set of generators of $S_1$ is equivalent to the action of any circuit that measures any complete set of generators of $S_2$ if and only if the two codes are equivalent under this definition.} if their stabilizer groups $S_1$ and $S_2$ satisfy $S_1 \subset (S_2 \cup -S_2)$ and $S_2 \subset (S_1 \cup -S_1)$.
We use this notion of code equivalence to relate two encoding circuits (or two unencoding circuits) as follows.

\begin{problem}[Relating (un)encoding circuits]
\label{prob:compare-encoding-unencoding-circuits}
Consider two (un)encoding circuits defined by $n_1,k_1,C_1$ and $n_2,k_2,C_2$ respectively.
First find if the (un)encoding circuits are related, which is the case if and only if the corresponding stabilizer codes $\comm{n_1,k_1,C_1}$ and $\comm{n_2,k_2,C_2}$ are equivalent.
If the (un)encoding circuits are related, find a Clifford unitary $C_L$, condition matrices $A_x,A_z,A$ and a vector $m_0$ that relates the second circuit to the first according to \cref{fig:endcoding-unencoding-circuits}.
\end{problem}

Note that for fixed outcomes $m_1$, $m_2$ the equalities~(up to a global phase) in \cref{fig:encoding-circuits-compare,fig:unencoding-circuits-compare} are equivalent.
Indeed, \cref{fig:encoding-circuits-compare} can be obtained from \cref{fig:unencoding-circuits-compare} by applying the conjugate-transpose to the product 
of the linear maps corresponding to the circuits, and vice versa.
A solution to \cref{prob:compare-encoding-unencoding-circuits} is given in~\cref{proc:compare-encoding-unencoding-circuits}. 
The correctness proof is~\cref{prop:encoding-comparison-solution} with $C = C_2^\dagger C_1$, the runtime is
\begin{equation}\label{eq:unencoding-difference-runtime}
O(n^\omega) 
\end{equation}
and is dominated by computing $C_2^\dagger C_1$.

\begin{figure*}[htp]
\begin{procedure}[\texttt{Relating (un)encoding circuits}] 
\label{proc:compare-encoding-unencoding-circuits}
\begin{algorithmic}[1]
\Blank
\Input 
\begin{itemize}
    \item Unitaries $C_1$ and $C_2$ and integers $n_1$, $n_2$, $k_1$ and $k_2$ that specify (un)encoding circuits for $\comm{n_1,k_1,C_1}$ and $\comm{n_2,k_2,C_2}$.
\end{itemize}

\Output One of the following:
\begin{itemize}
    \item \texttt{Not related} ($\comm{n_1,k_1,C_1}$ and $\comm{n_2,k_2,C_2}$ are not equivalent).
    \item \texttt{Related} ($\comm{n_1,k_1,C_1}$ and $\comm{n_2,k_2,C_2}$ are equivalent), Clifford unitary $C_\Delta$, condition matrices $A,A_x,A_z$ 
    and vector $m_0$ which relate the circuits as in \cref{prob:compare-encoding-unencoding-circuits} are satisfied.
\end{itemize}
\If{$n_1 \ne n_2 $ or $k_1 \ne k_2$} \Return \texttt{False}.  \Comment Code dimensions do not match \EndIf
\State Set $C \leftarrow C_2^\dagger C_1$, $n \leftarrow n_1$, $k \leftarrow k_1$.
\State Set $A \leftarrow 0_{(n-k)\times(n-k)}, A_x \leftarrow 0_{ k\times(n-k)}, A_z \leftarrow 0_{ k\times(n-k)}$, $m_0 = 0^{n-k}$.
\For{$P,j \in \{ (C Z_{j'} C^\dagger,j') : j' \in [n-k] \}$ } \Comment for $C_1 Z_{j} C_1^\dagger$, stabilizers of $\comm{n_1,k_1,C_1}$
\If{$x(P) \ne 0 $ or $ z(P)_{[n-k+1,n]} \ne 0$} 
\State \Return \texttt{Not related} \Comment $\pm C_1 Z_{j} C_1^\dagger$ is not a stabilizer of $\comm{n_2,k_2,C_2}$
\Else 
\State $A_j = z(P)_{[n-k]}$, $(m_0)_j = s(P)/2$
\EndIf
\EndFor
\State Compute $A_x,A_z, C_\Delta$ from the equations $C Z_{j+n-k} C^\dagger = Z^{(A_x)_j} \otimes  C_\Delta Z_j C_\Delta^\dagger$,
\Statex $C X_{j+n-k} C^\dagger = Z^{(A_z)_j} \otimes  C_\Delta X_j C_\Delta^\dagger,~j \in [k]$
\State \Return \texttt{Related}, $C_\Delta$, $A$, $A_x$, $A_z$, $m_0$ 
\Blank \Comment See \cref{fig:encoding-comparison-solution} for an illustration of relation to $C = C_2^\dagger C_1$
\end{algorithmic}
\end{procedure}
\end{figure*}

\cref{proc:compare-encoding-unencoding-circuits} is useful outside of the context of this section. For example, it can be used for general from circuit classification, 
that is determining if an action of the circuit corresponds to some common circuit family such as measurement of a set of commuting Pauli operators.

\subsection{Comparing general form circuits}
\label{sec:comparing-gen-form}

Here we provide the general form circuit comparison~\cref{alg:general-form-circuit-comparison} that solves the circuit comparison~\cref{prob:stabilizer-circuit-comparison} for general form circuits.
We first give a high-level overview of the algorithm, then discuss its correctness and explain how data produced by the algorithm can be used to find a correction unitary to append to the first circuit to achieve equivalence.
At the end of this section we discuss the algorithm's \runtime{}.

\cref{alg:general-form-circuit-comparison} consists of two main steps.
The first step checks if the left and right Clifford unitaries are consistent with the two circuits having equivalent action.
If they are, the algorithm proceeds to the second step, which checks if the condition matrices are consistent with the two circuits having equivalent action.
The algorithm returns \texttt{True} if the checks in both of these steps are passed, and \texttt{False} otherwise.
In the next two paragraphs we give a high-level overview of these two steps.

The first step of \cref{alg:general-form-circuit-comparison} makes use of the fact that equivalence of $\comm{\ki,k_1,\Ci_1}$, $\comm{\ki,k_2,\Co_2}$ and $\comm{\ko,k_1,D_1}$, $\comm{\ko,k_2,D_2}$ is a necessary condition for the two general form circuits to have equivalent circuit action, so the algorithm returns \texttt{False}
if the codes are not equivalent.
If the codes are equivalent, the algorithm transforms the first general form circuit $\mathcal{C}_1$ (see \cref{fig:general-form-circuit-comparison}(a)) into an equivalent circuit $\mathcal{C}'_1$
(see \cref{fig:general-form-circuit-comparison}(b)) using the relating (un)encoding~\cref{proc:compare-encoding-unencoding-circuits} from~\cref{sec:compare-encoding-unencoding}.
The circuit $\mathcal{C}'_1$ more closely resembles the second general form circuit $\mathcal{C}_2$ -- in particular $\mathcal{C}'_1$ has the same left and right Clifford unitaries as $\mathcal{C}_2$.
At this point, it is clear that circuits $\mathcal{C}'_1$ and $\mathcal{C}_2$ are equivalent up to a conditional Pauli when $C_L$ is equal to $C_R$ up to a Pauli.
If this this the case we further transform $\mathcal{C}'_1$ into the general form circuit $\mathcal{C}''_1$ (see \cref{fig:general-form-circuit-comparison}(c)) by consolidating the various conditional Paulis, making use of the compression map~\cref{alg:compress-randomness} from~\cref{sec:outcome-compression-map} to do so.
We then further transform $\mathcal{C}''_1$ to $\mathcal{C}'''_1$ (see \cref{fig:general-form-circuit-comparison}(d)) by pulling the Clifford unitary $C_L^\dagger$ through the conditional Paulis to meet $C_R$, transforming the conditional Paulis it is pulled through.

\begin{figure}[htb]
    \loadfig{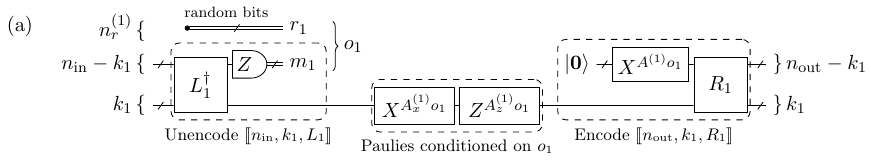}
    \loadfig{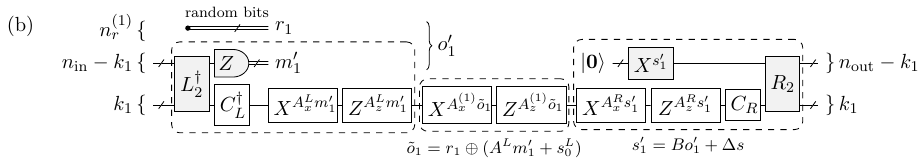}
    \loadfig{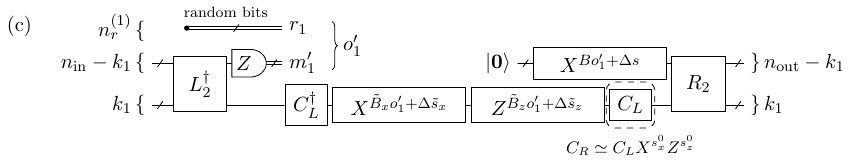}
    \loadfig{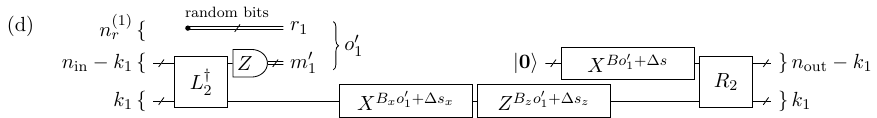}
    \caption[General form circuit comparison. Transformation step.]{
    To verify that general form circuits $\mathcal{C}_1$ and $\mathcal{C}_2$ are equivalent, \cref{alg:general-form-circuit-comparison} passes through a sequence of equivalent circuits: (a) $\mathcal{C}_1$, (b) $\mathcal{C}'_1$, (c) $\mathcal{C}''_1$ and (d) $\mathcal{C}'''_1$.
    (b) If code pairs $\comm{\ki,k_1,\Ci_1}$, $\comm{\ki,k_1,\Co_2}$ and $\comm{\ko,k_1,D_1}$,$\comm{\ko,k_1,D_2}$ are equivalent, we form the circuit $\mathcal{C}'_1$ with unencoding and encoding circuits related by $C_L,A^L,A_x^L,A_z^L,m_0^L$~(\cref{fig:unencoding-circuits-compare}) and $C_R,A^R,A_x^R,A_z^R,m_0^R$~(\cref{fig:encoding-circuits-compare}) respectively.
    The action of  $\mathcal{C}'_1$ matches that of $\mathcal{C}_1$ when $m_1 = (m_0^L + A^L m'_1)$.
    (c) If $C_R \simeq C_L X^{s^0_x} Z^{s^0_z}$, we form $\mathcal{C}''_1$ by consolidating the conditional Paulis in $\mathcal{C}'_1$.
    (d) We form $\mathcal{C}'''_1$ by transforming the conditional Paulis in $\mathcal{C}''_1$ by the Clifford $C_L$.
    }
    \label{fig:general-form-circuit-comparison}
\end{figure}

The second step of \cref{alg:general-form-circuit-comparison}, beginning on \cref{line:comparison-outcome-compress-b}, compares the condition matrices of $\mathcal{C}'''_1$ and $\mathcal{C}_2$.
To do so, we split each condition matrix into two sub-matrices for each of $\mathcal{C}'''_1$ and $\mathcal{C}_2$.
The first sub-matrix determines the dependence of each conditional Pauli on the measurement outcomes, and these sub-matrices must be equal for $\mathcal{C}'''_1$ and $\mathcal{C}_2$ if the circuits have equivalent action.
The second sub-matrix determines the dependence of each conditional Pauli on the random bits and do not need to be equal for $\mathcal{C}'''_1$ and $\mathcal{C}_2$ to have equivalent equivalent, but must satisfy a weaker equivalence condition which we check using the compression map~\cref{alg:compress-randomness}.

\begin{figure*}[p]
\begin{algorithm}[\texttt{General form circuit comparison}] 
\label{alg:general-form-circuit-comparison}
\begin{algorithmic}[1]
\Blank
\Input 
\begin{itemize}
    \item Two general form circuits $\mathcal{C}_1,\mathcal{C}_2$ with $\ki$ input qubits, $\ko$ output qubits, 
    $n^{(1)}_r$, $n^{(2)}_r$ random bits,
    left Clifford unitaries $\Ci_1,\Ci_2$, right Clifford unitaries $\Co_1,\Co_2$,
    $k_1$, $k_2$ inner qubits,
    condition matrices $A^{(1)},A^{(1)}_x,A^{(1)}_z$ and $A^{(2)},A^{(2)}_x,A^{(2)}_z$.
\end{itemize}
\Output One of the following:
\begin{itemize}
    \item \texttt{False} (the general form circuits are not equivalent)
    \item \texttt{True} (the general form circuits are equivalent), matrices $M_1,M_2$ and vectors $u_1,u_2$ such that conditions in \cref{prob:stabilizer-circuit-comparison} are satisfied.
\end{itemize}
\Blank \hrulefill \Comment Check equivalence of $\comm{\ki,k_1,\Ci_1}$, $\comm{\ki,k_2,\Co_2}$ and of $\comm{\ko,k_1,D_1}$, $\comm{\ko,k_2,D_2}$
\State\label{line:comparison-unencoding}Relate the unencoding circuits for $\comm{\ki,k_1,\Ci_1}$, $\comm{\ki,k_2,\Co_2}$  \Comment \cref{proc:compare-encoding-unencoding-circuits}
\If{ $\comm{\ki,k_1,\Ci_1}$, $\comm{\ki,k_2,\Co_2}$ are not equivalent}
\label{line:comparison-left-stab-fail}\Return \texttt{False} 
\Else
\Comment See \cref{fig:unencoding-circuits-compare}
\label{line:comparison-unencoding-diff}\State Let $C_L, A^L,A_x^L,A_z^L, m_0^L $ relate the unencoding circuits for $\comm{\ki,k_1,\Ci_1}$, $\comm{\ki,k_2,\Co_2}$ 
\EndIf
\State\label{line:comparison-encoding}Relate the encoding circuits for $\comm{\ko,k_1,D_1}$, $\comm{\ko,k_2,D_2}$  \Comment \cref{proc:compare-encoding-unencoding-circuits}
\If{ $\comm{\ko,k_1,D_1}$, $\comm{\ko,k_2,D_2}$ are not equivalent} 
\label{line:comparison-right-stab-fail}\Return \texttt{False} 
\Else
\Comment See \cref{fig:encoding-circuits-compare}
\label{line:comparison-encoding-diff}\State Let $C_R, A^R,A_x^R,A_z^R, m_0^R$ relate the encoding circuits for $\comm{\ko,k_1,D_1}$, $\comm{\ko,k_2,D_2}$ 
\EndIf
\Blank \hrulefill \Comment Check equivalence up to Pauli unitaries
\If{ $C_R C_L^\dagger$ is not a Pauli unitary }
\label{line:comparison-clifford-fail}\Return \texttt{False} \Comment \cref{tab:equality-correction-unitaries}
\Else
\State\label{line:comparison-pauli-shift}Find $s^0_x,s_z^0$ from equation $ X^{s_x^0} Z^{s_z^0} \simeq C_L^\dagger C_R$
\EndIf
\Blank \hrulefill \Comment Simplify the unitary acting on the $k_1$ inner qubits in $\mathcal{C}'_1$
\State\label{line:comparison-b-matrix}Let $\Delta s = (A^{R})^{-1}(A^{(1)}(0^{n_r^{(1)}}\oplus m_{0}^{L})+m_{0}^{R})$ and $B = (A^R)^{-1} A^{(1)} ( I_{n^{(1)}_r} \oplus A^L)$.
\State\label{line:comparison-bxt-matrix}For $\sigma\in\{x,z\}$, $\tilde B_\sigma \leftarrow \left(\boldsymbol{0}_{k\times n_{r}^{(1)}}|A_{\sigma}^{L}\right) + A_{\sigma}^{(1)}\left(I_{n_{r}^{(1)}}\oplus A^{L}\right) + A_{\sigma}^{R}B$.
\State\label{line:comparison-delta-s}For $\sigma \in \{x,z\}$, let $\Delta \tilde s_\sigma \leftarrow s_\sigma^0 + A_{\sigma}^{(1)}(0^{n_r^{(1)}}\oplus m_{0}^{L}) + A_{\sigma}^{R} \Delta s $
\State\label{line:comparison-bx-matrix}Find $B_x$, $B_z, \Delta s_x, \Delta s_z$ from
\Blank~~~
$
 X^{(B_x|\Delta s_x) s} Z^{(B_z|\Delta s_z) s} \simeq C_L X^{ (\tilde B_x|\Delta \tilde s_x) s} Z^{ (\tilde B_z | \Delta \tilde s_z) s} C_L^\dagger ,\text{for all } s
$ \Comment \cref{prop:batch-pauli-images}
\Blank \hrulefill \Comment Check equivalence for a fixed outcome
\State\label{line:comparison-outcome-compress-b}Find outcome compression matrix $F_1$, condition matrices $B',B'_x,B'_z$ for a general form circuit with $k$ input qubits, 
condition matrices $B_{\ast,[n_r^{(1)}]}, (B_x)_{\ast,[n_r^{(1)}]}, (B_z)_{\ast,[n_r^{(1)}]}$. \Comment{\cref{alg:compress-randomness}}

\State\label{line:comparison-random-shift}Let $r_0$ be a solution to linear equations $ \left(\begin{array}{c} B' \\ \hline  B'_x \\ \hline  B'_z \end{array}\right) r_0 = \Delta s \oplus \Delta s_x \oplus \Delta s_z $
\If{there is no $r_0$ exists}
\label{line:comparison-pauli-fail}\Return\texttt{False} \Comment \cref{tab:equality-correction-unitaries}
\EndIf
\Blank \hrulefill \Comment Check equivalence for all measurement outcomes
\State\label{line:comparison-ams}Let $A^{(2,m)} \leftarrow A^{(2)}_{\ast,[n_r^{(2)}+1,\cdot]}$, $B^{(m)} \leftarrow  B_{\ast,[n_r^{(1)}+1,\cdot]}$, similarly define $A^{(2,m)}_x,A^{(2,m)}_z,B^{(m)}_x,B^{(m)}_z$
\If{$B^{(m)} \ne A^{(2,m)}$ or $B_x^{(m)} \ne A_x^{(2,m)}$ or $B_z^{(m)} \ne A_z ^{(2,m)}$ }
\label{line:comparison-conditional-pauli-on-meas-fail}\Return \texttt{False} \Comment \cref{tab:equality-correction-unitaries}
\EndIf
\Blank \hrulefill \Comment Check equivalence for all outcomes
\State\label{line:comparison-outcome-compress-a}Find outcome compression matrix $F_2$, condition matrices $A',A'_x,A'_z$ for a general form circuit with $k$ input qubits, 
condition matrices $A^{(2)}_{\ast,[n_r^{(2)}]}, (A_x^{(2)})_{\ast,[n_r^{(2)}]}, (A_z^{(2)})_{\ast,[n_r^{(2)}]}$. \Comment{\cref{alg:compress-randomness}}
\If{$ \left(\begin{array}{c} B' \\ \hline  B'_x \\ \hline  B'_z \end{array}\right) \ne 
      \left(\begin{array}{c} A' \\ \hline  A'_x \\ \hline  A'_z \end{array}\right)$}
\label{line:comparison-conditional-pauli-on-random-fail}\Return \texttt{False}  \Comment \cref{tab:equality-correction-unitaries}
\EndIf
\State\label{line:comparison-ok}\Return \texttt{True}, $M_1 = F_1 \oplus (A^L)^{-1}, M_2 = F_2 \oplus I_{\ki-k},~u_1 = F_1^{(-1)} r_0 \oplus s_0^L$, $u_2 = 0^{n_r^{(2)} + \ki - k}$
\end{algorithmic}
\end{algorithm}
\end{figure*}

In what follows, we show the correctness of~\cref{alg:general-form-circuit-comparison} by considering the following sequence of conditions which are checked in order by the algorithm:
\begin{enumerate}[noitemsep]
    \item[(A)] equivalence of $\comm{\ki,k_1,\Ci_1}$, $\comm{\ki,k_2,\Ci_2}$ and of $\comm{\ko,k_1,\Co_1}$, $\comm{\ko,k_2,\Co_2}$~(\cref{line:comparison-left-stab-fail,line:comparison-right-stab-fail}),
    \item[(B)] equivalence up to Pauli unitaries~(\cref{line:comparison-clifford-fail}),
    \item[(C)] equivalence for a fixed outcome vector~(\cref{line:comparison-pauli-fail}),
    \item[(D)] equivalence for all measurement outcomes but fixed random bit vector~(\cref{line:comparison-conditional-pauli-on-meas-fail}),
    \item[(E)] equivalence (equivalence for all outcomes, including random bits~(\cref{line:comparison-conditional-pauli-on-random-fail})).
\end{enumerate}
We show that each condition must hold if the two input circuits have equivalent action, and that the conditions become progressively stricter down the list, with the last condition being that the circuits are equivalent. 
To show the correctness of \cref{alg:general-form-circuit-comparison}, we verify that it sequentially checks these conditions, returning \texttt{False} if it encounters a condition that is not satisfied, and returns \texttt{True} if and only if all of the conditions hold.
While discussing each of these conditions, we also build up the contents of \cref{tab:equality-correction-unitaries}, which specifies a modification that can be applied to the first input circuit to make it equivalent to the second input circuit.  

(A) Let us show that the circuit equivalence of $\mathcal{C}_1$ and $\mathcal{C}_2$ implies the code equivalence of $\comm{\ni,k,\Ci_1}$ and $\comm{\ni,k,\Ci_2}$.
This follows from two observations.
First, all the Choi states of equivalent general form circuits are equal up to Pauli unitaries
and the Choi state's stabilizer groups are equal up to signs.
This is because, for fixed $j \in [2]$, \cref{fig:general-form-choi} implies that all the Choi states $\ket{\Psi_{o_j}}$ corresponding to $\mathcal{C}_j$'s stabilizer instrument are equal 
up to Pauli unitaries, that is $\ket{\Psi_{o_j}}$ equals $\ket{\Psi_{o'_j}}$ up to Pauli unitaries for any outcomes $o_j,o'_j$.
Second, the fact that the stabilizer groups of the Choi circuits of $\mathcal{C}_1$ and $\mathcal{C}_2$ match up to signs implies that the stabilizer groups of $\comm{\ni,k,\Ci_1}$ and $\comm{\ni,k,\Ci_2}$ match up to signs too (which is the defining property of code equivalence).
This is because the stabilizer group of $\comm{\ni,k,\Ci_j}$ is equal up to signs to the subgroup of the stabilizer group of $\ket{\Psi_{o_j}}$ supported on last $\ki$ qubits,
for all outcomes $o_j$~(see stabilizer generators in~\cref{fig:general-form-choi}).
Similarly, the circuit equivalence of $\mathcal{C}_1$ and $\mathcal{C}_2$ implies the code equivalence of $\comm{\no,k,\Co_1}$ and $\comm{\no,k,\Co_2}$.
We check the equivalence of these stabilizer groups using~\cref{proc:compare-encoding-unencoding-circuits} in~\cref{line:comparison-unencoding,line:comparison-encoding} and the algorithm returns \texttt{Fail} if it is not the case.

\begin{table}[htp]
    \centering
    \begin{tabular}{|l|c|c|}
    \hline 
    Reason & \multirow{2}{*}{\rotatebox[origin=c]{-90}{Line}} & Correction unitary $U$ to achieve equivalence \\
    for inequivalence &      & via appending $U$ to general form circuit $\mathcal{C}_1$ \\
    \hline 
    \hline 
    Circuits differ by & \multirow{2}{*}{\labelcref{line:comparison-clifford-fail}} & \multirow{2}{*}{$U = \Co_2 \left( I_{\ko-k_1} \otimes \left(C_L  C_R^\dagger \right)\right) \Co^\dagger_2$} \\
    a Clifford unitary & & \\
    \hline 
    Circuits differ by & \multirow{2}{*}{\labelcref{line:comparison-pauli-fail}} & \multirow{2}{*}{$U = \Co_2 \left( X^{\Delta s} \otimes X^{\Delta s_x} Z^{\Delta s_z} \right) \Co^\dagger_2$} \\
    a Pauli unitary & & \\
    \hline 
    Circuits differ by & \multirow{4}{*}{\labelcref{line:comparison-conditional-pauli-on-meas-fail}} & {Let $\Delta A= B^{(m)} - A^{(2,m)}, \Delta A_\sigma = B_\sigma^{(m)} - A_\sigma^{(2,m)}$} \\
    a Pauli unitary & & for $\sigma \in \{x,z\}$, $f(m_1) = (A^L)^{-1}(m_1 + s_0^L)$ \\
    conditioned on & &  \multirow{2}{*}{$U(m_1) = \Co_2 \left( X^{\Delta A f(m_1) } \otimes X^{\Delta A_x f(m_1) } Z^{\Delta A_z f(m_1)} \right) \Co^\dagger_2$} \\ 
    measurements & & \\
    \hline 
    Circuits differ by & \multirow{4}{*}{\labelcref{line:comparison-conditional-pauli-on-random-fail}} & Let $\tilde n_r = \mathrm{ncols}(A')$, $l = n_r^{(1)}-\tilde n_r$, for $\sigma \in \{x,z\}$ \\
    a Pauli unitary & & $\Delta B = B_{\ast,[n^{(1)}_r]} + (A'|\mathbf{0}_{(\ko-k)\times l}), \Delta B_\sigma = (B_\sigma)_{\ast,[n^{(1)}_r]} + (A'_\sigma|\mathbf{0}_{k \times l})$ \\
    conditioned on & &  \multirow{2}{*}{$U(r_1) = \Co_2 \left( X^{\Delta B r_1} \otimes X^{ \Delta B_x r_1} Z^{ \Delta B_z r_1 } \right) \Co^\dagger_2$} \\
    random bits & & \\
    \hline 
    \end{tabular}
    \caption[General form circuit comparison recovery from fails]{\label{tab:equality-correction-unitaries}
             Cases when the comparison of general form circuits $\mathcal{C}_1$, $\mathcal{C}_2$ using \cref{alg:general-form-circuit-comparison} returns \texttt{False}
             and a correction unitary exists that can be appended to $\mathcal{C}_1$ to achieve circuit equivalence.
             We use the notation $r_1$ for random bit vector and $m_1$ for measurement outcome vector of $\mathcal{C}_1$ when defining conditional correction Pauli unitaries.
             Variables in the equations for the correction unitaries are defined in the pseudo-code of \cref{alg:general-form-circuit-comparison}.
             Note that it is not possible to correct Pauli unitaries conditioned on random bits by modifying $\mathcal{C}_1$ when number of columns $\mathrm{ncols}(A')$ of $A'$
             is greater than $n_r^{(1)}$; in this case $\mathcal{C}_2$ must be modified instead.
            } 
\end{table}

(B) Given that $\comm{\ni,k,\Ci_1},\comm{\ni,k,\Ci_2}$ are equivalent and $\comm{\no,k,\Co_1},\comm{\no,k,\Co_2}$ are equivalent since condition (A) has passed, we compute the relation between the
corresponding (un)encoding circuits~(\cref{line:comparison-unencoding-diff,line:comparison-encoding-diff})
and transform the first general form circuit $\mathcal{C}_1$ into an equivalent circuit $\mathcal{C}'_1$ as shown in~\cref{fig:general-form-circuit-comparison}(a).
If we ignore all the Pauli unitaries, $\mathcal{C}'_1$ is almost the same circuit as the second general form circuit $\mathcal{C}_2$,
except the product $C_R C_L^\dagger $ on the inner qubits. 
The stabilizer groups of the Choi states $\ket{\Psi_{o_1}}$ and $\ket{\Psi_{o_2}}$ match up to sign (which must be the case if $\mathcal{C}_1$ and $\mathcal{C}_2$ are equivalent) if and only if $C_L^\dagger C_R$ is a Pauli unitary.
We check this in~\cref{line:comparison-clifford-fail}, and the algorithm returns \texttt{Fail} if it is not the case.
We can fix this by appending a correction unitary to $\mathcal{C}_1$ from~\cref{tab:equality-correction-unitaries} that effectively 
cancels $C_R C_L^\dagger $ on the inner qubits.

Given that the algorithm has not yet returned \texttt{False}, we can assume from here onward that $C_R C_L^\dagger$ is a Pauli unitary.
Before moving on to condition (C), we first simplify the unitary acting on the inner qubits of $\mathcal{C}'_1$ to form the circuit $\mathcal{C}''_1$ as shown in \cref{fig:general-form-circuit-comparison}(c) to make the next steps more streamlined.
To see why $B$ and $\Delta s$ are set as specified in \cref{line:comparison-b-matrix}, note that their defining relation is $s_1' = B o_1' + \Delta s$, and make use of the encoding and unencoding circuit relations in going from $\mathcal{C}_1$ to $\mathcal{C}_1'$ are $s_1' = m_0^R + A^R s_1$ and $m_1' = m_0^L + A^L m_1$ respectively, where $s_1 = A^{(1)} o_1$ (and making use of the definitions $o_1 = r_1 \oplus m_1$, $o_1' = r_1 \oplus m_1'$). 
To see why $\tilde B_x$, $\tilde B_z$, $\Delta \tilde s_x$ and $\Delta \tilde s_z$ are set as specified on \cref{line:comparison-bxt-matrix,line:comparison-delta-s}, note that they must satisfy $\tilde B_\sigma o_1' + \Delta \tilde s_\sigma = A_\sigma^L m_1' + A_\sigma^{(1)} \tilde o_1 + A_\sigma^R s_1' + s_\sigma^0$ to ensure that $\mathcal{C}''_1$ has the same action as $\mathcal{C}'_1$, and then make use of the fact that $\tilde o_1 = r_1 \oplus (A^L m'_1 + s_0^L)$ and $s_1' = B o_1' + \Delta s$.

By the end of \cref{line:comparison-bx-matrix}, we have the circuit $\mathcal{C}'''_1$ (see \cref{fig:general-form-circuit-comparison}(d)), which is very close to the general form circuit $\mathcal{C}_2$, except that the Paulis are conditioned on affine rather than linear outcome maps (i.e. the presence of the $\Delta s_x$, $\Delta s_z$ and $\Delta s$ terms). 
Due to these similarities, and in particular since both circuits have the same left and right Cliffords $\Ci_2$ and $\Co_2$, the stabilizer generators of the Choi states (see \cref{fig:general-form-choi}) of both $\mathcal{C}'''_1$ and $\mathcal{C}_2$ can be expressed as:
\begin{eqnarray*}
& \Co_2 Z_{j_\out} \Co_2^\dagger \otimes (-1)^{s_{j_\out}}I_{\ki}, & j_\out \in [\ko-k], \\
& \Co_2 Z_{\ko - k + j} \Co_2^\dagger \otimes (-1)^{s_{j}} (\Ci_2 Z_{\ki - k + j} \Ci_2^\dagger)^\ast, & j \in [k],\\
& I_{\ko} \otimes  (-1)^{s_{j_\inn}} (\Ci_2 Z_{j_\inn} \Ci_2^\dagger)^\ast, & j_\inn \in [\ki - k], \\
& \Co_2 X_{\ko - k + j} \Co_2^\dagger \otimes  (-1)^{s_{j}} (\Ci_2 X_{\ki - k + j} \Ci_2^\dagger)^\ast,& j\in [k],
\end{eqnarray*}
where the only difference is captured by the sign bits,
\begin{equation}
\label{eq:sign-vectors}
s(\mathcal{C}_1''') = \left(\begin{array}{c} B \\ \hline  B_x \\ \hline (\mathbf{0}|I_{\ki - k}) \\ \hline B_z \end{array}\right)(r_1 \oplus m'_2) + 
  \left(\begin{array}{c} \Delta s \\ \hline \Delta s_x \\ \hline 0  \\ \hline \Delta s_z \end{array}\right), 
s(\mathcal{C}_2) = \left(\begin{array}{c} A^{(2)} \\ \hline  A_x^{(2)} \\ \hline (\mathbf{0}|I_{\ki - k}) \\ \hline A^{(2)}_z \end{array}\right)(r_2 \oplus m_2). 
\end{equation}
For the linear maps enacted by $\mathcal{C}'''_1$ and $\mathcal{C}_2$ upon outcomes $o'_1$ and $o_2$ 
to be equivalent\footnote{Linear maps are equivalent when they are equal up to a global phase},
the values of the sign vectors $s(\mathcal{C}_1''')$ and $s(\mathcal{C}_2)$ must be equal.
From the third block of these matrix equations, we see that it is necessary that $m_2 = m'_2$.

(C) Next we check the equivalence for a fixed outcome. 
We check if there exists an outcome $m'_2 \oplus r'_1$ of $\mathcal{C}'''_1$ upon which it enacts the same linear map (up to a global phase) as $\mathcal{C}_2$ upon outcome $0^{n_r^{(2)} + \ki - k}$.
By \cref{def:instrument-equality} of instrument equivalence, such an outcome of $\mathcal{C}'''_1$ must exist if $\mathcal{C}_1$ and $\mathcal{C}_2$ are equivalent.
Since $m_2 = m'_2$, the equality of the sign vectors $s(\mathcal{C}_1''')$, $s(\mathcal{C}_2)$ in \cref{eq:sign-vectors}
simplifies to 
\begin{equation}
\label{eq:random-shift}
 \left(\begin{array}{c} B_{\ast,[n_r^{(1)}]} \\ \hline  (B_x)_{\ast,[n_r^{(1)}]} \\ \hline  (B_z)_{\ast,[n_r^{(1)}]} \end{array}\right) r'_1 + 
 \left(\begin{array}{c} \Delta s \\ \hline \Delta s_x \\ \hline \Delta s_z \end{array}\right) = 0
\end{equation}
We check if such an $r'_1$ exists using the compression map~\cref{alg:compress-randomness} in~\cref{line:comparison-outcome-compress-b}.
The compression map replaces the outcome vector with a compressed outcome vector, and replaces the matrix in the equation with a matrix in reduced row echelon form.
We solve the resulting equation in~\cref{line:comparison-random-shift} via back substitution.
The vector $r'_1 $ is then $F_1^{(-1)} r_0$ where $F_1^{(-1)}$ is the right inverse of the outcome compression matrix $F_1$. 
The matrix $F_1^{(-1)}$ corresponds to the embedding matrix computed by the compression map~\cref{alg:compress-randomness}.
If no solution exists, the circuits are not equivalent~(\cref{line:comparison-pauli-fail}).
The correction unitary in~\cref{tab:equality-correction-unitaries} appended to the end of $\mathcal{C}_1$ has the effect of 
setting $\Delta s$, $\Delta s_x$, $\Delta s_z$ to zero and ensuring that $r'_1 = 0$ is a solution to \cref{eq:random-shift}.

(D) Next we check equivalence for all measurement outcomes. 
Consider sign vector $s(\mathcal{C}_2)$ related to the linear map enacted by $\mathcal{C}_2$ upon $0^{n^{(2)}_r} \oplus m_2$.
It must be equal to $s(\mathcal{C}'''_1)$ upon outcome $r'_1 \oplus m_2$. 
Therefore, for all $m_2$:
\begin{equation}
\label{eq:measurement-phases}
  A^{(2,m)}m_2 = B^{(m)}m_2,~A_x^{(2,m)}m_2 = B_x^{(m)}m_2,~A_z^{(2,m)}m_2 = B_z^{(m)}m_2.
\end{equation}
for the above matrices defined in \cref{line:comparison-ams}.
We check that this is the case in \cref{line:comparison-conditional-pauli-on-meas-fail}.
If this is not the case, correction unitary in~\cref{tab:equality-correction-unitaries} appended to the end of $\mathcal{C}_1$ has the effect of 
setting $B^{(m)}$, $B_x^{(m)}$, $B_z^{(m)}$ to $A^{(2,m)}$, $A_x^{(2,m)}$, $A_z^{(2,m)}$.

(E) It remains to check equivalence for all outcomes, including random bit vectors $r_1,r_2$.
Assuming \cref{eq:random-shift,eq:measurement-phases} and introducing $r = r_1 + r'_1$, the equality 
of sign vectors  $s(\mathcal{C}_1''')$, $s(\mathcal{C}_2)$ simplifies to 
$$
  \left(\begin{array}{c} B_{\ast,[n_r^{(1)}]} \\ \hline  (B_x)_{\ast,[n_r^{(1)}]} \\ \hline  (B_z)_{\ast,[n_r^{(1)}]} \end{array}\right) r = 
  \left(\begin{array}{c} A^{(2)}_{\ast,[n_r^{(1)}]} \\ \hline  (A_x^{(2)})_{\ast,[n_r^{(1)}]} \\ \hline  (A^{(2)}_z)_{\ast,[n_r^{(1)}]} \end{array}\right) r_2
$$
for all possible values of $r,r_2$. In other words, the image spaces of the stacked matrices must be the same. 
We check for this by using the compression map~\cref{alg:compress-randomness} in~\cref{line:comparison-outcome-compress-a}.
We then compare the results of applying the compression maps in~\cref{line:comparison-conditional-pauli-on-random-fail}.
If the compression maps are applied to the general form circuits with condition matrices $B_{\ast,[n_r^{(1)}]}, (B_x)_{\ast,[n_r^{(1)}]},(B_z)_{\ast,[n_r^{(1)}]}$,
$A^{(2)}_{\ast,[n_r^{(1)}]}, (A_x^{(2)})_{\ast,[n_r^{(1)}]},(A^{(2)}_z)_{\ast,[n_r^{(1)}]}$ result in the same new condition matrices, we use compression matrices $F_1,F_2$
to establish map correspondence between $r$ and $r_2$, that is $F_1 r = F_2 r_2$. We use this to construct $M_1, M_2$ in \cref{line:comparison-ok}.
When the equality in~\cref{line:comparison-conditional-pauli-on-random-fail} is false, the effect of the correction unitary is to 
replace matrices $B_{\ast,[n_r^{(1)}]}, (B_x)_{\ast,[n_r^{(1)}]},(B_z)_{\ast,[n_r^{(1)}]}$ with matrices $A',A'_x,A'_z$ padded to match the size
of the matrices they are replacing. 
After step (E) we have ensured that signs in~\cref{eq:sign-vectors} match and so the algorithm returns \texttt{True}.

We analyze the \runtime{} of the general form circuit comparison~\cref{alg:general-form-circuit-comparison} in terms of $n = \max(\ki,\ko)$
and the number of random bits of the general form circuits being compared $\nr^{(1)}, \nr^{(2)}$.
We show below that the algorithm runtime is 
\begin{equation} \label{eq:general-form-comparison-runtime}
O(n\cdot (n^{\omega-1} + \nr^{(1)} + \nr^{(2)}) ).
\end{equation}
The stabilizer group equivalence check for $\comm{\ki,k_1,\Ci_1}$ and $\comm{\ki,k_2,\Co_2}$ and also for $\comm{\ko,k_1,D_1}$ and $\comm{\ko,k_2,D_2}$ relies on~\cref{proc:compare-encoding-unencoding-circuits} and 
contributes $O(n^\omega)$. 
The equivalence check up to Pauli unitaries \cref{line:comparison-clifford-fail,line:comparison-pauli-shift} contributes $O(n^\omega)$.

Simplifying the unitary acting on the inner qubits contributes $O(n_r^{(1)}n^{\omega -1} + n^{\omega})$.
This is because computing matrix products of rectangular and square matrices in \cref{line:comparison-b-matrix,line:comparison-bxt-matrix}
contributes $O(n_r^{(1)}n^{\omega -1} + n^{\omega})$~(see~\cref{eq:rectangular-matrix-multiplication-copmplexity}).
The \runtime{} of \cref{line:comparison-delta-s} is the  \runtime{} of matrix vector multiplication and is negligible in comparison 
with matrix-matrix multiplications.
The \runtime{} of \cref{line:comparison-bx-matrix} is of finding square and rectangular matrix products and is $O(n_r^{(1)}n^{\omega -1} + n^{\omega})$.

Computing the outcome compression maps in \cref{line:comparison-outcome-compress-a,line:comparison-outcome-compress-b} is $O(n_r^{(2)}n^{\omega -1})$
and $O(n_r^{(1)}n^{\omega -1})$ according to~\cref{eq:compression-map-runtime}.
The rest of the steps are matrix comparisons~\cref{line:comparison-conditional-pauli-on-meas-fail,line:comparison-conditional-pauli-on-random-fail}, and solving a linear equation for a matrix in reduced
row echelon form~\cref{line:comparison-random-shift}, with \runtime{} negligible compared with the other steps.
In summary, the \runtime{} of the general form comparison algorithm is given by \cref{eq:general-form-comparison-runtime} as required.

    
    

\newpage
\section{Logical action of stabilizer circuits}
\label{sec:logical-general-form}

Given a circuit that implements a logical operation on encoded information, it can be useful to identify what that logical action is. 
It may also be desirable to verify that the logical action matches that of some known reference circuit. 
For example, one may wish to check that a circuit performs some lattice surgery operations on two patches of surface code implements the joint $XX$ or $ZZ$ measurement of the logical qubits.

In \cref{sec:general-form-circuit-for-logical-action} we provide \cref{alg:logical-general-form} which finds a general form circuit with action matching the logical action of an input logical operation circuit.
This algorithm also flags when the input circuit is not a logical operation circuit, and we explain how to identify a correction unitary to append to the input circuit to rectify this.
Then, in \cref{sec:logical-action-verification} we provide \cref{alg:logical-action-verificiaton} which checks if an input logical operation circuit $\mathcal{C}$ implements the same action as a given reference circuit $\mathcal{C}_\text{ref}$.
This makes use of \cref{alg:logical-general-form} and \cref{alg:general-from-choi} by first finding general form circuits equivalent to $\mathcal{C}_\text{ref}$ and to the logical action of $\mathcal{C}$, before using \cref{alg:general-form-circuit-comparison} to compare these two general form circuits.

\subsection{General form circuit for logical action}
\label{sec:general-form-circuit-for-logical-action}

Building on the definitions in \cref{sec:encoding-circuits}, we provide a formal definition of the problem we wish to solve:

\begin{problem}[Logical action]
\label{prob:general-form-for-logical-action}
Consider a stabilizer circuit $\mathcal{C}$ with outcome vector $v$, $\ni$ input qubits and $\no$ output qubits.
Further consider the stabilizer codes $\comm{\ki,k_\inn,C_\inn}$ and $\comm{\ko,k_\out,C_\out}$.
Check if $\mathcal{C}$ is a logical operation circuit with input and output codes $\comm{\ki,k_\inn,C_\inn}$ and $\comm{\ko,k_\out,C_\out}$.
If $\mathcal{C}$ is a logical operation circuit, find a general form circuit $\mathcal{C}_\mathcal{L}$ with outcome vector $o_\mathcal{L}$, full column rank matrix $M_\mathcal{L}$ and vector $v_{\mathcal{L},0}$
defining the equivalent outcomes such that the logical action of $\mathcal{C}$ upon outcome $v = v_{\mathcal{L},0} + M_\mathcal{L} o_\mathcal{L}$ has the same as action as $\mathcal{C}_\mathcal{L}$ upon outcome $o_\mathcal{L}$.
\end{problem}

Next we describe \cref{alg:logical-general-form}, which solves this problem in three main steps.
In the first step, the algorithm computes a general form circuit $\mathcal{C}'_\text{gen}$ for the circuit $\mathcal{C}' = \mathcal{E}(k_\inn,C_\inn) \circ \mathcal{C} \circ \mathcal{E}^\dagger(k_\out,C_\out)$, which encodes into the input code, applies the circuit, and then unencodes from the output code (see \cref{fig:logical-action-general-form}(a)).
The second step analyzes $\mathcal{C}'_\text{gen}$ to check that the circuit $\mathcal{C}$ is indeed a logical operation circuit. 
We show that this check can be performed by inspecting the matrix $M'$ which specifies the map between the outcomes of $\mathcal{C}'$ and the equivalent general form circuit $\mathcal{C}'_\text{gen}$.
At this point, it may seem that the problem has been solved since we have a general form circuit and a map specifying its equivalence to $\mathcal{C}'$, which above we claimed which encodes into the input code, applies the circuit, and then unencodes from the output code. 
However, the encoding circuit at the start of $\mathcal{C}'$ more precisely encodes into a code \emph{equivalent to}  the input code $\comm{\ki,k_\inn,C_\inn}$, but with a random syndrome $s_\inn$.
We seek a general form circuit $\mathcal{C}_\mathcal{L}$ which has equivalent action to $\mathcal{C}'$ when $s_\inn$ is trivial, which corresponds to the case where $\mathcal{C}'$ begins by encoding into precisely the code $\comm{\ki,k_\inn,C_\inn}$.
Below we will see that the incorporation of $s_\inn$ strengthens the power of the \cref{alg:logical-general-form} considerably.

To find a general form circuit $\mathcal{C}_\mathcal{L}$ equivalent to the logical action of $\mathcal{C}$, the third step of \cref{alg:logical-general-form} is to split the matrix $M'$ and the condition matrices $A',A'_x,A'_z$ of $\mathcal{C}'_\text{gen}$ into two components (see \cref{fig:logical-action-general-form}(b)). 
The set of first components of these matrices specifies the matrices needed to define $\mathcal{C}_\mathcal{L}$, while the set of second components of these matrices specifies the difference between the action of $\mathcal{C}'$ when the syndrome $s_\inn$ is zero (corresponding to the logical action of $\mathcal{C}$) and when it is not zero (see \cref{fig:logical-action-general-form}(c)).
The first two outputs of \cref{alg:logical-general-form}, marked as (1) and (2), are the solution to \cref{prob:general-form-for-logical-action} and the last two outputs, marked as (3) and (4), describe how the action is modified when $s_\inn$ is non-zero.
Note that the information in (3) and (4) is not explicitly called for by \cref{prob:general-form-for-logical-action} but it can be very useful to know the logical action of a circuit for any input code syndrome $s_\inn$, for example to analyze its fault-tolerance properties.

\begin{figure}[htp]
    \centering
     \begin{subfigure}[b]{\textwidth}
        \centering
        \loadfig{fig-logical-instrument-conjugation}
        \caption{Circuit $\mathcal{C}'$, defining the logical action of $\mathcal{C}$ for input and output codes $\comm{\ni,k_\inn,C_\inn}$, $\comm{\no,k_\out,C_\out}$.}
    \end{subfigure}
    \vskip 7mm
    \begin{subfigure}[b]{\textwidth}
        \centering
        \loadfig{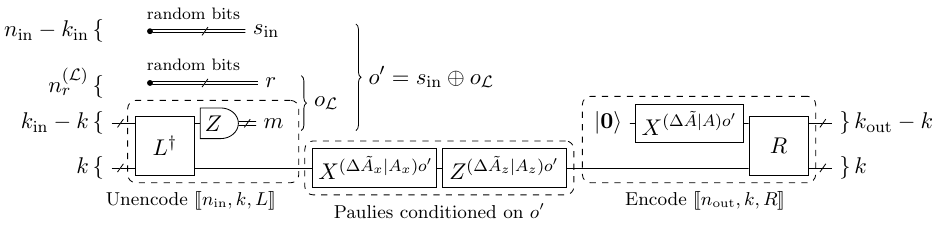}
        \caption{A general form circuit $\mathcal{C}'_\text{gen}$, equivalent to $\mathcal{C}'$ with outcome map $v' = M' o' + v'_0$. 
        For convenience, we split the random bit register as $r' = s_\inn \oplus r$ and split the condition matrices as $A' = (\Delta \tilde A | A)$, $A'_x = (\Delta \tilde A_x | A_x)$ and $A'_Z = (\Delta \tilde A_z | A_z)$, into $(\ni-k_\inn)$-column and  $(n^{(\mathcal{L})}_r + (k_\inn - k))$-column components.
        }
    \end{subfigure}
    \vskip 7mm
    \begin{subfigure}[b]{\textwidth}
        \centering
         \loadfig{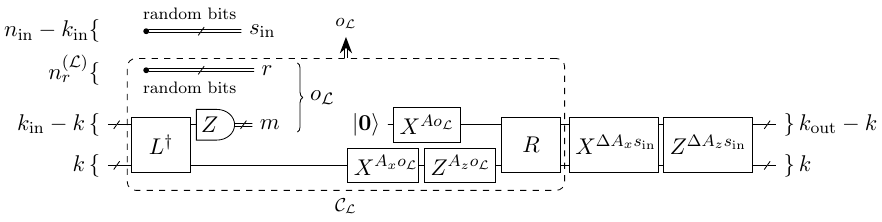}
        \caption{A circuit which is equivalent to $\mathcal{C}'$. 
        This circuit is built from a general form circuit $\mathcal{C}_\mathcal{L}$, which is equivalent to $\mathcal{C}'$ when $s_\inn$ is trivial, followed by Paulies conditioned on $s_\inn$ with outcome vector $s_\inn \oplus o_\mathcal{L}$.
        }
    \end{subfigure}
    \caption[Logical action general form circuit]{The circuits in (a) and (c) are equivalent with outcome map $s_\inn \oplus v \oplus s_\out = s_\inn \oplus (M_\mathcal{L} o_\mathcal{L} + v_{\mathcal{L},0} + \Delta M s_\inn) \oplus \Delta A s_\inn$ provided that $\mathcal{C}$ is a logical operation circuit.
    }
    \label{fig:logical-action-general-form}
\end{figure}
\begin{figure*}[htp]
\begin{algorithm}[\texttt{Logical action}] \label{alg:logical-general-form}
\begin{algorithmic}[1]
\Blank
\Input 
\begin{itemize}[noitemsep]
    \item A stabilizer circuit $\mathcal{C}$ with $\ki$ input qubits, $\ko$ output qubits and $\nM$ outcomes,
    \item Stabilizer codes $\comm{\ki,k_\inn,C_\inn}$and $\comm{\ko,k_\out,C_\out}$.
\end{itemize}
\Output One of the following:
\begin{itemize}[noitemsep]
    \item \texttt{False} (the circuit $\mathcal{C}$ is not a logical operation circuit with input code $\comm{\ki,k_\inn,C_\inn}$ and output code $\comm{\ko,k_\out,C_\out}$).
    \item \texttt{True} (the circuit is a logical operation circuit) and the following:
    \begin{itemize}[noitemsep]
        \item[(1)] General form circuit $\mathcal{C}_\mathcal{L}$ equivalent to  $\mathcal{C}' = \mathcal{E}(k_\inn,C_\inn) \circ \mathcal{C} \circ \mathcal{E}^\dagger(k_\out,C_\out)$ when the random outcome of the encoding circuit $\mathcal{E}(k_\inn,C_\inn)$ is $s_\inn = 0$,
        \item[(2)] Logical outcome map $o_\mathcal{L} \mapsto M_\mathcal{L} o_\mathcal{L} + v_{\mathcal{L},0} $ from outcomes of $\mathcal{C}_\mathcal{L}$ to outcomes of $\mathcal{C}$,
        \item[(3)] Difference operator: $k_\out$-qubit conditional Pauli unitary $X^{(\Delta A_x) s_\inn} Z^{(\Delta A_z) s_\inn}$,
        \item[(4)] Difference map: $\Delta v = \Delta M s_\inn$, output syndrome map $s_\out = (\Delta A) s_\inn$.
    \end{itemize}
\end{itemize}
Upon the output \texttt{True}, the outputs of the algorithm satisfy conditions in \cref{fig:logical-action-general-form}.
Additionally, the matrix $M_\mathcal{L}$ has full column rank and is in $(\nr^{(\mathcal{L})},\nm^{(\mathcal{L})})$-split reduced form for $\nr^{(\mathcal{L})},\nm^{(\mathcal{L})}$ being the 
number of random bit and measurement outcomes of $\mathcal{C}_\mathcal{L}$.
The entries of $v_{\mathcal{L},0}$ corresponding to the row rank profile of the left and right parts of $M_\mathcal{L}$ are zero.
\Blank
\State Construct stabilizer circuit $\mathcal{C}' = \mathcal{E}(k_\inn,C_\inn) \circ \mathcal{C} \circ \mathcal{E}^\dagger(k_\out,C_\out)$.
\State\label{line:logical-operation-general-form}Find a general form $\mathcal{C}'_\mathrm{gen}$ of $\mathcal{C}'$ (\cref{alg:general-from-choi})
(that is find left and right Clifford unitaries $\Ci,\Co$, condition matrices $A',A'_x,A'_z$), matrix $M'$ in $(\nr',\nm')$-split reduced echelon form
and vector $v'_0$ which determine the mapping between outcomes of $\mathcal{C}'_\mathrm{gen}$ and  $\mathcal{C}'$.
\Blank \hrulefill \Comment Check if $\mathcal{C}$ is a logical operation circuit
\For{ $j \in [\no-k_\out]$ } \Comment Check if $s_\out$ are redundant outcomes dependent only on $s_\inn$
\If{$(\ni-k_\inn) + \nM + j$ is in row rank profile of $M'_{*,[\nr']}$}\label{line:logical-operation-random-outcome-check}
\State\label{line:logical-operation-random-outcome-fail}\Return \texttt{False} \Comment Syndrome outcome $(s_\out)_j$ is a random outcome
\EndIf
\If{$(\ni-k_\inn) + \nM + j$ is in row rank profile of $M'_{*,[\nr'+1,\nr' + \nm']}$}\label{line:logical-operation-inp-dep-outcome-check}
\State\label{line:logical-operation-inp-dep-outcome-fail}\Return \texttt{False} \Comment Syndrome outcome $(s_\out)_j$ is an input-dependent outcome
\EndIf
\If{ $ M'_{(\ni-k_\inn) + \nM + j,[\nr'-(\ni-k_\inn)+1,\nr']}  \ne 0 $}\label{line:logical-operation-random-outcome-dependency-check}
\State\label{line:logical-operation-random-outcome-dependency-fail}\Return \texttt{False} \Comment Syndrome outcome $(s_\out)_j$ depends on random outcomes of $\mathcal{C}$
\EndIf
\EndFor
\If{ $(v'_0)_{(\ni-k_\inn) + \nM + [\no-k_\out]} \ne 0 $}\label{line:logical-operation-non-zero-check}
\State\label{line:logical-operation-non-zero-fail}\Return \texttt{False} \Comment Syndrome $s_\out$ is not zero when syndrome $s_\inn$ is zero
\EndIf
\State\label{line:logical-operation-input-output-syndrome-map}$\Delta A \leftarrow M'_{(\ni-k_\inn) + \nM + [\no-k_\out],[\ni-k_\inn]}$
\Blank \hrulefill \Comment Construct outputs (1)-(4)
\State\label{line:logical-operation-condition-matrices}Let $\mathcal{C}_\mathcal{L}$ be a general form circuit with left and right Clifford unitaries $\Ci,\Co$ 
and condition matrices $A = A'_{\ast,[\ni-k_\inn+1,\cdot]}$, $A_\sigma = (A'_\sigma)_{\ast,[\ni-k_\inn+1,\cdot]}$ for $\sigma \in \{x,z\}$.
\State\label{line:logical-operation-pauli-difference}Find $\Delta A_x, \Delta A_z$ from the equality \Comment{\cref{prop:batch-pauli-images}} 
$$
 X^{(\Delta A_x) s} Z^{(\Delta A_z) s} \simeq \Co \left( X^{(\Delta \tilde A) s} \otimes \left( X^{(\Delta \tilde A_x) s }  Z^{(\Delta \tilde A_z) s }  \right) \right) \Co^\dagger,
$$
where $\Delta \tilde A = A'_{\ast,[\ni-k_\inn]}$, $\Delta \tilde A_x = (A'_x)_{\ast,[\ni-k_\inn]}$, and $\Delta \tilde A_z = (A'_z)_{\ast,[\ni-k_\inn]}$.
\State\label{line:logical-operation-outcome-map}$M_\mathcal{L} \leftarrow M'_{[\ni - k_\inn + 1,\ni - k_\inn + \nM],[\ni - k_\inn + 1,\cdot]}$, $v_{L,0} = (v'_0)_{[\ni - k_\inn + 1,\ni - k_\inn+\nM]}$
\State\label{line:logical-operation-outcome-difference}$\Delta M \leftarrow M'_{[\ni - k_\inn + 1,\ni - k_\inn + \nM],[\ni - k_\inn]}$
\State \Return \texttt{True}, (1) $\mathcal{C}_\mathcal{L}$, (2) $M_\mathcal{L},v_{\mathcal{L},0}$, (3) $\Delta A_x, \Delta A_z$, (4) $\Delta A, \Delta M$.
\end{algorithmic}
\end{algorithm}
\end{figure*}

Next we explain the correctness of~\cref{alg:logical-general-form}.
While working through this explanation, we also describe how to find a correction unitary as specified in~\cref{tab:logical-action-correction-unitaries} which can be appended to the input circuit $\mathcal{C}$ when the algorithm returns \texttt{False} (due to $\mathcal{C}$ not being a logical operation circuit) in order to modify $\mathcal{C}$ to form a logical operation circuit.

Let us first discuss in more detail how to check that $\mathcal{C}$ is a logical operation circuit  using the general form circuit $\mathcal{C}'_\text{gen}$ and the outcome mapping defined by the matrix $M'$ for the circuit $\mathcal{C}' = \mathcal{E}(k_\inn,C_\inn) \circ \mathcal{C} \circ \mathcal{E}^\dagger(k_\out,C_\out)$ (see \cref{fig:logical-action-general-form}(a)).
Consider a general form circuit $\mathcal{C}''_\text{gen}$ of $\mathcal{C}'' = \mathcal{E}(k_\inn,C_\inn) \circ \mathcal{C}$ (which is the sub-circuit of $\mathcal{C}'$ consisting of $\mathcal{E}(k_\inn,C_\inn)$ followed by $\mathcal{C}$). 
For the circuit $\mathcal{C}$ to be a logical operation circuit with input code $\comm{\ni,k_\inn,C_\inn}$ 
and output code $\comm{\no,k_\out,C_\out}$, it is necessary that the left stabilizer group of $\mathcal{C}''_\text{gen}$
contains $\comm{\no,k_\out,C_\out}$ up to signs.
In other words, it is necessary that each outcome bit in $s_\out$ of $\mathcal{C}'$ is a redundant outcome.
We check for this necessary condition in \cref{line:logical-operation-random-outcome-check,line:logical-operation-inp-dep-outcome-check}
by ensuring that each outcome bit in $s_\out$ is neither random, nor input-dependent, making use of the fact that $M'$ is in split reduced echelon form (see \cref{def:split-echelon-form} and \cref{fig:outcome-mapping-structure}). 

The logical operation circuit~\cref{def:logical-operation-circuit} requires that the circuit output is encoded in $\comm{\no,k_\out,C_\out}$ when the circuit input is encoded in $\comm{\ni,k_\inn,C_\inn}$. 
To ensure this requirement, it is necessary that outcomes in $s_\out$ do not depend on the random outcomes of circuit $\mathcal{C}$.
We check this in \cref{line:logical-operation-random-outcome-dependency-check}.
We can also amend $\mathcal{C}$ if this necessary condition is not satisfied by including the conditional Pauli correction at the end of $\mathcal{C}$ as shown in \cref{tab:logical-action-correction-unitaries}, resulting in a new circuit which satisfies the condition.

Even when all the aforementioned necessary conditions are met, it is possible that $\mathcal{C}$ is not a logical operation circuit as it could map states encoded in $\comm{\ni,k_\inn,C_\inn}$ to states encoded in a code which is equivalent to but not equal to $\comm{\no,k_\out,C_\out}$ (i.e. with a fixed non-zero syndrome).
We check that this is not the case in \cref{line:logical-operation-non-zero-check}.
It this check fails, we can correct for it by appending the Pauli unitaries shown in \cref{tab:logical-action-correction-unitaries} to the end of $\mathcal{C}$.
When all the necessary conditions we have discussed are met, it must be that $\mathcal{C}$ is a logical operation circuit and the output code syndrome $s_\out$ must be a linear function 
of the input code syndrome $s_\inn$, that is $s_\out = (\Delta A) s_\inn$.

The matrix $\Delta A$ is extracted from $M'$ in \cref{line:logical-operation-input-output-syndrome-map}.
This is possible because the matrix $M'$ is in split reduced echelon form and those entries of $v'_0$ which correspond to random outcomes are zero and so the first $\ni - k_\inn$ bits of the outcome vector of $\mathcal{C}'$ that are equal to the input code syndrome $s_\inn$ are also equal to the first $\ni - k_\inn$ random bits of the general form circuit $\mathcal{C}'_\text{gen}$.
For the same reason, the first $\ni - k_\inn$ columns of the condition matrices $A',A'_x,A'_z$ determine the dependence of the action of $\mathcal{C}'$ 
on the input code syndrome $s_\inn$. 
We use this to compute the difference Pauli operators in~\cref{line:logical-operation-pauli-difference}.
The remaining columns of $A',A'_x,A'_z$ are precisely the condition matrices of $\mathcal{C}_\mathcal{L}$~(\cref{line:logical-operation-condition-matrices}).

The mapping between the outcomes of the general form circuit $\mathcal{C}_\mathcal{L}$ and the outcomes of $\mathcal{C}$ is given by a sub-matrix $M_L$ of $M'$ in 
\cref{line:logical-operation-outcome-map}. 
This is because circuit $\mathcal{C'}$ starts with allocating $\ni-k_\inn$ random bits and top-left $(\ni-k_\inn)\times(\ni-k_\inn)$ block is identity.
This structure and the fact that all outcomes in $s_\out$ are redundant also implies that the sub-matrix $M_L$ is in split reduced echelon form.
Similarly, the outcome difference map $\Delta M$ for when input code syndrome is non-trivial is captured by the first $\ni-k_\inn$ columns of $M'$ as in \cref{line:logical-operation-outcome-difference}.

\begin{table}[htp]
    \centering
    \begin{tabular}{|l|c|c|}
    \hline 
    Case when $\mathcal{C}$  & \multirow{3}{*}{\rotatebox[origin=c]{-90}{Line}} & Correction unitary $U$ to ensure  \\
    is not a logical         &      & that $\mathcal{C}$ is a logical operation circuit \\
    operation circuit        &      & via appending $U$ to $\mathcal{C}$ \\
    \hline 
    \hline 
    Output syndrome $s_\out$  & \multirow{4}{*}{\labelcref{line:logical-operation-random-outcome-dependency-fail}} & {Let $\Delta a= M'_{(\ni-k_\inn) + \nM + j,[\nr'-(\ni-k_\inn)+1,\nr']}$~(\cref{line:logical-operation-random-outcome-dependency-check})} \\
    depends on random & &  $U = C_\out^\dagger X_j C_\out$ conditioned on the random outcomes of \\
    outcomes of $\mathcal{C}$ & & $\mathcal{C}$ indicated by $\Delta a$ \\ 
    \hline 
    Output syndrome $s_\out$ & \multirow{3}{*}{\labelcref{line:logical-operation-non-zero-fail}} & Let $\Delta v = (v'_0)_{(\ni-k_\inn) + \nM + [\no-k_\out]}$ ~(\cref{line:logical-operation-non-zero-check})\\
    when input syndrome & & $U = C_\out^\dagger( X^{\Delta v} \otimes I_{k_\out}  )C_\out$ \\
    $s_\inn$ is zero & &   \\
    \hline 
    \end{tabular}
    \caption[Logical operation general form recovery from fails]{\label{tab:logical-action-correction-unitaries}
             Cases when computing the logical action of stabilizer circuit $\mathcal{C}$ using \cref{alg:logical-general-form} fails a check and returns \texttt{False}
             and there exists a correction unitary to append to $\mathcal{C}$ to ensure that check is passed.
             Variables in the equations for the correction unitaries are defined in the pseudo-code of \cref{alg:logical-general-form}.
            } 
\end{table}

Next we discuss the \runtime{} of  \cref{alg:logical-general-form}.
The \runtime{} of \cref{alg:logical-general-form} is limited by the \runtime{} of \cref{line:logical-operation-general-form}:
\begin{equation}
\label{eq:logical-operation-general-form-runtime}
O\left(\nmax\cdot\left(\Nu+(\Ncnd+\nM)(\nr+\nm + (\ni - k_\inn) + (\no-k_\out) + 1) + \nmax(\nmax + \nr)\right)\right).    
\end{equation}
This equation follows from \cref{eq:general-form-runtime} and the following bounds on the parameters of $\mathcal{C}'$.
The circuit $\mathcal{C'}$ has $\nM + \no - k_\out$ outcomes, of which at most $n_r + \ni - k_\inn + \no - k_\out$ are random and at most $\nm + \no - k_\out$ are input-dependent.
We also include the \runtime{} $O(\nmax^3)$ of simulating the unitaries $C_\inn$ and $C_\out^\dagger$ when computing the general form $\mathcal{C}'_\text{gen}$.

\cref{alg:logical-general-form} has a number of applications.
For example, we can use the additional outputs (3) and (4) of the algorithm to determine which input code syndromes do not cause logical errors and classify all syndromes according to the logical errors they cause.
First, code syndromes that do not cause logical errors must not cause Pauli X or Z errors on the logical qubits, so such $s_\inn$ must belong to the intersection of kernels of $\Delta A_x$,  $\Delta A_z$ returned as output (3).
Additionally, these syndromes must not flip the logical measurement outcomes,
that is they should belong to the kernel of $M_L^{(-1)} (\Delta M ) s_\inn$, where $\Delta M$ is returned as part of output (4) of the algorithm.
The syndromes of the input code that do not cause logical errors form a linear space $L_0$.
Taking the quotient of the space of all possible syndromes by $L_0$ lets us partition 
syndromes into the sets that cause the same logical error.

\subsection{Logical action verification}
\label{sec:logical-action-verification}

It may be desirable to verify that the logical action of given logical operation circuit matches the action of some known reference circuit. 
We can state the logical operation circuit comparison problem in a form that is similar to the stabilizer circuit comparison \cref{prob:stabilizer-circuit-comparison}:

\begin{problem}[Logical circuit comparison]
\label{prob:logical-action-verification}
Consider:
\begin{itemize}[noitemsep]
    \item a stabilizer circuit $\mathcal{C}$ with outcome vector $v$, $\ni$ input qubits and $\no$ output qubits,
    \item an input code $\comm{\ni,k_\inn,C_\inn}$ and an output code $\comm{\ko,k_\out,C_\out}$,
    \item a reference stabilizer circuit $\mathcal{C}_\text{ref}$ with outcome vector $v_\text{ref}$, $k_\inn$ input and $k_\out$ output qubits.
\end{itemize}

Find if $\mathcal{C}$ is a logical operation circuit~(\cref{def:logical-operation-circuit}) with input code $\comm{\ni,k_\inn,C_\inn}$ and output code $\comm{\ko,k_\out,C_\out}$.
If it is, find if the logical action~(\cref{def:logical-action}) of $\mathcal{C}$ is equivalent to the action of $\mathcal{C}_\text{ref}$.
If the actions are equivalent, find matrices $M,M_\text{ref}$ and vectors $u,u_\text{ref}$ such that
the logical action of $\mathcal{C}$ upon outcome $v$ is the same as the action of $\mathcal{C}_\text{ref}$ upon outcome $v_\text{ref}$ if and only if $M(v + u) = M_\text{ref}(v_\text{ref} + u_\text{ref})$.
\end{problem}

The main difference between the above problem and the stabilizer circuit comparison \cref{prob:stabilizer-circuit-comparison} 
is that in \cref{prob:stabilizer-circuit-comparison} we must first check if $\mathcal{C}$ is a logical operation circuit and then analyze its logical action with respect to the input and output codes.
Similarly to the stabilizer circuit comparison \cref{prob:stabilizer-circuit-comparison}, when the circuits are equivalent and the matrix $M_\text{ref}$ has a left inverse, we can find a map between the outcome vector $v$ of $\mathcal{C}$ and the outcome vector $v_\text{ref}$ of $\mathcal{C}_\text{ref}$ as
$$
 v_\text{ref} = u_\text{ref} + M^{(-1)}_\text{ref}M(v+u).
$$
Consider an example in which one wishes check that a circuit which performs a lattice surgery operation on two patches of surface code implements the joint $XX$ measurement of the logical qubits encoded in the patches.
In this case, the input circuit can be highly complex, producing a long output vector $v$ containing a long list of measurement outcomes.
The reference circuit on the other hand consists of an $XX$ measurement on a pair of qubits, which produces a single output bit $v_\text{ref}$, and as such it must be that $M_\text{ref} = 1$. 
This trivially has a left inverse $M^{(-1)}_\text{ref} = I_1$, and the above map will provide the lattice surgery circuit outcomes corresponding the logical $XX$ or $ZZ$ measurement outcome.

Below we present a logical action verification \cref{alg:logical-action-verificiaton}, which solves the logical action verification \cref{prob:logical-action-verification} in three main steps.
The first step is to verify that the circuit circuit $\mathcal{C}$ is a logical operation circuit and find a general form circuit $\mathcal{C}_\mathcal{L}$ equivalent to its logical action~(\cref{def:logical-action}).
The second step is to find a general form circuit equivalent to $\mathcal{C}_\text{ref}$ using \cref{alg:general-from-choi}.
The third step is to compare the logical action general form circuit $\mathcal{C}_\mathcal{L}$ and the general form circuit equivalent to $\mathcal{C}_\text{ref}$ 
using the general form comparison~\cref{alg:general-form-circuit-comparison}.

\begin{algorithm}[\texttt{Logical action verification}]
\label{alg:logical-action-verificiaton}
\begin{algorithmic}[1]
\Blank
\Input 
\begin{itemize}[noitemsep]
    \item A stabilizer circuit $\mathcal{C}$ with $\ni$ input qubits, $\no$ output qubits.
    \item Stabilizer codes $\comm{\ni,k_\inn,C_\inn}$,  $\comm{\no,k_\out,C_\out}$.
    \item A stabilizer circuit $\mathcal{C}_\text{ref}$ with $k_\inn$ input qubits, $k_\out$ output qubits.
\end{itemize}
\Output One of the following:
\begin{itemize}[noitemsep]
    \item \texttt{False} (the circuit is not a logical operation circuit with input code $\comm{\ki,k_\inn,C_\inn}$ and output code $\comm{\ko,k_\out,C_\out}$ equivalent to $\mathcal{C}_\text{ref}$).
    \item \texttt{True} (the circuit is a logical operation circuit equivalent to $\mathcal{C}_\text{ref}$) and matrices $M,M_\text{ref}$, vectors $u,u_{\text{ref}}$ 
    defining the equivalent outcomes as in~\cref{prob:logical-action-verification}.
\end{itemize}
\State Solve logical action \cref{prob:general-form-for-logical-action} for  $\mathcal{C}$, $\comm{\ni,k_\inn,C_\inn}$,  $\comm{\no,k_\out,C_\out}$
\Blank \Comment \cref{alg:logical-general-form}
\If{$\mathcal{C}$ is a not a logical operation circuit}
\State\label{line:logical-action-verification-logical-action-fail}\Return \texttt{False} \Comment \cref{tab:logical-action-correction-unitaries}
\Else 
    \State Let $\mathcal{C}_\mathcal{L}$ be a general form circuit for logical action of $\mathcal{C}$ 
    \Statex[2] with the outcome mapping between $\mathcal{C}$ and $\mathcal{C}_\mathcal{L}$ defined by $M_L$, $v_{0,L}$
\EndIf
\State Let $\mathcal{C}_\text{gen}$ be a general form circuit equivalent to $\mathcal{C}_\text{ref}$, \Comment \cref{alg:general-from-choi}
\Statex[1]  with the outcome mapping defined by $M_\text{gen},v_{0,\text{gen}}$.
\State Compare general form circuits $\mathcal{C}_\mathcal{L}$ and $\mathcal{C}_\text{gen}$ \Comment \cref{alg:stabilizer-circuit-comparison}
\If{$\mathcal{C}_\mathcal{L}$ and $\mathcal{C}_\text{gen}$ are equivalent with matrices $\tilde M_L, \tilde M_\text{gen}$ and vectors $\tilde u_L, \tilde u_\text{gen}$ \\ ~~~~~~defining the equivalent outcomes } 
\State\label{line:logical-action-verification}$M \leftarrow \tilde M_L M_L^{(-1)}$,
$u \leftarrow \tilde u_L + M^{(-1)}_L v_{0,L}$,
$M_\text{ref} \leftarrow \tilde M_\text{gen} M_\text{gen}^{(-1)}$,
$u_\text{ref} \leftarrow \tilde u_\text{gen} + M^{(-1)}_\text{gen} v_{0,\text{gen}}$
\State \Return \texttt{True}, $M$, $M_\text{ref}$, $u$, $u_\text{ref}$
\Else 
\State\label{line:logical-action-verification-inequivalence-fail}\Return \texttt{False} \Comment \cref{tab:equality-correction-unitaries}
\EndIf
\end{algorithmic}
\end{algorithm}

Now we briefly discuss the correctness of \cref{alg:logical-action-verificiaton}. 
The logical action of $\mathcal{C}$ is equivalent to the action of the reference circuit if and only if corresponding general form circuits are equivalent.
The expressions for the matrices $M,M_\text{ref}$ and vectors $u,u_\text{ref}$ follow from: (1) expressing the outcome vector $o_L$ of $\mathcal{C}_\mathcal{L}$
in terms of the outcome vector $v$ of $\mathcal{C}$ as 
$
 o_L = M_L^{(-1)}(v+v_{0,L})
$, 
(2) expressing the outcome vector $o$ of $\mathcal{C}_\text{gen}$
in terms of the outcome vector $v_\text{ref}$ of $\mathcal{C}_\text{ref}$ as 
$
o = M_\text{gen}^{(-1)}(v_\text{ref}+v_{0,\text{gen}})
$ and (3) substituting expressions for $o$ and $o_L$ into the condition $\tilde M_L(o_l + \tilde u_l) =  \tilde M_\text{gen}(o + \tilde u_\text{gen})$
for when the action of $\mathcal{C}_\mathcal{L}$ upon $o_L$ is equivalent to the action of $\mathcal{C}_\text{gen}$ upon $o$.

The \runtime{} of our logical action verification \cref{alg:logical-action-verificiaton} is dominated by the general form circuit computation subroutines. 
We parameterize $\mathcal{C}$ as specified in \cref{sec:circuit-parameters}, and use the same notation to parameterize the reference circuit $\mathcal{C}_\text{ref}$ but with the superscript $^\text{(ref)}$ on each parameter.

The \runtime{} analysis of our logical action~\cref{alg:logical-general-form} presented in the previous subsection~(\cref{eq:logical-operation-general-form-runtime})
and the \runtime{} analysis of our general form \cref{alg:general-from-choi}~(\cref{eq:general-form-runtime}) and general form comparison 
\cref{alg:general-form-circuit-comparison}~(\cref{eq:general-form-comparison-runtime}) imply that the runtime of our logical action verification \cref{alg:logical-action-verificiaton} is:
\begin{equation}
    \begin{array}{c}
    O\left( \nmax\cdot\left(\Nu+(\Ncnd+\nM)(\nr+\nm + (\ni - k_\inn) + (\no-k_\out) + 1) + \nmax(\nmax + \nr)\right) + \right. \\ 
    \left. \nmax^\text{(ref)} \cdot\left(\Nu^\text{(ref)}+(\Ncnd^\text{(ref)}+\nM^\text{(ref)})(\nr^\text{(ref)}+\nm^\text{(ref)} + 1) + \nmax(\nmax^\text{(ref)} + \nr^\text{(ref)})\right) \right).
    \end{array}
\end{equation}

When the logical action verification \cref{alg:logical-action-verificiaton} fails a check and returns \texttt{False}, it can be useful to find a ``correction'' unitary that can be appended to $\mathcal{C}$
ensure that test is passed.
The algorithm returns \texttt{False} in two cases.
In the first case~(\cref{line:logical-action-verification-logical-action-fail}), the circuit $\mathcal{C}$ is not a logical operation circuit.
In the second case~(\cref{line:logical-action-verification-inequivalence-fail}), the logical action is not equivalent to the action of the reference circuit.
In the first case, it is sometimes possible to ensure that $\mathcal{C}$ is a logical operation by appending conditional Pauli unitaries as computed within \cref{alg:logical-general-form} and provided in \cref{tab:logical-action-correction-unitaries}.
In the second case, it is sometimes possible to ensure the action equivalence by appending Clifford unitaries or conditional Pauli unitaries to 
the general form circuit $\mathcal{C}_\mathcal{L}$ as described in \cref{tab:equality-correction-unitaries}.
These unitaries act on the logical qubits of the output code $\comm{\no,k_\out,C_\out}$ and are conditioned on outcomes $o_L$ of $\mathcal{C}_\mathcal{L}$. 
They can be easily translated to unitaries to append to the end of circuit $\mathcal{C}$ to ensure the circuit equivalence. 
For example, the conditional Pauli unitary  $X_j^{\ip{c_x,o_L}} Z_j^{\ip{c_z,o_L}}$ becomes 
\begin{equation}
    C_\out \left( X_j^{\ip{(M_L^{(-1)})^T c_x,v + v_{0,L}}} Z_j^{\ip{(M_L^{(-1)})^T c_z,v + v_{0,L}}} \right) C_\out^\dagger.
\end{equation}

\newpage
\newpage
\section{Logical action of code deformation circuits (analytic)}
\label{sec:code-deformation-action}

Here we consider the logical action of the class of \emph{code deformation circuits}, which measure a set of commuting Pauli operators and then apply conditional Paulis. 
Such circuits form the basis of many fault-tolerance schemes, for instance they can be used to sequentially deform stabilizer codes, to perform syndrome extraction for stabilizer and Floquet codes~\cite{Hastings2021}, and to implement lattice surgery operations. 
In this section, we first formulate a precise problem logical action of code deformation \cref{prob:logical-measure}.
While this is a special case of the logical action~\cref{prob:general-form-for-logical-action} which is solved by \cref{alg:logical-general-form}, here we present \cref{thm:logical-action-from-common-symplectic-basis} which provides an analytic solution to \cref{prob:logical-measure}.
We then provide an example of how this can be applied to analyze an infinite family of lattice surgery circuits in \cref{sec:lattice-surgery-example}.
In \cref{sec:general-form-two-groups}, we find a general form circuit equivalent to measuring two stabilizer groups, which is a result that may be of independent interest, and which we leverage in our proof of \cref{thm:logical-action-from-common-symplectic-basis} which is given in \cref{sec:two-stab-groups-gen-form}.

\begin{problem}[Logical action of code deformation]
\label{prob:logical-measure}
Consider the $n$-qubit stabilizer group $S_\inn$ and a circuit $\mathcal{M}$ that first measures a set of commuting Pauli operators that generate an $n$-qubit stabilizer group $M$, and then applies some Pauli unitaries conditioned on the measurement outcomes.
Additionally consider stabilizer codes $\comm{n,k_\inn,C_\inn}$ and $\comm{n,k_\out,C_\out}$ 
with stabilizer groups $S_\inn$ and $S_\out \subset M \cdot (S_\inn \cap M^\perp)$.

Check whether $\mathcal{M}$ is a logical operation circuit with input and output codes $\comm{n,k_\inn,C_\inn}$ and $\comm{n,k_\out,C_\out}$.
If so, find a general form circuit $\mathcal{M}_\mathcal{L}$ equivalent to the logical action of $\mathcal{M}$ (\cref{def:logical-action}) and an associated relation between the outcomes of $\mathcal{M}$ and $\mathcal{M}_\mathcal{L}$.
\end{problem}

Let us clarify what the associated relation is between the outcomes of $\mathcal{M}$ and $\mathcal{M}_\mathcal{L}$ mentioned above.
When the logical action of $\mathcal{M}$ is equivalent to a general form circuit $\mathcal{M}_\mathcal{L}$
there exists a relation $\mathcal{R}$ between the outcomes $v_M$ of $\mathcal{M}$ and the outcomes $m$ of $\mathcal{M}_\mathcal{L}$
such that the following property holds: 
The action of  $\mathcal{E}(k_\inn,C_\inn) \circ \mathcal{M} \circ \mathcal{E}^\dagger(k_\out,C_\out)$ 
upon outcome $\textbf{0} \oplus v_M \oplus \textbf{0}$ is equivalent to the action of $\mathcal{M}_\mathcal{L}$ upon outcome $m$ if and only if the relation $\mathcal{R}$ is true for $v_M$ and $m$.
We say that $\mathcal{R}$ is a relation associated with the equivalence of the logical action of the circuit $\mathcal{M}$ and the circuit $\mathcal{M}_\mathcal{L}$.

\begin{figure}[ht]
    \centering
    \loadfig{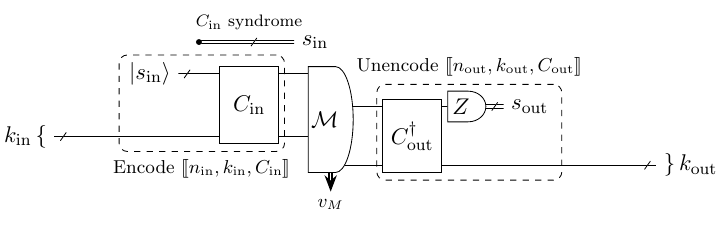}
    \caption[Circuit defining logical action of measuring a stabilizer group]{\label{fig:logical-instrument-measure}
    The circuit $\mathcal{E}(k_\inn,C_\inn) \circ \mathcal{M} \circ \mathcal{E}^\dagger(k_\out,C_\out)$ with outcome vector $v' = s_\inn \oplus v_M \oplus s_\out$ used for defining the logical action of a logical operation circuit $\mathcal{M}$ 
    with outcome vector $v_M$, input code $\comm{\ni, k_\inn, C_\inn}$ and output code $\comm{\no, k_\out, C_\out}$~(\cref{prob:logical-measure}).
    }
\end{figure}

It is useful to define a symplectic basis that clarifies how any two stabilizer groups relate to each other.
In \cref{sec:sim-diag-of-groups} we  prove any two $n$-qubit stabilizer groups have a common symplectic basis, defined as:
\begin{definition}[Common symplectic basis of two stabilizer groups]
\label{def:common-symplectic-basis}
A common symplectic basis $\mathcal{B}$ for two $n$-qubit stabilizer groups $S_\inn$ and $M$ is a symplectic basis that can be grouped into ten ordered sets $\Zdelta,\Xdelta,\Zcap,\Xcap,\Zs,\Xs,\Zm,\Xm,\Z,\X$ as 
in \cref{fig:common-symplectic-basis}, and such that $S_\inn = \ip{ \Zdelta \cup \Zcap \cup \Zs }$ and  $M = \ip{ \Xdelta \cup \Zcap \cup \Zm }$.
\end{definition}

\begin{figure}[htp]
    \centering
     \begin{subfigure}[b]{\textwidth}
        \centering
        \includegraphics[scale=0.39]{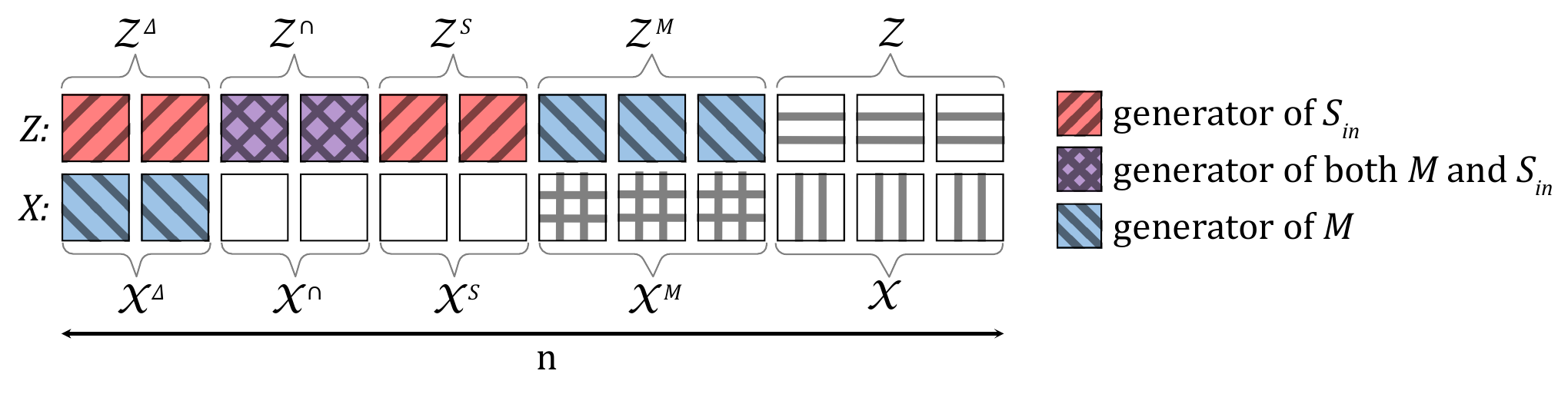}
        \caption{A common symplectic basis $\mathcal{B}$ of $n$-qubit stabilizer groups $S_\inn$ and $M$~(\cref{def:common-symplectic-basis})
        consists of $2n$ Pauli operators.
        Boxes correspond to basis elements of $\mathcal{B}$, only vertically adjacent elements anti-commute.}
        \label{fig:common-symplectic-basis}
    \end{subfigure}
    \vskip 1cm 
    \begin{subfigure}[b]{\textwidth}
        \centering
        \includegraphics[scale=0.39]{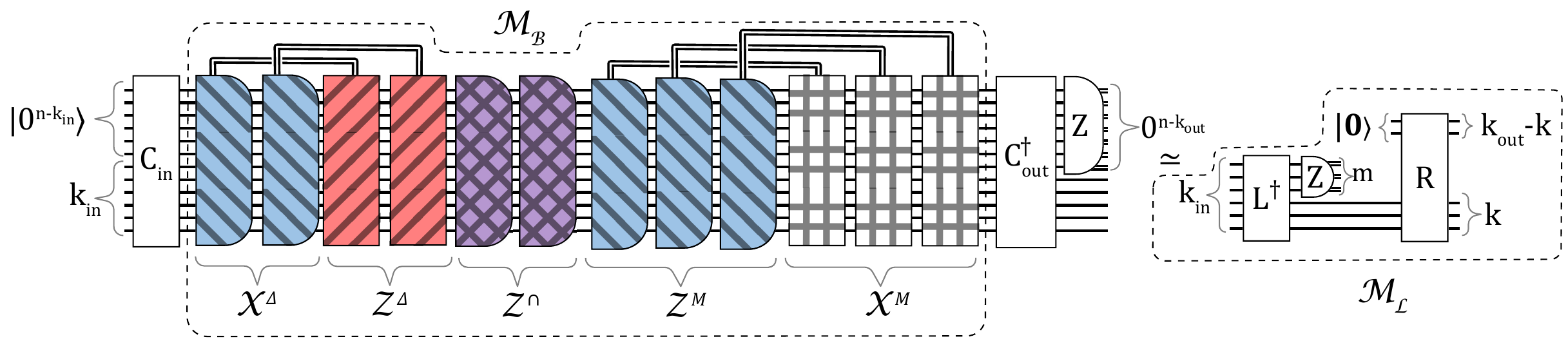}
        \caption{
        A general form circuit $\mathcal{M}_\mathcal{L}$ equivalent to the logical action of circuit $\mathcal{M}_\mathcal{B}$ with respect to input and output codes $\comm{n,k_\inn,C_\inn}$ and $\comm{n,k_\out,C_\out}$, with stabilizer groups $S_\inn$
        and $S_\out \subset M \cdot(S_\inn \cap M^\perp)$ as described in \cref{thm:logical-action-from-common-symplectic-basis}.
        The circuit $\mathcal{M}_\mathcal{B}$ measures a complete set of generators of $M$ and applies conditional Paulis
        determined by the common symplectic basis $\mathcal{B}$ of $S_\inn$ and $M$ (see in~\cref{def:common-symplectic-basis-circuit}).
        }
        \label{fig:common-symplectic-basis-circuit}
    \end{subfigure}
    \vskip 1cm
    \begin{subfigure}[b]{\textwidth}
        \centering
        \includegraphics[scale=0.39]{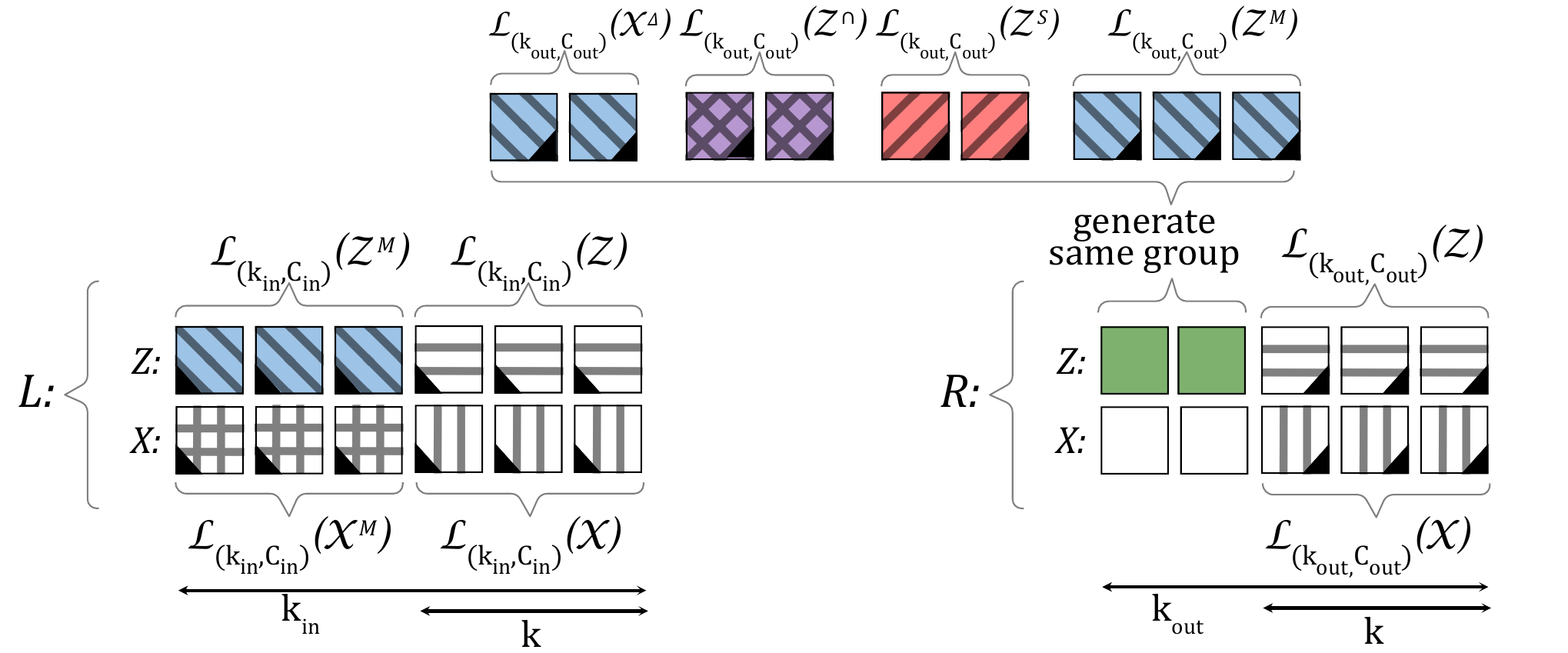}
        \caption{
        The left and right Clifford unitaries $L$ and $R$ of the general form circuit $\mathcal{M}_\mathcal{L}$, where $\mathcal{M}_\mathcal{L}$ has $k$ inner qubits, zero random bits and all-zero condition matrices $A,A_x,A_z$.
        To specify the $k_\inn$-qubit and $k_\out$-qubit Clifford unitaries $L$ and $R$, we show the images of the standard $k_\inn$-qubit and $k_\out$-qubit Pauli bases under each.
        Most of these images correspond to images of specific elements of the symplectic basis $\mathcal{B}$ under the logical operator maps $\mathcal{L}_{(k_\inn,C_\inn)}$ and $\mathcal{L}_{(k_\out,C_\out)}$ (which map $n$-qubit logical operators representatives for codes $\comm{n,k_\inn,C_\inn}$ and $\comm{n,k_\out,C_\out}$ to elements of the $k_\inn$-qubit and $k_\out$-qubit logical Pauli group
        as described in \cref{eq:logical-operator-map}).
        Those images of Pauli basis elements under $R$ marked in green are found however from the combination of specific elements of the symplectic basis $\mathcal{B}$ under  $\mathcal{L}_{(k_\out,C_\out)}$.
        } 
        \label{fig:common-symplectic-basis-logical-action}
    \end{subfigure}
    \caption[Common symplectic basis and logical action of code deformation]{
    \label{fig:logical-action-from-common-symplectic-basis}
    Analyzing the logical action with respect to input and output codes $S_\inn$ and $ S_\out \subset M\cdot(S_\inn \cap M^\perp)$ of a code deformation circuit which measures the stabilizer group $M$.
    }
    
\end{figure}

We restrict our attention in \cref{prob:logical-measure} to symplectic-basis code deformation circuits defined as follows (note that any code deformation circuit $\mathcal{M}$ that measures a group $M$ is equivalent up to conditional Paulis to a symplectic code deformation circuit $\mathcal{M}_\mathcal{B}$ that measures $M$).
\begin{definition}[Symplectic-basis code deformation circuit]
\label{def:common-symplectic-basis-circuit}
Given a common symplectic basis $\mathcal{B}$ of two $n$-qubit stabilizer groups $S_\inn$ and $M$, we define a stabilizer circuit $\mathcal{M}_\mathcal{B}$~(\cref{fig:common-symplectic-basis-circuit}) with outcome vector $v_M$:
\begin{equation}
    \begin{array}{cl}
        (i) & \text{measure } \Xdelta, \Zcap, \Zm \text{(a full set of generators of }M), \\
        (ii) &  \text{apply Pauli unitaries }\Zdelta_j \text{ conditioned on measuring }\Xdelta_j\text{ being one, for }j \in [|\Zdelta|],\\
        (iii) &  \text{apply Pauli unitaries }\Xm_j\text{ conditioned on measuring }\Zm_j\text{ being one, for } j \in [|\Zm|]. \\
    \end{array} \nonumber
\end{equation}
\end{definition}
By construction, the choice of the corrections in $\mathcal{M}_\mathcal{B}$ ensures that it is a logical operation circuit for 
any output code that satisfies the condition of \cref{prob:logical-measure}, that is $S_\out \subset M \cdot (S_\inn \cap M^\perp)$.
To see that this is the case requires showing that $s_\out = 0$ when $s_\inn =0$ for $\mathcal{E}(k_\inn,C_\inn) \circ \mathcal{M}_\mathcal{B} \circ \mathcal{E}^\dagger(k_\out,C_\out)$. 
First note that the state prior to the measurement in \cref{fig:common-symplectic-basis-circuit}, which corresponds to the state after applying the encoding circuit $\mathcal{E}(k_\inn,C_\inn)$ with $s_\inn =0$, is stabilized by all elements of $S$. 
From the common symplectic basis \cref{fig:common-symplectic-basis}, it is clear that the expectation value of all elements of $S_\inn \cap M^\perp = \Zcap \cup \Zs$ are unchanged by measuring $M$ and from applying corrections in $\Zdelta \cup \Xm$.
The corrections ensure that in addition to $S_\inn \cap M^\perp$, all elements of $M$ are also stabilized by the state which is fed into the encoding circuit $\mathcal{E}^\dagger(k_\out,C_\out)$, such that $s_\out =0$.

A solution to  \cref{prob:logical-measure} with $\mathcal{M}=\mathcal{M}_{\mathcal{B}}$ is given in \cref{thm:logical-action-from-common-symplectic-basis} as follows:

\begin{theorem}[Logical action of code deformation]
\label{thm:logical-action-from-common-symplectic-basis}
Consider the $n$-qubit stabilizer groups $S_\inn$ and $M$ with a common symplectic basis $\mathcal{B} =  \ip{\Zdelta,\Xdelta,\Zcap,\Xcap,\Zs,\Xs,\Zm,\Xm,\Z,\X}$ (\cref{def:common-symplectic-basis}, \cref{fig:common-symplectic-basis}), 
and input and output codes $\comm{n,k_\inn,C_\inn}$ and $\comm{n,k_\out,C_\out}$ with stabilizer groups $S_\inn$ and $S_\out \subset M \cdot (S_\inn \cap M^\perp)$.

The symplectic-basis code deformation circuit $\mathcal{M}_{\mathcal{B}}$~(\cref{def:common-symplectic-basis-circuit}, \cref{fig:common-symplectic-basis-circuit}) is a logical operation circuit with input and output codes $\comm{n,k_\inn,C_\inn}$ and $\comm{n,k_\out,C_\out}$.
The logical action of $\mathcal{M}_{\mathcal{B}}$ is equivalent to the general form circuit
$\mathcal{M}_\mathcal{L}$ with zero random bits, $k = |\Z|$ inner qubits,
left $k_\inn$-qubit Clifford unitary $\Ci$, right $k_\out$-qubit  Clifford unitary $\Co$ defined using
using maps $\mathcal{L}_{(k_\inn,C_\inn)}$ and $\mathcal{L}_{(k_\out,C_\out)}$~\cref{eq:logical-operator-map},
(\cref{fig:common-symplectic-basis-logical-action}):
\begin{equation}
\label{eq:sm-left-right}
\arraycolsep=1pt
\begin{array}{rclcrclrcl}
 \langle \Co Z_j \Co^\dagger &:& j \in [k_\out-k]  \rangle &~=~~& \langle   \mathcal{L}_{(k_\out,C_\out)}(P) &:& P \in \Xdelta \cup \Zcap \cup \Zs \cup \Zm \rangle & & \\
 \Ci Z_{k_\inn - k + j} \Ci^\dagger&=&\mathcal{L}_{(k_\inn,C_\inn)}(\Z_j),& ~~~~& \Co Z_{k_\out - k + j} \Co^\dagger&=&\mathcal{L}_{(k_\out,C_\out)}(\Z_j), & ~~j  &\in& [k], \\
 \Ci Z_j \Ci^\dagger &=& \mathcal{L}_{(k_\inn,C_\inn)}(\Zm_j),&~~~~&  & & & j& \in  & [k_\inn - k], \\
 \Ci X_{k_\inn - k + j} \Ci^\dagger &=& \mathcal{L}_{(k_\inn,C_\inn)}(\X_j),& ~~~~&\Co X_{k_\out - k + j} \Co^\dagger &=& \mathcal{L}_{(k_\out,C_\out)}(\X_j),& ~~j  &\in& [k], 
\end{array}
\end{equation}
All condition matrices of $\mathcal{M}_\mathcal{L}$ are zero.
Measurement outcomes of $\Zm_1,\ldots, \Zm_{k-k_\inn}$ in $\mathcal{M}_\mathcal{B}$ are input-dependent and are equal to measurement outcome vector $m$ of $\mathcal{M}_\mathcal{L}$, 
measurement outcomes of $\Xdelta$ in $\mathcal{M}_\mathcal{B}$ are random, 
and the measurement outcomes of $\Zcap$ of $\mathcal{M}_\mathcal{B}$ are zero.
\end{theorem}

\subsection{Example: lattice surgery for repetition codes}
\label{sec:lattice-surgery-example}

Here we provide an example demonstrating how the analytic results of \cref{thm:logical-action-from-common-symplectic-basis} can be used to extract the logical action of an infinite family of code deformation circuits. 

In this example, we consider $2 d$ qubits laid out in two rows, where $ d \in \mathbb{N}$.
Initially two logical qubits are each encoded in a $d$-qubit repetition code (one code in each row of qubits) which together have the stabilizer group $S_{d}$~(left of \cref{fig:example-sm}). 
Then a two-step code deformation circuit $\mathcal{C}_d$ performs lattice surgery by measuring the generators of a $2d$-qubit coupled repetition code with stabilizer group $M_d$~(right of \cref{fig:example-sm}), before returning to the initial code by measuring the generators of $S_{d}$.
We assume that there are Pauli corrections after each of the two code deformation step such that they are symplectic-basis code deformation circuits.
Then $\mathcal{C}_d = \mathcal{M}_{\mathcal{B}_d} \circ \mathcal{M}_{\mathcal{B}'_d}$, where $\mathcal{B}_d$ is a common symplectic basis for $S_{d}$ and $M_d$ (see \cref{fig:common-symplectic-basis-example}) and $\mathcal{B}'_d$ is a common symplectic basis for $M_d$ and $S_{d}$.
Our goal in this example is to find the logical action of $\mathcal{C}_d$ with respect to input and output codes which are both equal to $\comm{2d,2,C_{S_{d}}}$ for an encoding unitary $C_{S_{d}}$ for $S_{d}$ which is specified in \cref{fig:example-c}.

\begin{figure}[htp]
    \centering
     \begin{subfigure}[b]{\textwidth}
        \centering
        \includegraphics[width=\textwidth]{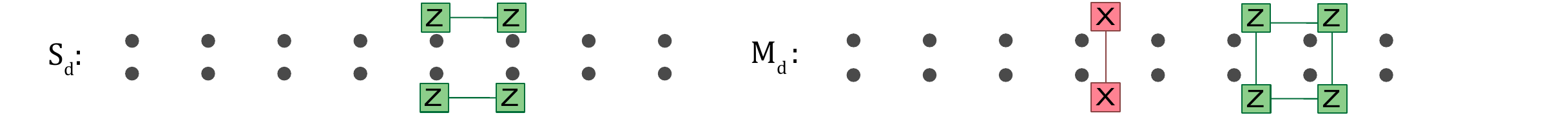}
        \caption{
        The stabilizer group $S_{d}$ for a pair of $d$-qubit repetition codes (left) and the stabilizer group $M_d$ of a coupled repetition code (right).
        $S_{d}$ is generated by horizontal neighbor pairs $Z_{(i,j)}Z_{(i,j+1)}$ with $j \in [d-1]$ on each row $i \in [2]$.
        $M_d$ is generated by vertical neighbor pairs $X_{(1,j)}X_{(2,j)}$ with $j \in [d]$, and squares $Z_{(1,j)}Z_{(1,j+1)} Z_{(2,j)}Z_{(2,j+1)}$ with $j \in [d-1]$.
        Logical operator representatives are fixed by subfigure (c) for $S_{d}$ to be $\{\bar{X}_i = \prod_j X_{(i,j)} ,\bar{Z}_i = Z_{(i,d)}\}$ for $i \in [2]$, and for $M_d$ to be $\{\bar{X} = \prod_j X_{(1,j)}, \bar{Z} = Z_{(1,d)}Z_{(2,d)}\}$. 
        }
        \label{fig:example-sm}
    \end{subfigure}
    \begin{subfigure}[b]{\textwidth}
        \centering
        \includegraphics[width=\textwidth]{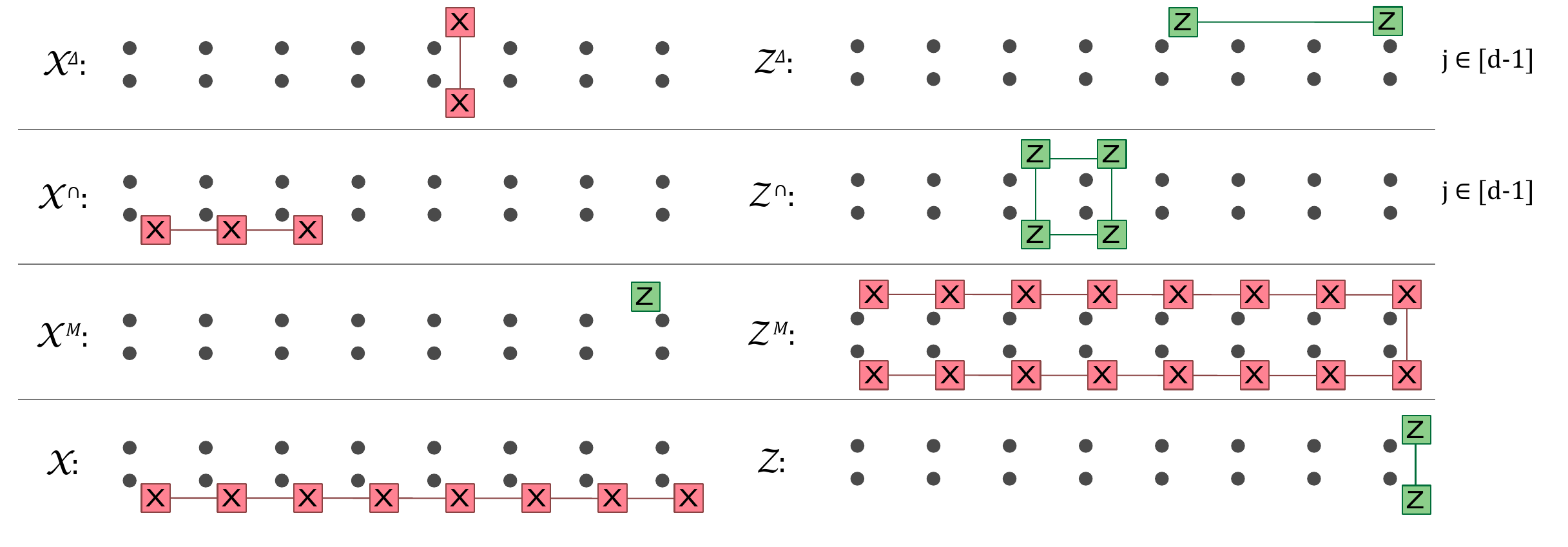}
        \caption{
        The common symplectic basis $\mathcal{B}_d$ for stabilizer groups $S_{d}$ and $M_d$, where $|\Zdelta| = |\Xdelta| = |\Zcap| = |\Xcap| = d-1$, 
        $|\Zs| = |\Xs| = 0$, $|\Zm| = |\Xm| = |\Z| = |\X| = 1$.
        The common symplectic basis $\mathcal{B}'_d$ for stabilizer groups $M_d$ and $S_{d}$ is obtained from $\mathcal{B}_d$ interchanging $\Zdelta \leftrightarrow  \Xdelta$ and $\Zs \leftrightarrow \Zm$, $\Xs \leftrightarrow \Xm$.
        }       \label{fig:common-symplectic-basis-example}
    \end{subfigure}
    \begin{subfigure}[b]{\textwidth}
        \centering
        \includegraphics[width=\textwidth]{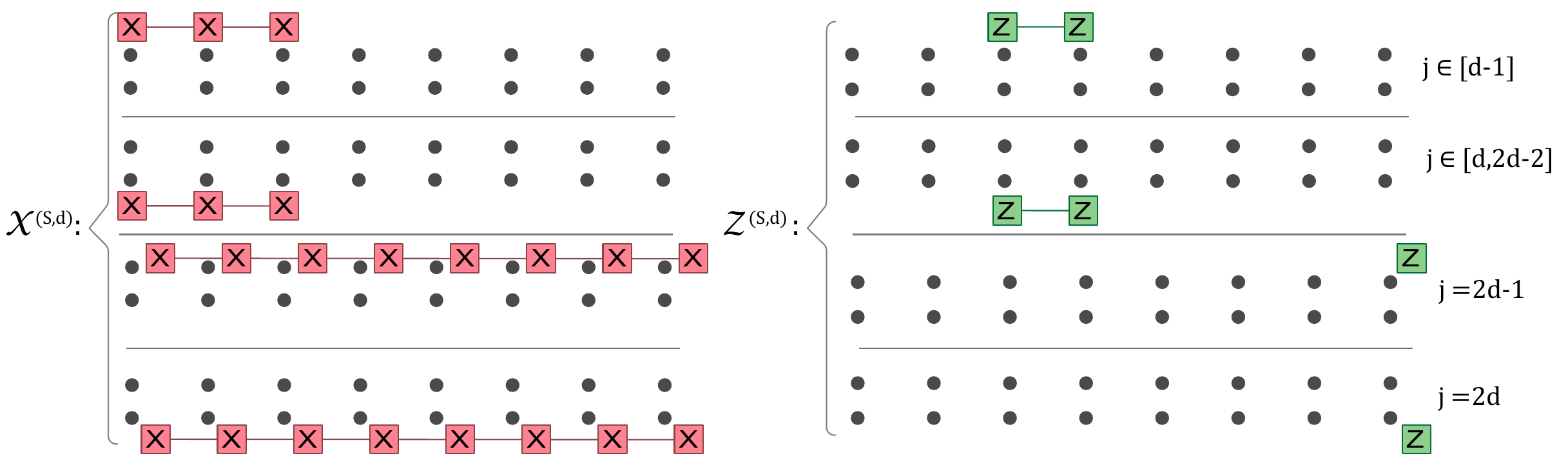}
        \includegraphics[width=\textwidth]{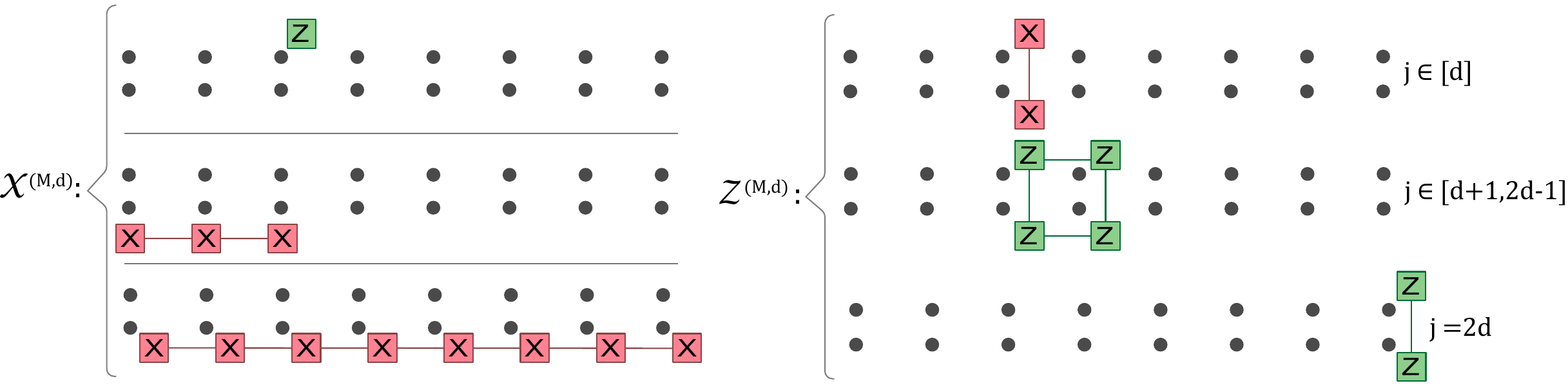}
        \caption{ Encoding unitaries $C_{S_{d}}, C_{M_{d}}$ for codes $\comm{2d,2,C_{S_{d}}}$, $\comm{2d,1,C_{M_{d}}}$ with stabilizer groups $S_{d}$ and $M_d$,
        specified by their action on basis Paulis such that $\mathcal{X}^{\inn}_j = C_{S_{d}} X_j (C_{S_{d}})^{\dagger}$, $\mathcal{Z}^{\inn}_j = C_{S_{d}} Z_j (C_{S_{d}})^{\dagger}$ and $\mathcal{X}^{\out}_j = C_{M_{d}} X_j (C_{M_{d}})^{\dagger}$, $\mathcal{Z}^{\out}_j = C_{M_{d}} Z_j (C_{M_{d}})^{\dagger}$.
        Note that for $j \in [2d - 2]$, the images under $C_{S_{d}}$ are stabilizers of $S_d$, while for $j = 2d-1,2d$ they are logical operators of $S_d$.
        Similarly, for $j \in [d - 1]$, images under $C_{M_{d}}$ are stabilizers of $M_d$, while for $j = 2d$ they are logical operators of $S_d$.
        }
        \label{fig:example-c}
    \end{subfigure}
    \caption[Common symplectic basis example]{
    \label{fig:common-symplectic-basis-example-all}
    Stabilizer groups, symplectic Pauli basis and encoding unitaries to analyze the example of lattice surgery for repetition codes on $2d$ qubits arranged in two rows using \cref{thm:logical-action-from-common-symplectic-basis}.
    To clarify any ambiguity in this figure, explicit definitions are provided in~\cref{app:sm-equations}.
    }
    
\end{figure}

To find the logical action of $\mathcal{C}_{d}$, we first find and then compose the logical actions of the circuits $\mathcal{M}_{\mathcal{B}_d}$ and $\mathcal{M}_{\mathcal{B}'_d}$.
To do so, we define encoding unitaries $C_{S_{d}}$ and $C_{M_{d}}$ for codes $\comm{2d,2,C_{S_{d}}}$ and $\comm{2d,1,C_{M_{d}}}$ with stabilizer groups $S_{d}$ and $M_d$ as specified in \cref{fig:example-c}.  
Applying \cref{thm:logical-action-from-common-symplectic-basis} to the first step yields a general form circuit $\mathcal{M}_\mathcal{L}$ with action equivalent to the logical action of $\mathcal{M}_{\mathcal{B}_d}$ with respect to input and output codes $\comm{2d,2,C_{S_{d}}}$ and $\comm{2d,1,C_{M_{d}}}$.
The circuit $\mathcal{M}_\mathcal{L}$
has $k=1$ inner qubit, a $k_\inn=2$-qubit left Clifford $L$ and a $k_\out=1$-qubit right Clifford $R$ which are defined by the equations:
\begin{equation}
\begin{array}{rclcrclcrcl}
L Z_1 L^\dagger & = & X_1 X_2,&~~& L Z_2 L^\dagger &=& Z_1 Z_2,&~~& L X_2 L^\dagger &=& X_2, \\
 & &                          &~~& R Z_1 R^\dagger & = & Z_1 ,&~~& R X_1 R^\dagger &=& X_1. \\
\end{array}
\end{equation}
These equations can be identified from the specification of $L$ and $R$ in terms of the logical operator maps $\mathcal{L}_{(k_\inn,C_{S_{d}})}(P)$ and $\mathcal{L}_{(k_\out,C_{M_{d}})}(P)$ (according to \cref{thm:logical-action-from-common-symplectic-basis}), which in turn can be read off \cref{fig:example-c} by considering the pre-images of the logical representatives given by the common symplectic basis of $S_d$ and $M_d$ under $C_{S_{d}}$ and $C_{M_{d}}$ respectively (according to~\cref{eq:logical-operator-map}).
We can check that $R=I$ and $L = \mathrm{H}_1 \mathrm{CX}_{1,2}$.

Similarly, applying \cref{thm:logical-action-from-common-symplectic-basis} to the second step yields a general form circuit $\mathcal{M}_\mathcal{L'}$ with action equivalent to the logical action of $\mathcal{M}_{\mathcal{B'}_d}$ with respect to input and output codes $\comm{2d,1,C_{M_{d}}}$ and $\comm{2d,2,C_{S_{d}}}$.
The circuit $\mathcal{M}_\mathcal{L'}$
has $k'=1$ inner qubit, a $k'_\inn=1$-qubit left Clifford $L'$ and a $k'_\out=2$-qubit right Clifford $R'$ which are defined by the equations:
\begin{equation}
\begin{array}{rclcrclcrcl}
 & &                          &~~& L' Z_1 (L')^\dagger & = & Z_1 ,&~~& L' X_1 (L')^\dagger &=& X_1. \\
R' Z_1 (R')^\dagger & = & X_1 X_2,&~~& R' Z_2 (R')^\dagger &=& Z_1 Z_2,&~~& R' X_2 (R')^\dagger &=& X_2, \\
\end{array}
\end{equation}
We have $L'=I$ and $R' = \mathrm{H}_1 \mathrm{CX}_{1,2}$.

Using these results, we see that  circuit $\mathcal{C}_d = \mathcal{M}_{\mathcal{B}_d} \circ \mathcal{M}_{\mathcal{B}'_d}$
has logical action $\mathcal{M}_\mathcal{L} \circ \mathcal{M}_\mathcal{L'}$ of measuring $X \otimes X$ followed by $Z\otimes I$
conditioned on the outcome $1$ of $X \otimes X$ with input and output code $S_{d}$.
This follows from the following circuit identities: 
\begin{equation}
\begin{tikzpicture}
\begin{yquant}
qubit {} q1;
qubit {} q2;
nobit  sp;
qubit {} q3;
qubit {} q4;
[name=c1] cnot q2 | q1;
[name=c2] H q1;
[name=c3] dmeter {$Z$} q1; 
[name=c4] inspect {$m$} q1;
discard q1;
hspace {1mm} q1,q2;
[name=d1] init {$\ket{0}$} q1;
[name=d2] H q1;
[name=d3] cnot q2 | q1;
hspace {1mm} q1,q2;
discard q1,q2;
text {$\simeq$} q1;
hspace {-1mm} q1,q2;
init {} q1;
init {} q2;
cnot q2 | q1;
H q1;
dmeter {$Z$} q1; 
inspect {$m$} q1;
discard q1;
init {$\ket{m}$} q1;
box {$X^m$} q1;
H q1;
cnot q2 | q1;
hspace {1mm} q1,q2;
discard q1,q2;
text {$\simeq$} q1;
hspace {-1mm} q1,q2;
discard q3,q4;
hspace {25mm} q3,q4;
text {$\simeq$} q3;
align q3,q4;
init {} q3;
init {} q4;
cnot q4 | q3;
H q3;
[shape=yquant-dmeter, name=mz] box {$Z$} q3; 
H q3;
cnot q4 | q3;
box {$Z^m$} q3;
hspace {1mm} q3,q4;
discard q3,q4;
hspace {1mm} q3,q4;
text {$\simeq$} q3;
hspace {1mm} q3,q4;
init {} q3;
init {} q4;
[shape=yquant-dmeter, name=mxx] box {$XX$} (q3,q4); 
box {$Z^m$} q3;
hspace {1mm} q3,q4;
discard q3,q4;

\node[draw, dashed,rounded corners,fit=(c1) (c2) (c3) (c4), inner xsep=0.5mm, label={above:{\footnotesize $\mathcal{M}_\mathcal{L}$}}] (circ2) {};
\node[draw, dashed,rounded corners,fit=(d1) (d2) (d3), inner xsep=0.5mm, label={above:{\footnotesize $\mathcal{M}_\mathcal{L'}$}}] (circ2) {};
\node[fit=(mz),inner ysep = 3pt, label={above:{\footnotesize $m$}}] (circ3) {};
\node[fit=(mxx),inner ysep = 3pt, label={above:{\footnotesize $m$}}] (circ4) {};
\draw[thick,double] (mz.north)--(circ3.north);
\draw[thick,double] (mxx.north)--(circ4.north);
\end{yquant}
\end{tikzpicture}
\end{equation}

The outcome $m$ of the logical $XX$ measurement is equal to the outcome of $\Zm_1$ in $\mathcal{M}_{\mathcal{B}_d}$.
The rest of the outcomes in the circuits $\mathcal{M}_{\mathcal{B}_d}$ and $\mathcal{M}_{\mathcal{B}'_d}$
are either zero or random. 
The actions of the circuits $\mathcal{M}_{\mathcal{B}_d}$ and $\mathcal{M}_{\mathcal{B}'_d}$
do not depend on the random outcomes because of the choice of conditional Pauli unitaries.

\newpage
\subsection{General form for measuring two stabilizer groups}
\label{sec:general-form-two-groups}

Here we consider a circuit which sequentially measures two stabilizer groups, and find an equivalent general form circuit~(\cref{lem:general-form-two-stab-groups}) and find the stabilizers of the circuit's Choi states~(\cref{cor:common-symplectic-basis-choi}).
To state these results more precisely, it is useful to first define the syndrome extraction circuit in \cref{def:common-symplectic-basis-circuit-no-corrrections}.

\begin{definition}[Syndrome extraction circuit]
\label{def:common-symplectic-basis-circuit-no-corrrections}
Given a common symplectic basis $\mathcal{B}$ of two $n$-qubit stabilizer groups $S_\inn$ and $M$, we define the circuit $\mathcal{S}_\mathcal{B}$ with outcome vector $v_S$:
\begin{equation}
    \begin{array}{cl}
        (i) & \text{measure } \Zdelta, \Zcap, \Zs \text{(a full set of generators of }S_\inn).
    \end{array} \nonumber
\end{equation}
\end{definition}

In \cref{lem:general-form-two-stab-groups}, we derive a general form circuit equivalent to $\mathcal{S}_\mathcal{B} \circ \mathcal{M}_\mathcal{B}$, where the syndrome extraction circuit $\mathcal{S}_\mathcal{B}$~(\cref{def:common-symplectic-basis-circuit-no-corrrections}) measures stabilizer group $S_\inn$ and symplectic-basis code deformation circuit $\mathcal{M}_\mathcal{B}$~(\cref{def:common-symplectic-basis-circuit}) measures stabilizer group $M$, and where $\mathcal{B}$ is a common symplectic basis $\mathcal{B}$ of $S_\inn$ and $M$.

\begin{lemma}[General form circuit for measuring two stabilizer groups]
\label{lem:general-form-two-stab-groups}
Consider two $n$-qubit stabilizer groups $S_\inn$ and $M$ with a common symplectic basis $\mathcal{B}$.
Additionally, consider circuits $\mathcal{S}_\mathcal{B}$ and $\mathcal{M}_\mathcal{B}$
defined from $\mathcal{B}$ in \cref{def:common-symplectic-basis-circuit-no-corrrections,def:common-symplectic-basis-circuit}.

Then the circuit $\mathcal{S}_\mathcal{B} \circ \mathcal{M}_\mathcal{B}$ is equivalent to 
a general form circuit $\mathcal{C}_\text{gen}$ with:
\begin{itemize}[noitemsep]
    \item $k = |\mathcal{Z}|$ inner qubits and zero random bits. 
    \item Left and right Clifford unitaries $L =  C_\mathcal{B}$ and $R = C_\mathcal{B}~(H^{\otimes |\Zdelta| } \otimes I_{n - |\Zdelta|} )$, where is the Clifford unitary that maps the standard Pauli basis to $\mathcal{B}$. 
    \item Condition matrices $A_x$ and $A_z$ are zero, while $A$ is an $(n-k)\times(n-k)$ matrix equal to $\textbf{0}_{|\Zdelta|\times|\Zdelta|} \oplus I_{|\Zcap|+|\Zs|} \oplus \textbf{0}_{|\Zm|\times|\Zm|}$.
\end{itemize}
Circuits $\mathcal{S}_\mathcal{B}$ and $\mathcal{C}_\text{gen}$ have the same action when their outcome vectors $v_S \oplus v_M$ and $m$ satisfy:
$$m = (v_S(\Zdelta_j))_{j \in |\Zdelta|} \oplus (v_S(\Zcap_j))_{j \in |\Zcap|} \oplus (v_S(\Zs_j))_{j \in |\Zs|} \oplus (v_M(\Zm_j))_{j \in |\Zm|}.$$
The outcomes $(v_M(\Xm_j))_{j \in |\Zm|}$ are random and the action of $\mathcal{S}_\mathcal{B} \circ \mathcal{M}_\mathcal{B}$ does 
not depend on them. 
The outcomes $(v_M(\Zcap_j))_{j \in |\Zcap|}$ are redundant and equal to $(v_S(\Zcap_j))_{j \in |\Zcap|}$. 
\end{lemma}

\begin{figure}[ht]
    \centering
    \begin{subfigure}[b]{\textwidth}
        \centering
            \loadfig{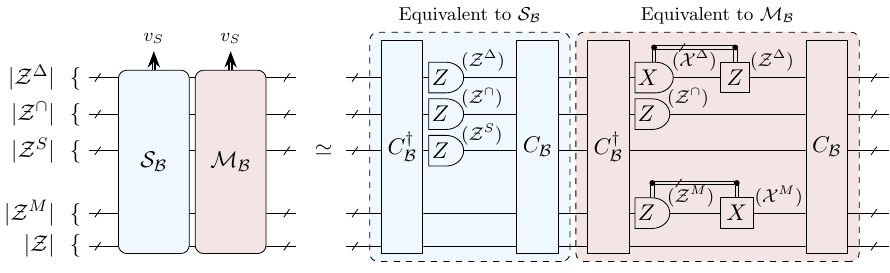}
            \caption{ \label{fig:measure-two-stab-groups-clifford} Replacing $\mathcal{S}_\mathcal{B}$ and $\mathcal{M}_\mathcal{B}$ by equivalent circuits using Clifford unitary $C_\mathcal{B}$ and single-qubit measurements. Clifford  unitary $C_\mathcal{B}$ is the Clifford unitary that maps the standard Pauli basis to $\mathcal{B}$.}
    \end{subfigure}
    \vskip 0.5cm
    \begin{subfigure}[b]{\textwidth}
        \centering
            \loadfig{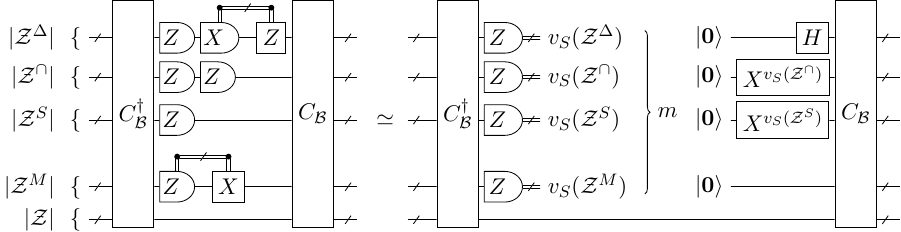}
            \caption{\label{fig:measure-two-stab-groups-destructive}Replacing non-destructive measurements with destructive measurements and qubit allocations.}
    \end{subfigure}
    \caption[Reducing the measurement of two stabilizer groups to a special case]{
        The circuit $\mathcal{S}_\mathcal{B} \circ \mathcal{M}_\mathcal{B}$ and equivalent circuits, for syndrome extraction circuit $\mathcal{S}_\mathcal{B}$~(\cref{def:common-symplectic-basis-circuit-no-corrrections}) and symplectic-basis code deformation circuit $\mathcal{M}_\mathcal{B}$~(\cref{def:common-symplectic-basis-circuit}) given a common symplectic basis $\mathcal{B}$ of 
        groups $S_\inn$ and $M$. 
        Other than $C_\mathcal{B}$, all operations in the equivalent circuits act on single qubits.
        On $Z$ measurements, the superscript $^{(\Zdelta)}$ indicates that the outcomes are those that would be obtained for measuring $\Zdelta$, while on conditional Pauli $Z$s, the superscript $^{(\Zdelta)}$ indicates they are equivalent to conditionally applying elements of $\Zdelta$.
        }
    \label{fig:measure-two-stab-groups}
\end{figure}

\begin{proof}
The proof strategy is to transform the circuit $\mathcal{S}_\mathcal{B} \circ \mathcal{M}_\mathcal{B}$ to the general form 
circuit  $\mathcal{C}_\text{gen}$ via a series of equivalence-preserving transformations. 
The first step is to replace both $\mathcal{S}_\mathcal{B}$ and $\mathcal{M}_\mathcal{B}$ with equivalent circuits that use the Clifford unitary $C_\mathcal{B}$, single-qubit non-destructive $X$ and $Z$ measurements and conditional Paulis as shown in \cref{fig:measure-two-stab-groups-clifford}.
The second step is to remove the adjacent $C_\mathcal{B}$ and $C^\dagger_\mathcal{B}$ from the right circuit 
in \cref{fig:measure-two-stab-groups-clifford} to obtain the left circuit in \cref{fig:measure-two-stab-groups-destructive}.
The third step is to replace all non-destructive single-qubit measurements with equivalent destructive measurements and qubit allocations as in \cref{fig:measure-two-stab-groups-destructive}.
Finally, we see that the right circuit in  \cref{fig:measure-two-stab-groups-destructive} is equivalent to $\mathcal{C}_\text{gen}$ 
with the required outcome relation.
\end{proof}

For our proof of \cref{lem:general-form-two-stab-groups} in \cref{sec:two-stab-groups-gen-form} we will find the following \cref{cor:common-symplectic-basis-choi} of \cref{lem:general-form-two-stab-groups}  useful:

\begin{corollary}[Choi state of measuring two stabilizer groups]
\label{cor:common-symplectic-basis-choi}
Consider $n$-qubit stabilizer groups $S_\inn$ and $M$ with a common symplectic basis $\mathcal{B}$~(\cref{def:common-symplectic-basis}, \cref{fig:common-symplectic-basis}).
Additionally, consider stabilizer circuit $\mathcal{M}_\mathcal{B}$~(\cref{def:common-symplectic-basis-circuit}, \cref{fig:common-symplectic-basis-circuit})
that measures a full set of generators of $M$ and applies conditional Pauli unitaries (defined by $\mathcal{B}$),
and stabilizer circuit $\mathcal{S}_\mathcal{B}$~(\cref{def:common-symplectic-basis-circuit-no-corrrections}) that measures a full set of generators of $S_\inn$ (defined by $\mathcal{B}$).
The stabilizers of the Choi circuit of $\mathcal{S}_\mathcal{B} \circ \mathcal{M}_\mathcal{B}$ upon outcome $v_S \oplus v_M$ are:
\begin{equation}
\arraycolsep=1pt
\begin{array}{rclcrcll}
P & \otimes&  (-1)^{v_S(P)} I, & ~~ & ~~P  &\in& \Zcap \cup \Zs, & \text{(output stabilizers)}, \\
 P & \otimes& I, & ~~ & ~~P  &\in& \Xdelta \cup \Zm, & \text{(output stabilizers)}, \\
 I & \otimes&  (-1)^{v_S(P)} P^\ast, & ~~ & ~~P  &\in& \Zdelta \cup \Zcap \cup \Zs, & \text{(measured observables)}, \\
 I & \otimes&  (-1)^{v_M(P)} P^\ast, & ~~ & ~~P  &\in& \Zm, &  \text{(measured observables)},  \\
P & \otimes&  P^\ast, & ~~ & ~~P  &\in& \X \cup \Z, & \text{(preserved observables)}.
\end{array}
\end{equation}
The action of the circuit $\mathcal{S}_\mathcal{B} \circ \mathcal{M}_\mathcal{B}$ and the signs of the Choi state stabilizers do not depend on the outcomes $v_M(P)$ for $P \in \Xdelta$, these outcomes are random.
\end{corollary}

\subsection{Proof of logical action theorem}
\label{sec:two-stab-groups-gen-form}

Here we prove \cref{thm:logical-action-from-common-symplectic-basis}.

\begin{figure}[htp]
    \centering
         \begin{subfigure}[b]{\textwidth}
        \centering
        \loadfig{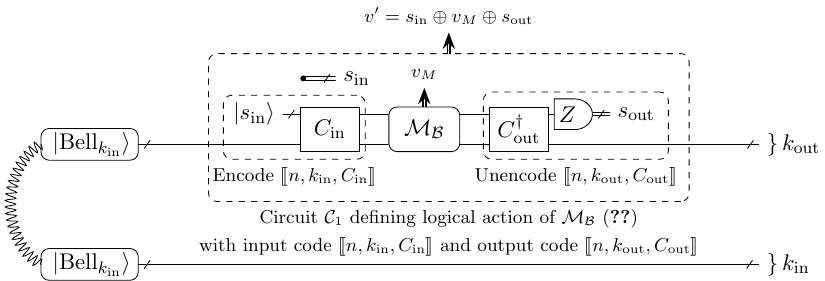}
        \caption{Choi circuit $\text{Choi}(\mathcal{C}_1)$ of circuit $\mathcal{C}_1$ defining logical action of $\mathcal{M}_\mathcal{B}$~(\cref{fig:common-symplectic-basis-circuit}).}
        \label{fig:logical-action-m}
    \end{subfigure}
    \vskip 0.2cm
     \begin{subfigure}[b]{\textwidth}
        \centering
        \loadfig{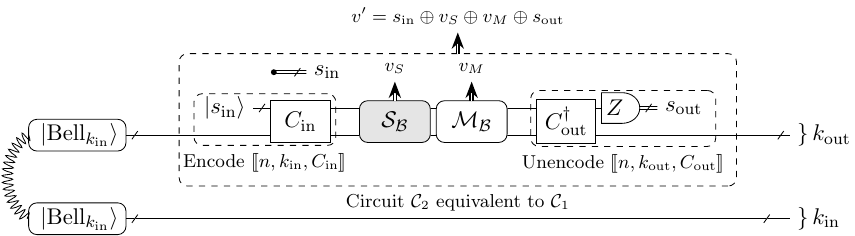}
        \caption{
        Choi circuit $\text{Choi}(\mathcal{C}_2)$ of circuit $\mathcal{C}_2$. Circuits $\text{Choi}(\mathcal{C}_2)$ and $\text{Choi}(\mathcal{C}_1)$ are equivalent. 
         }
        \label{fig:logical-action-sm-choi-state}
    \end{subfigure}
    \vskip 0.2cm
    \begin{subfigure}[b]{\textwidth}
        \centering
        \loadfig{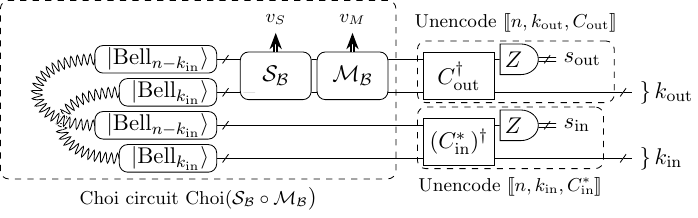}
        \caption{
        Circuit $\mathcal{C}_3$, which is equivalent to circuits $\text{Choi}(\mathcal{C}_1)$ and $\text{Choi}(\mathcal{C}_2)$.
        }
        \label{fig:action-sm-choi-state}
    \end{subfigure}
    \caption[Deriving logical action of measuring $M$ from common symplectic basis]{Equivalent circuits $\text{Choi}(\mathcal{C}_1)$ (a), $\text{Choi}(\mathcal{C}_2)$ (b) and $\mathcal{C}_3$ (c) used to derive the logical action~(\cref{fig:common-symplectic-basis-logical-action}) of the circuit $\mathcal{M}_\mathcal{B}$~(\cref{fig:common-symplectic-basis-circuit}) specified by the common symplectic basis $\mathcal{B}$~(\cref{fig:common-symplectic-basis}) of stabilizer groups $S_\inn$ and $M$.
    The codes $\comm{n,k_\inn,C_\inn}$ and $\comm{n,k_\out,C_\out}$ have stabilizer groups $S_\inn$ and $S_\out$, where $S_\out$ is a subgroup of $M \cdot (S \cap M^\perp)$.
    Given $\mathcal{B}$, the circuit $\mathcal{S}_\mathcal{B}$~(\cref{def:common-symplectic-basis-circuit-no-corrrections}) measures a full set of generators of $S_\inn$, while the circuit $\mathcal{M}_\mathcal{B}$~(\cref{def:common-symplectic-basis-circuit}) measures a full set of generators of $M$ and applies conditional Paulies.
    } 
    \label{fig:action-to-logical-action}
\end{figure}

\begin{proof}[Proof of \cref{thm:logical-action-from-common-symplectic-basis}]
Our goal is to prove that $\mathcal{C}_1 = \mathcal{E}(k_\inn,C_\inn) \circ \mathcal{M}_\mathcal{B} \circ \mathcal{E}^\dagger(k_\out,C_\out)$ (with trivial input syndrome) is equivalent to the general form circuit $\mathcal{M}_\mathcal{L}$ with the parameters claimed in \cref{thm:logical-action-from-common-symplectic-basis}.
To achieve this goal, our strategy is to derive complete sets of stabilizer generators of $\text{Choi}(\mathcal{C}_1)$ (\cref{fig:logical-action-m}) and compare these with the complete sets of stabilizer generators of $\text{Choi}(\mathcal{M}_\mathcal{L})$, which we know in terms of the parameters of the general form circuit $\mathcal{M}_\mathcal{L}$ from our earlier analysis in \cref{fig:general-form-choi}. 
There is freedom in the choice of generator sets.
Given the set of Choi states produced by $\text{Choi}(\mathcal{C}_1)$, to aid comparison we seek sets of generators which differ only by their phase for different Choi states in the set, and where the generators of the set fall into three disjoint subsets: (1) generators supported on the first $k_\out$ qubits, (2) generators supported on the last $k_\inn$ qubits, and (3) generators which form a symplectic basis when restricted to first $k_\out$ qubits.

The key to finding stabilizers of $\text{Choi}(\mathcal{C}_1)$ in satisfying conditions (1), (2) and (3) is to relate them to the stabilizers of $\text{Choi}(\mathcal{S}_\mathcal{B} \circ \mathcal{M}_\mathcal{B})$ shown in \cref{fig:action-sm-choi-state}, described in \cref{cor:common-symplectic-basis-choi}.
The relation between the stabilizers of
$\text{Choi}(\mathcal{C}_1)$ and $\text{Choi}(\mathcal{S}_\mathcal{B} \circ \mathcal{M}_\mathcal{B})$ follows from the fact that circuits $\text{Choi}(\mathcal{C}_1)$~(\cref{fig:logical-action-m}) and $\mathcal{C}_3$~(\cref{fig:action-sm-choi-state}) are equivalent as specified in \cref{fig:action-to-logical-action}. 
Let us postpone for now the proof that $\text{Choi}(\mathcal{C}_1)$ and $\mathcal{C}_3$ are equivalent (which we will see holds for all syndromes) and first prove that the theorem is true given that this equivalence holds.

Now let us derive the stabilizer generators of $\text{Choi}(\mathcal{C}_1)$ upon the outcome $s_\inn = 0$ to find the general form circuit $\mathcal{M}_\mathcal{L}$. 
By the equivalence of $\text{Choi}(\mathcal{C}_1)$ and $\mathcal{C}_3$, we can do this by finding the stabilizer generators of $\mathcal{C}_3$ upon the outcome $s_\inn = 0$.
The circuit $\mathcal{C}_3$ is constructed by first applying $\text{Choi}(\mathcal{S}_\mathcal{B} \circ \mathcal{M}_\mathcal{B})$ (for which we know the stabilizer generators from \cref{cor:common-symplectic-basis-choi}) and then applying the unencoding circuits for
$\comm{n,k_\inn,C_\inn}$ and $\comm{n,k_\out,C_\out}$.
Recall that \cref{eq:encoding-conjugation-action} specifies the action of an unencoding map given a specific outcome, which implies that when $s_\inn = s_\out = 0$,
the state output by $\mathcal{C}_3$ is stabilized by:
\begin{equation}
\label{eq:sm-stab}
\arraycolsep=1pt
\begin{array}{rclcrcll}
 \Lout(P) & \otimes&  I, & ~~ & ~~P  &\in& \Zcap \cup \Zs, & \text{(output stabilizers)}, \\
 \Lout(P) & \otimes& I, & ~~ & ~~P  &\in& \Xdelta \cup \Zm, & \text{(output stabilizers)}, \\
   I & \otimes& \Lin(P^\ast), & ~~ & ~~P  &\in& \Zdelta \cup \Zcap \cup \Zs, & \text{(measured observables)}, \\
  (-1)^{v_M(P)} I & \otimes&\Lin(P^\ast), & ~~ & ~~P  &\in& \Zm, &  \text{(measured observables)},  \\
 \Lout(P) & \otimes& \Lin( P^\ast ), & ~~ & ~~P  &\in& \X \cup \Z, & \text{(preserved observables)}.
\end{array}
\end{equation}
It is not hard to check that the map $\Lin$ is applied to Paulies in its domain $S_\inn^\perp$ 
and that $\Lin(P^\ast)=I$ for $P \in  \Zdelta \cup \Zcap \cup \Zs \subset S_\inn $.
We need to check that $\Lout$ is applied to Paulies in its domain $S_\out^\perp$ .
For the stabilizer group $S_\out \subset M \cdot (S_\inn \cap M^\perp)$ 
we have $(M \cdot (S_\inn \cap M^\perp))^\perp \subset S_\out^\perp$.
Note that $\ip{\Zcap \cup \Zs \cup \Xdelta \cup \Zm} = M \cdot (S_\inn \cap M^\perp)$
and so $(M \cdot (S_\inn \cap M^\perp))^\perp = (M \cdot (S_\inn \cap M^\perp)) \cdot \ip{\X \cup \Z}$,
and therefore $\ip{\Zcap \cup \Zs \cup \Xdelta \cup \Zm \cup \X \cup \Z} \subset S_\out^\perp$,
where the equalities can be seen to be true from the symplectic basis definition in \cref{fig:common-symplectic-basis}. 
We conclude that $\Lout$ is applied to Paulies in its domain.
Note that map $\mathcal{L}$ preserves the commutation: 
$$
\comm{ \mathcal{L}_{\comm{n,k,C}}(P),\mathcal{L}_{\comm{n,k,C}}(Q) } = \comm{P,Q}.
$$
The Paulis in \cref{eq:sm-stab} are stabilizer generators of $\mathcal{C}_3$, and by equivalence are also stabilizer generators of $\text{Choi}(\mathcal{C}_1)$.
Moreover, they satisfy conditions (1), (2) and (3) such that they can be matched to the stabilizer generators of a general form circuit with appropriate choices of the left and the right Clifford unitaries $L$ and $R$.
We use \cref{fig:general-form-choi} to find the expression for the left and the right Clifford unitaries of $\mathcal{M}_\mathcal{L}$ in \cref{eq:sm-left-right}
from \cref{eq:sm-stab}.
This also shows that measurement outcome vector $m$ of $\mathcal{M}_\mathcal{L}$
is $( v(\Zm_j) )_{j \in [\Zm] }$.
The outcomes $v(P)$ for $P \in \Xdelta$ of $\mathcal{M}_\mathcal{B}$ in $\mathcal{C}_1$ and $\text{Choi}(\mathcal{C}_1)$ are random, 
because this is the case for outcomes $v(P)$ for $P \in \Xdelta$ in $\text{Choi}(\mathcal{S}_\mathcal{B} \circ \mathcal{M}_\mathcal{B})$ and $\mathcal{C}_3$
according to \cref{cor:common-symplectic-basis-choi}.
 
It remains to prove that $\text{Choi}(\mathcal{C}_1)$ and $\mathcal{C}_3$ are equivalent, which we do in two steps by showing that both $\text{Choi}(\mathcal{C}_1)$ and $\mathcal{C}_3$ are equivalent to another circuit $\text{Choi}(\mathcal{C}_2)$~(\cref{fig:logical-action-sm-choi-state}).

Circuits $\text{Choi}(\mathcal{C}_1)$ and $\text{Choi}(\mathcal{C}_2)$ are equivalent because their sub-circuits $\mathcal{C}_1$
and $\mathcal{C}_2$ are equivalent. 
Indeed $\mathcal{C}_2$ is obtained from $\mathcal{C}_1$ by inserting $\mathcal{S}_\mathcal{B}$~(highlighted in grey) between 
$\mathcal{E}(k_\inn,C_\inn)$ and $\mathcal{M}_\mathcal{B}$.
Before $\mathcal{S}_\mathcal{B}$ is applied the state is already stabilized by a stabilizer group equal 
to $S_\inn$ up to signs. 
The measurement outcome $v_S$ is redundant and is a linear function of $s_\inn$ and the action of $\mathcal{S}_\mathcal{B}$ on the state is trivial.

Circuit $\mathcal{C}_3$ is equivalent to $\text{Choi}(\mathcal{C}_2)$ because it is obtained from $\text{Choi}(\mathcal{C}_2)$
via a series of three circuit transformations preserving the circuit equivalence $ \text{Choi}(\mathcal{C}_2) \rightarrow \mathcal{C}_{3,1} \rightarrow \mathcal{C}_{3,2}  \rightarrow \mathcal{C}_3$ .
First we obtain $\mathcal{C}_{3,1}$ by replacing the part of $ \text{Choi}(\mathcal{C}_2)$ which
allocates the random bits $s_\inn$ and initializes the ($n-k_\inn$)-qubit register in $\ket{s_\inn}$
with the initialization of a pair of ($n-k_\inn$)-qubit registers in $\ket{\text{Bell}_{(n-k_\inn)}}$
and performing destructive $Z$ measurements on the second ($n-k_\inn$)-qubit register with the outcome $s_\inn$.
Second we obtain $\mathcal{C}_{3,2}$ by pulling the unitary $C_\inn$ through the Bell state $\ket{\text{Bell}_{(n)}}$ it is acting on, 
using the identity $(C_\inn\otimes I_n)\ket{\text{Bell}_{(n)}} = (I_n \otimes C^T_\inn)\ket{\text{Bell}_{(n)}}$.
Third, we change the order in which the destructive $Z$ measurements with outcome $s_\inn$ and the sub-circuit $\mathcal{S}_\mathcal{B}$ appear in the circuit to obtain $\mathcal{C}_3$. 
We move the destructive $Z$ measurements with outcome $s_\inn$ to the end of the circuit such that the outcome $s_\inn$ becomes redundant and the
outcome $v_S$ becomes random.
\end{proof}
 
\cref{thm:logical-action-from-common-symplectic-basis} can be extended to apply to the case when $s_\inn$ is non-zero, resulting in the inclusion of $s_\inn$-dependent condition matrices in the equivalent general form circuit. 
In the proof, we derived \cref{eq:sm-stab} assuming that $s_\inn$ and $s_\out$ are zero, resulting in simplified expressions for the phases of stabilizers, but relaxing this assumption will result in additional phases
from the maps $\mathcal{G}_{\comm{n,k_\inn,C_\inn}}$, $\mathcal{G}_{\comm{n,k_\out,C_\out}}$~(\cref{eq:logical-operator-map}).

\section*{Acknowledgements}
Special thanks to Aarthi Sundaram and Nicolas Delfosse for helpful feedback, and guidance regarding existing literature.

\newpage

\DeclareEmphSequence{\itshape,\bfseries,\mdseries}

\renewcommand{\listtheoremname}{List of Algorithms}
\listoftheorems[ignoreall,show=algorithm]

\renewcommand{\listtheoremname}{List of Problems}
\listoftheorems[ignoreall,show=problem]

\renewcommand{\listtheoremname}{List of Definitions}
\listoftheorems[ignoreall,show=definition]

\renewcommand{\listtheoremname}{List of Theorems and Lemmas}
\listoftheorems[ignoreall,show=theorem,show=lemma,show=corollary]

\listoffigures
\listoftables

\renewcommand{\listtheoremname}{List of Procedures}
\listoftheorems[ignoreall,show=procedure]

\renewcommand{\listtheoremname}{List of Propositions}
\listoftheorems[ignoreall,show=proposition]

\bibliographystyle{plain}
\bibliography{references}

\DeclareEmphSequence{\bfseries,\mdseries}

\appendix


\newpage





\section{Additional mathematical material}
\label{app:mathematical-material}

\subsection*{Clifford and Pauli unitaries}


We have a more efficient algorithm for composing a Clifford unitary with a Pauli exponent $\exp(i \pi / 4)$ because of the following equation: 

\begin{equation}
\label{eq:pauli-exp}
e^{i \nicefrac{\pi}{4}P }Q e^{-i \nicefrac{\pi}{4}P }
=
\frac{I + i P}{\sqrt2} \cdot Q \cdot  \frac{I - i P}{\sqrt2}
=
\left\{ 
\begin{array}{cc}
    Q, & \text{ when } \comm{P,Q} = 0 \\
    iPQ, & \text{ when } \comm{P,Q} = 1 
\end{array}
\right.
\end{equation}
The following proposition allows  us to further reduce amount of calculations we need to perform to compose a Clifford unitary with a Pauli exponent $\exp(i \pi P/ 4)$
and are used to justify correctness of \cref{alg:compose-exp}.

\begin{proposition}[Preimage and sign]
\label{prop:preimage-sign}
Let $C$ be a Clifford unitary and $P$ a Pauli operator described as in \cref{def:desc-pauli},
then $s(P \cdot C Z_k C^\dagger) = s(P) + s(C Z_k C^\dagger)  + 2 \comm{C^\dagger X^{x(P)} C, Z_k }$
and $s(P \cdot C X_k C^\dagger) = s(P) + s(C X_k C^\dagger)  + 2 \comm{C^\dagger X^{x(P)} C, X_k }$.
\end{proposition}
\begin{proof}
The result follows from the following equations: 
\begin{align*}
P \cdot C Z_k C^\dagger 
& = i^{s(P)} Z^{z(P)} X^{x(P)} C Z_k C^\dagger = \\
& = i^{s(P)}(-1)^{\comm{X^{x(P)}, C Z_k C^\dagger }} Z^{z(P)} C Z_K C^\dagger X^{x(P)} = \\
& = i^{s(P)+s(C Z_k C^\dagger)}(-1)^{\comm{ C^\dagger X^{x(P)} C, Z_k }}  Z^{z(P)+z(C Z_k C^\dagger)} X^{x(P) + x(C Z_k C^\dagger)}
\end{align*}
We conclude that $s(P \cdot C Z_k C^\dagger ) = s(P) + s(C Z_k C^\dagger) + 2\comm{C^\dagger X^{x(P)} C, Z_k}$.
Result for $s(P \cdot C X_k C^\dagger)$ follows similarly.
\end{proof}

The proposition below is used to justify correctness of \cref{alg:general-from-choi}.
\begin{proposition}[CSS Clifford action]
\label{prop:css-cliiford-action}
\label{prop:linear-reversible-clifford}
Let $A$ be an invertible $n \times n$ matrix over $\f_2$.
For any computational basis state $|a\rangle$ defined by $n$-bit vector $a$, let unitary $U_A |a\rangle = | A a \rangle$.
Then the following identities hold 
\begin{align}
    U_A Z^a U_A^\dagger = Z^{A^{-T} a} \\
    U_A X^a U_A^\dagger = X^{Aa}
\end{align}
and $U_A$ is a Clifford unitary with description described by the following $(2n+2)\times(2n+2)$ matrix 
as discussed in \cref{def:desc-clifford} and \cref{fig:desc-pauli-and-clifford}:
$$
\left(\begin{array}{c|c|c}
A^{-1} & 0 & 0\\
\hline 0 & A^{T} & 0\\
\hline 0 & 0 & 0
\end{array}\right)
$$
\end{proposition}
\begin{proof}
Let us start with notation:
\begin{align*}
Z = |0\rangle \langle 0 | + (-1) |1\rangle \langle 1 |  = \sum_{z \in \{0,1\}} (-1)^z |z\rangle \langle z|, &
\quad X = |0\rangle \langle 1 | + |1\rangle \langle 0 | = \sum_{z \in \{0,1\}} |z\rangle \langle z + 1|, \\
Z^b = Z_1^{b_1} \otimes \ldots \otimes Z_l^{b_l} = \sum_{c \in \{0,1\}^l} (-1)^{\ip{c,b}} |c\rangle \langle c|, &
\quad X^b = X_1^{b_1} \otimes \ldots \otimes X_l^{b_l} = \sum_{c \in \{0,1\}^l} |c\rangle \langle c + b|.
\end{align*}
We observe that the image $U_A Z^b U_A^\dagger$ is:
\begin{align*}
& \sum_{c \in \{0,1\}^l} (-1)^{\ip{c,b}} |Ac\rangle \langle Ac|  
=  \sum_{c' \in \{0,1\}^l} (-1)^{\ip{A^{-1}c',b}} |c'\rangle \langle c'| = \\
& =  \sum_{c' \in \{0,1\}^l} (-1)^{\ip{c',A^{-T}b}} |c'\rangle \langle c'|  = Z^{A^{-T}b} \text{ where } A^{-T} = (A^{-1})^T.
\end{align*}
Similarly we observe that image of $U_A X^b U_A^\dagger$ is:
\begin{align*}
& \sum_{c \in \{0,1\}^l} |Ac\rangle \langle A(c + b)| 
=  \sum_{c' \in \{0,1\}^l}  |c'\rangle \langle c' + Ab|  
=  X^{Ab} .
\end{align*}
\end{proof}

The following is useful for general form circuit comparison~\cref{alg:general-form-circuit-comparison}:

\begin{proposition}[Batch Pauli conjugation]
\label{prop:batch-pauli-images}
Let $C$ be a Clifford unitary, denote the top-left $2n\times 2n$ binary-symplectic part of the description $M(C)$~(\cref{def:desc-clifford}) as 
$$
\left(
\begin{array}{c|c}
      A_{z,x} & A_{x,x} \\ \hline
      A_{z,z} & A_{x,z}
\end{array}
\right).
$$
Then for any vector $v$ and matrix $\hat A$ we have $C X^{\hat A v} C^\dagger \simeq X^{\hat A_x v} Z^{\hat A_z v}$ where $\hat A_x = A^T_{x,z} \hat A$, $\hat A_z =  A^T_{z,z} \hat A$.
Similarly for any vector $v$ and matrix $\tilde A$ we have $C Z^{\tilde A v} C^\dagger \simeq X^{\tilde A_x v} Z^{\tilde A_z v}$ where $\tilde A_x = A^T_{x,x} \tilde A$ and $\tilde A_z =  A^T_{z,x} \tilde A$.
\end{proposition}
\begin{proof}
Recall that for the standard basis vectors $e_j$ we have $Z^{e_j} = Z_j$. 
According to \cref{def:desc-clifford} we have 
$$
C Z^{e_j} C^\dagger = C Z_j C^\dagger \simeq X^{(A_{x,x})_j} Z^{(A_{z,x})_j} = X^{A^T_{x,x} e_j} Z^{(A^T_{z,x}) e_j}
$$
By linearity, for any $\tilde v$ we have: 
$$
C Z^{\tilde v} C^\dagger \simeq X^{A^T_{x,x} \tilde v} Z^{(A^T_{z,x}) \tilde v}
$$
Now using $\tilde v = \tilde Av$ we have $\tilde A_x = A^T_{x,x} \tilde A$ and $\tilde A_z =  A^T_{z,x} \tilde A$.

Similarly, according to \cref{def:desc-clifford} we have 
$$
C X^{e_j} C^\dagger = C X_j C^\dagger \simeq X^{(A_{x,z})_j} Z^{(A_{z,z})_j} = X^{A^T_{x,z} e_j} Z^{(A^T_{z,z}) e_j}
$$
and therefore $\hat A_x = A^T_{x,z} \hat A$ and $\hat A_z =  A^T_{z,z} \hat A$.
\end{proof}

The following proposition helps establish correctness of deallocation step in \cref{alg:outcome-complete-stab-sim}.

\begin{proposition}[Generalized controlled Pauli conjugation]
\label{prop:pauli-power-conjuagtion-by-ctrl-pauli}
For integer $j$, matrix $A$, vectors $a,r$ with entries in $\f_2$ such that $a_j = 0$ the following identity holds: 
$$
 \Lambda(X_j,Z^a) X^{Ar} \Lambda(X_j,Z^a) = X^{ (I + e_j a^T) A r}.
$$
Multiplying $A$ by $I + e_j a^T$ is equivalent to adding the sum of rows of $A$ indicated by $a$ to row $j$.
\end{proposition}
\begin{proof}
Recall that according to \cref{eq:controlled-pauli-image} we have: 
$$
 \Lambda(X_j,Z^a) X^{Ar} \Lambda(X_j,Z^a) = X^{Ar} X_j^{\comm{Z^a,X^{Ar}}} = X^{Ar} X_j^{\ip{Ar,a}} = X^{Ar + e_j \ip{A^T a,r}} 
$$
Now, using identity $v\ip{u,w} = \ip{v u^T, w}$ for column vectors $v,u,W$ we have 
$$
e_j \ip{A^T a,r} = \ip{e_j (A^T a)^T, r} = \ip{(I+e_j a^T)A}.
$$
Finally, note that $(A^T a)$ is the sum of rows of $A$ indicated by $a$ and $e_j (A^T a)^T$ is a 
matrix with all rows being zero, except row $j$ that is equal to $A^T a$.
\end{proof}

\subsection*{Encoding and unencoding circuits}

The following proposition helps establish correctness of encoding and unencoding circuit comparison algorithm 
\cref{proc:compare-encoding-unencoding-circuits}.

\begin{proposition}[Encoding circuit decomposition]
\label{prop:encoding-comparison-solution}
Let $C$ be a Clifford unitary such that code $\comm{n,k,C}$ has a stabilizer group 
equivalent to the stabilizer group of $\comm{n,k,I_n}$, that is $C$ maps $Z_j$ for $j \in [n-k]$
to products of $Z_j$ for $j \in [n-k]$ up to a sign. 

Then images under conjugation by $C$ are of the following form: 
\begin{align}
C Z_j C^\dagger = (-1)^{(m_0)_j} Z^{A_j} \otimes I_k, & ~~ j \in [n-k], \\
C Z_{j+n-k} C^\dagger = Z^{(A_x)_j} \otimes  C_\Delta Z_j C_\Delta^\dagger, & ~~ j \in [k], \\
C X_{j+n-k} C^\dagger = Z^{(A_z)_j} \otimes  C_\Delta X_j C_\Delta^\dagger, & ~~ j \in [k] 
\end{align}
for some vector $m_0$, matrices $A,A_x,A_z$ and Clifford unitary $C_\Delta$.
Additionally, the Choi state of an encoding circuit with syndrome $m$ implementing the linear map
$
 \ket{\phi} \mapsto C (\ket{m} \otimes  \ket{\phi})
$~(parameterized by $m$)
is equal to the Choi state of the map $\ket{\phi} \mapsto \ket{m'} \otimes C_\Delta X^{A_x m'} Z^{A_z m'} \ket{\phi}$ (parameterized by $m'$)
for $m' = A^{-1} (m_0 + m)$.
\end{proposition}
\begin{proof}
We first establish that images of $X$ and $Z$ under conjugation by $C$ have required form.
Images $C Z_{j+n-k} C^\dagger$, $C X_{j+n-k} C^\dagger$ for $j \in [k]$ must commute with $Z_j$ for $j \in [n-k]$
and are, therefore, products of Pauli $Z$ over qubits $[n-k]$ up to signs.
This is because images $C Z_j C^\dagger$ for $j \in [n-k]$ generate $\ip{Z_1,\ldots,Z_{n-k}}$ up to signs, by the equivalence of 
the stabilizer groups of $\comm{n,k,C}$ and $\comm{n,k,I_n}$.

Next we establish the following equations for the Choi state stabilizer generators:
\begin{align}
     \label{eq:code-stab} (-1)^{\ip{m', e_j }} Z_j \otimes I_k \otimes I_k, & ~~ j \in [n-k], \\
     \label{eq:cl-x-stab} (-1)^{\ip{A_x m', e_{j}}} I_{n-k} \otimes \left( C_\Delta Z_{j} C^\dagger_\Delta \right) \otimes Z_{j}, & ~~ j \in [k], \\
     \label{eq:cl-z-stab} (-1)^{\ip{A_z m',  e_{j}}} I_{n-k} \otimes \left( C_\Delta X_{j} C^\dagger_\Delta \right) \otimes X_{j}, & ~~ j \in [k].
\end{align}
We gradually replace stabilizer generators with equivalent ones to arrive at the required result.
By the definition of the Choi state for linear map 
$
 \ket{\phi} \mapsto C (\ket{m} \otimes  \ket{\phi})
$
the stabilizer generators consist of the three sets:
\begin{align*}
     \{ (-1)^{m_j} C Z_j C^\dagger \otimes I_k, & ~~ j \in [n-k] \}, \\ 
     \{ C X_{j+n-k} C^\dagger \otimes X_{j}, & ~~ j \in [k] \}, \\
     \{ C Z_{j+n-k} C^\dagger \otimes Z_{j}, & ~~ j \in [k] \}
\end{align*}

We first multiply the generators from the first set with each other, so that after the modifications
the stabilizers from the first set are $Z_1,\ldots,Z_{n-k}$ up to phases.
Using inner-product notation, the first set of generators can be rewritten as 
$$
(-1)^{\ip{m + m_0,e_j }} Z^{A^T e_j} \otimes I_k,~~ j \in [n-k]. \\
$$
This is because when syndrome is $m$, we have stabilizer $(-1)^{m_j} C Z_j C^\dagger = (-1)^{(m_0+m)_j} Z^{A_j}$, 
using $C Z_j C^\dagger = (-1)^{(m_0)_j} Z^{A_j} \otimes I_k$.
Now we use $(m_0 + m)_j = \ip{m_0 + m, e_j}$ and $A_j = A^T e_j$.
This implies that the Choi state is stabilized by 
$$
(-1)^{\ip{ m + m_0, (A^T)^{-1} e_j }} Z^{A^T (A^T)^{-1} e_j} \otimes I_k = (-1)^{\ip{ A^{-1} (m + m_0), e_j }} Z^{e_j} \otimes I_k,~~ j \in [n-k], \\
$$
because each $(A^T)^{-1} e_j$ is just a sum of $e_{j'}$ for some set of $j'$.
The above stabilizer are the same as those in~\cref{eq:code-stab} because $m' = A^{-1} (m_0 + m)$.

Next we multiply the stabilizers from the second and third sets by the generators from the first set, 
so the generators from the second and third sets are now supported on the last $k$ qubits.
The generators from the second set must be multiplied by $Z^{(A_x)_j}$ up to phases, 
so $Z^{(A_x)_j}$ gets replaced by the phase 
$$
(-1)^{\ip{ m', (A_x)_j }} = (-1)^{\ip{ m', A_x^T e_j }} = (-1)^{\ip{ A_x m', e_j }}
$$
Similarly, $Z^{(A_z)_j}$ in the generators from the third set is replaced by phase
$$
(-1)^{\ip{ m', (A_z)_j }} = (-1)^{\ip{ A_z m', e_j }}
$$
The phase dependent on $m'$ appears in \cref{eq:cl-x-stab,eq:cl-z-stab}.
Equations  \cref{eq:code-stab,eq:cl-x-stab,eq:cl-z-stab} match the Choi state stabilizer generators of 
 $\ket{\phi} \mapsto \ket{m'} \otimes C_\Delta X^{A_x m'} Z^{A_z m'} \ket{\phi}$, as required.
\end{proof}

To connect \cref{prop:encoding-comparison-solution} to  \cref{proc:compare-encoding-unencoding-circuits}, set the unitary $C$ in \cref{prop:encoding-comparison-solution} to be $C = C_2^\dagger C_1$.

\newpage

\subsection*{Stabilizer instruments}

Given a set of stabilizer states $\{ \ket{\psi}_r \}_{r \in R}$ we say that Pauli operator $P$ is a \emph{phase operator} of $\{ \ket{\psi}_r \}_{r \in R}$
when $P \ne \pm I$ and for all $r \in R : P \ket{\psi}_r  = \pm \ket{\psi}_r$. 
For any phase operator $P$ we define the corresponding \emph{phase function} $f_P$ via 
$
 P\ket{\psi}_r = (-1)^{f_P(r)}\ket{\psi}_r.
$
The set of Choi states corresponding to a stabilizer circuit and the related quantum instrument motivates the following:
\begin{definition}[Linear family of stabilizer states]
\label{def:linear-family}
Consider a family of $n$-qubit stabilizer states  $\{ \ket{\psi}_r \}_{r \in R}$, 
we say that it is a linear family of stabilizer states if 
\begin{itemize}[noitemsep]
    \item[(i)] there exist a set of $n$ independent commuting phase operators of $\{ \ket{\psi}_r \}_{r \in R}$
    \item[(ii)] for any $m$ phase operators $P_1,\ldots,P_l$ set $\{ (f_{P_1}(r),\ldots,f_{P_l}(r)) : r \in R \}$ is a coset\footnote{Coset of a vector space $V$ is a subset $v + L \subset V$ where $v \in V$ is a vector and $L \subset V$ is a subspace of $V$.} of $\f_2^l$
\end{itemize}
\end{definition}
A corollary of \cref{thm:general-form} and \cref{fig:general-form-choi} is that the set of Choi states of any stabilizer circuit is a linear family.
A natural question to ask is that if any linear family of stabilizer states corresponds to a stabilizer circuit.
We will show that to answer this question affirmatively, the linear family must have an additional property. 
To state this property we need the following definition: 

\begin{definition}[Phase-complete family of stabilizer states]
\label{def:phase-complete}
Consider a linear family of $n$-qubit stabilizer states $\{ \ket{\psi_r} \}_{r \in R}$ and set $K \subset [n]$. 
Consider the set $\mathcal{F}_K$ of all phase operators of $\{ \ket{\psi_r} \}_{r \in R}$ supported on $K$.
We call $\{ \ket{\psi_r} \}_{r \in R}$ phase-complete with respect to $K$
when $-I \in \ip{\mathcal{F}_K}$ and for any $P_1,\ldots,P_l$ such that $-I,P_1,\ldots,P_l$ are independent generators of $\ip{\mathcal{F}_K}$  
set $\{ (f_{P_1}(r),\ldots,f_{P_l}(r)) : r \in R \} = \f_2^l$.
\end{definition}

The following describes which sets of stabilizer states correspond to stabilizer circuits:

\begin{lemma}[Phase completeness]
\label{lem:choi-state-criteria}
A set of $n$-qubit stabilizer states $\{ \ket{\psi_r} \}_{r \in R}$ is a set of Choi states 
of some stabilizer circuit $\mathcal{C}$ with $\ki$ input qubits and $\ko = n - \ki$ output qubits
if and only if $\{ \ket{\psi_r} \}_{r \in R}$ is phase-complete linear family of stabilizer states with respect to set $K = [\ko+1,\ko+\ki]$.
\end{lemma}

Before we proceed with the proof of~\cref{lem:choi-state-criteria} note that it allows us to give an alternative definition of a 
stabilizer instrument without the need to mention a quantum circuit: a quantum instrument is a stabilizer instrument if 
its Choi states are a phase-complete linear family of stabilizer states.
When completing this work we became aware of an equivalent definition of a stabilizer instrument, the Definition~6 in \cite{LogicalBlocks}.
Results of this section can be re-derived using this alternative definition.
\begin{proof}[Proof of \cref{lem:choi-state-criteria}]
We first show that the set of Choi states of any stabilizer circuit $\mathcal{C}$ is 
a phase-complete linear family of stabiliser states.
This is achieved in three steps:
(A1) construct independent commuting phase operators
(A2) show that image of several phase functions together is an image of an affine map, 
which is a coset of a linear space
(A3) show phase completeness property by expressing image of the related phase functions as an image of a full row rank linear operator.
Second, we show that any phase-complete linear family of stabilizer states $\{ \ket{\psi_r} \}_{r \in R}$ can be associated with a general form circuit. 
This is also achieved in three steps:
(B1) show that the set of outcomes $R$ can be replaced by $R' = \f_2^{n'_r}$ for some integer $n'_r$,
(B2) show that all $\ket{\psi_r} \simeq C\ket{A'r'}$ for $r' \in \f_2^{n'_r}$ for some Clifford unitary $C$ and matrix $A'$
(B3) use the second stage of \cref{alg:general-from-choi} to construct the general form circuit.

(A1) \cref{thm:general-form} and \cref{fig:general-form-choi} imply that the set of Choi state of $\mathcal{C}$
is $\{ \ket{\phi_o} \}_{o \in \f_2^{n_O}}$, where $\ket{\phi_o}$ is a Choi state of a linear map enacted by the general form circuit 
upon outcome $o$ described in~\cref{fig:general-form-choi}.
Removing phases from the stabiliser generators in the right half of~\cref{fig:general-form-choi} gives us a set of 
$n = \ko + \ki$ independent commuting phase operators $Q_1,\ldots,Q_n$ of $\ket{\phi_o}$:
\begin{alignat*}{3}
\Co Z_{j_\out} \Co^\dagger & \otimes I_{\ki}, & j_\out \in [\ko-k] \\
\Co Z_{\ko - k + j} \Co^\dagger & \otimes (\Ci Z_{\ki - k + j} \Ci^\dagger)^\ast, & j \in [k] \\
I_{\ko} & \otimes (\Ci Z_{j_\inn} \Ci^\dagger)^\ast,& j_\inn \in [\ki-k] \\
\Co X_{\ko - k + j} \Co^\dagger & \otimes (\Ci X_{\ki - k + j} \Ci^\dagger)^\ast, & j \in [k] 
\end{alignat*}

(A2) We start by relating a phase function to an image affine map.
Every phase operator $P$ of $\{ \ket{\phi_o} \}_{o \in \f_2^{n_O}}$ is the product of these operators up to a sign,
that is for every phase operator $P$ there exist vector $a_P$ and constant $s_P$ such that: 
$$
    P = (-1)^{s_P} Q_1^{(a_P)_1} \ldots Q_n^{(a_P)_n}.
$$
Phase function $f_P$ is then given by
$$
f_P = s_P + (a_P)_1 (f_{Q_1}) + \ldots + (a_P)_n f_{Q_n} = s_P + (a_P)^T ({f_Q}_1,\ldots,{f_Q}_n)
$$
We introduce matrix $A'$ to relate $({f_Q}_1,\ldots,{f_Q}_n)$ to outcome $o$ of the general form:
$$
 A' = \left(\begin{array}{c} A \\ \hline  A_x \\ \hline (\mathbf{0}_{\ki - k \times n_r}|I_{\ki - k}) \\ \hline  A_z \end{array}\right),
$$
and have $({f_Q}_1(o),\ldots,{f_Q}_n(o)) = A' o$. For any $l$ phase functions we have:
$$
(f_{P_1}(o),\ldots,f_{P_l}(o)) = (s_{P_1},\ldots,s_{P_l}) + \left(\begin{array}{c} a^T_{P_1} \\ \hline  \vdots \\ \hline a^T_{P_l} \end{array}\right) A o,
$$
which immediately implies that image of $(f_{P_1},\ldots,f_{P_l})$ is a coset of $\f^l_2$.

(A3) We show the phase-completeness property for a convenient set of phase operators $P_1,\ldots,P_l$
$$
I_{\ko}  \otimes (\Ci Z_{j_\inn} \Ci^\dagger)^\ast, j_\inn \in [\ki-k]
$$
first, and then use an argument similar to one in (A2) to show that the phase complete property holds 
for any set of phase operators $-I,P'_1,\ldots,P'_l$ such that $-I,P'_1,\ldots,P'_l$ generate $\mathcal{F}_{[\ko+1,n]}$~(see~\cref{def:phase-complete}).
For the chosen $P_1,\ldots,P_l$ we have 
$$
(f_{P_1}(o),\ldots,f_{P_l}(o)) = (\mathbf{0}_{\ki - k \times n_r}|I_{\ki - k}) o,
$$
and so the image is $\f_2^l$. Similarly to the argument in (A2), we have 
$$
(f_{P'_1}(o),\ldots,f_{P'_l}(o)) = 
(s_{P_1},\ldots,s_{P_l}) + 
\hat A 
(\mathbf{0}_{(\ki - k) \times n_r}|I_{\ki - k}) o,~~
\hat A = \left(\begin{array}{c} (a_{P'_1})^T_{[\ko+1,n-k]} \\ \hline  \vdots \\ \hline (a_{P'_l})^T_{[\ko+1,n-k]} \end{array}\right)
$$
Matrix $\hat A$ is an invertible matrix because $-I,P'_1,\ldots,P'_l$ generate $\mathcal{F}_{[\ko+1,n]}$, 
therefore matrix $\hat A (\mathbf{0}_{(\ki - k) \times n_r}|I_{\ki - k})$ has column rank $l$ and the image of $(f_{P'_1}(o),\ldots,f_{P'_l}(o))$ is $\f_2^l$.

Now we show that any phase-complete linear family of stabilizer states corresponds to a general form circuit following the approach outlined above.
(B1) Let $Q_1,\ldots,Q_n$ be some independent commuting phase operators of $\{ \ket{\psi_r} \}_{r \in R}$. 
Image of $(f_{Q_1},\ldots,f_{Q_n})$ is a coset $v + L$. 
We first replace $Q_1,\ldots,Q_n$ by $(-1)^{v_1} Q_1,\ldots,(-1)^{v_n} Q_n$ so that image of the corresponding phase functions is $L$.
Let us now set $n'_r$ to the rank of $L$ and let $A'$ be a basis matrix of $L$, so that $L = A' \f_2^{n'_r}$.
We can label $\ket{\psi}_r$ by $r' = (A')^{(-1)} r$ and define relabelled set of states $\{ \ket{\phi}_{r'} \}_{r' \in \f_2^{n'_r}}$.
(B2) Let $C$ be any Clifford unitary such that $C Z_j C^\dagger = Q_j$. 
Note that $\ket{\phi}_{r'}$ is stabilized by $(-1)^{(A'r')_j} CZ_j C^\dagger$ and so $\ket{\phi}_{r'} \simeq C\ket{A'r'}$.
(B3) We apply the second part of~\cref{alg:general-from-choi} to the set of states $\{ \ket{\phi}_{r'} \}_{r' \in \f_2^{n'_r}}$
described by Clifford unitary $C$ and matrix $A'$. We can always find matrix $\tilde A_m$ in \cref{line:general-form-outcome-relation}
according to~\cref{prop:relabelling}~(discussed below)
because $\{ \ket{\phi}_{r'} \}_{r' \in \f_2^{n'_r}}$ is phase-complete with respect to $[\ko+1,\ko+\ki]$.
\end{proof}

The following proposition is used in the proof above and is also used to justify correctness of~\cref{alg:general-from-choi}.

\begin{proposition}[Outcome relabelling for phase complete sets of stabilizer states]
\label{prop:relabelling}
Consider $(\ko+\ki)$-qubit s set of stabilizer state $\{ \tilde C \ket{A r} \}_{r \in \f_2^{n_r}}$, where $\tilde C$ has a property described by \cref{fig:bipartition-family-equality}
for integer $k$, Clifford unitaries $B,D$ and $n_1 = \ko, n_2 = \ki$.
Suppose that the set of stabilizer states is phase-complete with respect to $[\ko+1,\ko+\ki]$, then 
there exist invertible matrix $\tilde A$ such that $(A \tilde A)_{[\ko+1,\ko-k]} =  (\mathbf{0}_{(\ki - k) \times n_r}|I_{\ki - k})$.
\end{proposition}
\begin{proof}
First note that Phase operators $$
I_{\ko}  \otimes (\Ci Z_{j} \Ci^\dagger)^\ast, j  \in [\ki-k]
$$
generate $\mathcal{F}_{[\ko+1,n]}$~(see~\cref{def:phase-complete}). 
The image of the corresponding phase functions is $A_{[\ko+1,\ko+\ki-k]} r$ and is $\f_2^{\ki - k}$ according to the phase-completeness.
For this reason,  $A_{[\ko+1,\ko+\ki-k]}$ must have rank $\ki-k$ and so \cref{proc:block-reshape} described below can be applied to $A_{[\ko+1,\ko+\ki-k]}$.
\end{proof}

\subsection*{Linear algebra}

The following is used to justify correctness of ~\cref{alg:general-from-choi}.
\begin{proposition}[Product of split reduced echelon form matrices] \label{prop:matrix-structure-propagation}
Consider $n\times n$ matrix $R$ in $(n_r,n_m)$-split reduced echelon form~(\cref{def:split-echelon-form}) and $m \times n$ matrix $M$ in $(n,0)$-split echelon form,
then the product $MR$ is in $(n_r,n_m)$-split reduced echelon form.
\end{proposition}
\begin{proof}
The result follows from two observations.
First, row $j$ of $M R$ is equal to the linear combination of rows of $R$ with the coefficients given by row $j$ of $M$.
Second, matrix $B$ is in reduced row echelon form if and only if it consists of leading one columns and all other column are 
linear combinations of previous leading one columns.

First check that $(MR)^T_{[n_r]}$ is in reduced row echelon form.
We derive locations of leading one columns of $(MR)^T_{[n_r]}$.
For $j \in [n]$, let $s_j$ be maximum index of a non-zero entry of first $j$ rows of $R_{[n_r]}$.
Sequence $s_j$ is monotonically non-decreasing, because $R_{[n_r]}$ is in row reduced echelon form.
Additionally, $s_{j-1} + 1 = s_{j}$  when $j$ is a leading one column of $R^T_{[n_r]}$ 
and $s_{j-1} = s_{j}$ otherwise.
Similarly, for $j \in [n]$, let $l_j$ be maximum index of a non-zero entry of first $j$ rows of $M_{[n_r]}$.
The maximum index of a non-zero entry of row $j$ of $(MR)_{[n_r]}$ is at most $s_{l_j}$ 
because the row $(MR)_{j,[n_r]}$ is a linear combination of first $l_j$ rows of $R_{n[r]}$.

Leading one columns of $(MR)^T_{[n_r]}$ are exactly those $j'$ for which $s_{l_{j'-1}} + 1 = s_{l_{j'}}$.
This is because equality $s_{l_{j'-1}} + 1 = s_{l_{j'}}$ is true if and only if $l_{j'-1} + 1 = l_{j'}$ and $s_{l_{j'}-1} + 1 = s_{l_{j'}}$.
In other words, $l_{j'}$ and $s_{l_{j'}}$ are leading one columns of $M^T_{[n_r]}$ and $R^T_{[n_r]}$
and so $(MR)_{j'} = 0^{s_{l_{j'}-1}}1 0^{n - s_{l_{j'}}}$.
Rows $j''$ of $(MR)_{[n_r]}$ for which $s_{l_{j''-1}} = s_{l_{j''}}$ are linear combinations 
of previous leading one rows of $(MR)_{[n_r]}$.

Observation $(MR)_{j'} = 0^{s_{l_{j'}-1}}1 0^{n - s_{l_{j'}}}$ implies that $(MR)_{j',[n_r+1,n_m]}$ is $0^{n_m}$ for $j'$ being leading one columns of $(MR)^T_{[n_r]}$.
Finally, matrix $(MR)^T_{[n_r+1,n_m+n_r]}$ is in reduced row echelon form similarly to how $(MR)^T_{[n_r]}$ is in reduced row echelon form because 
$(MR)^T_{[n_r+1,n_m+n_r]} = (M R_{[n_r+1,n_m+n_r]})^T$.
\end{proof}

\subsection*{General form circuits}

A corollary of the general circuit comparison~\cref{alg:general-form-circuit-comparison} is the following 
\begin{lemma}[General form equivalence]
Let $\mathcal{C}_1$, $\mathcal{C}_2$ be two equivalent general form circuits with $\ki$ inputs, $\ko$ outputs and $n_r^{(1)},n^{(2)}_r$ random bits,
then the general form circuits have the same number of inner qubits $k$.
Additionally, there exist:
\begin{itemize}[noitemsep]
    \item full row rank matrices $F_1$, $F_2$ with $n_r^{(1)},n^{(2)}_r$ columns
    \item an invertible $(\ki -k)\times(\ki -k)$ matrix $A_m$
    \item vectors $\Delta r, \Delta m$ of sizes $n_r^{(1)},\ki - k$
\end{itemize}
such that the general form circuits $\mathcal{C}_1,\mathcal{C}_2$ implement the same linear map upon outcomes $o_1,o_2$ if and only if 
$$
 (F_1 \oplus A_m) (o_1 + \Delta r \oplus \Delta m) = (F_2 \oplus I_{\ki - k}) o_2. 
$$
\end{lemma}
 
\newpage

\section{Additional algorithmic material}

\label{app:procedures}

Here we provide a description of some procedures that we make use of. 

\subsection*{Bit-level procedures}

In \cref{tab:bit string-procedures}, we list procedures involving bit strings directly.

\begin{table}[htp]
\centering
    \begin{tabular}{|c|c|c|c|}
    \hline 
    \multirow{2}{*}{Procedure}  & \multirow{2}{*}{Mathematical action} & \multirow{2}{*}{Run time} & \multirow{2}{*}{Bit-string run time} \\
    & & & \\
    \hline 
    \hline 
    \texttt{xor(}$a,b$\texttt{)} &  $a+b$ & $O(n)$ & $O(1)$ \\
    \hline 
    \texttt{and(}$a,b$\texttt{)} &  $a \wedge b$ & $O(n)$  & $O(1)$\\
    \hline 
    \texttt{weight(}$a$\texttt{)} & $\mathrm{wt}(a)$  & $O(n)$  & $O(1)$  \\
    \hline 
    \texttt{set\_bit(}$a,c,j$\texttt{)}  &  $a_j \leftarrow c $ & $O(1)$  & $O(1)$ \\
    \hline 
    \texttt{check\_bit(}$a,j$\texttt{)} &  Is $a_j = 1$ ? & $O(1)$  & $O(1)$   \\
    \hline 
    \texttt{nonzero\_bit(}$a$\texttt{)}  &  $j$ where $a_j = 1 $ & $O(n)$  & $O(1)$  \\
    \hline 
    \texttt{is\_zero(}$a$\texttt{)}&  Is $a = 0\ldots0$ ?  & $O(n)$  & $O(1)$  \\
    \hline
    \end{tabular}
    \caption[Bit-wise operations]{
    Mathematical action and run time complexities of basic bit string procedures.
    Note that the \bitruntime{} is $O(1)$ for all.
    The procedure inputs are as follows: 
    $a,b$ are bit strings of size $n$.
    $c$ is an individual bit 
    and $j$ is a bit index.
    }
    \label{tab:bit string-procedures}
    \label{tab:bit-string-procedures}
\end{table}

\subsection*{Linear algebra}

The following procedure is used in \cref{alg:general-from-choi} in \cref{line:general-form-outcome-relation}.

\begin{procedure}[\texttt{block-reshape}]
\label{proc:block-reshape}
\begin{algorithmic}[1]
\Blank
\Input $m\times n$ matrix $A$ over $\f_2$ with rank $m$
\Output $n \times n$ invertible matrix $R$, $m\times m$ invertible matrix $F$ such that  $A R= (\mathbf{0}_{m \times (n-m)} | F )$ and $R$ 
is in $(n-m,m)$-split reduced echelon form~(\cref{def:split-echelon-form}).
\State\label{line:block-reshape-rref} Find $m\times m$ invertible matrix $F'$ such that $(F')^{-1} A'$ is in reduced row echelon form
\State\label{line:block-reshape-leading-ones} Set $R$ to $I_{n\times n}$, let $L = l_1,\ldots,l_m$ be the increasing sequence of leading one columns of $F'A'$
\For{ $j \in [n]$ } 
\If{column $j$ of $((F')^{-1}A')$ does not contain leading $1$} 
\For{$k \in [m]$} $R_{n+1-l_k,n+1-j} \leftarrow ((F')^{-1}A')_{k,j}$ \EndFor
\EndIf
\EndFor
\State\label{line:block-reshape-permute} Permute columns of $R$ so that columns $n+1-l_m,\ldots,n+1-l_1$ become last $m$ columns, that is for increasing sequence $j_1,\ldots,j_{n-m}$
of integers in $[n] \setminus L$ apply column permutation $\sigma$ defined as 
$$
 \sigma(k) = n+1-j_{n-m+1-k}, k \in [n-m],~\sigma(k+m) = n+1-l_{m+1-k}, k \in [m]
$$
\State $F \leftarrow $ matrix $F'$ with reversed columns 
\State \Return $R$, $F$
\end{algorithmic}
\end{procedure}

The key idea behind \cref{proc:block-reshape} is to pick $m$ linear-independent columns of $A$ with an additional property
that every other column that is linear dependent, is linear dependent only on the columns to the right from it.
To illustrate the procedure, consider an example matrix $A$ of rank $m=3$ with columns $1,4,7$ being the linear independent ones:
$$
\left(
\begin{array}{cccccccc}
 \mathbf{k_1} & e f_1+g h_1 & c f_1+d h_1 & \mathbf{h_1} & b f_1 & a f_1 & \mathbf{f_1} & 0 \\
 \mathbf{k_2} & e f_2+g h_2 & c f_2+d h_2 & \mathbf{h_2} & b f_2 & a f_2 & \mathbf{f_2} & 0 \\
 \mathbf{k_3} & e f_3+g h_3 & c f_3+d h_3 & \mathbf{h_3} & b f_3 & a f_3 & \mathbf{f_3} & 0 \\
\end{array}
\right)
$$ 
Finding such set of column is achieved by computing row-reduced echelon form in \cref{line:block-reshape-rref} of matrix $A'$ equal to matrix $A$ with reversed columns: 
$$
A' = F'B,~F' = \left(
\begin{array}{ccc}
 f_1 & h_1 & k_1 \\
 f_2 & h_2 & k_2 \\
 f_3 & h_3 & k_3 \\
\end{array}
\right),~B = 
\left(
\begin{array}{cccccccc}
 0 & \mathbf{1} & a & b & \mathbf{0} & c & e & \mathbf{0} \\
 ~ & \mathbf{~} & ~ & ~ & \mathbf{1} & d & g & \mathbf{0} \\
 ~ & \mathbf{~} & ~ & ~ & \mathbf{~} & ~ & ~ & \mathbf{1} \\
\end{array}
\right)
$$ 
Columns with leading ones correspond to the $m=3$ linear independent columns of $A$,
and the values in the other columns of $B$ give the coefficients to express linear dependent columns using the linear independent ones.
Sequence $L$ in~\cref{line:block-reshape-leading-ones} is $2,5,8$ in our example.
Before the column permutation in~\cref{line:block-reshape-permute}, matrix $R$ constructed by the procedure is lower-triangular 
and product $RA$ consist of reversed columns of $F$ separated by some zero columns
$$
R = \left(
\begin{array}{cccccccc}
 \mathbf{1} & ~ & ~ & \mathbf{~} & ~ & ~ & \mathbf{~} & ~ \\
 \mathbf{~} & 1 & ~ & \mathbf{~} & ~ & ~ & \mathbf{~} & ~ \\
 \mathbf{~} & 0 & 1 & \mathbf{~} & ~ & ~ & \mathbf{~} & ~ \\
 \mathbf{~} & g & d & \mathbf{1} & ~ & ~ & \mathbf{~} & ~ \\
 \mathbf{~} & 0 & 0 & \mathbf{~} & 1 & ~ & \mathbf{~} & ~ \\
 \mathbf{~} & 0 & 0 & \mathbf{~} & 0 & 1 & \mathbf{~} & ~ \\
 \mathbf{~} & e & c & \mathbf{~} & b & a & \mathbf{1} & ~ \\
 \mathbf{~} & 0 & 0 & \mathbf{~} & 0 & 0 & \mathbf{~} & 1 \\
\end{array}
\right),~ A R = 
\left(
\begin{array}{cccccccc}
 k_1 & 0 & 0 & h_1 & 0 & 0 & f_1 & 0 \\
 k_2 & 0 & 0 & h_2 & 0 & 0 & f_2 & 0 \\
 k_3 & 0 & 0 & h_3 & 0 & 0 & f_3 & 0 \\
\end{array}
\right)
$$
Column $j$ of $R$ can be interpreted as coefficients of a linear combinations of columns of $A$ needed to obtain column $j$ of $AR$.
Matrix $R$ is lower-triangular because each linear dependent column of $A$ depends on the linear independent columns to the right of it.
After the column permutation in~\cref{line:block-reshape-permute} matrix $R$ becomes: 
$$
R = \left(
\begin{array}{ccccc|ccc}
 ~ & ~ & ~ & ~ & ~ & \mathbf{1} & \mathbf{~} & \mathbf{~} \\
 1 & ~ & ~ & ~ & ~ & \mathbf{~} & \mathbf{~} & \mathbf{~} \\
 ~ & 1 & ~ & ~ & ~ & \mathbf{~} & \mathbf{~} & \mathbf{~} \\
 g & d & ~ & ~ & ~ & \mathbf{~} & \mathbf{1} &  \mathbf{~} \\
 ~ & ~ & 1 & ~ & ~ & \mathbf{~} & \mathbf{~} & \mathbf{~} \\
 ~ & ~ & ~ & 1 & ~ & \mathbf{~} & \mathbf{~} & \mathbf{~} \\
 e & c & b & a & ~ & \mathbf{~} & \mathbf{~} & \mathbf{1} \\
 ~ & ~ & ~ & ~ & 1 & \mathbf{~} & \mathbf{~} & \mathbf{~} \\
\end{array}
\right),
AR = \left(
\begin{array}{cccccccc}
 0 & 0 & 0 & 0 & 0 & k_1 & h_1 & f_1 \\
 0 & 0 & 0 & 0 & 0 & k_2 & h_2 & f_2 \\
 0 & 0 & 0 & 0 & 0 & k_3 & h_3 & f_3 \\
\end{array}
\right)
$$
and has the required shape.

Runtime of the \cref{proc:block-reshape} is
\begin{equation}
\label{eq:block-reshape-runtime}
O(nm^{\omega -1})
\end{equation}
dominated by reduced echelon form computation and is $O(nm^{\omega -1})$ using Algorithm 8 in~\cite{FastGaussianElimination}.

\subsubsection*{Complexity of basic linear algebra operations}

\paragraph{Matrix multiplication}
This is used to analyze many cases when matrix multiplications are required, for example, to derive \cref{eq:bulk-deallocation-complexity}
used in complexity analysis of \cref{alg:outcome-complete-stab-sim}.
Complexity of matrix multiplication of square $n \times n$ matrices over $\f_2$ is a well-studied problem.
In practical implementations, such as \cite{ABH}, the complexity of matrix multiplication is $O(n^\omega)$ for $\omega = \log_2 7$.
Note that multiplication of matrices over $\f_2$ has been optimized taking into account the architecture of modern CPUs and is very efficient in practice, 
see benchmark results in \cite{ABH}.
We reduce the multiplication of rectangular matrices to multiplication of square matrices by subdividing rectangular matrices into square blocks.
For example, complexity of multiplying $m\times n$ by $n \times n$ matrix is 
\begin{equation}
\label{eq:rectangular-matrix-multiplication-copmplexity}
O(n\max(m,n)\min(m,n)^{\omega-2}) = O(n(m+n)\min(m,n)^{\omega-2})
\end{equation}
Indeed, if $m > n$, we multiply $\lceil m/n \rceil$ $n\times n$ blocks in $O(n\cdot m \cdot n^{\omega-2})$.
If $n > m$, we multiply $1 \times \lceil n/m \rceil$ matrix of $n \times n$ blocks by $\lceil n/m \rceil \times \lceil n/m \rceil$ matrix of $n \times n$ blocks in $O(\lceil n/m \rceil^2 m^\omega) = O( n \cdot n \cdot m^{\omega-2})$. 
We use naive matrix multiplication for matrices of $n \times n$ blocks.

\paragraph{Reduced row echelon form}
This is a core subroutine to solve many linear algebra problems and its complexity used for 
complexity analysis of \cref{proc:block-reshape} in \cref{line:block-reshape-rref}, that is used in \cref{alg:general-from-choi}.
Additionally it is indirectly used in \cref{proc:support-subgroup}, \cref{proc:partition}.
Given a matrix $m\times n$ matrix $A$ of rank $r$, we can represent $A = BR$ for invertible square matrix $B$ and matrix $R$ in reduced-row echelon form.
Using Algorithm 8 in \cite{FastGaussianElimination} we can find $B$ and $R$ in $O(mnr^{\omega -2})$.
Matrix $B$ constructed by Algorithm 8 in \cite{FastGaussianElimination} has certain block structure up to column permutations, and so finding its inverse requires only $O(r^\omega)$ operations.
For this reason, we can find $B$, $B^{-1}$ and $R$ in 

\begin{equation}
\label{eq:rref-complexity}
O(mnr^{\omega -2})
\end{equation}

\paragraph{Multiplication by a matrix in reduced and split reduced echelon forms}

Multiplication of an $m \times n$ matrix $A$ by a $n \times (l_1 + l_2)$ matrix $B$ in $(l_1,l_2)$-split reduced echelon form~(\cref{def:split-echelon-form}) is used in 
\cref{alg:general-from-choi} in \cref{line:general-form-matrix-assignment}.
A matrix in split reduced echelon form is sparse, with first $l_1$ columns of weight at most $n - l_1 + 1$ and the last $l_2$ columns of weight one.
Multiplying $m \times n$ matrix by a column of weight $w$ has \runtime{} $O(mw)$, so runtime of $AB$ is 
\begin{equation}
\label{eq:split-reduced-echelon-mult-complexity}
O(m \cdot ( l_1 (n - l_1) + l_2))
\end{equation}
When $A^T$ is in reduced row-echelon form matrix $A$ is equal to $(\frac{I}{A'})$ up to row permutations, for some dense matrix $A'$.
In this case complexity of computing the product is even lower and equals to 
\begin{equation}
\label{eq:reduced-split-reduced-echelon-mult-complexity}
O((m-n) \cdot ( l_1 (n - l_1) + l_2))
\end{equation}
Above equation is also useful for complexity analysis of \cref{line:general-form-matrix-assignment} in \cref{alg:general-from-choi}.

\paragraph{Finding a basis of the kernel of a matrix}

This is used to justify runtime of \cref{proc:support-subgroup}.
Kernel of the matrix $A$ can be computed by first decomposing matrix into product $BR$ of a square invertible matrix $B$
and matrix $R$ in reduced row-echelon form.
Kernel basis of $R$ can be easily constructed as illustrated by an example below.
We show how basis $b_{(2)},b_{(4)},b_{(5)},b_{(7)}$ of the kernel of a $3\times7$ matrix $R$ in reduced row-echelon form. 
\begin{equation}
\label{fig:kernel-basis-from-reduced-row-echelon-form}
\begin{array}{rc|lcccccr|}
     &  & \boldsymbol{1} & {\ca a_{1,2}} & \boldsymbol{0} & {\ca a_{1,4}} & {\ca a_{1,5}} & \boldsymbol{0} & {\ca a_{1,7}}\\
    R & = & \boldsymbol{0} & {\cb 0} & \boldsymbol{1} & {\cb a_{2,4}} & {\cb a_{2,5}} & \boldsymbol{0} & {\cb a_{2,7}}\\
     &  & \boldsymbol{0} & {\cc 0} & \boldsymbol{0} & {\cc 0} & {\cc 0} & \boldsymbol{1} & {\cc a_{3,7}}\\
    \hline
    b_{(2)} & = & \boldsymbol{\ca a_{1,2}} & 1 & \boldsymbol{\cb 0} & 0 & 0 & \boldsymbol{\cc 0} & 0\\
    b_{(4)} & = & \boldsymbol{\ca a_{1,4}} & 0 & \boldsymbol{\cb a_{2,4}} & 1 & 0 & \boldsymbol{\cc 0} & 0\\
    b_{(5)} & = & \boldsymbol{\ca a_{1,5}} & 0 & \boldsymbol{\cb a_{2,5}} & 0 & 1 & \boldsymbol{\cc 0} & 0\\
    b_{(7)} & = & \boldsymbol{\ca a_{1,7}} & 0 & \boldsymbol{\cb a_{2,7}} & 0 & 0 & \boldsymbol{\cc a_{3,7}} & 1 \\
\end{array}
\end{equation}
Kernel bases of $A$ and $R$ are the same, therefore for $m\times n$ matrix $A$ of rank $r$, the complexity of finding kernel basis is 
\begin{equation} \label{eq:kernel-basis-complexity}
O(mnr^{\omega -2} + nr).
\end{equation}
This is because kernel basis matrix found in this way is has $n-r$ columns of weight one and only has an $(n-r) \times r$
dense sub-matrix.

\paragraph{Completing a matrix to a full rank matrix}

This is used to justify runtime of \cref{proc:partition} in \cref{line:partition-full-rank-completion}.
An $m \times n$ matrix $A$ can be completed to a full rank matrix by decomposing matrix into product $BR$ of a square invertible matrix $B$
and matrix $R$ in reduced row-echelon form.
Next matrix $R$ can be easily completed to a full-rank matrix:

\begin{equation}
\label{fig:complete-basis}
\begin{array}{rc|lcccccr|}
     &  & \boldsymbol{1} & { a_{1,2}} & \boldsymbol{0} & { a_{1,4}} & { a_{1,5}} & \boldsymbol{0} & { a_{1,7}}\\
    R & = & \boldsymbol{0} & { 0} & \boldsymbol{1} & { a_{2,4}} & { a_{2,5}} & \boldsymbol{0} & { a_{2,7}}\\
     &  & \boldsymbol{0} & { 0} & \boldsymbol{0} & { 0} & { 0} & \boldsymbol{1} & { a_{3,7}}\\
    \hline
    ~& ~ & 0 & 1 & 0 & 0 & 0 & 0 & 0 \\
\hat R & = & 0 & 0 & 0 & 1 & 0 & 0 & 0\\
    ~& ~ & 0 & 0 & 0 & 0 & 1 & 0 & 0\\
    ~& ~ & 0 & 0 & 0 & 0 & 0 & 0 & 1 \\
\end{array}
\end{equation}

The full-rank completion of $A$ is $\tilde A = (\frac{A}{\hat R})$.
The bottleneck for the completion procedure is computing row-echelon form in $O(mn\min(m,n)^{\omega -2})$
using estimate $\min(m,n)$ for rank of $A$ and \cref{eq:rref-complexity}.
Because completion matrix is $1$-sparse we conclude that full-rank completion complexity is 

\begin{equation}
\label{eq:full-rank-completion-complexity}
O(mn\min(m,n)^{\omega -2})
\end{equation}

\subsection*{Procedures for Pauli and Clifford unitary manipulation}

Here we justify \runtime{} and \bitruntime{} for procedures in \cref{tab:data-structure-requirements}. 
We omit justification of \runtime{} of the following procedures: 
\begin{itemize}[noitemsep]
    \item \texttt{init\_pauli}, \texttt{x\_bits}, \texttt{z\_bits}, \texttt{xz\_phase},
    \item \texttt{set\_x\_bits}, \texttt{set\_z\_bits}, \texttt{mult\_phase},
    \item \texttt{copy\_cliff}, \texttt{num\_qubits}, \texttt{image},
    \item \texttt{right\_mult\_swap}, \texttt{tensor\_prod}, \texttt{add\_qubits}, \texttt{remove\_qubits}.
\end{itemize}

The \runtime{} and \bitruntime{} of \texttt{disentange} and \texttt{preimage} has been justified in \cref{sec:pauli-and-clifford-procedures}.

\subsubsection*{Hermitian conjugate, commutator and product of Pauli unitaries }

The \runtime{} and \bitruntime{} of 
\texttt{is\_hermitian}, \texttt{commutator} and \texttt{prod\_pauli} in \cref{tab:data-structure-requirements} follows from three key observations below 
and \cref{tab:bit-string-procedures} of \runtime{} of bit string procedures.
First, the Hermitian conjugate of a Pauli unitary can be written as:
$$
P^\dagger = (i^{s(P)} Z^{x(P)} X^{x(P)})^\dagger = i^{s(P)}(-1)^{1+\mathrm{wt}(x(P) \wedge z(P))} Z^{x(P)} X^{x(P)}
$$
So checking if $P$ is hermitian requires checking if $\mathrm{wt}(x(P) \wedge z(P)) = 0 \text{ mod }2$.
Second, the commutator of two Pauli unitaries can be written as 
$$
\begin{array}{rcl}
  \comm{P,Q} & = & \comm{ Z^{x(P)} X^{x(P)}, Z^{x(Q)} X^{x(Q)} } \\
  & = & \comm{ Z^{x(P)} , Z^{x(Q)} X^{x(Q)} } +  \comm{ X^{x(P)}, Z^{x(Q)} X^{x(Q)} } \\
  & = &  \mathrm{wt}(x(P) \wedge z(Q)) + \mathrm{wt}(z(P) \wedge x(Q)) \text{ mod } 2 
 \end{array}
$$
Third, the product of two Pauli unitaries can be written as: 
$$
PQ = i^{s(P)} Z^{x(P)} X^{x(P)} i^{s(Q)} Z^{x(Q)} X^{x(Q)} = i^{s(P)+s(Q)}(-1)^{\mathrm{wt}(x(P) \wedge z(Q))} Z^{z(P)+z(Q)} X^{x(P)+x(Q)} 
$$

\subsubsection*{Initialization, product and inverse for Clifford unitaries}

In this subsection we discuss  \runtime{} and \bitruntime{} of initialization procedures, product and inverse
as they are closely related.

\paragraph{\texttt{init\_cliff}}
During the initialization we are given images and their phases. 
Those just get stored in $M(C)$ in $O(n^2)$.
Additionally we need to compute phases of pre-images, which can be achieved by computing the inverse of $C$ without using the pre-image phases information. 
The inverse can also be computed by  using matrix-matrix and matrix-vector multiplications~\cite{DehaeneMoor2003}. 
For this reason, the \runtime{} is $O(n^\omega)$ and \bitruntime{} is $O(n^2)$.

\paragraph{\texttt{init\_cliff\_css}} 
Initialization of a CCS Clifford is essentially copying entries of inputs $A$ and $B$ into appropriate locations in $M(C)$.
The exact relation between $A$, $B$ and $C$ is described by \cref{prop:css-cliiford-action}.
For this reason  \runtime{} and \bitruntime{}  are $O(n^2)$.

\paragraph{\texttt{inverse\_cliff}}
The inverse of Clifford unitary has  \runtime{} and \bitruntime{}  of $O(n^2)$ because it consists of two simple steps: 
rearranging of binary symplectic part of $M(C)$ as in \cref{eq:symplectic-matrix-inverse} and exchanging image phases for preimage phases.
Both of the steps can be done in $O(n^2)$ in both complexity models.

\paragraph{\texttt{prod\_cliff}}
The \runtime{} for the product of Clifford unitaries is $O(n^\omega)$. 
This is because binary symplectic part of $M(C_1 C_2)$ and phases of images can be found using matrix-matrix and matrix-vector multiplications~\cite{DehaeneMoor2003}. 
The fact that computing inverses is $O(n^2)$, implies that phases of preimages can also be computed in $O(n^\omega)$ be relating them 
to phases of images of $C^\dagger_2 C^\dagger_1$.

The \bitruntime{} for the product of Clifford unitaries is  $O(n^\omega)$. 
This is because \bitruntime{} of matrix-matrix multiplication is $O(n^2)$ and \bitruntime{} of matrix-vector multiplication is $O(n)$
when matrix columns or rows are stored as bit strings.

\subsubsection*{Left multiplication of Clifford unitaries \texttt{left\_mult\_$\ast$} }

\label{app:pauli-and-clifford-manipulation}
\begin{figure*}[ht]
\begin{procedure}[\texttt{left\_mult\_exp$^*$(}$C,P$\texttt{)}] \label{alg:compose-exp} 
\begin{algorithmic}[1]
\Blank
\Input Clifford unitary $C$, Pauli operator $P$.
\Output $C \leftarrow e^{i\pi/4 P} C$
\State Let $K$ be indices of qubits on which $P$ is supported, $k = |K|$
\State $\tilde P = C^\dagger P C$, $\tilde P_x = C^\dagger X^{x(P)} C$.
\Blank \Comment Update image phases in Clifford unitary representation $M(C)$
\For{$j$ in $[n]$}
\Blank \Comment Below we interpret integers mod $4$ as bit vectors of length $2$
\If{\label{line:exp-x-im-phase}$x(\tilde P)_j = 1$}  $M(C)_{j,[2n+1,2n+2]} \leftarrow s(C Z_k C^\dagger) + 1 + s(P) + 2 x(\tilde P_x)_j \text{ mod }4$ \EndIf
\If{\label{line:exp-z-im-phase}$z(\tilde P)_j = 1$} $M(C)_{n+j,[2n+1,2n+2]} \leftarrow s(C X_k C^\dagger) + 1 + s(P) + 2 z(\tilde P_x)_j \text{ mod }4$ \EndIf 
\EndFor
\Blank \Comment Update pre-images in Clifford unitary representation $M(C)$
\For{ $j$ in $[k]$ }
\Blank \Comment Below we interpret Pauli unitaries as bit vectors of length $2n+2$
\If{\label{line:exp-z-preimage}$\comm{Z_{K(j)},P} = 1$} $M(C)_{\ast,n + K(j)} \leftarrow i (C^\dagger Z_{K(j)} C) \tilde P$  \EndIf
\If{\label{line:exp-x-preimage}$\comm{X_{K(j)},P} = 1$} $M(C)_{\ast,K(j)} \leftarrow i (C^\dagger Z_{K(j)} C) \tilde P$  \EndIf
\EndFor
\end{algorithmic}
\end{procedure}
\end{figure*}

In this subsection we first focus on implementing \texttt{left\_mult\_exp$^*$(}$C,P$\texttt{)}~(\cref{alg:compose-exp}).
Left and right multiplication by any other Clifford unitary can be built using this procedure.
This is because any $n$-qubit Clifford unitary can be written as product of $O(n)$ exponents $e^{i\pi/4 P}$ \cite{CliffordExpDecomposition}.
For example, $CZ = \exp(i \pi \ket{11}\bra{11}) \simeq e^{-i \pi Z_1 /4 }e^{i \pi Z_1 /4 }e^{\pi Z_1 Z_2 /4 }$
and similarly generalized Controlled-Pauli $\Lambda(P,Q) \simeq e^{-i \pi P /4 }e^{i \pi Q /4 }e^{\pi P Q /4 }$.
Procedures
\begin{itemize}[noitemsep]
  \item \texttt{left\_mult\_swap}
  \item \texttt{left\_mult\_pauli},
  \item \texttt{left\_mult\_ctrl\_pauli},
\end{itemize}
can be implemented by 
using a constant number of calls to \texttt{left\_mult\_exp$^*$}.
Using identity $C e^{i\pi P/4 } C = e^{i\pi C P C^\dagger /4} C$ right multiplication 
my $e^{i\pi/4 P}$ can be implemented using the left multiplication by Pauli exponent.
In practice, more efficient implementations the reduction to \texttt{left\_mult\_exp$^*$(}$C,P$\texttt{)} are possible and preferable.
In the remainder of this section we discuss correctness of \cref{alg:compose-exp}, analyse its \runtime{}
and \bitruntime{}.
We conclude with the correctness proof of the procedure \texttt{disentangle}~(\cref{alg:disentanlgle}).

Next we briefly discuss correctness of the above \cref{alg:compose-exp}. 
First consider how images $C X_k c^\dagger$ and $C Z_k C^\dagger$ change after left-multiplication by $e^{i\pi P/4 }$. 
Using \cref{eq:pauli-exp} and $\comm{CPC^\dagger,Q} = \comm{P,C^\dagger Q,C}$ we have: 
\begin{equation}
\begin{array}{rcl}
e^{i \pi P/4} C Z_k C^\dagger e^{-i \pi P/4} & = & (iP)^{\comm{C^\dagger P C,Z_k}} C Z_k C^\dagger \\
e^{i \pi P/4} C X_k C^\dagger e^{-i \pi P/4} & = & (iP)^{\comm{C^\dagger P C,X_k}} C X_k C^\dagger \\
\end{array}
\end{equation}
In \cref{line:exp-x-im-phase}, condition $x(\tilde P)_j = 1$ is equivalent to $\comm{C^\dagger P C,Z_k} = 1$.
Similarly, in \cref{line:exp-z-im-phase}, condition $z(\tilde P)_j = 1$ is equivalent to $\comm{C^\dagger P C,X_k} = 1$.
We see that we update exactly those image phases that change after composing with $e^{i \pi P/4}$.
The equation for updating the image phases follows from \cref{prop:preimage-sign}
and observation that $x(\tilde P_x)_j = \comm{\tilde P_x,Z_j}$,  $z(\tilde P_x)_j = \comm{\tilde P_x,X_j}$.

Second we consider how preimages $C^\dagger X_k C$ and $C^\dagger Z_k C$ change after left-multiplication by $e^{i\pi P/4 }$. 
Using \cref{eq:pauli-exp} we have: 
\begin{equation}
\begin{array}{rcl}
C^\dagger e^{-i \pi P/4} Z_k e^{i \pi P/4} C & = & C^\dagger Z_k C (i C^\dagger P C)^{\comm{P,Z_k}}  \\
C^\dagger e^{-i \pi P/4} X_k e^{i \pi P/4} C & = & C^\dagger X_k C (i C^\dagger P C)^{\comm{P,X_k}}  \\
\end{array}
\end{equation}
In \cref{line:exp-z-preimage}, condition $x(P)_{K(j)} = 1$ is equivalent to $\comm{P,Z_{K(j)}} = 1$.
Similarly, in \cref{line:exp-x-preimage}, condition $z(P)_{K(j)} = 1$ is equivalent to $\comm{P,X_{K(j)}} = 1$.
We see that we update exactly those preimages that change after composing $C$ with $e^{i \pi P/4}$.

Next we analyze \runtime{} and \bitruntime{} of \cref{alg:compose-exp}
in terms of number of qubits $n$ on which $|C|$ is defined and weight $|P|$ of $P$.
We show that the \runtime{}  is $O(n\cdot|P|)$. 
The \runtime{} of computing preimage of $P$ and $X^x(P)$ is $O(n\cdot|P|)$.
The \runtime{} of updating phases of images is $O(|C|)$.
The \runtime{} of updating preimages is $O(n\cdot|P|)$.

Next we show that \bitruntime{} of \cref{alg:compose-exp} is $O(|P|)$.
Recall that \bitruntime{} of a product of two Pauli unitaries is $O(1)$. For this reason,
the \bitruntime{} of computing preimage of $P$ and $X^x(P)$ is $O(|P|)$.
Similarly, the \bitruntime{} of updating preimages is $O(|P|)$.
Next lest us show that \bitruntime{} of updating phases of images is $O(1)$.
Instead of updating image phases in a for loop we can use bit-wise operations as follows: 
\begin{algorithmic}[1]
\State $s_{\pm 1}, s_{(i)} = s(P) + 1 \text{ mod 4 }$  
\State $\Delta_{\pm 1} \leftarrow (x(\tilde P_x) \oplus z(\tilde P_x)) + s_{(i)}\cdot M(C)_{[2n],2n+2} + (s_{\pm 1})^{2n}$  \Comment{ $+$ is bit-wise XOR}
\State $M(C)_{[2n],2n+2} \leftarrow M(C)_{[2n],2n+2} + s_{(i)} \cdot (x(\tilde P) \oplus z(\tilde P))$ 
\State $M(C)_{[2n],2n+1} \leftarrow M(C)_{[2n],2n+1} +\Delta_{\pm 1} \wedge (x(\tilde P) \oplus z(\tilde P))$ \Comment{ $\wedge$ is bit-wise AND}
\end{algorithmic}
The correctness of above lines follows from separately considering cases $s_{(i)}= 0$  and $s_{(i)}= 1$
and the following expression for adding two two-bit integers modulo four: 
$$
 (a_1\cdot 2^1 + a_0) + (b_1 \cdot 2^1 + b_0) \text{ mod } 4 = ( ( a_1 + b_1 + a_0 b_0 ) \text{ mod }2)\cdot 2^1 + (a_0 + b_0\text{ mod }2 ).
$$

\subsubsection*{Correctness of \texttt{disentangle} }


\begin{proposition}[\texttt{disentangle} correctness]
\label{prop:disentangle}
Procedure \texttt{disentangle}~(\cref{alg:disentanlgle}) is correct.
\end{proposition}
\begin{proof}
To avoid confusion arising from the Clifford unitary $C$ being modified throughout the algorithm, we define the static Clifford unitary $C_0$ as the initial Clifford unitary that is fed into the algorithm, and other static Clifford unitaries throughout.
Note that $Z_jC_0\ket{0^n} = C_0\ket{0^n}$ implies $x(C_0^\dagger Z_j C_0)=0$.
At the end of step 2, $C$ is replaced by $C_1$, where $C_1 = \Lambda(C_0 X_{j'}C^\dagger,  C_0 Z_{j'} C_0^\dagger ~ Z_j) ~ C_0 ~ \mathrm{SWAP}_{j,j'}$.
To see that $C_1|0^n\rangle=C_0|0^n\rangle$, note that $C_0 Z_{j'} C_0^\dagger Z_j$ stabilizes $C_0 \mathrm{SWAP}_{j,j'}|0^n\rangle$ (using $x(C_0^\dagger Z_j C_0)=0$ and $\mathrm{SWAP}_{j,j'}|0^n\rangle = |0^n\rangle$), and then apply \cref{eq:controlled-pauli-stabilized}.
Moreover, we see using \cref{eq:controlled-pauli-image} that $C_1 Z_j C_1^\dagger$ is now $C_0~ \mathrm{SWAP}_{j,j'} Z_j \mathrm{SWAP}_{j,j'} C_0^\dagger \cdot C_0 Z_{j'} C_0^\dagger ~ Z_j = Z_j$ since $C_0~ \mathrm{SWAP}_{j,j'} Z_j \mathrm{SWAP}_{j,j'} ~C_0^\dagger$ commutes with $C_0 Z_{j'} C_0^\dagger ~Z_j$ (because $x(C_0^\dagger Z_j C_0)=0$) but anticommutes with $C_0X_{j'}C_0^\dagger$.
Note that $C_1 Z_j C_1^\dagger = Z_j$ implies the $j$th bit of $x(C_1^\dagger X_j C_1)$ is one (since $C_1^\dagger X_j C_1$ anticommutes with $C_1^\dagger Z_j C_1$).

In the second part of the algorithm, there are two cases.
The first case is when the $j$th bit of $z(C_1^\dagger X_j C_1)$ is zero, triggering the if condition and $C$ is then set to $C_2 = \Lambda(C_1 Z_j C_1^\dagger,C_1 X_j C_1^\dagger ~ X_j) ~ C_1$.
To see that $C_2|0^n\rangle = C_1|0^n\rangle$ note that $C_1 Z_{j} C_1^\dagger$ stabilizes $C_1|0^n\rangle$, and make use of \cref{eq:controlled-pauli-stabilized}.
To see that the image of $Z_j$ is unchanged, i.e., that $C_2 Z_j C_2^\dagger = Z_j$, apply \cref{eq:controlled-pauli-image} making use of the fact that the $j$th bit of $x(C_1^\dagger X_j C_1)$ is one.
Moreover, by \cref{eq:controlled-pauli-image}, image of $X_j$ becomes $C_2 X_j C_2^\dagger = C_1 X_j C_1^\dagger ~ C_1 X_j C_1^\dagger X_j = X_j$ since $C_1 X_j C_1^\dagger$ commutes with $C_1 X_j C_1^\dagger X_j$ (by the if condition) but anticommutes with $C_1 Z_{j}C_1^\dagger$.

The second case is when the $j$th bit of $z(C_1^\dagger X_j C_1)$ is one, triggering the else condition and $C$ is set to $C_3 = \Lambda(C_1 Z_j C_1^\dagger, X_j ~ C_1 ~iZ_j X_j ~ C_1^\dagger)~ e^{i \frac{\pi}{4} C_1 Z_j C_1^\dagger} ~ C_1$.
To see that $C_3|0^n\rangle = C_1|0^n\rangle$, first note that since $C_1 Z_j C_1^\dagger$ has $C_1\ket{0^n}$ as an eigenstate, so too does $e^{i \frac{\pi}{4} C_1 Z_j C_1^\dagger}$, such that $e^{i \frac{\pi}{4} C_1 Z_j C_1^\dagger} C_1\ket{0^n} =e^{i \frac{\pi}{4}} C_1\ket{0^n}$.
This implies that the state $e^{i \frac{\pi}{4} C_1 Z_j C_1^\dagger}C_1\ket{0^n}$ is stabilized by $C_1 Z_j C_1^\dagger$, which (using \cref{eq:controlled-pauli-stabilized}) tells us that $C_3 \ket{0^n} = e^{i\frac{\pi}{4}}C_1\ket{0^n}$, and thus $C_3|0^n\rangle = C_1|0^n\rangle$ (up to an unimportant phase).
To see that the image of $Z_j$ is unchanged, i.e., that $C_3 Z_j C_3^\dagger = Z_j$, one can implement a direct calculation making use of \cref{eq:exponentiated-pauli-image}, \cref{eq:controlled-pauli-image}, and that $C_1 Z_j C_1^\dagger = Z_j$.
To see that the image of $X_j$ becomes $C_3 X_j C_3^\dagger = X_j$, one can implement a direct calculation making use of \cref{eq:exponentiated-pauli-image}, \cref{eq:controlled-pauli-image}, and that $C_1 Z_j C_1^\dagger = Z_j$ and that the $j$th bit of $x(C_1^\dagger X_j C_1)$ and the $j$th bit of $z(C_1^\dagger X_j C_1)$ are both one.
\end{proof}


\section{Bulk deallocation of qubits in stabilizer simulation algorithms}
\label{app:bulk-deallocation-of-qubits}

We discuss how to reduce the cost of deallocation in outcome-complete stabilizer simulation,
a similar idea applies to outcome-specific stabilizer simulation.
\cref{tab:outcome-complete-simulation} shows that qubit deallocation is an expensive operation.
Instead of deallocating qubits every time the deallocation is requested, we can delay deallocation to the end of the simulation algorithm.
Next we describe an efficient approach to a bulk deallocation of qubits.

To deallocate many qubits in a stabilizer circuit that requires $n_\mathrm{max}$ qubits, has $\ko$ output qubits and $n_r$ random outcomes we use bipartite normal form 
for a family of stabilizer states described in \cref{sec:partition}. 
We apply the bipartite normal form~\cref{proc:partition} to the family of states $D\ket{Ar}$ describing the result of outcome complete simulation 
for bipartition of $n_\mathrm{max}$ qubits into $\ko$ and $n_\mathrm{max}-\ko$ qubits.
Because the output qubits should be disentangled by the end of the simulation, the number of Bell pairs across the partition is zero.
The bipartite normal form procedure finds Clifford unitaries $D_1$, $D_2$ on $\ko$ and $n_\mathrm{max}-\ko$ qubits
such that $D U_F\ket{a} \simeq (D_1 \otimes D_2) \ket{a}$ for all computational basis states $\ket{a}$, 
and so the family $D\ket{Ar}$ can be expressed as $(D_1 \otimes D_2)\ket{F^{-1} Ar}$.
Last $n_\mathrm{max}-\ko$ rows of $F^{-1}A$ must be zero and $D_2$ must map Pauli $Z$ to products of Pauli $Z$ operators,
because the qubits being deallocated must all be in zero state.
The result of outcome complete simulation is the family of states $D_1| (F^{-1} A)_{[\ko]} r\rangle $.

We conclude with runtime analysis of the bulk deallocation.
Complexity of computing basis change matrix $F$, its inverse and bipartite normal form is $n^\omega_\mathrm{max}$
because the number of Bell pairs across partition is zero.
Computing product $F^{-1} A$ is the complexity of multiplying $n_\mathrm{max} \times n_\mathrm{max}$ matrix by $n_\mathrm{max} \times n_r$ matrix 
over $\f_2$ and is $O(n_\mathrm{max} (n_\mathrm{max}+n_r) \min(n_\mathrm{max},n_r)^{\omega -2})$ according to \cref{eq:rectangular-matrix-multiplication-copmplexity}.
The overall complexity of the algorithm is
\begin{equation} \label{eq:bulk-deallocation-complexity}
    O(n_\mathrm{max}^{\omega-1}(n_r+n_\mathrm{max})).
\end{equation}

\newpage

\section{Pauli propagation} 

\begin{algorithm}[\texttt{Pauli propagation}]
\label{alg:pauli-propagation} 
\begin{algorithmic}[1]
\Blank
\Input 
\begin{itemize}[noitemsep,topsep=0pt]
\item A stabilizer circuit $\mathcal{C}$ with $\nM$ outcomes and $\no$ output qubits
\end{itemize}
\Output 
\begin{itemize}[noitemsep,topsep=0pt]
    \item A stabilizer circuit $\mathcal{C}'$ with $\nM$ outcomes  and $\no$ output qubits that has no conditional Pauli unitaries
    \item An $\nM \times \nM$ matrix $M$ and $\nM$-dimensional vector $v_0$
    \item An $\no \times \nM$ condition matrices $A_x$, $A_z$, $\no$ dimensional vectors $v_x,v_z$
\end{itemize}
such that $\mathcal{C}$ upon outcome vector $v$ followed by $X^{A_x v + v_x}Z^{A_z v + v_z}$ enacts the same ( up to a global phase ) linear map 
as $\mathcal{C}'$  upon outcome vector $v' = Mv + v_0$. 
Circuit $\mathcal{C}'$ is the circuit $\mathcal{C}$ with conditional Pauli unitaries removed. 
Matrix $M$ is a unit lower-triangular matrix and is the identity matrix when restricted to rows and columns corresponding 
to the random bit allocation operations.
\Blank
\State $\mathcal{C}_\text{rem} \leftarrow \mathcal{C}$,  
\State $\mathcal{C}' \leftarrow$ empty circuit,  $n'_O \leftarrow 0$ 
\State Let $A_x,A_z$ be $n \times n_O$ matrix, $v_x,v_z$ are $n$-dimensional zero vectors
\State Let $M$ be $n'_O \times n'_O$ matrix, and $v_0$ be $n'_O$ dimensional zero vector

\While{$\mathcal{C}_\text{rem}$ is not empty}
\State Let $g$ be the first operation in $\mathcal{C}_\text{rem}$, remove $g$ from $\mathcal{C}_\text{rem}$
\Blank  \hrulefill \Comment{allocation}
\If{$g$ allocates qubit $j$,\label{line:pauli-propagate-allocate}}
\State Append zero row to $A_x,A_z$ and zero entry to $v_x,v_z$ 
\State Append $g$ to $\mathcal{C'}$
\ElsIf{$g$ deallocates qubit $j$ \label{line:pauli-propagate-deallocate}} \Comment Assumes that the qubit is in zero
\State Remove row $j$ from $A_x,A_z$ and entry $j$ from $v_x,v_z$ 
\State Append $g$ to $\mathcal{C'}$
\ElsIf{$g$ allocates a random bit\label{line:pauli-propagate-random}}
\State Append zero column to $A_x,A_z,M$, add zero row to $M$, $n'_O \leftarrow n'_O + 1$, $M_{n'_O,n'_O} \leftarrow 1$
\State Append $g$ to $\mathcal{C'}$
\Blank \hrulefill \Comment{unitaries}
\ElsIf{$g$ is a unitary $U$\label{line:pauli-propagate-unitary}} 
\State Find $A'_x,A'_z,v'_x, v'_z$ from $ X^{A'_x v + v'_x } Z^{A'_z v + v'_z } \simeq U \left( X^{A_x v + v_x } Z^{A_z v + v_z } \right) U^\dagger $ for all $v$
\State Replace $A_x,A_z,v_x,v_z$ by $A'_x,A'_z,v'_x,v'_z$
\State Append $g$ to $\mathcal{C'}$
\ElsIf{$g$ applies a Pauli unitary $P$ if $\langle c\rangle = c_0$, \\$\quad\quad\quad\quad\quad$ where $\langle c\rangle$ is the parity of outcomes indicated by $c\in \mathbb{F}_2^{n_M}$,\label{line:pauli-propagate-conditional}}
\State $A_x \leftarrow A_x + x(P) c^T$,  $A_z \leftarrow A_z + z(P) c^T $ 
\State $v_x \leftarrow v_x + (c_0+1)\cdot x(P)$, $v_z \leftarrow v_z + (c_0+1) \cdot  z(P)$
\Blank \hrulefill\Comment{measurements}
\ElsIf{\label{line:pauli-propagate-measure}$g$ measures Pauli $P$ }
  \State Append row $A_x^T z(P) + A_z^T x(P)$ to $M$
  \State Append zero column to $M$, $n'_O \leftarrow n'_O + 1$,  $M_{n'_O,n'_O} \leftarrow 1$
  \State Append $\ip{x(P),v_z} + \ip{z(P),v_x}$ to $v_0$, append $g$ to $\mathcal{C'}$
\EndIf
\EndWhile
\State \Return $\mathcal{C'}$, $A_x,A_z,v_x,v_z$, $M,v_0$
\end{algorithmic}
\end{algorithm}

\section{Common symplectic basis of two stabilizer groups}
\label{sec:sim-diag-of-groups}

We show that any two stabilizer groups $S$ and $M$ have a common symplectic basis~(\cref{def:common-symplectic-basis})
by providing an efficient \cref{alg:common-symplectic-basis}.
Correctness of the algorithm is established by three key properties.
First, the number of independent generators of $S / (S \cap M^\perp)$
and $M / (M \cap S^\perp)$ computed in \cref{line:common-basis-m-cosets,line:common-basis-s-cosets} are the same. 
Second, matrix $A$ constructed in \cref{line:common-basis-diagonalize} is indeed invertible.
Third, $M$ and $S$ are isomorphic to the direct sum of $S \cap M$
and some quotient groups as follows: 
$$
\begin{array}{ccc}
  M & \simeq & \left( M / (M \cap S^\perp) \right) \oplus \left( (M \cap S^\perp) / (S \cap M) \right) \oplus S \cap M,  \\
  S & \simeq & \left( S / (S \cap M^\perp) \right) \oplus \left( (S \cap M^\perp) / (S \cap M) \right) \oplus S \cap M,
\end{array}
$$
so that $S = \ip{ \Zdelta \cup \Zs \cup \Zcap}$ and $M = \ip{ \Xdelta \cup \Zm \cup \Zcap}$.
\begin{figure*}
\begin{algorithm}[\texttt{Common symplectic basis}]
\label{alg:common-symplectic-basis}
\begin{algorithmic}[1]
\Blank
\Input $n$-qubit stabilizer groups $S$, $M$
\Output Common symplectic basis $\mathcal{B} = \ip{ \Zdelta,\Xdelta,\Zcap,\Xcap,\Zs,\Xs,\Zm,\Xm,\Z,\X }$ of $S$,$M$, that is
$S = \ip{ \Zdelta \cup \Zcap \cup \Zs }$, $M = \ip{ \Xdelta \cup \Zcap \cup \Zm }$ 
\State Let $\bar g_1,\ldots,\bar g_{k_\Delta}$ be independent generators of $S / (S \cap M^\perp)$ \label{line:common-basis-s-cosets}
\State Let $\bar h_1,\ldots,\bar h_{k_\Delta}$ be independent generators of $M / (M \cap S^\perp)$ \label{line:common-basis-m-cosets}
\State Let $g_j,h_j$ be representatives of cosets $\bar g_j, \bar h_j$ for $j \in [k_\Delta]$ 
\State Let $A_{i,j} = \comm{h_i, g_j}$, let $\Zdelta_j = g_j$, $\Xdelta_j = \prod_{i \in [k_\Delta]} h_i^{A^{-1}_{j,i}}$ \label{line:common-basis-diagonalize}
\State Let $\Zcap$ be independent generators of $S \cap M$ 
\State Let $\bar e_1, \ldots, \bar e_{k_S}$ be independent generators of $(S \cap M^\perp) / (S \cap M)$
\State Let $\Zs_j$ be a representative of coset $\bar e_j$ for $j \in k_S$   
\State Let $\bar f_1, \ldots, \bar f_{k_M}$ be independent generators of $(M \cap S^\perp) / (S \cap M)$
\State Let $\Zm_j$ be a representative of coset $\bar f_j$ for $j \in k_M$ 
\State Find $\Xcap,\Xs,\Xm,\Z,\X$ by completing the partially specified symplectic basis 
$\Zdelta,\Xdelta,\Zcap,\Zs,\Zm$ to a full symplectic basis. \label{line:common-basis-completion}
\State \Return $\ip{ \Zdelta,\Xdelta,\Zcap,\Xcap,\Zs,\Xs,\Zm,\Xm,\Z,\X }$
\end{algorithmic}
\end{algorithm}
\end{figure*}

The first and the second properties follow from the following proposition:
\begin{proposition}[Group's induced symplectic space] \label{prop:coset-structure}
Let $M$ and $S$ be two stabilizer groups,  
then $M / (M \cap S^\perp)$ and $S / (S \cap M^\perp)$ are 
commutative groups of the same size. 
For $\bar g$ from $M / (M \cap S^\perp)$ and $\bar h$ from $S / (S \cap M^\perp)$
the map $\comm{\bar g,\bar h}$ is well-defined as 
$$
  \comm{\bar g,\bar  h} = \comm{g,h} \text{ for any representatives } g,h \text{ of cosets } \bar g, \bar h. 
$$
The map $\comm{\bar g,\bar h}$ is non-degenerate, that is: 
\begin{align*}
    \text{if } \bar h \in S / (S \cap M^\perp) \text{ is such that } \forall \bar g \in M / (M \cap S^\perp) : \comm{\bar g,\bar h} = 0 & \text{, then } \bar h = M \cap S^\perp \\
    \text{if } \bar g \in M / (M \cap S^\perp) \text{ is such that } \forall \bar h \in S / (S \cap M^\perp) : \comm{\bar g,\bar h} = 0 & \text{, then } \bar g = S \cap M^\perp.
\end{align*}
\end{proposition}

Before proceeding with the proof, let us clarify how above proposition implies that matrix $A$ constructed in \cref{line:common-basis-diagonalize} is indeed invertible.
If square $k_\Delta \times k_\Delta$ matrix $A$ were not invertible, then there would exist 
$\alpha_1 ,\ldots, \alpha_{k_\Delta}$ from $\f_2$ that are not all zero but yield a trivial row sum
$$
    \sum_{ i \in [k_\Delta]}\alpha_i A_i = \mathbf{0}.
$$
In other words, following definition of $A$ in \cref{line:common-basis-diagonalize}, for all $j \in k_\Delta$ we have 
$$
\sum_{ i \in [k_\Delta]}\alpha_i \comm{h_i, \Zdelta_j} = \comm{ \sum_{ i \in [k_\Delta]}\alpha_i h_i, \Zdelta_j} = 0.
$$
Then for $h = \sum_{ i \in [k_\Delta]}\alpha_i h_i \notin M \cap S^\perp$ and for all $j \in k_\Delta$ we have $\comm{h,\Zdelta_j}$ = 0
which is a contradiction to $\comm{\bar g,\bar  h}$ being non-degenerate.

\begin{proof}[Proof of \cref{prop:coset-structure}]
Let us first show that for $\bar g$ from $M / (M \cap S^\perp)$ and for $\bar h$ from $S / (S \cap M^\perp)$
the map $\comm{\bar g,\bar h}$ is well-defined. 
Let $g$ and $h$ be representatives of cosets $\bar g,\bar h$. Then any other representative of the same coset 
can be written as $gr$ and $hs$ for $r$ from $M \cap S^\perp$ and $s$ from $S \cap M^\perp$. 
We have: 
$$
\comm{gr,hs} = \comm{g,hs} + \comm{r,hs} = \comm{g,h} + \comm{g,s} + \comm{r,h} + \comm{r,s}
$$
Now $\comm{g,s} = 0$ because $g$ is from $M$ and $s$ from $S \cap M^\perp \subset M^\perp$. 
Analogously, $\comm{h,r} = 0$ because $h\in S$, $r\in M \cap S^\perp$, and 
$\comm{r,s} = 0$ because $r\in M \cap S^\perp \subset M$, $s\in S \cap M^\perp$. We see that $\comm{gr,hs} = \comm{g,h}$.

Next we show that $M / (M \cap S^\perp)$ and $S / (S \cap M^\perp)$ have the same size.
First note that  $M / (M \cap S^\perp)$ has zero generators if and only if $S / (S \cap M^\perp)$ 
has zero generators. Indeed, $M / (M \cap S^\perp)$ has zero generators if and only if $ M \subset S^\perp$; 
$ M \subset S^\perp$ if and only if $  S \subset M^\perp$;  $  S \subset M^\perp$ if and only if 
$S / (S \cap M^\perp)$ has zero generators.
 
Now we show that if $M / (M \cap S^\perp)$ has $k \ge 1$ generators, then  $S / (S \cap M^\perp)$  has at least $k$ generators. 
Suppose $\bar g_1,\bar g_2,\ldots,\bar g_k$ are some generators of $M / (M \cap S^\perp)$. 
Let $\bar h_1$ be an element of  $S / (S \cap M^\perp)$ that anti-commutes with $\bar g_1$.
Such an $\bar h_1$ must exist because $\bar g_1$ represents a non-trivial coset of  $M / (M \cap S^\perp)$,
so it is not equal to $(M \cap S^\perp)$.
Adjust generators $\bar g_2,\ldots,\bar g_k$ so they commute with $\bar h_1$, multiplying anti-commuting generators by $\bar g_1$.
Let $\bar h_2$ be an element of  $S / (S \cap M^\perp)$ that anti-commutes with $\bar g_2$. 
It must exist, because $\bar g_2$ represents a non-trivial coset of  $M / (M \cap S^\perp)$.
Coset $\bar h_2$ is not equal to $\bar h_1$, because $\bar h_1$ commutes with $\bar g_2$ and the commutator map $\comm{,}$
is well-defined on the cosets.
We can also adjust $\bar h_2$ to commute with $\bar g_1$, by multiplying it with by $\bar h_1$, if necessary.

Suppose that we have constructed independent $\bar h_1,\ldots,\bar h_l$ such that $\comm{\bar h_i,\bar g_j} = \delta_{i,j}$
for $i = 1,\ldots,l$ and $j = 1,\ldots,k$.
Let $\bar h_{l+1}$ be an element of  $S / (S \cap M^\perp)$  that anti-commutes with $\bar g_{l+1}$.
It must exist, because $\bar g_{l+1}$ represents a non-trivial coset of  $M / (M \cap S^\perp)$.
It must be independent of  $\bar h_1,\ldots,h_l$, because they all commute with $\bar g_{l+1}$ and 
so does the product of any of their subset. 
Finally, multiply $\bar h_{l+1}$ by some of $\bar h_1,\ldots,h_l$
so $\comm{h_i,g_j} = \delta_{i,j}$ for  $i = 1,\ldots,l+1$ and $j = 1,\ldots,l+1$.
Adjust $\bar g_{l+2},\ldots,\bar g_{k}$ so they all commute with $h_{k+1}$. 
By induction we see that $|S / (S \cap M^\perp)| \ge |M / (M \cap S^\perp)|$.
Using the same argument we see that the number of generators of $|M / (M \cap S^\perp)| \ge |S / (S \cap M^\perp)|$.
 
Finally we can show that $\comm{\bar h,\bar g}$ is non-degenerate. 
Let us fix $\bar h$ and assume that for any $\bar g$, $\comm{\bar h,\bar g} = 0$.
For this reason all $h$ from $\bar h$ commute with all $g \in M $ and so $\bar h =( M \cap S^\perp)$.
A similar argument applies to $\bar g$ for which $\comm{\bar h,\bar g} = 0$ for all possible $\bar h$.
\end{proof}

\section{Equations for common symplectic basis example}

\label{app:sm-equations}

Consider groups $S_d, M_d$  defined on $2d$ qubits~(\cref{fig:example-sm}). For convenience we index qubits with two indices $(i,j)$, $i \in [2]$
and $j \in d$, with $i$ corresponding to the row index and $j$ corresponding to a column index $j$
of a qubit in~\cref{fig:common-symplectic-basis-example-all}.
\begin{align}
\label{eq:example-sm}
    S_d & = \ip{ Z_{(i,j)}Z_{(i,j+1)} : i \in [2], j \in [d-1] } \\
    M_d & = \ip{ Z_{(1,j)}Z_{(1,j+1)} Z_{(2,j)}Z_{(2,j+1)}, X_{(1,j')}X_{(2,j')} : i \in [2], j \in [d-1], j' \in [d] } 
\end{align}
A common symplectic basis $\mathcal{B}_d$~(\cref{fig:common-symplectic-basis-example}) for $S_d,M_d$ is given by the following equations: 
\begin{equation}
    \label{eq:common-symplectic-basis-example}
    \arraycolsep=1pt
\begin{array}{lclclcl}
    \mathcal{X}^{\Delta}_{j} & = &  X_{(1,j)}X_{(2,j)}, j \in [d-1],& ~~~~ & \mathcal{Z}^{\Delta}_{j} & = & Z_{(1,j)}Z_{(1,d)}, j \in [d-1]  \\
    \mathcal{X}^{\cap}_{j} & = & \prod_{i \in [j] } X_{(2,i)}, j \in [d-1],& ~~~~ & \mathcal{Z}^{\cap}_{j} & = & Z_{(1,j)}Z_{(1,j+1)} Z_{(2,j)}Z_{(2,j+1)}, j \in [d-1]  \\
    \mathcal{X}^{S} & =  & \varnothing & ~~~~ & \mathcal{Z}^{S} & = &  \varnothing \\
    \mathcal{X}^{M} & = & \{ Z_{1,d} \} & ~~~~ & \mathcal{Z}^{M} & = &  \{ \prod_{j \in d} X_{(1,j)}X_{(2,j)} \} \\
    \mathcal{X} & = & \{ \prod_{j \in d} X_{(2,j)} \} & ~~~~ & \mathcal{Z} & = &  \{  Z_{1,d}Z_{2,d} \} \\
\end{array}
\end{equation}
We consider Clifford unitaries $C_\inn^{(d)}$ such that $\comm{2d,2,C_\inn^{(d)}}$ has stabilizer group $S_d$ are given by the following equations,
via corresponding symplectic basis $\mathcal{X}^{\inn},\mathcal{Z}^{\inn}$~\cref{fig:example-c}:
\begin{equation}
    \label{eq:example-cin}
    \arraycolsep=1pt
\begin{array}{lclclcl}
    \mathcal{X}^{\inn}_{j} & = & \prod_{i \in [j] } X_{(1,i)}, j \in [d-1],& ~~~~ & \mathcal{Z}^{\inn}_{j} & = &   Z_{(1,j)}Z_{(1,j+1)}, j \in [d-1]  \\
    \mathcal{X}^{\inn}_{(d-1) + j} & = & \prod_{i \in [j] } X_{(2,i)}, j \in [d-1],& ~~~~ & \mathcal{Z}^{\inn}_{(d-1)+j} & = &   Z_{(2,j)}Z_{(2,j+1)}, j \in [d-1]  \\
    \mathcal{X}^{\inn}_{2d-1} & = & \prod_{i \in [d] } X_{(1,i)},& ~~~~ & \mathcal{Z}^{\inn}_{2d-1} & = & Z_{(1,d)}  \\
    \mathcal{X}^{\inn}_{2d} & = & \prod_{i \in [d] } X_{(2,i)},& ~~~~ & \mathcal{Z}^{\inn}_{2d} & = & Z_{(2,d)}  \\
\end{array}
\end{equation}
Clifford unitaries $C_\out^{(d)}$ such that $\comm{2d,1,C_\out^{(d)}}$ has stabilizer group $M_d$ are given by the following equations
via corresponding symplectic basis $\mathcal{X}^{\out},\mathcal{Z}^{\out}$~\cref{fig:example-c}:
\begin{equation}
    \label{eq:example-cout}
    \arraycolsep=1pt
\begin{array}{lclclcl}
    \mathcal{X}^{\out}_{j} & = & Z_{(1,j')}, j' \in [d],& ~~~~ & \mathcal{Z}^{\out}_{j} & = & X_{(1,j)}X_{(2,j)}, j' \in [d]  \\
    \mathcal{X}^{\out}_{d + j} & = & \prod_{i \in [j] } X_{(2,i)}, j' \in [d],& ~~~~ & \mathcal{Z}^{\out}_{d+j} & = & Z_{(1,j)}Z_{(1,j+1)} Z_{(2,j)}Z_{(2,j+1)}, j \in [d-1]  \\
     &   & & ~~~~ & \mathcal{Z}^{\out}_{2d} & = & Z_{(1,d)}Z_{(2,d)}  \\
\end{array}
\end{equation}

\end{document}